\newif\ifarxiv
\arxivtrue 
\newif\ifmsc 
\ifarxiv \mscfalse \else \msctrue \fi
\newif\iftodo
\todofalse 

\ifmsc
\RequirePackage{amsmath} 
\RequirePackage{amsthm}
\RequirePackage{hyperref}
\RequirePackage[capitalize]{cleveref}
\documentclass{msc}
\usepackage[shortlabels]{enumitem}
\fi

\ifarxiv
\documentclass{Arxiv}
\fi

\usepackage[utf8]{inputenc}
\usepackage[T1]{fontenc}
\usepackage[english]{babel}

\usepackage{amssymb} 
\usepackage{graphicx}
\usepackage{stmaryrd}
\usepackage{mathrsfs}
\usepackage{bm}

\usepackage{tikz-cd}
\usepackage{tikz}
\usetikzlibrary{shapes}
\usepackage{ebproof}

\ifmsc
\newtheorem{example}[therm]{Example}
\Crefname{therm}{Theorem}{Theorems}
\Crefname{example}{Example}{Examples}
\renewenvironment*{theorem}{\begin{therm}}{\end{therm}}
\RequirePackage{xcolor} 
\definecolor{blue(pigment)}{rgb}{0.2, 0.2, 0.6}
\definecolor{darkgreen}{rgb}{0.0, 0.5, 0.0}
\definecolor{darkred}{HTML}{9F000F}   
\hypersetup
{
    colorlinks=true,
    linkcolor=darkred, 
    filecolor=magenta,            
    urlcolor=blue(pigment),
    citecolor=darkgreen
}
\fi

\usepackage{MSCS}
\usepackage{Thomas}
\usepackage{quiver}

\usepackage[numbers]{natbib}

\usepackage{caption} 
\usepackage{subcaption}

\SetLabelAlign{center}{\hss#1\hss}

\ifarxiv 
\usepackage{authblk}
\author[1]{Thomas Ehrhard}
\author[2]{Aymeric Walch}
\affil[1]{Université Paris Cité, CNRS, Inria, IRIF, F-75013, Paris, France}
\affil[2]{Université Paris Cité, CNRS, IRIF, F-75013, Paris, France}
\affil[ ]{\texttt{\{ehrhard,walch\}@irif.fr}}
\fi

\iftodo
\usepackage{todonotes}
\else 
\usepackage[disable]{todonotes}
\fi

\begin{document}

\title{Coherent Taylor expansion as a bimonad}

\ifmsc
\lefttitle{Coherent Taylor expansion as a bimonad}
\righttitle{}

\papertitle{Article}

\jnlPage{1}{00}
\jnlDoiYr{2023}
\doival{10.1017/xxxxx}

\begin{authgrp}
  \author{Anonymous authors}
\end{authgrp}

\fi 

\ifmsc 
\begin{abstract}
    We extend the recently introduced setting of coherent
  differentiation by taking into account not only differentiation,
  but also Taylor expansion in categories which are not necessarily
  (left) additive.
  The main idea consists in extending summability into an infinitary
  functor which intuitively maps any object to the object of its
  countable summable families.
  This functor is endowed with a canonical structure of a bimonad.
  In a linear logical categorical setting, Taylor expansion is then
  axiomatized as a distributive law between this summability functor
  and the resource comonad (aka.~exponential). This distributive law
  allows to extend the
  summability functor into a bimonad on the coKleisli category of the
  resource comonad: this extended functor computes the Taylor
  expansion of the (nonlinear) morphisms of the coKleisli category.
  We also show how this categorical axiomatization of Taylor expansion
  can be generalized to arbitrary cartesian categories, leading to a
  general theory of Taylor expansion formally similar to that of
  cartesian differential categories, although it does not require the
  underlying cartesian category to be left additive.
  We provide several examples of concrete categories which arise in
  denotational semantics and feature such analytic structures.

\end{abstract} 

\begin{keywords} Categorical semantics; Denotational semantics; Linear Logic; Differential lambda calculus; Taylor expansion;
  Coherent differentiation.
\end{keywords}
\fi

\maketitle

\ifarxiv 
\begin{abstract}
    We extend the recently introduced setting of coherent
  differentiation by taking into account not only differentiation,
  but also Taylor expansion in categories which are not necessarily
  (left) additive.
  The main idea consists in extending summability into an infinitary
  functor which intuitively maps any object to the object of its
  countable summable families.
  This functor is endowed with a canonical structure of a bimonad.
  In a linear logical categorical setting, Taylor expansion is then
  axiomatized as a distributive law between this summability functor
  and the resource comonad (aka.~exponential). This distributive law
  allows to extend the
  summability functor into a bimonad on the coKleisli category of the
  resource comonad: this extended functor computes the Taylor
  expansion of the (nonlinear) morphisms of the coKleisli category.
  We also show how this categorical axiomatization of Taylor expansion
  can be generalized to arbitrary cartesian categories, leading to a
  general theory of Taylor expansion formally similar to that of
  cartesian differential categories, although it does not require the
  underlying cartesian category to be left additive.
  We provide several examples of concrete categories which arise in
  denotational semantics and feature such analytic structures.

\end{abstract} 
\fi

\listoftodos

\paragraph*{Differentiation, Lambda-Calculus and Linear Logic.}
Linear Logic (LL) is a logical system that arose from
semantics~(\cite{Girard87}), following the observation that some
interesting models of the $\lambda$-calculus were actually the
coKleisli category $\kleisliExp$
of some category $\categoryLL$ of linear morphisms.
That is, a morphism from object $X$ to $Y$ can be seen as a morphism
in $\categoryLL(!X, Y)$, where $\categoryLL$ is a category of linear
morphisms and $\oc\_$ is a comonad on $\categoryLL$.
The syntactical counterpart of a morphism in $\categoryLL$ is a
proof/program that uses its input (data or hypothesis) exactly once,
and the syntactical version of $\oc\_$ features a construction (called
promotion) that makes a resource duplicable and discardable.

It turns out that in multiple models of LL, the morphisms
$f \in \categoryLL(!X, Y)$ are differentiable in some sense, strongly
suggesting that differentiation of proofs and programs should be
considered as a natural LL operation.
Remember  (see~\cite{Dieudonne69} for instance) that for any Banach spaces \(E\) and \(F\), a function
\(f:E\to F\) is differentiable
at \(x \in E \) if there is a
neighborhood \(U\) of \(0\) in \(E\) and a linear and continuous
function \(\phi:E\to F\) such that, for all \(u\in U\)
\begin{align}
  \label{eq:derivative-def}
  f(x+u)=f(x)+\phi(u)+o(\|u\|)\,.
\end{align}
When \(\phi\) exists, it is unique and is denoted as
\(\derive{f}\).
When \(\derive{f}\) exists for all \(x\in E\),
the function \(\derivenoarg{f} :E\to\mathcal L(E,F)\), where
\(\mathcal L(E,F)\) is the Banach space of linear and continuous
functions \(E\to F\), is called the \emph{differential} of \(f\).
This function can itself admit a differential and so on. When all
these iterated differentials exist one says that \(f\) is
\emph{smooth} and the \(n\)th derivative of \(f\) is a function
\(\derivennoarg f n :E\to\mathcal L_n(E,F)\) where
\(\mathcal L_n(E,F)\) is the space of continuous \(n\)-linear
symmetric functions \(E^n\to F\).
It can even happen that \(f\) is locally (or even globally) expressed
using its iterated derivatives by means of the \emph{Taylor Formula},
\begin{equation}
  f(x+u)=\sum_{n=0}^\infty\frac1{n!} \deriven{f}{n} (u,\dots,u) \, .
\end{equation}
When this holds locally at any point \(x\), \(f\) is said to be \emph{analytic}, see 
\cite{Whittlesey65}.

Based on categorical models of LL where objects are similar to vector spaces
and morphisms
$f \in \categoryLL(!X, Y)$ are analytic functions from $X$ to $Y$,
the differential
\(\lambda\)-calculus (\cite{EhrhardRegnier02}) and differential LL
provide a logical and syntactical account of differentiation.
If $\Gamma, x : A \vdash P : B$ and $\Gamma \vdash Q : A$, then one
can define in these systems, by purely syntactical means, a program
$\Gamma, x : A \vdash \frac{\partial P}{\partial x} \cdot Q : B$ whose
denotational semantics in such models is the derivative of the
interpretation of \(P\) wrt.~the variable \(x\), linearly applied to
the semantics of \(Q\).
This syntactic derivative can be seen as a version of $P$, where
exactly one call to the variable $x$ that occurs during the
computation of \(P\) is replaced with a call to $Q$: this explains why
\(x\) is still free in \(\frac{\partial P}{\partial x} \cdot Q\) in
general.
This provides a new approach of finite approximations of programs
 by a syntactical version of the Taylor Formula whose effect is to
translate \(\lambda\)-calculus application \((P)Q\) into the
differential \(\lambda\)-calculus in such a way that
\begin{equation} \label{eq:Taylor-syntax} %
  (\lambda x. P)Q \text{ reduces to }
  \sum_{n=0}^{\infty} \frac{1}{n!} \left(
    \frac{\partial^n P}{\partial x^n} (\underbrace{Q, \ldots,
      Q}_{\text{n times}}) \right) [0/x]\,.
\end{equation}
The term of rank $n$ in this formal sum corresponds to the part of the
computation that uses the input $Q$ exactly $n$ times.
Applying this transformation hereditarily to all the applications
occurring in a \(\lambda\)-term, it becomes possible to turn it into
an infinite sum of strongly normalizing \emph{resource
  terms}, see~\cite{Ehrhard08}.

Differentiation is deeply connected with \emph{addition}, as it can
already be seen in its definition~\cref{eq:derivative-def} and in the
syntactical Taylor expansion~\cref{eq:Taylor-syntax}.
As a result, the differential $\lambda$-calculus is always endowed
with an unrestricted operation of sum between terms.
Similarly, all categorical models of the differential
\(\lambda\)-calculus and of differential LL (see~\cite{Blute06,
  Blute09}) were using categories where the hom-sets 
  were commutative monoids, and all models of 
  the syntactical Taylor expansion (\cite{Manzonetto12}) require the hom-sets 
  to feature arbitrary countable sums.
The only available operational interpretation of such a sum being
erratic choice, these formalisms are inherently non-deterministic.
This is in sharp contrast with the uniformity property of the sum that
can be observed in the syntax. For example, if the term of
\cref{eq:Taylor-syntax} reduces to a variable, then only one term of
the sum is non-zero.
The position $n$ of this term gives the number of times the term $Q$
is evaluated during the weak head reduction strategy as shown
in~\cite{EhrhardRegnier02}.
Furthermore, many deterministic (or probabilistic) models of LL such
as those based on the relational model (like coherence spaces and
probabilistic coherent spaces) feature morphisms in
$\categoryLL(!X, Y)$ which are clearly analytic in some sense,
although the hom-sets do not feature a commutative monoid structure.

\paragraph*{Coherent differentiation.}
Recently, \cite{Ehrhard22a} observed that, in a setting
where all coefficients are non-negative, differentiation survives to
strong restrictions on the use of addition.
Consider for instance a function \([0,1]\to[0,1]\) which is smooth on
\([0,1)\) and all of whose iterated derivatives are everywhere \(\geq 0\)%
\footnote{This actually implies that \(f\) is analytic.}.
If \(x,u\in[0,1]\) are such that \(x+u\in[0,1]\) then
\begin{equation} \label{eq:taylor-bounded}
  f(x)+f'(x)u\leq f(x+u)\in[0,1] \, .
\end{equation}
This makes sense even if
\(f'(1)=\infty\), which can happen: take \(f(x)=1-\sqrt{1-x}\).
So if \(S\) is the set of all such pairs $\Spair{x}{u}$ that we call
\emph{summable}, we can consider the function
\[ \Dbin(f)
: \Spair{x}{u}\mapsto \Spair{f(x)}{\derive{f} \cdot u} \] as a map
\(S\to S\).
This basic observation is generalized in~\cite{Ehrhard23-cohdiff} to a wide
range of categorical models \(\cL\) of LL including coherence spaces,
probabilistic coherence spaces \emph{etc.}~where hom-sets have only a
\emph{partially defined} addition.
In these \emph{summable categories}, we obtain an endofunctor \(\S :
\cL\to\cL\) equipped with an additional structure which allows defining 
summability and (partial) finite sums in a very general way that turns
$\S$ into a monad.
Differentiation is then axiomatized as a distributive law
$\oc \S \naturalTrans \S \oc$ between this monad (similar to the
tangent bundle monad of a tangent category, see~\cite{Rosicky84}) and
the resource comonad \(\oc\_\) of the LL structure of the category%
\footnote{Which by the way needs not be a fully-fledged LL model.}
\(\cL\).
Indeed, this distributive law allows one to extend \(\Sfun\) to
$\kleisliExp$, the coKleisli category of \(\oc\_\), and this extension
\(\Dbin:\kleisliExp \to \kleisliExp\) inherits the monad structure of
$\S$.
A category equipped with such a differentiation is then called a
\emph{coherent differential category}.
It was also observed that the functor $\S$ is often \emph{representable},
following the terminology of tangent categories (see \cite{Cockett14}), 
meaning that 
$\S = \Dbimon \linarrow \_$ (the category \(\cL\) being most often
symmetric monoidal closed, with internal hom of \(X\) and \(Y\)
denoted as \(X\linarrow Y\)), where $\Dbimon = \Sone \with \Sone$
(\(\Sone\) is the unit of the tensor product of \(\cL\), and $\with$ is 
a cartesian product).
This object $\Dbimon$ can be endowed with the structure of a comonoid
from which the monad structure of $\S$ arise.
Differentiation then boils down to the existence of a coalgebra
structure $\Dbimon \arrow \oc \Dbimon$ on \(\Dbimon\) which is
compatible with its comonoid structure.

\paragraph*{Contributions of this work.}
The goal of the present article is to study the Taylor expansion in this
setting of partial sums.
We believe that this is a crucial step towards providing generic tools
to define a denotational semantics of the differential
$\lambda$-calculus and of the syntactical Taylor expansion in a much
more general setting than the current state of the art
of~\cite{Manzonetto12}.
The starting idea is that for the same reason that 
$\Dbin f(x) = \Spair{f(x)}{\derive{f}  u}$ is defined, it should be possible to 
define directly a second order approximant.
\[\D f \Striple{x}{u_1}{u_2} = \Striple{f(x)}{\derive{f} \cdot u_1}
  {\frac{1}{2} \deriven{f}{2} \cdot (u_1, u_1) + \derive{f} \cdot u_2
  } \, . \] 
  The term $u_1$ should be seen as a first order variation and
the term $u_2$ as a second order variation. So
$\D f \Striple{x}{u_1}{u_2}$ gives the components (sorted by order) of
the order 2 Taylor approximation of $f$ on the variation $u_1 +
u_2$. 
Similarly to \cref{eq:taylor-bounded}, this sum is lower
than $f(x+u_1+u_2)$, so it is well-defined.  This idea should work for
any order, and going to the limit, for an infinite amount of
coefficient, yielding an operator which provides all the terms of the
Taylor expansion.

The first step is to introduce the \emph{$\Sigma$-summabibility structure} $\S$, 
an infinitary counterpart of the summability
structure previously mentioned. This $\Sigma$-summability structure is 
infinitary in the sense that an ``element'' of \(\S(X)\) is a
\(\Nat\)-indexed family of elements of \(X\) whose infinitary sum is
well-defined,
the underlying notion of infinitary sum is the one of $\Sigma$-monoid, 
see~\cite{Haghverdi00}.
It turns out that 
$\S$ is not only a monad, but also a comonad. The monadic and comonadic 
structures interact well, turning $\S$ into a \emph{bimonad}
(\cref{sec:summability-structure}).
Surprisingly, the whole Taylor expansion operation is again a
distributive law $\oc \S \naturalTrans \S \oc$ following the exact same
properties as coherent differentiation.
This distributive law allows, as it is standard, to extend the functor
\(\Sfun\) into a functor \(\Tayfun:\Kloc\cL\to\Kloc\cL\) 
which implements Taylor expansion and
inherits the monadic structure of \(\Sfun\).
One more axiom is added, ensuring intuitively that the maps are
analytic in the sense that they coincide with their Taylor expansion.
This axiom was the missing piece to ensure that $\D$
inherits from $\S$ the structure of a bimonad.
We call a category equipped with such Taylor expansion an \emph{analytic
  category} (\cref{sec:Taylor}).

Similarly to coherent differentiation, 
in many concrete models of LL, the functor $\S$ is representable and 
is equal to
$\Dbimon \linarrow \_$ where this time
$\Dbimon = \Sone \with \Sone \with \cdots$ (\(\Nat\)-indexed cartesian
product).
A category is called \emph{representably $\Sigma$-additive} if this functor 
$\S$ is a $\Sigma$-summability structure.
This object $\Dbimon$ can be endowed with the structure of a bimonoid
that completely determines the bimonad structure of $\S$.
The analytic Taylor expansion ---~that is, the aforementioned distributive law
$\oc \S \naturalTrans \S \oc$~--- then boils down to a
coalgebra $\Dbimon \arrow \oc \Dbimon$ compatible with the bimonoid
structure (\cref{sec:elementary}). 
We call \emph{representably analytic categories} 
such instance of analytic categories.
This coalgebra structure always exist when the exponential 
(the \(\Oc\_\) comonad) is free, that is, when the category is 
Lafont, \cite{Lafont88}.
This implies the following result.

\begin{theorem}
  Any representably $\Sigma$-additive Lafont category 
  is a representably analytic category.
\end{theorem}

We provide many examples of representably analytic Lafont categories.
First, we discuss in \cref{sec:examples-taylor} that the relational model and the 
weighted relational model are Lafont representably
analytic categories.
We compute explicitly the action of the
\(\Tayfun\) functor on morphisms,
showing that $\D$ performs the expected Taylor expansion.
It means that our theory of analytic category
covers the most important models of differential LL and 
of syntactical Taylor expansion, so 
and should be a strict generalization of the theory of \cite{Manzonetto12}.

Then, we provide many examples which were not captured by this previous 
theory.
The first example is based on a notion of coherence space
introduced in~\cite{Lamarche95}, which seems deeply related to Scott
semantics.
The second one is based on nonuniform coherence spaces that were
introduced in~\cite{BucciarelliEhrhard01} and are known to have two
different exponentials, one of them being the free exponential.
In this setting, the free exponential provides an example of
representably analytic category, but we show that the non-free one is not
a representably analytic category, yet is a representably coherent
differential category in the sense of~\cite{Ehrhard23-cohdiff}.
We also mention the fact that the usual Girard's coherence spaces, with
their free exponential, are a representably analytic category, and last we
deal with the case of probabilistic coherence spaces (introduced
in~\cite{DanosEhrhard11}) whose only known exponential has been shown to
be free.
They provide yet another example of representably analytic category where
morphisms are analytic functions with non-negative real coefficients.

\paragraph*{Taylor expansion in cartesian closed categories.}

Because $\kleisliExp$ is a cartesian closed category, it can be
interesting to drift away from the structure of $\categoryLL$ by
only looking at the structure induced in $\kleisliExp$.
This is what happened with differentiation.
It was first categorically axiomatized in a typical LL setting with
additive categories, introducing a notion of \emph{differential
  categories} in~\cite{Blute06}.
Differentiation was then carried to the setting of cartesian
left-additive categories, introducing cartesian differential
categories in~\cite{Blute09}, leading to successful uses of
differentiation outside the realm of LL.
Unsurprisingly, the coKleisli categories of differential LL categories are
instances of cartesian differential categories, but the latter are
more general than the former and cover more examples of categories
where differentiation is available.
Similarly, left summability structures and \emph{cartesian coherent
  differential categories} were introduced in ~\cite{Walch23} to
axiomatize coherent differentiation directly in any cartesian
category.
They arise in particular as the coKleisli category of coherent
differential categories, and are at the same time a conservative
generalization of the cartesian differential categories.

Since analytic categories are very similar to coherent differential
categories, it is possible to introduce in a very similar way a notion
of \emph{cartesian analytic category}
(\cref{sec:cartesian-Taylor}).
We can then define the notion of a 
\emph{cartesian closed analytic category}.
This more direct axiomatization should provide the foundation for
working on the denotational semantics of syntactical Taylor expansion,
but also provide a categorical framework for Taylor expansion
independent of LL.

\begin{figure}
\begin{tikzpicture}[every text node part/.style={align=center}, scale=1]

  \node (DC) [ellipse, draw] at (0, 0) {Differential \\ Category [1]};
  \node (CDC) [ellipse, draw] at (5, 0) {Coherent \\ Differential \\ Category [4]};
  \node (CarDC) [ellipse, draw] at (0, -4) {\small Cartesian \\ \small Differential \\ \small Category [2]};
  \node (CarCDC) [ellipse, draw] at (5, -4) { \small Cartesian \\ \small Coherent 
  Differential \\ \small  Category [5]};
  
    \draw [->] (DC) -- (CDC) node[midway, above] {$\subseteq$[6]};
  \draw [->] (DC) -- (CarDC) node[pos=0.5, right]{coKleisli [2]};
  \draw[->] (CarDC) -- (CarCDC) node[midway, below] {$\subseteq$[5]};
  \draw[->] (CDC) -- (CarCDC) node[pos=0.45, left] {coKleisli [5]};

  \node (Taylor) [ellipse, draw] at (2.5, -6) {Model \\ of Taylor \\ 
  expansion [4]};
  \node (CohTaylor) [ellipse, draw, blue(pigment)] at (8, -1.5) 
  {Analytic \\ category};
  \node (CarTaylor) [ellipse, draw, blue(pigment)] at (8, -6) 
  {Cartesian Analytic \\ Category};

  \draw[->, blue(pigment)] (CohTaylor) -- (CarTaylor) node[pos=0.5, right, blue(pigment)] {coKleisli};
  \draw[->, dashed] (Taylor) -- (CarTaylor) node[midway, below] {\ref{fig:map-of-concepts-3}};

  \draw[->] (Taylor.west) to[out = -180, in = -90] node[pos=0.5, below] {} (CarDC.south);
  \draw[->, dashed] (CarTaylor) to[out = -180, in = -90] 
  node[pos=0.7, right]{\ref{fig:map-of-concepts-2}} (CarCDC.south);
  \draw[->, dashed] (CohTaylor.north) to[out = 90, in = 0]
  node[pos=0.4, above] {\ref{fig:map-of-concepts-1}} (CDC.east);

  \draw[dashed, darkred] (-2, -2.5) -- (10, -2.5);
  \node [darkred] at (-1.4, -2.25){LL};
  \node [darkred] at (-1.6, -2.75){$\lambda$-calculus};
  \end{tikzpicture}

  \caption{Map of concepts,
  [1] \cite{Blute06} 
  [2] \cite{Blute09}
  [3] \cite{Manzonetto12}
  [4] \cite{Ehrhard23-cohdiff}
  [5] \cite{Walch23}
  [6] \cite{Walch22-internship}
  }
  \label{fig:map-of-concepts}
  \labeltext{(1)}{fig:map-of-concepts-1}
  \labeltext{(2)}{fig:map-of-concepts-2}
  \labeltext{(3)}{fig:map-of-concepts-3}
  \end{figure} 
\paragraph{Map of concepts.} We summarize in \cref{fig:map-of-concepts} 
how the different concepts relate with each other.
  We expect the dashed arrows to hold, but the proofs are postponed to further work.

  The arrows \ref{fig:map-of-concepts-1} and \ref{fig:map-of-concepts-2}
  should be a consequence of the fact that the Taylor functor performs Taylor 
  expansion at all order, and in particular includes Taylor expansion at order one, 
  which yield precisely the definition of the derivative in 
  \cref{eq:derivative-def}.
  The reverse arrow does not hold, a coherent differential category is not always 
  an analytic category, we provide such example in \cref{sec:nucs}.
    
  We expect the arrow \ref{fig:map-of-concepts-3} to hold 
  because our notion of analytic categories cover both the relational model 
  and the weighted relational model, and the Taylor functor 
  performs the expected Taylor expansion, as discussed in \cref{sec:examples-taylor}.
  These two models are the two principal models of syntactical Taylor
  expansion following the definition of \cite{Manzonetto12}, so our new theory 
  captures the main examples of the previous theory.  

\paragraph*{Mates and distributive law.} 
One contribution of this article is to exhibit the crucial role 
played by the mate construction (\cref{sec:mate}) in the representable case, both in 
the setting of \cite{Ehrhard23-cohdiff} and in the setting of 
this article. 
In the representable case, the commutative bimonoid $\Dbimon$ 
induces a bimonad $\_ \tensor \Dbimon$.
This bimonad is the left adjoint of 
$\S = \Dbimon \linarrow \_$, so the mate construction 
induces a bimonad structure on $\S$ that turns out to
be precisely the one described in \cref{sec:bimonad}.
Then the mate construction also induces a bijection 
between natural transformations $\Sdl : \oc\S \naturalTrans \S \oc$ 
and $\Sdlmate : (\oc \_ \tensor \Dbimon) \naturalTrans \oc (\_ \tensor \Dbimon)$, 
and it turns out that this bijection preserves distributive laws.
It provides a crucial step when showing that Taylor expansion 
boils down to a coalgebra on $\Dbimon$.

\paragraph*{Related work.}
There might be connections between the work presented here and the
recent article~\cite{Kerjean23} where an account of Taylor
expansion in differential LL is provided, based on the use of a
resource exponential modality \(\Oc\_\) which has not only its standard
comonad structure, but also a monad structure. 
The article conjectures that this monad structure is compatible with the 
comonadic structure of $\Oc\_$, turning $\Oc\_$ into a
bimonad.
In our setting, this is not the exponential modality which features a
bimonad structure, but the infinitary summability functor \(\Sfun\)
which does not at all play the same role: for instance, in LL models,
the functor \(\Sfun\) preserves cartesian products whereas the \(\Oc\_\)
functor turns the cartesian product into a tensor product.
Another difference between the two approaches is that, being based on
differential LL, \cite{Kerjean23} is based on additive categories
whereas one of our main motivations is to deal with Taylor expansion
in settings where addition is only partially defined.
A more detailed analysis of the possible connections between the two
approaches is definitely necessary.

\ifarxiv
\tableofcontents
\fi

The first two sections of this article consist of reminders on 
distributive law and monoidal monads. We encourage the reader familiar 
with those notions to skip these sections, and only refer 
to them when necessary.

\section{Background: distributive laws}

\label{sec:dl}

In this section, we review the well-known notion of distributive laws
first introduced in~\cite{Beck69}.
Distributive laws provide a solution to the following problem: given 
a (co)monad on a category $\category$ and some 
structure on $\category$, is it possible to extend that structure 
to the (co)Kleisli category of that (co)monad? 
Dually, is it possible to lift that structure to the 
(co)Eilenberg-Moore category of that (co)monad?

This categorical background will be important for this paper, since 
coherent differentiation and Taylor expansion are both formalized 
by such a distributive law, allowing to extend a monad $\S$ 
encoding partial summability to a monad $\D$ on the coKleisli 
of the exponential comonad $\Oc\_$, see \cref{sec:Taylor}. 

This theory was mainly developed in~\cite{Street72} in a general 
$2$-categorical setting (see~\cite{Maranda65} for a definition of 
a $2$-category). We will refer to~\cite{Power02}
as it is a more accessible introduction.
For accessibility reason, we only unfold the notion in the 
special case of the $2$-category $\categorycat$, whose objects 
are the (small) categories, whose morphisms are functors and whose 
$2$-morphisms are natural transformations. 

\subsection{Distributive laws for monads and extension to Kleisli categories}
\label{sec:monad-dl}

If $\mtriple = (\monad, \munit, \msum)$ is a monad on
$\cat$, then its Kleisli category $\kleisli$ is the category with the
same objects as $\cat$ and such that
$\kleisli(X, Y) = \cat(X, \monad Y)$. The identity is given by
$\munit_X \in \kleisli(X, X)$ and the composition of
$f \in \kleisli(X, Y)$ with $g \in \kleisli(Y, Z)$ is defined as
\begin{equation}
    g \kleislicomp f = \msum_Y \comp \monad g \comp f \, .
\end{equation} 
There is a functor 
$\kleisliL : \cat \arrow \kleisli$ such that 
$\kleisliL X = X$ and for any $f \in \cat(X, Y)$,
$\kleisliL f = \munit \comp f \in \kleisli(X, Y)$.
Then for any $f \in \kleisli(X', X)$, 
$h \in \category(X, Y)$ and $g \in \kleisli(Y, Y')$,
\begin{equation} \label{eq:KleisliL-comp}
\kleisliL h \kleislicomp f = \monad h \comp f \quad  \quad
  g \kleislicomp \kleisliL h = g \comp h \, .
\end{equation}
These equations are obtained by a straightforward computation using 
naturality and the triangle identities of the monad.

Throughout this section, we assume that 
$\mtriple[1] = (\monad[1], \munit[1], \msum[1])$ is a monad 
on a category $\cat$, 
$\mtriple[2] = (\monad[2], \munit[2], \msum[2])$ is a monad 
on a category $\catbis$, and
$\mtriple[3] = (\monad[3], \munit[3], \msum[3])$ is a monad 
on a category $\catter$.
All the following definitions and result can be found 
in section 4 of~\cite{Power02}.

\begin{definition} \label{def:extension} 
    A functor 
    $\extension{\F} : \kleisli[1] \arrow \kleislibis[2]$
    is an extension of the functor $\F : \cat \arrow \catbis$ 
    if the following diagram commutes.
    \[
    \begin{tikzcd}
    \cat \arrow[d, "{\kleisliL[1]}"'] \arrow[r, "{\F}"] & \catbis \arrow[d, "{\kleisliL[2]}"] \\
    \kleisli[1] \arrow[r, "{\extension{\F}}"']            & \kleislibis[2]                              
    \end{tikzcd}
    \]
\end{definition}

\begin{definition}
    \label{def:extension-natural} 
    Let $\F, \G : \cat \arrow \catbis$ 
    with respective extensions 
    $\extension{\F}, \extension{\G} : \kleisli[1] \arrow \kleislibis[2]$. 
    Let $\alpha : \F \naturalTrans \G$ be a natural transformation. Then 
	$\kleisliL[2] \alpha_X  \in \kleislibis[2](\extension{\F}X, \extension{\\G}X)$.
    We say that $\alpha$ extends to $\extension{\F}$ and $\extension{\G}$ 
    if $\kleisliL[2] \alpha$ is a natural transformation 
    $\extension{\F} \naturalTrans \extension{\G}$.
\end{definition}

\begin{remark} In \cite{Power02}, the extensions of functors and 
natural transformations to the Kleisli category 
are called liftings. We prefer the terminology extension that 
is widely used in the literature, in order to make a clear distinction
with the notion of lifting to the (co)Eilenberg-Moore 
category that is introduced~\cref{sec:dl-em}
\end{remark}

These extensions are related 
to the notion of distributive laws.

\begin{definition}
    \label{def:dl}
    A natural transformation 
    $\dl : \F \mtriple[1] \naturalTrans \mtriple[2] \F$  is a \emph{distributive law} 
    between $\F$ and two monads $\mtriple[1]$ and $\mtriple[2]$ 
    if the following two diagrams commute.
    \[\begin{tikzcd}
        {\F} \\
        {\F \monad[1]} & {\monad[2] \F}
        \arrow["{\F \munit[1]}"', from=1-1, to=2-1]
        \arrow["{\dl}"', from=2-1, to=2-2]
        \arrow["{\munit[2] \F}", from=1-1, to=2-2]
    \end{tikzcd}
    \begin{tikzcd}
    \F \monad[1]^2 \arrow[d, "{\F \msum[1]}"'] \arrow[r, "{\dl \monad[1]}"] & \monad[2] \F \monad[1] \arrow[r, "{\monad[2] \dl}"] & \monad[2]^2 \F \arrow[d, "{\msum[2] \F}"] \\
    \F \monad[1] \arrow[rr, "\dl"']                                          &                                                               & \monad[2] \F                             
    \end{tikzcd}\]
\end{definition}

\begin{remark} We write
	$\F \mtriple[1] \naturalTrans \mtriple[2] \F$ instead 
	of $\F \monad[1] \naturalTrans \monad[2] \F$ because it 
    makes it obvious what 
	part is playing the role of the monad. We will keep this convention through
	the article.
\end{remark}

\begin{definition}
    \label{def:dl-morphism}
    Let $\dl : \F \mtriple[1] \naturalTrans \mtriple[2] \F$ and 
    $\dlbis : \G \mtriple[1] \naturalTrans \mtriple[2] \G$ two distributive laws.
    A morphism between $\dl$ and $\dlbis$ is a natural transformation
    $\alpha : \F \naturalTrans \G$ such that the diagram below commutes.
    \[
    \begin{tikzcd}
    \F \monad[1] \arrow[r, "\dl"] \arrow[d, "{\alpha \monad[1]}"'] & \monad[2] \F  \arrow[d, "{\monad[2] \alpha}"]  \\
    \G \monad[1] \arrow[r, "\dlbis"']                           & \monad[2] \G                            
    \end{tikzcd}
    \]
\end{definition}

\begin{theorem} (part one of Theorem 4.8 of \cite{Power02})
    \label{thm:extension-and-dl} 
Any extension
$\extension{\F} : \kleisli[1] \arrow \kleislibis[2]$ induces
a distributive law
$\dl : \F \mtriple[1] \naturalTrans \mtriple[2] \F$.
This law is given by the image of 
$\id_{\monad[1]X} \in \kleisli[1](\monad[1] X, X)$ by $\extension{F}$
\[ \dl_X \defEq \extension{\F}(\id_{\monad[1] X}) 
    \in \kleislibis[2](\F \monad[1] X, \F X) 
    = \catbis(\F \monad[1] X, \monad[2] \F X). \]
Conversely, any distributive law 
    $\dl : \F \mtriple[1] \naturalTrans \mtriple[2] \F$ induces an extension
    $\extension{\F} : \kleisli[1] \arrow \kleislibis[2]$ that maps an object 
    $X$ to $\extension{\F} X \defEq \F X$ and a morphism
    $f \in \kleisli[1](X, Y) = \cat(X, \monad[1] Y)$ to
    \[\extension{\F}f \defEq \begin{tikzcd}
        {\F X} & {\F \monad[1]Y} & {\monad[2] \F Y}
        \arrow["{\F f}", from=1-1, to=1-2]
        \arrow["{\dl_Y}", from=1-2, to=1-3]
    \end{tikzcd} \in \kleislibis[2](\F X, \F Y) \, .\]
    Those two constructions are inverse of each other, so there is a bijection 
    between extensions and distributive laws.
\end{theorem}

\begin{remark} \label{rem:dl-compositional}
This bijection is compositional. 
The natural transformation $\id_{\monad} : \monad \naturalTrans \monad$
is a distributive law between the identity functor $\idfun^{\cat}$ 
and the monad $\mtriple$, associated to the extension 
$\idfun^{\kleisli}$.
We can check that if $\extension{\F}$ is an extension 
of $\F : \cat \arrow \catbis$ whose associated distributive 
law is $\dl$, and if $\extension{\G}$ is an extension 
of $\G : \catbis \arrow \catter$ whose associated distributive 
law is $\dlbis$, then $\extension{\G} \extension{\F}$ is 
an extension of $\G \F$ whose associated distributive law is 
defined as the following pasting diagram.
    \begin{equation} \label{eq:dl-composition}
    \begin{tikzcd}
            {\cat} & {\catbis} & {\catter} \\
            {\cat} & {\catbis} & {\catter}
            \arrow["\F", from=1-1, to=1-2]
            \arrow["{\monad[1]}"', from=1-1, to=2-1]
            \arrow["{\monad[2]}"{description}, from=1-2, to=2-2]
            \arrow["{\monad[3]}", from=1-3, to=2-3]
            \arrow["\G", from=1-2, to=1-3]
            \arrow["\F"', from=2-1, to=2-2]
            \arrow["\G"', from=2-2, to=2-3]
            \arrow["{\dl}"', shorten <=7pt, shorten >=7pt, Rightarrow, from=2-1, to=1-2]
            \arrow["{\dlbis}"', shorten <=7pt, shorten >=7pt, Rightarrow, from=2-2, to=1-3]
        \end{tikzcd}
    \end{equation}
        which corresponds to the following composition: 
        \begin{tikzcd}
            {\G\F \monad[1]} & {\G \monad[2] \F} & {\monad[3] \G \F}
            \arrow["{\G \dl}", from=1-1, to=1-2]
            \arrow["{\dlbis \F}", from=1-2, to=1-3]
        \end{tikzcd}.
\end{remark}
The result below is also proved in~\cite{Power02}. The proof is simple, 
but this result is crucial in the development of this article.

\begin{theorem} (part two of Theorem 4.8 of \cite{Power02}) \label{thm:extension-and-dl-morphism} 
    Let $\dl : \F \mtriple[1] \naturalTrans \mtriple[2] \F$ and 
    $\dlbis : \G \mtriple[1] \naturalTrans \mtriple[2] \G$ be two distributive laws
    and let $\extension{\F}, \extension{\G}$ be their associated extensions.
    Then a natural transformation 
    $\alpha : \F \naturalTrans \G$ is a morphism between the distributive laws 
    $\dl$ and $\dlbis$ if and only if 
    $\alpha$ extends to $\extension{\F}$ and $\extension{\G}$.
\end{theorem}

\begin{proof} 
    This is a direct consequence of \cref{eq:KleisliL-comp} and naturality of $\alpha$ 
    in $\cat$.
\end{proof}

\begin{remark}
    The bijections of~\cref{thm:extension-and-dl}
    and~\cref{thm:extension-and-dl-morphism} give an isomorphism of $2$-categories. 
    We can define a $2$-category $\extensioncat$ whose objects are the monads 
    $\mtriple[1]$ on small 
    categories\footnote{For foundational issues.}. The hom-set
    $\extensioncat(\mtriple[1], \mtriple[2])$ consists in
    the pairs $(\F, \extension{\F})$ where $\F : \cat \arrow \catbis$
    is a functor and $\extension{\F} : \kleisli[1] \arrow \kleislibis[2]$ 
    is an extension of $\F$. 
    A $2$-cell
    $(\F, \extension{\F}) \naturalTrans (\G, \extension{\G})$ is a natural
    transformation $\alpha : \F \naturalTrans \G$ that extends to
    $\extension{\F}$ and $\extension{\G}$ in the sense of 
    \cref{def:extension-natural}.
    We can also define a $2$-category $\dlcat$ whose objects are also
    the monads $\mtriple$ on small categories.
    A morphism between $\mtriple[1]$ and $\mtriple[2]$
    consists of a pair $(\F, \dl)$ where $\F : \cat \arrow \catbis$
    is a functor and
    $\dl : \F \mtriple[1] \naturalTrans \mtriple[2] \F$ is a distributive law.
    The composition is given as in \cref{rem:dl-compositional}.
    The $2$-cells are the morphisms of distributive laws.
Then following \cref{rem:dl-compositional},
the bijection of \cref{thm:extension-and-dl} and 
\cref{thm:extension-and-dl-morphism}
is an isomorphism of $2$-categories between $\dlcat$ and $\extensioncat$.
\end{remark}

\subsection{Distributive law for comonads and extension to 
coKleisli categories}
\label{sec:comonad-dl}

The coKleisli category of a comonad $\comtriple$ is defined in a dual way 
to the Kleisli category of a monad.
If $\comtriple = (\comonad, \comunit, \comsum)$ 
is a comonad on $\cat$, then its coKleisli category $\cokleisli$ 
is the category with the same objects as $\cat$ and such that 
$\cokleisli(X, Y) = \cat(\comonad X, Y)$. The identity is given by 
$\comunit_X \in \cokleisli(X, X)$ and the composition of 
$f \in \cokleisli(X, Y)$ with $g \in \cokleisli(Y, Z)$
is defined as 
\begin{equation}
    g \cokleislicomp f = g \comp \comonad f \comp \comsum \, .
\end{equation} 
There is a functor 
$\cokleisliL : \cat \arrow \cokleisli$ such that 
$\cokleisliL X = X$ and for any $f \in \cat(X, Y)$,
$\cokleisliL f = f \comp \comunit \in \cokleisli(X, Y)$.

The notion of structure extension to the coKleisli category of a comonad 
is the same as the one of~\cref{def:extension,def:extension-natural}
for Kleisli categories.
Extensions to the coKleisli categories are also 
related to distributive law, except that this time the 
natural transformation goes in the opposite direction and has type
$\comtriple[2] \F \naturalTrans \F \comtriple[1]$.

\begin{definition}  \label{def:codl}
    A natural transformation $\dl : \comtriple[2] \F \naturalTrans \F \comtriple[1]$ 
    is a \emph{distributive law} 
    between $\F$ and two comonads $\comtriple[1]$ and $\comtriple[2]$ 
    if the following two diagrams commute.
    \[\begin{tikzcd}
        {\comonad[2] \F} & {\F \comonad[1]} \\
        & \F
        \arrow["{\dl}", from=1-1, to=1-2]
        \arrow["{\comunit[2] \F}"', from=1-1, to=2-2]
        \arrow["{\F \comunit[1]}", from=1-2, to=2-2]
    \end{tikzcd} \quad
    \begin{tikzcd}
        {\comonad[2] \F} && {\F \comonad[1]} \\
        {\comonad[2]^2 \F} & {\comonad[2] \F \comonad[1]} & {\F \comonad[1]^2}
        \arrow["\dl", from=1-1, to=1-3]
        \arrow["{\comonad[2] \dl}"', from=2-1, to=2-2]
        \arrow["{\dl \comonad[2]}"', from=2-2, to=2-3]
        \arrow["{\F \comsum[2]}", from=1-3, to=2-3]
        \arrow["{\comsum[2] \F}"', from=1-1, to=2-1]
    \end{tikzcd}\]
\end{definition}

\begin{definition} \label{def:codl-morphism}
    Let $\dl : \comtriple[2] \F \naturalTrans \F \comtriple[1]$ and 
    $\dlbis : \comtriple[2] \G \naturalTrans \G \comtriple[1]$
    two distributive laws.
    A morphism between $\dl$ and $\dlbis$ is a natural transformation
    $\alpha : \F \naturalTrans \G$ such that the diagram below commutes.
    \[ \begin{tikzcd}
		{\comonad[2] \F} & {\F \comonad[1]} \\
		{\comonad[2] \G} & {\G \comonad[1]}
		\arrow["{\dl}", from=1-1, to=1-2]
		\arrow["{\comonad[2] \alpha}"', from=1-1, to=2-1]
		\arrow["{\dlbis}"', from=2-1, to=2-2]
		\arrow["{\alpha \comonad[1]}", from=1-2, to=2-2]
	\end{tikzcd}  \]
\end{definition}

\begin{theorem} (Theorem 5.15 of \cite{Power02})
    \label{thm:coextension-and-dl-morphism} 
    \label{thm:coextension-and-dl} 
    \begin{enumerate}
        \item There is a bijection between extensions and distributive laws.
        \item A natural transformation is a morphism of distributive laws 
        if and only if it extends to the coKleisli categories.
    \end{enumerate}
\end{theorem}

Again, this bijection is compositional in a similar way to 
\cref{rem:dl-compositional} and provides an isomorphism of 
$2$-categories.

\subsection{Distributive laws and lifting to (co)Eilenberg-Moore category}
\label{sec:dl-em}

There are two other notions of distributive law, by taking the dual direction 
of \cref{def:dl,def:codl}. 

\begin{definition} \label{def:dl-em}
    A natural transformation 
    $\dl : \mtriple[2] \F \naturalTrans \F \mtriple[1]$  is a \emph{distributive law} 
    between $\F$ and two monads $\mtriple[1]$ and $\mtriple[2]$ 
    if two diagrams analogue to \cref{def:dl} commute.
    A natural transformation 
    $\dl : \F \comtriple[1] \naturalTrans \comtriple[2] \F$  is a \emph{distributive law} 
    between $\F$ and two comonads $\comtriple[1]$ and $\comtriple[2]$ 
    if two diagrams analogue to \cref{def:codl} commute.
	Morphisms of distributive law are defined in the same 
    way as \cref{def:dl-morphism} and \cref{def:codl-morphism}.
\end{definition}

Those notions of distributive law are tied to the notion of lifting of a functor and
natural transformations to the Eilenberg-Moore categories
$\algebras[1]$ and $\algebras[2]$ of the monads (and similarly
the coEilenberg-Moore categories $\coalgebras[1]$ and $\coalgebras[2]$ of the comonads). 
A lifting of a functor 
$\cat \arrow \catbis$ is a functor $\lifting{F} : \algebras[1] \arrow \algebras[2]$
such that the diagram below commutes (here, $\forget$ is the forgetful functor),
see \cite{Power02} for more details.
\[
    \begin{tikzcd}
    \cat \arrow[r, "{\F}"] & \catbis  \\
    \algebras[1]  \arrow[u, "{\forget}"] \arrow[r, "{\lifting{\F}}"'] & \algebras[2] \arrow[u, "{\forget}"']                             
    \end{tikzcd}
\]
We summarize the different configurations for a distributive law, 
depending on the direction of the natural transformation and the 
choice between monads and comonads in \cref{fig:dl}.

\begin{figure}
\begin{center}
\begin{tabular}{| c | c | c |}
    \hline
    Distributive law & Role of $H$ and $K$ & Type of extension \\
    \hline
    $\F \ltriple \naturalTrans \rtriple \F$ & Monads &   Extension of $\F$ to Kleisli \\
    $\rtriple \F \naturalTrans \F \ltriple$ & Monads &   Liftings of $\F$ to Eilenberg-Moore \\
    $\F \lcomtriple \naturalTrans \rcomtriple \F$ & Comonads & Liftings of $\F$ to coEilenberg-Moore \\
    $\lcomtriple \F \naturalTrans \F \rcomtriple$ & Comonads & Extension of $\F$ to coKleisli  \\
    \hline
\end{tabular}
\end{center}
\caption{Configurations for distributive laws}
\label{fig:dl}
\end{figure}

The central notion of this article, Taylor expansion, 
is given by a distributive law 
$\Sdl : \oc \S \naturalTrans \S \oc$ between a functor
$\S$ that represents partial summability (see 
\cref{sec:summability-structure}) and
the resource comonad $\oc\_$ of LL (see \cref{def:resource-comonad}).
It turns out that $\S$ is a monad, and that 
$\Sdl : \oc \S \naturalTrans \S \oc$ is also a distributive law 
between the \emph{functor} $\oc\_$, and the \emph{monad}
$\S$.
As such, the following notion of distributive law of a comonad 
over a monad is central to this article. 

\begin{definition} 

Let $\mtriple = (\monad, \munit, \msum)$ be a monad on a category $\cat$, and 
$\comtriple = (\comonad, \comunit, \comsum)$
be a comonad on $\cat$.
A distributive law of the comonad $\comtriple$ 
over the monad $\mtriple$ is a natural transformation 
$\dl : \comtriple \mtriple \naturalTrans \mtriple \comtriple$
that is both a distributive law 
$\dl : \comonad \mtriple \naturalTrans \mtriple \comonad$ (\cref{def:dl})
and a distributive law 
$\dl : \comtriple \monad \naturalTrans \monad \comtriple$ (\cref{def:codl}).
\end{definition}

Similarly, there are distributive laws between two monads, as well as distributive 
laws between two comonads, with similar definitions 
as above. In fact, distributive laws were first 
introduced in~\cite{Beck69} between two monads.
The notion of distributive law of a comonad 
over another comonad is also of interest in the article, because the functor 
$\S$ admits a comonadic structure and $\Sdl : \oc \S \naturalTrans 
\S \oc$ is a distributive law of the comonad $\oc\_$ over 
the comonad $\S$.

Observe that the following assertions are equivalent, by definition.
\begin{enumerate}
	\item $\dl$ is a distributive law of the comonad $\comtriple$ 
	over the monad $\mtriple$,
	\item $\dl$ is a distributive law 
	$\comonad \mtriple \naturalTrans \mtriple \comonad$
	and $\comunit$ and $\comsum$ are morphisms of distributive laws,
	\item $\dl$ is a distributive law 
	$\comtriple \monad \naturalTrans \monad \comtriple$
	and $\munit$ and $\msum$ are morphisms of distributive laws,
\end{enumerate}
where the distributive laws $\comonad^2 \mtriple \naturalTrans \mtriple \comonad^2$,
$\comtriple \monad^2 \naturalTrans \monad^2 \comtriple$, $\mtriple \naturalTrans 
\mtriple$ and $\comtriple \naturalTrans \comtriple$ involved above 
are given by compositionality
of the distributive laws, see~\cref{rem:dl-compositional}.

By \cref{thm:coextension-and-dl,thm:coextension-and-dl-morphism},
a distributive law of a comonad over a monad provides 
a monad $\extension{\mtriple}$ on $\cokleisli$ that extends $\mtriple$.
By \cref{thm:extension-and-dl,thm:extension-and-dl-morphism},
such distributive law also provides 
a comonad $\extension{\comtriple}$ on $\kleisli$ 
that extends $\comtriple$.
We can also check that $(\cokleisli)_{\extension{\mtriple}} =
(\kleisli)_{\extension{\comtriple}}$, so a distributive law of a comonad over 
a monad allows to combine them in arbitrary ways. This result was 
first proved in~\cite{VanOsdol71}.

There is also another notion of distributive laws 
of the monad $\mtriple$ over the comonad $\comtriple$ 
obtained by reversing the arrows. 
Such distributive 
laws are associated to lifting of the monad $\mtriple$ to the coEilenberg-Moore 
category of $\comtriple$, and lifting of 
$\comtriple$ to the Eilenberg-Moore category of $\mtriple$.

Those notions of distributive laws between monads and comonads allow 
in particular to define the notion of bimonad as a functor equipped 
with a structure of monad and a structure of comonad together with a distributive 
law expressing that these two structures commute in some sense. 
Bimonads are formally similar to the notion of Hopf bialgebra in 
a symmetric monoidal category, apart that the role of the symmetry is 
played here by the distributive law.

\begin{definition}[6.2 of \cite{Mesablishvili11}] 
    \label{def:bimonad} 
	Assume that $\bimonmon = (\bimon, \munit, \msum)$ is a monad on $\category$
	and $\bimoncomon = (\bimon, \comunit, \comsum)$ is a comonad on $\category$.
	Then $(\bimon, \munit, \msum, \comunit, \comsum)$ is a 
	$\bimonswap$-bimonad if $\bimonswap$ is a distributive law
	$\bimonmon \bimoncomon \naturalTrans \bimoncomon \bimonmon$
	of the monad $\bimonmon$ over the comonad $\bimoncomon$,
	and a distributive law $\bimoncomon \bimonmon \naturalTrans \bimonmon \bimoncomon$
	of the comonad $\bimoncomon$ over the monad $\bimonmon$,
	and if the diagrams below commute.
	\[ \begin{tikzcd}
	  HH &  \\
	  H & \idfun
	  \arrow["\msum"', from=1-1, to=2-1]
	  \arrow["{\comunit \ntcomph \comunit}", from=1-1, to=2-2]
	  \arrow["\comunit", from=2-1, to=2-2]
	\end{tikzcd} \quad 
	\begin{tikzcd}
	  \idfun & H \\
	   & HH
	  \arrow["\munit \ntcomph \munit"', from=1-1, to=2-2]
	  \arrow["\munit", from=1-1, to=1-2]
	  \arrow["\comsum", from=1-2, to=2-2]
	\end{tikzcd} \quad 
	\begin{tikzcd}
	  \id & H \\
	  & \id
	  \arrow["\munit", from=1-1, to=1-2]
	  \arrow["\comunit", from=1-2, to=2-2]
	  \arrow[Rightarrow, no head, from=1-1, to=2-2]
	\end{tikzcd} \] 
	\[ \begin{tikzcd}
	  {H H} & H & HH \\
	  HHHH && HHHH
	  \arrow["\msum", from=1-1, to=1-2]
	  \arrow["\comsum", from=1-2, to=1-3]
	  \arrow["{\comsum \ntcomph \comsum}"', from=1-1, to=2-1]
	  \arrow["{H \bimonswap H}"', from=2-1, to=2-3]
	  \arrow["{\msum \ntcomph \msum}"', from=2-3, to=1-3]
	\end{tikzcd} \]
	Recall that $\ntcomph$ is defined as the horizontal composition of natural 
	transformations: if 
	$F, G : \cat \arrow \catbis$, $F', G' : \catbis \arrow \catter$, 
	$\alpha : F \naturalTrans G$ and $\beta : F' \naturalTrans G'$,
	then $\beta \ntcomph \alpha : GF \naturalTrans G' F'$ is defined as
	\begin{equation} \label{eq:horizontal-composition}
		(\beta \ntcomph \alpha)_X \defEq \beta_{F' X} \comp G \alpha_X = 
	G' \alpha_X \comp \beta_{F X} \, .
	\end{equation}
  \end{definition}

  \begin{remark}
    There is a difference between $\tau$-bimonads and bimonads as introduced
    in 4.1 of \cite{Mesablishvili11}.
    Every $\tau$-bimonad is a bimonad, but the converse is not true.
    Still, $\tau$ is a part of the $\tau$-bimonad structure and as such 
    should not be explicitly referenced by the terminology.
    Furthermore, all the bimonads under consideration in our article are $\tau$-bimonads.
    Thus, we will refer to $\tau$-bimonads simply as bimonads. 
  \end{remark}
  
  \begin{remark} \label{rem:bimonad-involutive}
	The four diagram making $\bimonswap$ a distributive law
	$\bimonmon \bimoncomon \naturalTrans \bimoncomon \bimonmon$
	are exactly the same as the diagrams making $\bimonswap$ 
	a distributive law $\bimoncomon \bimonmon \naturalTrans \bimonmon \bimoncomon$,
	except that the arrows involving $\bimonswap$ are reversed. In particular, 
	if $\bimonswap$ is involutive then any of the two assumptions implies the other.
  \end{remark}

\section{Background: symmetric monoidal monads as distributive laws}
\label{sec:smm}

An important concept in the theory of monads (and in this article) 
is the concept of lax symmetric monoidal 
functor and lax symmetric monoidal monad. This concept 
is important because
if $\smcat$ is a symmetric monoidal category
and $\mtriple$ is a monad on $\smcat$, then 
$\kleislism$ inherits from $\smcat$ the structure of a symmetric 
monoidal category as first shown in \cite{Day70}. 
As expected, this notion is then deeply connected to the notion of 
distributive laws, and we discuss this connection in this section.

Lax monoidal structures can also be expressed in terms of strength, see 
\cite{Kock70,Kock72}. We detail this process as the notion of strength
crucially allows us to define a theory of Taylor expansion with regard to only one parameter
in \cref{sec:compatibility-product}.

\subsection{Distributive laws on product categories}
\label{sec:product-category}

Let us recall first some fact and notation about product categories.
The category $\categorycat$ whose objects are the (small) categories 
and whose morphisms are the functors 
is a cartesian category, with terminal object the category 
$\categoryunit$ which contains one object and one morphism, and whose
categorical product is defined as follows.

	Given two categories $\cat[1]$ and $\cat[2]$, the product category
	$\cat[1] \times \cat[2]$ is the category whose objects are the pairs 
	$(X_1, X_2)$ with $X_1 \in \objects(\cat[1])$ and $X_2 \in \objects(\cat[2])$ 
	and whose morphisms are the pairs $(f_1, f_2)$ with $f_1 \in \cat[1](X_1, Y_1)$ and
	$f_2 \in \cat[2](X_2, Y_2)$.

	For any functors $\F_1 : \cat[1] \arrow \catbis[1]$ and $\F_2 : \cat[2] \arrow \catbis[2]$,
	we can define the functor $\F_1 \times \F_2 : \cat[1] \times \cat[2] \arrow \catbis[1] \times 
	\catbis[2]$ by $(\F_1 \times \F_2) (X_1, X_2) \defEq (\F_1 X_1, \F_2 X_2)$ and 
	$(\F \times \F_2) (f_1, f_2) \defEq (\F_1 f_1, \F_2 f_2)$.

	Given $\F_1 : \cat \arrow \catbis[1]$ and $\F_2 : \cat \arrow \catbis[2]$, we can
	define $\prodPair{F_1}{F_2} : \cat \arrow \catbis[1] \times \catbis[2]$
	by $\prodPair{F_1}{F_2} X = (\F_1 X, \F_2 X)$ and
	$\prodPair{F_1}{F_2} f = (\F_1 f, \F_2 f)$.
	For any category $\cat$, we can define the functor 
	$\diagonal^{\cat} : \cat \arrow \cat \times \cat$
	by $\diagonal^{\cat}(X) = (X, X)$ and $\diagonal^{\cat} f = (f, f)$. That is,
	$\diagonal^{\cat} = \prodPair{\idfun}{\idfun}$.

	Given the functors  $\F_1, G_1  : \cat[1] \arrow \catbis[1]$ and 
	$\F_2, \G_2 : \cat[2] \arrow \catbis[2]$ and two natural transformation 
	$\alpha_1 : \F_1 \naturalTrans \G_1$ and $\alpha_2 : \F_2 \naturalTrans \G_2$, 
	we can define the natural transformation $(\alpha_1, \alpha_2) : \F_1 \times \F_2
	\naturalTrans \G_1 \times \G_2$ by 
	$(\alpha_1, \alpha_2)_{(X_1, X_2)} \defEq (\alpha_{1, X_1}, \alpha_{2, X_2})$. 

\begin{definition} \label{def:product-monad}
	Given a monad $\mtriple[1] = (\monad[1], \munit[1], \msum[1])$ on $\cat[1]$ and 
	a monad $\mtriple[2] = (\monad[2], \munit[2], \msum[2])$ on $\cat[2]$, we can 
	define the monad $\mtriple[1] \times \mtriple[2]$ on $\cat[1] \times \cat[2]$ whose
	unit is $(\munit[1], \munit[2])$ and whose sum is $(\msum[1], \msum[2])$. 
\end{definition}

We can check that $(\cat[1] \times \cat[2])_{\mtriple[1] \times \mtriple[2]}
	= \kleisli[1] \times \kleisli[2]$, so the following lemmas 
    make sense.

\begin{lemma} \label{prop:extension-diagonal}
	The functor $\diagonal^{\kleisli} : \kleisli \arrow \kleisli \times 
	\kleisli $
	is an extension of $\diagonal^{\cat}$
	whose associated distributive law 
	is $\id_{\prodPair{M}{M}} : \prodPair{M}{M} \naturalTrans \prodPair{M}{M}$.
\end{lemma}

\begin{proof} Straightforward computation.
\end{proof}

\begin{lemma} \label{prop:extension-product} 	
	Let $\F : \cat[1] \arrow \catbis[1]$ and $\G : \cat[2] \arrow \catbis[2]$. 
	Let $\extension{F}$ be an extension of $F$ and 
    $\extension{G}$ an extension of $G$ whose associated distributive laws 
    are respectively $\dl$ and $\dlbis$. 
    Then $\extension{F} \times \extension {G}$ is 
	an extension of $F \times G$ whose associated distributive law
	is $(\dl, \dlbis)$.
\end{lemma}

\begin{proof} Straightforward computation.
\end{proof}

\begin{lemma} \label{prop:extension-pairing}
	Let $\F : \cat \arrow \catbis[1]$ and $\G : \cat \arrow \catbis[2]$. 
	Let $\extension{F}$ be an extension of $F$ and 
    $\extension{G}$ an extension of $G$ whose associated distributive laws 
    are respectively $\dl^{\F}$ and $\dlbis$. 
    Then $\prodPair{\extension{F}}{\extension {G}}$ is 
	an extension of $\prodPair{F}{G}$ whose associated distributive law
	is $(\dl, \dlbis)$.
\end{lemma}

\begin{proof} Observe that $\prodPair{\extension{F}}{\extension{G}}
	= (\extension{F} \times \extension{G}) \diagonal^{\kleisli}$, 
    so the result follows from \cref{prop:extension-product,prop:extension-diagonal} 
    and the compositionality of the bijection between extensions and distributive laws
\end{proof}

\begin{lemma} \label{prop:extension-constant} 
	For any object $A$ of $\category$, the constant endofunctor 
 	$A^{\kleisli} : \kleisli \arrow \kleisli$ is an 
 	extension of the constant endofunctor $A : \category \arrow \category$. 
	Its associated distributive law is 
	$\munit \in \cat(A, \monad A)$
\end{lemma}

\begin{proof} The fact that $A^{\kleisli}$ is an extension of $A$ is immediate.
	Its associated distributive law is
	$A^{\kleisli}(\id_{\monad X}) = \id_A^{\kleisli} = \munit$.
\end{proof}

\subsection{Lax monoidal functor and lax monoidal monads}
\label{sec:monoidal-functor}

We recall in this section basic definitions on symmetric monoidal 
category and monoidal monads, we refer the reader to~\cite{Mellies09}
for more details. Then we show that any symmetric monoidal monad can be 
interpreted as a distributive law. This explains why a symmetric 
monoidality structure on a monad allows to extend the 
symmetric monoidal structure of the category on which 
the monad is defined to its Kleisli category.

\begin{definition}
A monoidal category $\smtuple$ is a category $\smcat$ equipped 
with a bifunctor $\sm : \smcat \times \smcat \arrow \smcat$ called 
the tensor product, an object 
$\smone$ called the unit of the tensor, two natural isomorphisms
$\smunitR_X \in \smcat(X \sm \smone, X)$ and 
$\smunitL_X \in \smcat(\smone \sm X, X)$ called the right 
and left unitors, and a natural isomorphism 
$\smassoc_{X, Y, Z} \in \smcat((X \sm Y) \sm Z, X \sm (Y \sm Z))$
called the associator.
These isomorphisms are subject to commutations that we will not recall 
here.
The category is symmetric monoidal if there is an additional 
isomorphism $\smsym_{X,Y} \in \smcat(X \sm Y, Y \sm X)$ compatible 
with the monoidal structure.
\end{definition}
\begin{definition} \label{def:smf}
    A lax symmetric monoidal functor from a symmetric monoidal category
    $\smtuple[1]$ to another symmetric monoidal category
    $\smtuple[2]$ is a tuple
    $(\smf, \smfzero, \smftwo)$ where $\smf : \smcat[1] \arrow \smcat[2]$
    is a functor, $\smfzero \in \smcat[2](\smone[2], \smf \smone[1])$
    and $\smftwo_{X,Y} \in \smcat[2](\smf X \sm[2] \smf Y, \smf (X \sm[1] Y))$
    is a natural transformation
    that are compatible with the monoidal structure
\[\begin{tikzcd}
	{\smf X \sm[2] \smone[2]} & {\smf X \sm[2] \smf \smone[1] } & {\smf(X \sm[1] \smone[1])} \\
    {\smf X} & & {\smf X} 
	\arrow["{\smf X \sm[2] \smfzero}", from=1-1, to=1-2]
	\arrow["\smftwo", from=1-2, to=1-3]
	\arrow["{\smunitR[2]}"', from=1-1, to=2-1]
	\arrow["{\smf \smunitR[1]}", from=1-3, to=2-3]
	\arrow[Rightarrow, no head, from=2-1, to=2-3]
\end{tikzcd} \quad
\begin{tikzcd}
	{\smone[2] \sm[2] \smf X} & {\smf \smone[1] \sm[2] \smf X} & {\smf(\smone[1] \sm[1] X)} \\
	{\smf X} & & {\smf X}
	\arrow["{\smfzero \sm[2] \smf X}", from=1-1, to=1-2]
	\arrow["\smftwo", from=1-2, to=1-3]
	\arrow["{\smunitL[2]}"', from=1-1, to=2-1]
	\arrow["{\smf \smunitL[1]}", from=1-3, to=2-3]
	\arrow[Rightarrow, no head, from=2-1, to=2-3]
\end{tikzcd} \] 
\[ 
\begin{tikzcd}
	{(\smf X \sm[2] \smf Y) \sm[2] \smf Z} & {\smf (X \sm[1] Y) \sm[2] \smf Z} &
	{\smf ((X \sm[1] Y) \sm[1] Z)} \\
	{\smf X \sm[2] (\smf Y \sm[2] \smf Z)} & {\smf X \sm[2] \smf(Y \sm[1] Z)}  &
	{\smf (X \sm[1] (Y \sm[1] Z))}
	\arrow["{\smassoc[2]}"', from=1-1, to=2-1]
	\arrow["{\smf \smassoc[1]}", from=1-3, to=2-3]
	\arrow["{\smftwo \sm[2] \smf Z}", from=1-1, to=1-2]
	\arrow["\smftwo", from=1-2, to=1-3]
	\arrow["{\smf X \sm[2] \smftwo}"', from=2-1, to=2-2]
	\arrow["\smftwo"', from=2-2, to=2-3]
\end{tikzcd}\]
and compatible with the symmetry.
\[ \begin{tikzcd}
	{\smf X \sm[2] \smf Y} & {\smf(X \sm[1] Y)} \\
	{\smf Y \sm[2] \smf X} & {\smf(Y \sm[1] X)}
	\arrow["{\smftwo_{X, Y}}", from=1-1, to=1-2]
	\arrow["{\smsym[2]_{\smf X, \smf Y}}"', from=1-1, to=2-1]
	\arrow["{\smftwo_{Y, X}}"', from=2-1, to=2-2]
	\arrow["{\smf \smsym_{X, Y}}", from=1-2, to=2-2]
\end{tikzcd} \]	
A strong symmetric monoidal functor is a lax symmetric
monoidal functor such that $\smfzero$ is an iso, 
and $\smftwo$ is a natural isomorphism.
\end{definition}

\begin{definition} A monoidal natural transformation between
two monoidal functors $(\F, \smfzero_\F, \smftwo_\F)$ and 
$(\G, \smfzero_\G, \smftwo_\G)$ 
is a natural transformation $\alpha : \F \naturalTrans \G$
such that the following diagram commutes.
\[\begin{tikzcd}
	& {\smone[2]} \\
	{\F \smone[1]} && {\G \smone[1]}
	\arrow["{\smfzero_{\F}}"', from=1-2, to=2-1]
	\arrow["{\smfzero_{\G}}", from=1-2, to=2-3]
	\arrow["\alpha"', from=2-1, to=2-3]
\end{tikzcd} \quad
\begin{tikzcd}
	{\F X \sm[2] \F Y} & {\F(X \sm[1] Y)} \\
	{\G X \sm[2] \G Y} & {\G(X \sm[1] Y)}
	\arrow["{\alpha \sm[2] \alpha}"', from=1-1, to=2-1]
	\arrow["\smftwo_{\F}", from=1-1, to=1-2]
	\arrow["\smftwo_{\G}"', from=2-1, to=2-2]
	\arrow["\alpha", from=1-2, to=2-2]
\end{tikzcd}\]
\end{definition}

There is in fact a 2-category in which the objects are the (small) monoidal
categories, the morphisms are the lax monoidal functors, and the
2-cells are the monoidal natural transformations. The composition of two
lax monoidal functors is given by the composition of the functor and 
a suitable composition of their associated natural transformations. The 
compositions of $2$-cells are the same as the compositions in $\categorycat$.

\begin{definition} \label{def:monoidal-monad}
A lax symmetric monoidal monad on a symmetric monoidal
category $\smtuple$ is the data of $(\monad, \munit, \msum, \smfzero, \smftwo)$
such that $(\monad, \smfzero, \smftwo)$ is a lax symmetric
monoidal functor from $\smtuple$ to itself, 
$(\monad, \munit, \msum)$ is a monad on $\smcat$, and such that $\munit$ and $\msum$
are monoidal natural transformations. 
This last assumption corresponds to the fact that 
$\smfzero = \munit_{\smone} \in \smcat(\smone, \monad \smone)$ 
and to the commutation of the following diagrams.
\[ \begin{tikzcd}
	{X \sm Y} \\
	{\monad X \sm \monad Y} & {\monad (X \sm Y)}
	\arrow["{\munit_{X} \sm \munit_{Y}}"', from=1-1, to=2-1]
	\arrow["{\smftwo_{X, Y}}"', from=2-1, to=2-2]
	\arrow["{\munit_{X \sm Y}}", from=1-1, to=2-2]
\end{tikzcd} \]
\[ \begin{tikzcd}[column sep=large]
	{\monad^2 X \sm \monad^2 Y} & {\monad (\monad X \sm \monad Y)} & {\monad^2 (X \sm Y)} \\
	{\monad X \sm \monad Y} && {\monad (X \sm Y)}
	\arrow["{\msum_X \sm \msum_Y}"', from=1-1, to=2-1]
	\arrow["{\smftwo_{\monad X, \monad Y}}", from=1-1, to=1-2]
	\arrow["{\monad \smftwo_{X, Y}}", from=1-2, to=1-3]
	\arrow["{\msum_{X \sm Y}}", from=1-3, to=2-3]
	\arrow["{\smftwo_{X, Y}}"', from=2-1, to=2-3]
\end{tikzcd} \]
\end{definition}

Lax symmetric monoidal monads are a well studied notion because they are related 
to the extension of the symmetric monoidal structure to the Kleisli category.
\begin{theorem} (page 30 of \cite{Day70}) \label{thm:sm-extension}
    If $\mtriple$ is a symmetric monoidal monad on $\smcat$, then
	the structure of symmetric monoidal category of $\smcat$ extends 
	to $\kleislism$.
\end{theorem}

It turns out that lax monoidal monads are an example of distributive laws. 
This observation sheds light on \cref{thm:sm-extension} above, 
and is doubtlessly folklore, but seems to be often overlooked 
in the literature. We can trace this observation 
to~\cite{Guitart80}\footnote{With the difference that the naturality
of the symmetric monoidal structure is shown through the use of 
strengths, see~\cref{sec:commutative-monad}.,}.

The diagrams of \cref{def:monoidal-monad} correspond to the fact
that the natural transformation 
$\smftwo : (\_ \tensor \_) (\monad \times \monad) \naturalTrans \monad (\_ \tensor \_)$ is
a distributive law between the functor $\_ \tensor \_$ and the monads
$\mtriple \times \mtriple$ and $\mtriple$. 
By \cref{thm:extension-and-dl}, it means that $\_ \tensor \_$ extends to a functor 
$\_ \extension{\tensor} \_ : \kleislism \times \kleislism \arrow \kleislism$. 
Then the commutations of \cref{def:smf} corresponds to the fact that 
$\smunitL$, $\smunitR$, $\smassoc$ and $\smsym$ are morphisms of distributive laws.
The distributive laws involved with $\smassoc$ are given by compositionality, and 
the distributive laws involved with $\smunitL$ and $\smunitR$ are given by 
\cref{prop:partial-dl} below.

\begin{proposition} \label{prop:partial-dl} 
    For any objects $X$ and $Y$, the functor 
    $X \extension{\sm} \_$ is an extension of $X \sm \_$ with 
    associated distributive law 
    $\smftwo \ntcomp (\munit_X \sm \id_{\monad})$. Similarly, the functor 
    $\_ \extension{\sm} Y$ is an extension of $\_ \sm Y$ with 
    associated distributive law 
    $\smftwo \ntcomp (\id_{\monad} \sm \munit_Y)$.
  \end{proposition}
    
    \begin{proof} Observe that
        $\smone \sm \_ = (\_ \sm \_ ) \prodPair{1}{\idfun}$ and 
        $\smone \extension{\sm} \_ = (\_ \extension{\sm} \_ ) \prodPair{1^{\kleislism}}{\idfun}$.
        We conclude by compositionality of extensions and distributive laws,
        using the results of
        \cref{prop:extension-pairing} and \cref{prop:extension-constant}. 
    \end{proof}
    
By \cref{thm:extension-and-dl-morphism}, this is equivalent to the fact 
that $\smunitL$, $\smunitR$, $\smassoc$ and $\smsym$ extend to natural 
transformations on $\kleislism$. 
This is why the Kleisli category $\kleislism$ of a lax monoidal monad inherits 
from $\smcat$ the structure of a symmetric monoidal category.

\begin{remark} \label{rem:hopf-monad}
	Dually, there is a notion of a symmetric oplax monoidal functor
from a monoidal category
$\smtuple[1]$ to another monoidal category
$\smtuple[2]$. This is a tuple
$(\osmf, \osmfzero, \osmftwo)$ where $\osmf : \cat[1] \arrow \cat[2]$
is a functor, $\osmfzero \in \cat[2](\smf \smone[1], \smone[2])$
and $\osmftwo_{X,Y} \in \cat[2](\smf (X \sm[1] Y), \smf X \sm[2] \smf Y)$
is a natural transformation, and such that diagrams similar to the ones of 
\cref{def:smf} commute. There is also a notion of symmetric
oplax monoidal monad (also called Hopf Monad). 
In the same way that lax symmetric monoidal monads are related 
to the distributive laws of \cref{def:dl}, 
Hopf monads are related to the distributive laws of \cref{def:dl-em},
as observed in~\cite{Wisbauer08}.
It is not surprising then that a monad $\mtriple$
on $\smcat$ is a Hopf monad if and only if the symmetric monoidal structure 
of $\smcat$ lifts to the Eilenberg-Moore category of $\mtriple$, see~\cite{Moerdijk02}.

Finally, a symmetric (op)lax monoidal comonad is a comonad $\comtriple$ such that 
the functor $\comonad$ is symmetric (op)lax monoidal, and such that 
$\comunit, \comsum$ are monoidal natural transformations.
Again, the structure of an (op)lax monoidal comonad 
can be seen as a distributive law. Following \cref{fig:dl}, 
a lax symmetric monoidal comonad provides a lifting of the symmetric monoidal 
structure to the coEilenberg-Moore category (see \cite{Wisbauer08}), 
and an oplax symmetric monoidal comonad provides an extension of the 
symmetric monoidal structure to the coKleisli category. We summarize the 
different results in \cref{fig:lax-oplax}.
\end{remark}

\begin{figure}
    \begin{center}
    \begin{tabular}{| c | c | c |}
        \hline
        Monoidal structure & Role of the functor & Type of extension \\
        \hline
        Lax  & Monad &  Extension to Kleisli \\
        Lax & Comonad & Lifting to coEilenberg-Moore \\
        Oplax & Monad &   Lifting to Eilenberg-Moore \\
        Oplax & Comonad & Extension to coKleisli  \\
        \hline
    \end{tabular}
    \end{center}
    \caption{Extension and lifting of the monoidal structure}
    \label{fig:lax-oplax}
    \end{figure}

\subsection{Commutative monad} \label{sec:commutative-monad}
It is well-known that symmetric monoidal monads are the same as commutative monads, 
see \cite{Kock70,Kock72}. Let us recall what is a commutative monad. 
Let $\mtriple$ be a monad on a symmetric monoidal category
$\smcat$.

\begin{definition} \label{def:left-right-strength}
	A \emph{left strength} for $\mtriple$ is a natural transformation
$\strengthL_{X, Y} \in \smcat(X \sm \monad Y, \monad (X \sm Y))$ subject to 
the compatibility with the monoidal structure	
\[ \begin{tikzcd}
	{\smone \sm \monad X} & {\monad(\smone \sm X)} \\
    & {\monad X}
	\arrow["\smunitL_{\monad X}"', from=1-1, to=2-2]
	\arrow["\strengthL_{\smone, X}", from=1-1, to=1-2]
	\arrow["{\monad \smunitL_X}", from=1-2, to=2-2]
\end{tikzcd} 
\begin{tikzcd}
	{(X \sm Y) \sm \monad Z} && {\monad ((X \sm Y) \sm Z)} \\
	{X \sm (Y \sm \monad Z)} & {X \sm \monad (Y \sm Z)} & {\monad (X \sm (Y \sm Z))}
	\arrow["{\strengthL_{X \sm Y, Z}}", from=1-1, to=1-3]
	\arrow["{\smassoc_{X, Y, \monad Z}}"', from=1-1, to=2-1]
	\arrow["{X \sm \strengthL_{Y, Z}}"', from=2-1, to=2-2]
	\arrow["{\strengthL_{X, Y \sm Z}}"', from=2-2, to=2-3]
	\arrow["{\monad \smassoc_{X, Y, Z}}", from=1-3, to=2-3]
\end{tikzcd} \]
and the compatibility with the monad structure
\[ \begin{tikzcd}
	{X \sm Y} \\
	{X \sm \monad Y} & {\monad(X \sm Y)}
	\arrow["{X \sm \munit_Y}"', from=1-1, to=2-1]
	\arrow["{\strengthL_{X,Y}}"', from=2-1, to=2-2]
	\arrow["{\munit_{X \sm Y}}", from=1-1, to=2-2]
\end{tikzcd} \quad 
\begin{tikzcd}
	{X \sm \monad^2 Y} & {\monad(X \sm \monad Y)} & {\monad^2 (X \sm Y)} \\
	{X \sm \monad Y} && {\monad (X \sm Y)}
	\arrow["{\strengthL_{X, \monad Y}}", from=1-1, to=1-2]
	\arrow["{\monad \strengthL_{X, Y}}", from=1-2, to=1-3]
	\arrow["{X \sm \msum_Y}"', from=1-1, to=2-1]
	\arrow["{\strengthL_{X, Y}}"', from=2-1, to=2-3]
	\arrow["{\msum_{X \sm Y}}", from=1-3, to=2-3]
\end{tikzcd} \]	
A \emph{right strength} is a natural transformation 
$\strengthR_{X, Y} \in \smcat(\monad X \sm Y, \monad (X \sm Y))$ subject to similar
commutations. 
\end{definition}
When the category is symmetric monoidal, any left strength 
$\strengthL$ induces a right strength 
\begin{equation} \label{eq:left-right-strength}
    \strengthR_{X,Y} = \monad \smsym_{Y,X} \comp \strengthL_{Y,X} \comp \smsym_{X,Y}
\end{equation} 
and vice versa. This lead to the following naming convention in symmetric monoidal 
categories.

\begin{definition} \label{def:strength}
    A left-strength $\strengthL$ for a monad $\mtriple$ is also called a 
    \emph{strength}.
\end{definition}

\begin{definition} \label{def:commutative-monad}
A \emph{commutative monad} is a monad equipped with a strength
$\strengthL$ such that the following diagram commutes.
\[ \begin{tikzcd}
	& {\monad X \sm \monad Y} \\
	{\monad (\monad X \sm Y)} && {\monad (X \sm \monad Y)} \\
	{\monad^2(X \sm Y)} && {\monad^2(X \sm Y)} \\
	& {\monad (X \sm Y)}
	\arrow["{\strengthL_{\monad X, Y}}"', from=1-2, to=2-1]
	\arrow["{\monad \strengthR_{X, Y}}"', from=2-1, to=3-1]
	\arrow["{\msum_{X \sm Y}}"', from=3-1, to=4-2]
	\arrow["{\strengthR_{X, \monad Y}}", from=1-2, to=2-3]
	\arrow["{\monad \strengthL_{X, Y}}", from=2-3, to=3-3]
	\arrow["{\msum_{X \sm Y}}", from=3-3, to=4-2]
\end{tikzcd} \]
where $\strengthR$ is the right-strength obtained from $\strengthL$ by \cref{eq:left-right-strength}.
\end{definition}

It is well known that any commutative monad is a lax 
symmetric monoidal monad, (see~\cite{Kock70}), defining
$\smftwo \in \smcat(\monad X \sm \monad Y, \monad (X \sm Y))$ as the diagonal
of the square above. Conversely,
any symmetric monoidal monad is a commutative monad, see~\cite{Kock72}. 
The left and right strengths are defined from the lax monoidal structure by 
\[ \strengthL_{X,Y} \defEq 
\begin{tikzcd}
	{X \sm \monad Y} & {\monad X \sm \monad Y} & {\monad (X \sm Y)}
	\arrow["{\munit_X \sm \monad Y}", from=1-1, to=1-2]
	\arrow["{\smftwo_{X,Y}}", from=1-2, to=1-3]
\end{tikzcd} \]
\[ \strengthR_{X,Y} \defEq 
\begin{tikzcd}
	{\monad X \sm Y} & {\monad X \sm \monad Y} & {\monad (X \sm Y)}
	\arrow["{\monad X \sm \munit_Y}", from=1-1, to=1-2]
	\arrow["{\smftwo_{X,Y}}", from=1-2, to=1-3]
\end{tikzcd} \]

Recall from \cref{prop:partial-dl} that 
$\strengthL_{X, \_}$ is the distributive law associated to the extension
$X \extension{\sm} \_$ of $X \sm \_$, and $\strengthR_{\_, Y}$ is the 
distributive law associated to the extension 
$\_ \extension{\sm} Y$ of $\_ \sm Y$. So the equivalence between lax
symmetric monoidal monads and commutative monads can also be understood
as the fact that providing an extension 
$\_ \extension{\tensor} \_$ is the same as providing two extensions
$X \extension{\sm} \_$ and $\_ \extension{\sm} Y$ that follows 
the compatibility condition of the bifunctor theorem, see proposition 1 
of~\cite{Maclane71}.

\section{$\Sigma$-additive categories and left $\Sigma$-summability structures}

\label{sec:summability-structure}

\emph{Summability structures} and \emph{left summability structures}
have been introduced respectively in \cite{Ehrhard23-cohdiff} and \cite{Walch23}.
Both are categorical structures that give to the hom-sets the structure of 
a finite counterpart of the notions of \emph{partially additive monoids} (see 
\cite{Arbib80}) and \emph{$\Sigma$-monoid} (see \cite{Haghverdi00}).  
The difference between summability structures and left summability structures
is that in the former every morphism is linear with regard to
the sum (we will call this property \emph{additivity}), while this is not 
the case in the latter. Summability structures thus typically appear in models of 
LL $\categoryLL$, while left summability structures appear in the coKleisli
categories of their exponential $\kleisliExp$ or in other cartesian closed 
categories $\category$.

We introduce an infinitary counterpart of those structures, with a key difference
in our approach.
In \cite{Ehrhard23-cohdiff,Walch23}, the $\Sigma$-monoid structure on the hom-set arises 
naturally from the (left) summability structure itself.
Here, we assume that the category is already enriched over $\Sigma$-monoid, and a 
summability structure is simply a categorical structure that captures this 
sum. This leads to a simpler theory, that we strongly conjecture to be equivalent 
to the original one. In all the concrete models we know, the 
$\Sigma$-monoid structure arises naturally from the categorical 
structure of the model, such as in the representable case 
described in \cref{sec:elementary}.

We work in the framework of left additive structures, because it is more general 
and is necessary for \cref{sec:cartesian-Taylor}. Still, we will put
a lot of emphasis on the properties of the category of additive morphisms
(that is, morphisms that commutes with the sum), 
see \cref{sec:category-add,sec:bimonad}, and on summability in the
models of LL, see \cref{sec:summable-resource-category}.

\begin{remark} 
  In contrast with bare $\Sigma$-monoid enriched categories, 
  (left) summability structures provide us with an action on objects. 
  This additional structure gives us access to an internal description 
  of summable families which is crucial for representing Taylor expansions 
  (and more generally, the Faà~di~Bruno formula) within the model.
\end{remark}

\subsection{Categories enriched over $\Sigma$-monoids}

If $A$ is a countable set and $M$ is a set, an $A$-indexed family of 
elements of $M$ is a function $\vect x : A \arrow M$. 
We also write $\vect x = \family{x_a}$.
We consider in this section non-empty sets $M$ together with a partial function 
$\Sigma$ from indexed family on $M$ to $M$ called the sum. An indexed family 
$\family{x_a}$ is called summable when it is in the domain of $\Sigma$,
and we write its image as $\sum_{a \in A} x_a$.

\begin{notation} \label{notation:sum-def}
  Borrowing the notations from \cite{Tsukada22}, for any expressions 
  $e$ and $e'$ involving sums, we write \begin{itemize}
  \item $e \sumsub e'$ if whenever $e$ is defined, then $e'$ is defined 
  and $e = e'$;
  \item $e \sumiff e'$ if $e$ is defined if and only if $e'$ is defined and $e = e'$.
\end{itemize}
\end{notation}

\begin{definition}{\cite{Haghverdi00}} \labeltext{(Unary)}{ax:unary}
  \labeltext{(PA)}{ax:pa}
  The tuple $(M, \Sigma)$ is a \emph{$\Sigma$-monoid} if 
  the sum $\Sigma$ satisfies the following axioms.
  \begin{itemize}
    \item The unary sum axiom \ref{ax:unary}: any one element family 
    $x \in M$ is summable of sum $x$.
    \item The partition associativity axiom \ref{ax:pa}: 
    Let $\family{x_a}$ be an indexed family and let $\{A_i\}_{i \in I}$ be a partition 
    of $A$ where $I$ is at most countable and where the $A_i$ can be potentially empty. 
    Then \begin{equation} \label{eq:pa}
      \sum_{a \in A} x_a \sumiff 
      \sum_{i \in I} \left(\sum_{a \in A_i} x_a \right) \, .
    \end{equation}
  \end{itemize}
\end{definition}

It follows from \ref{ax:pa} that any subfamily of a summable family is also summable.  
In particular, it follows from \ref{ax:pa} and \ref{ax:unary} 
that the empty family is always summable (if $M$ is not empty). 
Let $0$ be its sum. For any family 
$\vect x = \family{x_a}$, we define 
\begin{equation} \label{eq:support}
  \supp{\vect x} = \{a \in A \St x_a \neq 0 \} \, .
\end{equation}
Then $0$ behaves as the neutral element, as shown in \cref{prop:zero-neutral}
below.

\begin{proposition} \label{prop:zero-neutral}
  Let $\vect x = \family{x_a}$ be an indexed family, and $A'$ a set such that 
  $\supp{\vect x} \subseteq A' \subseteq A$.
  Then $\family{x_a}$ is summable if and only if $\family<A'>{x_a}$ 
  is summable and the two sums are equal.
\end{proposition} 

\begin{proof}
  Assume that $\family<A'>{x_a}$ is summable. Then for all $a \in A$, define 
  $A_a = \{a\}$ if $a \in A'$, and $A_a = \emptyset$ otherwise. Then 
  for all $a \in A$, $\family<A_a>[a']{x_{a'}}$ is summable of sum $x_a$, by \ref{ax:unary}
  and because $0$ is the sum over the empty family.
  It follows from \ref{ax:pa} that $\family{\sum_{a' \in A_a} x_{a'}} = \family{x_a}$ is 
  summable of sum $\sum_{a \in A'} x_a$.
  Conversely, if $\family{x_a}$ is summable then $\family<A'>{x_{a'}}$ is summable 
  because it is a subfamily, and it is proved above that the sums are equal.  
\end{proof}

\begin{remark} \label{rem:summability-general}
  As discussed in~\cite{Manes86,Hines13}, the  ``only if'' assumption
  of the axiom \ref{ax:pa} is very strong, as it implies that the morphisms 
  are \emph{positive}, in the sense that if $x + y = 0$ then $x = y = 0$. 
  Because of this, our axiomatization of summability 
  leaves behind interesting models of LL in which the
  coefficients are not necessarily nonnegative, but in which maps are
  definitely analytics, such as Köthe spaces, see ~\cite{Ehrhard02},
  or finiteness spaces, see~\cite{Ehrhard05}.  

  For now, we keep this stronger and fundamentally 
  positive\footnote{In the sense explained above.}
  axiomatization because it suits quite well with the situations which
  occur in the denotational semantics of programming languages and of
  proofs.
  A more general axiomatization of summability structure based on
  the \emph{partial commutative monoids} of~\cite{Hines13} where the ``only if''
  assumption is dropped is currently a work in progress.
  It should allow one to recover the partial summability of both
  finiteness spaces and Köthe spaces.
\end{remark}

\begin{definition} \label{def:reindexing-bij}
  If $A$ and $B$ are indexing sets, \(\phi:A \arrow B\) is a bijection 
  $\vect x= \family{x_a}$ is a family, we define a B-indexed family
  $\Famact\phi\Vect x$ by $\Famact\phi\Vect x=\family<B>[b]{x_{\phi^{-1}(b)}}$.
\end{definition}

\begin{lemma} \label{prop:reindexing-bijection}
For any indexed family $\vect x = \family{x_a}$ and any 
bijection $\phi : A \arrow B$, $\vect x$ is summable if and 
only if $\Famact{\phi} \vect x$ is summable, 
and the two sums are equal.
\end{lemma}

\begin{proof}
Let $A_b = \{\phi^{-1}(b)\}$. By \ref{ax:unary}, $\family<A_b>{x_a}$
is summable of sum $x_{\phi^{-1}(b)}$.
Then by \ref{ax:pa}, $\family{x_a}$ is summable 
if and only if $\family<B>[b]{x_{\phi^{-1}(b)}}$ is summable and 
the two sums are equal.
\end{proof}

\begin{definition} \label{def:reindexing}
We generalize \cref{def:reindexing-bij} above.
Given two indexing set $A$ and $B$, an injection \(\phi:A \injection B\) 
and a family
$\vect x= \family{x_a}$ we define a $B$-indexed family
$\Famact\phi\Vect x= \family<B>[b]{y_b}$ by
\[ y_b=
  \begin{cases}
    x_a & \text{if }b\in\phi(A)\text{ and }\phi(a)=b\\
    0 & \text{otherwise}.
  \end{cases} \]
\end{definition}

\begin{lemma}
  The operation \(\phi\to\Famact\phi\) is functorial, that is %
  \(\Famact\Id\Vect x=\Vect x\) and %
  \(\Famact{\psi}\Famact\phi\Vect x=\Famact{(\psi\Comp\phi)}\Vect x\).
\end{lemma}

\begin{lemma} \label{prop:reindexing}
For any injection $\phi : A \injection B$ and any $A$-indexed family $\vect x$,
$\vect x$ is summable if and only if $\Famact{\phi} \vect x$ is summable, 
and the two sums are equals.
\end{lemma}

\begin{proof}
  Let  $\vect y = \Famact{\phi} \vect x = \family<B>[b]{y_b}$.
  By definition, $\supp{\vect y} \subseteq \im(\phi) \subseteq B$, 
  so by \cref{prop:zero-neutral}, $\vect y$ 
  is summable if and only if $\family<\im(\phi)>[b]{y_b}$
  is summable, and the two sums are equal. We conclude by 
  \cref{prop:reindexing-bijection}, using the 
  fact that $\phi$ is a bijection between $A$ and $\im(\phi)$.
\end{proof}

We introduce the following terminology, as an analogue of the terminology of 
additive and left additive categories used in differential categories and cartesian 
differential categories, see~\cite{Blute06,Blute09}. The monoid structure on 
the hom-set is replaced by a $\Sigma$-monoid structure.
\begin{definition}
A left $\Sigma$-additive category is a category $\category$ such that for any 
objects $X, Y \in \objects(\category)$, the hom-set $\category(X, Y)$ is a 
$\Sigma$-monoid, and such that the sum is left distributive over composition:
for all objects $X, Y, Z$, for all indexed family 
$\family<B>[b]{g_b \in \category(Y, Z)}$ and $f \in \category(X, Y)$,
 \begin{equation}
  \left( \sum_{b \in B} g_b \right) \comp f \sumsub \sum_{b \in B} g_b \comp f \, .
 \end{equation}
A morphism $h \in \category(Y, Z)$ is $\Sigma$-additive if for all indexed family 
$\family{f_a \in \category(X, Y)}$ 
\begin{equation}
  h \comp \left(\sum_{a \in A} f_a \right) \sumsub \sum_{a \in a} (h \comp f_a) \, .
\end{equation}
A $\Sigma$-additive category is a left $\Sigma$-additive category in which 
all morphisms are $\Sigma$-additive. A $\Sigma$-additive category can also be defined as 
a category enriched over the category of $\Sigma$-monoids, see~\cite{Haghverdi00}.
\end{definition}

We want to stress however that (left) $\Sigma$-additive categories 
are not necessarily (left) additive categories. The sum of a $\Sigma$-monoid 
may be defined on some infinite indexed families, but undefined on some finite 
indexed family. None of the examples of $\Sigma$-additive categories given 
in \cref{sec:elementary-Taylor-examples} are additive categories.

Let $0^{X, Y}$ be the zero of the $\Sigma$-monoid structure of
$\category(X, Y)$. It follows from the definition that for any
$f \in \category(X, Y)$, $0^{Y, Z} \comp f = 0^{X, Y}$, and that for
any $\Sigma$-additive $h \in \category(Y, Z)$, $h \comp 0^{X, Y} = 0^{X, Z}$.
We will often omit the superscript and simply write $0$.

Distributivity of the left and on the right implies double distributivity, thanks 
to the axiom \ref{ax:pa}.
\begin{proposition} \label{prop:double-distributivity}
For any left $\Sigma$-additive category $\category$, any indexed family 
$\family{f_a \in \category(X, Y)}$ and any indexed family 
$\family<B>[b]{h_b \in \category(Y, Z)}$
such that $h_b$ is $\Sigma$-additive for all $b \in B$,
\[ \left(\sum_{b \in B} h_b \right)  \comp \left(\sum_{a \in A} f_a \right)
\sumsub \sum_{(a, b) \in A \times B} (h_b \comp f_a)\,. \]
\end{proposition}

\begin{proof}
By left distributivity and additivity of the $h_b$
\[ \left(\sum_{b \in B} h_b \right)  \comp \left(\sum_{a \in A} f_a \right) 
\sumsub \sum_{b \in B} \left( h_b \comp \left( \sum_{a \in A} f_a \right) \right) 
\sumsub \sum_{b \in B} \sum_{a \in A} (h_b \comp f_a) \, .\]
But the sets $A \times \{b\}$ are a partition of 
$A \times B$, so \ref{ax:pa} ensures that 
\[ \sum_{(a, b) \in A \times B} (h_b \comp f_a) 
\sumiff \sum_{b \in B} \sum_{a \in A} (h_b \comp f_a) \, . \qedhere \] 
\end{proof}

\begin{corollary} \label{prop:sum-additive}
  For any summable family $\family<B>[b]{h_b}$ such that $h_b \in \category(Y, Z)$ 
  is $\Sigma$-additive for all $b \in B$, $\sum_{b \in B} h_b$ is $\Sigma$-additive.
\end{corollary}
\begin{proof}
 Let $\family{f_a \in \category(X, Y)}$ be a summable family. It follows from 
 \cref{prop:double-distributivity} that 
 \[ \left(\sum_{b \in B} h_b \right)  \comp \left(\sum_{a \in A} f_a \right)
\sumsub \sum_{(a, b) \in A \times B} (h_b \comp f_a) \]
Then, the sets $\{a\} \times B$ are a partition of 
$A \times B$, so it follows from \ref{ax:pa} that
\[ \sum_{(a, b) \in A \times B} (h_b \comp f_a) 
\sumiff \sum_{a \in A} \sum_{b \in B} (h_b \comp f_a) \] 
Finally, we know by left distributivity that 
\[ \sum_{b \in B} (h_b \comp f_a) = \left( \sum_{b \in B} h_b \right) 
\comp f_a \, . \] 
It follows that
\[ \left(\sum_{b \in B} h_b \right)  \comp \left(\sum_{a \in A} f_a \right) 
\sumsub \sum_{a \in A} \left( \left( \sum_{b \in B} h_b \right) \comp f_a 
\right) \, . \]
\end{proof}

\begin{remark} %
  \label{rk:infinite-summability-continuity}
  In~\cite{Walch23}, additive morphisms only preserve finite
  summations (including \(0\)), which corresponds to the usual
  algebraic notion of morphism of monoid.
  In our \(\Sigma\)-summability setting, additivity means not only
  preservation of finite sums, but also of all infinite sums when they 
  are defined.
  This means that $\Sigma$-additive morphisms also feature a property of
  \emph{continuity} whose precise nature depends on the category
  \(\category\) at stake.

  A very interesting situation occurs when the \(\Sigma\)-monoid
  structure of \(\category\) satisfies three additional properties
  \begin{enumerate}
    \item a family \(\Vect f\in\category(X,Y)^\Nat\) is summable as soon
    as, for any \emph{finite} set \(A\) and any injection
    \(\phi:A\to\Nat\), the family \(\Famact\phi\Vect f\) is summable;
  \item the preorder relation \(\leq\) defined on hom-sets of
    \(\category\) by \(f\leq g\) if \(\exists h\in\category(X,Y)\)
    such that \(g=f+h\) (where \(+\) is the binary addition induced by
    \((\S,\sequence{\Sproj_i}, \Ssum)\) on \(\category(X,Y)\)) is
    antisymmetric, that is, is an order relation;
  \item if for any \emph{finite} subset \(A'\) of $A$ one has
  $\sum_{a \in A'} f_a \leq f$, then $\sum_{a \in A} f_a \leq f$.
  \end{enumerate}
  The first condition turns the structure of $\sum$-monoid 
  into a \emph{partially additive monoid}, and the two other conditions 
  turn this partially additive monoid into an \emph{additive domain}, see Section~8.3 of~\cite{Manes86}.
  When these conditions hold, each hom-set \(\category(X,Y)\)
  is easily seen to be an \(\omega\)-complete partial order (ordered
  by \(\leq\) and having \(0\) as least element).
  The least upper bound (lub) of $(\sum_{i=0}^n f_i)_{n \in \N}$ is given by 
  $\psum f_i$. Then, $\Sigma$-additive
  morphisms are continuous because they 
  commute with the corresponding lubs in the sense that
  \(h\Comp{\sup_{i\in\Nat}g_i} = \sup_{i\in\Nat} (h \comp g_i) \) when \(h\in\category(Y,Z)\) is
  $\Sigma$-additive and \(\Vect g\in\category(X,Y)^\Nat\) is a monotone
  sequence.

  All the examples provided in \Cref{sec:Lafont-categories} are
  instances of this situation and, for that reason, feature general
  fixpoint operators allowing to combine our coherent Taylor
  structures with general recursion in the spirit
  of~\cite{Ehrhard22-pcf}.
\end{remark}

\subsection{$\Sigma$-summability structures}

\label{sec:summabibility-structure-def}

We assume in this section that $\category$ is a left $\Sigma$-additive 
category. 

\begin{definition} \label{def:pre-presummability-structure} A
  left $\Sigma$-summability structure on $\category$
  is a tuple $(\S, \sequence{\Sproj_i})$ where $\S$
  is a map (functional class) on objects 
  $\S: \objects(\category) \arrow \objects(\category)$  
  and each $\Sproj_i$ (for $i\in\Nat$) is a family of morphisms 
  $(\Sproj_{i,X}\in\category(\S X,X))_{X\in\objects(\category)}$
  called the projections such that:
  \begin{enumerate}
    \item[\ref{def:summability-structure-1}] 
    \labeltext{SS-1}{def:summability-structure-1} 
    For any $X$, $\sequence{\Sproj_{i, X}}$ is
    jointly monic. Joint monicity means that 
    for any $f, g \in \category(X, \S Y)$, if $\Sproj_{i, Y} \comp f 
    = \Sproj_{i,Y} \comp g$ for any $i \in \N$, then $f = g$.
  \item[\ref{def:summability-structure-2}]
    \labeltext{SS-2}{def:summability-structure-2} For any
    $\N$-indexed family, $\sequence{f_i \in \category(X, Y)}$ is 
    summable if and only if there exists a morphism 
    $f \in \category(X, \S Y)$ such that
    $\forall i\in\N, \Sproj_i \comp f = f_i$.
  \item[\ref{def:summability-structure-3}] 
  \labeltext{SS-3}{def:summability-structure-3}
  For any object $X$, the morphisms $\sequence{\Sproj_{i, X}}$ 
  are $\Sigma$-additive.
  \item[\ref{def:summability-structure-4}] 
  \labeltext{SS-4}{def:summability-structure-4}  
  For any $A$-indexed family $\family{f_a \in \category(X, \S Y)}$,
  if $\family<\N \times A>[(i, a)]{\Sproj_{i, Y} \comp f_a}$ 
  is summable then $\family{f_a}$
  is summable.
\end{enumerate}
The subscript $X$ is usually omitted when writing projections $\pi_{i,X}$. 
\end{definition}

\begin{definition}
 A left $\Sigma$-summability structure on a $\Sigma$-additive category
 is called a $\Sigma$-summability structure.
\end{definition}

\begin{notation}
  Because of \ref{def:summability-structure-1},
  the morphism $f$ given in \ref{def:summability-structure-2} 
  is unique and is called \emph{witness} of $\sequence{f_i}$. 
  We write $\Spairing{f_i} \defEq f$.
\end{notation}

Intuitively, \ref{def:summability-structure-1} and 
\ref{def:summability-structure-2} mean 
that $\S X$ can be seen as the object of "summable sequences". The morphism
$\Sproj_i$ maps a sequence to its $i$-th element.
Then, \ref{def:summability-structure-3} ensures that 
the sums between sequences is given by the sum on each component 
\[ \Sproj_i \comp \sum_{a \in A} \Spairing{f_{i,a}}
\sumsub \sum_{a \in A} f_{i,a} \text{ so } 
\sum_{a \in A} \Spairing{f_{i,a}} = \Spairing{\sum_{a \in A} f_{i, a}} \, . \]
Finally, \ref{def:summability-structure-4} crucially ensures that 
$\family{\Spairing{f_{i,a}}}$ is summable if and only if
$\family<\N \times A>[(i,a)]{f_{i,a}}$ is summable 
(the only if is a consequence of 
double distributivity, see \cref{prop:double-distributivity}). 
So the summability of 
sequences is completely determined by the global summability 
of all of their components.

\begin{example} 
  The models of Taylor expansion introduced in~\cite{Manzonetto12}
  are assumed to feature arbitrary countable sums. Then these sums 
  correspond to a particular (left) $\Sigma$-summability structure 
  defined as $\S X \defEq \with_{i \in \N} X$, $\Sproj_i \defEq \prodProj_i$
  (where $\with$ is a countable categorical product
  and $\prodProj_i$ is the $i$ projection). 
  This setting is covered in depth in \cref{sec:arbitrary-sum}.
  Note that in general, $\S X$ is not the cartesian product, precisely because the sum is 
  partial so the witness of a family of morphism needs not 
  always exist.
\end{example} 

\begin{notation}
By \ref{def:summability-structure-2}, 
the projections $\sequence{\Sproj_i}$ are summable with witness 
$\id$. Let $\Ssum_X = \sum_{i \in \N} \Sproj_{i,X} \in \category(\S X, X)$.
\end{notation}

We now assume that $\category$ is equipped with a left $\Sigma$-summability 
structure.
Observe that by left distributivity,
\begin{equation} \label{eq:sum-as-composition}
  \Ssum \comp \Spairing{f_i} \sumsub \sum_{i \in \N} (\Sproj_i \comp
\Spairing{f_i}) = \sum_{i \in \N} f_i \, . \end{equation}
Thus, the data of $(\S, \sequence{\Sproj_i}, \Ssum)$ completely 
characterizes the summability and the values of sums on $\N$-indexed 
families. 
By \cref{prop:reindexing}, this is enough to characterize the $\Sigma$-monoid 
structure on all indexed family, because any countable set has an injection 
into $\N$.
This explains why earlier versions of this paper define
summability structures first, and the $\Sigma$-monoid 
structure simply arise as a byproduct. The definition 
of sum of those earlier version were similar to  
the characterization given by \cref{prop:summability-alt} below.

\begin{proposition} \label{prop:summability-alt}
  A family $\vect f = \family{f_a \in \category(X, Y)}$ is summable if and only if 
  there exists an injection $\phi : A \injection \N$ and 
  a morphism $\Sfamily{\vect f}{\phi} \in \category(X, \S Y)$ 
  such that $\Sfamily{\vect f}{\phi}$ is a witness for 
  $\Famact{\phi} f$. That is, 
  \[ \Sproj_i \comp \Sfamily{\vect f}{\phi} = \begin{cases}
  f_a \text{ if } i = \phi(a) \text{ for some } a \in A \\
  0 \text{ otherwise.}
  \end{cases} \]
  Furthermore, $\sum_{a \in A} f_a = \Ssum \comp \Sfamily{\vect f}{\phi}$.
\end{proposition}

\begin{proof}
This is a direct consequence of
\ref{def:summability-structure-2} and \cref{prop:reindexing,eq:sum-as-composition}.
\end{proof}

The characterization of summability in terms of the 
existence of witnesses implies 
that properties involving summability can often be 
expressed in terms of the basic data of the summability structure. 
This is the case for instance of the additivity of morphisms.

\begin{proposition} \label{prop:additive} %
  For any morphism $h$, the following assertions are equivalent. \begin{enumerate}
    \item $h$ is $\Sigma$-additive;
    \item  $h \comp 0 = 0$ and 
    $\sequence{h \comp \Sproj_i}$ is summable with sum
    $h \comp \Ssum$.
  \end{enumerate}
\end{proposition}

\begin{proof}
  The forward implication is a direct consequence of the definition of
  additivity applied on the empty family and on the summable family
  $\sequence{\Sproj_i}$.
  For the backward implication, assume that $\vect f = \family{f_a}$ is summable.
  By \cref{prop:summability-alt}, there exists an injection 
  $\phi : A \injection \N$ and a witness $\Sfamily{\vect f}{\phi}$ of 
  $\Famact{\phi} \vect f$.
  Then we can check that 
  $\Spairing{h \comp \Sproj_i} \comp \Sfamily{\vect f}{\phi}$ is a witness for 
  $\Famact{\phi} \family{h \comp f_a}$. Indeed, 
  \[ h \comp \Sproj_i \comp \Sfamily{\vect f}{\phi}
  = \begin{cases}
    h \comp f_a \text{ if } i \in \phi(A) \text{ and } \phi(a) = i \\
  h \comp 0 = 0 \text{ otherwise.}
  \end{cases} \]
  Thus, by \cref{prop:summability-alt}, $\family{h \comp f_a}$ is summable with sum 
  \[ \sum_{a \in A} h \comp f_a 
  = \Ssum \comp \Spairing{h \comp \Sproj_i} \comp \Sfamily{\vect f}{\phi} 
  = h \comp \Ssum \comp \Sfamily{\vect f}{\phi}
  =  h \comp \sum_{a \in A} f_a \, .\] 
  It concludes the proof.
\end{proof}

\begin{notation} As for any family, we use the notation
  $\vect{\Sproj}$ for the family of projections 
  $\sequence{\Sproj_i}$.
\end{notation}

\subsection{The category of $\Sigma$-additive morphisms}

\label{sec:category-add}

We assume in this section that $\category$ is a left $\Sigma$-additive 
category.
It is easy to show that the identity morphism $\id_X \in \category(X, X)$ is 
$\Sigma$-additive, and that the composition of two $\Sigma$-additive morphisms is 
also $\Sigma$-additive. 

\begin{definition}
  Let $\categoryAdd$ be the category whose objects are those of 
  $\category$ and whose morphisms are the $\Sigma$-additive morphisms of $\category$
  We use $h'\compl h$ for $h'\comp h$ when $h$ and $h'$ are $\Sigma$-additive 
  to stress the fact that additivity is a form of linearity.
\end{definition}
There is an obvious forgetful functor from $\categoryAdd$ to $\category$
that we always keep implicit.
It follows from \cref{prop:sum-additive} that $\categoryAdd$ is 
a $\Sigma$-additive category. 
We now assume that $\category$ is equipped with a left $\Sigma$-summability 
structure $(\S, \vect \Sproj)$. 

\begin{lemma} \label{prop:Spairing-additive}
   For all summable family $\family{h_i \in \category(Y,Z)}$ of 
   $\Sigma$-additive morphisms,
   $\Spairing{h_i}$ is $\Sigma$-additive.
\end{lemma}

\begin{proof}
  Assume that $\family{f_a \in \category(X, Y)}$ is summable. 
  We first want to show that 
  $\family{\Spairing{h_i} \comp f_a}$ is summable. 
  By \ref{def:summability-structure-4}, it suffices to show that 
  $\family<\N \times A>[(i, a)]{h_i \comp f_a}$ is summable. This is the case, 
  by double distributivity (\cref{prop:double-distributivity}).
  Then, we show that
  \[ \Spairing{h_i} \comp \left(\sum_{a \in A} f_a \right) = \sum_{a \in A}
  (\Spairing{h_i} \comp f_a) \, .\]  
  This is done in a straightforward way using the additivity and the joint monicity 
  of the $\Sproj_i$.
\end{proof}

\begin{proposition}
  $(\S, \vect \Sproj)$ is 
  also a $\Sigma$-summability structure on $\categoryAdd$.
\end{proposition}

\begin{proof} Both
  \ref{def:summability-structure-1} and \ref{def:summability-structure-3} 
  hold by assumption. The property \ref{def:summability-structure-4} follows
  from the fact that the summability in $\categoryAdd$ coincides with the 
  summability in $\category$.
  Finally, \ref{def:summability-structure-2} follows from 
  \cref{prop:Spairing-additive} above.
\end{proof}

Recall \cref{prop:additive}: a morphism $h$ is $\Sigma$-additive if and only if
$\sequence{h \comp \Sproj_i}$ is summable with sum $h \comp \Ssum$. We define
\begin{equation}
  \S h \defEq \Spairing{h \comp \Sproj_i}
\end{equation} 
so that 
$\Sproj_i \comp \S h = h \comp \Sproj_i$. Furthermore, 
$\Ssum \comp \S h = h \comp \Ssum$ by \cref{eq:sum-as-composition}.
We can easily check by joint monicity of the $\Sproj_i$
that $\S h \comp \Spairing{x_i} = \Spairing{h \comp x_i}$, 
so $\S h$ consists intuitively in applying $h$ in each argument.

\begin{proposition} \label{prop:category-additive}
  $\S$ is an endofunctor on $\categoryAdd$, and the $\Sproj_i$ and 
  $\Ssum$ are natural transformations.
\end{proposition}

\begin{proof} Observe first that for all $h \in \categoryAdd(X, Y)$, 
  $\S h \in \categoryAdd(\S X, \S Y)$ by definition of $\S h$ and 
  \cref{prop:Spairing-additive}.
  Then 
  \[ \Sproj_i \comp \S \id_X = \id_{X} \comp \Sproj_i 
  = \Sproj_i = \Sproj_i \comp \id_{\S X} \]
  so $\S \id_X = \id_{\S X}$ by joint monicity of the $\Sproj_i$.
  Then, assume that
$h \in \category(X, Y)$ and $h' \in \category(Y, Z)$ are $\Sigma$-additive.
\[ \Sproj_i \comp \S h' \comp \S h 
= h' \comp \Sproj_i \comp \S h = h' \comp h \comp \Sproj_i
= \Sproj_i \comp \S (h' \comp h) \] 
so $\S (h' \comp h) = \S h' \comp \S h$ by joint monicity of the $\Sproj_i$.
Thus, $\S$ is a functor, and we already know that 
$\Sproj_i \compl \S h = h \compl \Sproj_i$, $\Ssum \compl \S h = h \compl \Ssum$,
so they are natural transformations.
\end{proof}

\begin{proposition} \label{prop:S-preserves-sum}
  For all family $\family{h_a}$ of $\Sigma$-additive morphisms, 
  \[ \S \left( \sum_{a \in A} h_a \right) \sumiff \sum_{a \in A} (\S h_a) \]
\end{proposition}

\begin{proof}
  Let $\phi : A \injection \N$ be an injection.
  Assume that $\family{h_a}$ is summable.  Then 
  $\sequence{h_i'} = \Famact{\phi}{\family{h_a}}$ is summable 
  by \cref{prop:reindexing}.
  So by double distributivity,
  \[ \family<\N \times \N>[(i, j)]{h_i' \compl \Sproj_j}
  = \family<\N \times \N>[(i, j)]{\Sproj_j \compl \S h_i'}  \]
  is summable. 
  It follows by \ref{def:summability-structure-4} that
  $\sequence{\S h_i'}$ is summable.
  But $\S 0 = 0$ (this directly follows from the joint monicity of the $\Sproj_i$).
  So $\sequence{\S h_i'} = \Famact{\phi}{\family{\S h_a}}$ and 
  $\family{\S h_a}$ is summable by 
  \cref{prop:reindexing}. 
  Furthermore, 
  \[ \Sproj_i \compl \S \left( \sum_{a \in A} h_a \right) 
  = \left( \sum_{a \in A} h_a \right) \compl \Sproj_i 
  \sumsub \sum_{a \in A} (h_a \compl \Sproj_i) \] 
  \[ \Sproj_i \compl \left( \sum_{a \in A} \S h_a \right)
  \sumsub \sum_{a \in A} (\Sproj_i \compl \S h_a) 
  = \sum_{a \in A} (h_a \compl \Sproj_i) \] 
  so by joint monicity of the $\Sproj_i$, the two sums are equal.
  
  To conclude the proof, assume that $\family{\S h_a}$ is summable and let us 
  show that $\family{h_a}$ is summable. As observed above, 
  $\sequence{\S h_i'} = \Famact{\phi}{\family{\S h_a}}$ so 
  this family is summable by \cref{prop:reindexing}.
  It follows by double distributivity that 
  $\family<\N \times \N>[(i, j)]{\Sproj_j \compl \S h_i'}
  = \family<\N \times \N>[(i, j)]{h_i' \compl \Sproj_j}$ 
  is summable. 

  Then, let $\Sinj_i \in \category(X, \S X)$ be the morphism such that 
  $\Sproj_i \comp \Sinj_i = \id_X$ and $\Sproj_j \comp \Sinj_i = 0$ if $j \neq i$.
  This morphism exists thanks to \ref{ax:unary} and \ref{def:summability-structure-2}.
  Then by left distributivity in $\category$, 
  $\family<\N \times \N>[(i, j)]{h_i' \comp \Sproj_j \comp \Sinj_0}$ 
  is summable. This indexed family 
  has values $h_i'$ on indexes $(i, 0)$,
  and values $0$ on any index $(i, j)$ such that $j \neq 0$, by additivity of 
  $h_i'$.
  So by \cref{prop:zero-neutral,prop:reindexing-bijection}, 
  $\sequence{h_i'}$ is summable and thus $\family{h_a}$ is summable by 
  \cref{prop:reindexing}.
\end{proof}

Finally, we show that 
naturality in $\categoryAdd$ 
is preserved by witness and sums.

\begin{lemma} 
  For any category $\catbis$ and all functors $\F, \G : \catbis \arrow \categoryAdd$, $0$ is a natural transformation 
  $\F \naturalTrans \G$.
  \end{lemma}
  \begin{proof}
    We have $0 \compl \F f = 0 = \G f \compl 0$ by 
  additivity of $\G f$.
  \end{proof}

\begin{definition} \label{def:summability-naturality}
  Let $\catbis$ be a category, and $F, G : \objects(\catbis) \arrow 
  \objects(\category)$ two maps.  
  Assume that for all $a \in A$,
  $\alpha^a = (\alpha^a_X)_{X \in \objects(\catbis)}$ is a family of 
  morphisms with $\alpha^a_X \in \category(F X, G X)$.
  We say that $\family{\alpha^a}$ is 
  summable if for all object $X$, 
  $\family{\alpha^a_X}$ is summable. 
  Then we write $\sum_{a \in I} \alpha^a$ the family 
  $(\sum_{a \in A} \alpha^a_X)_X$.
  If $A = \N$, we write 
  $\Spairing{\alpha^i}$ the family 
  $(\Spairing{\alpha^i_X})_X$.
\end{definition}

\begin{proposition} \label{prop:Spairing-natural}
  The sums and witnesses of natural transformations in $\categoryAdd$
  are also natural transformations.
  More precisely, for any category $\catbis$,
  for all functors 
  $\F, \G : \catbis \arrow \categoryAdd$ and for any
  family $\family{\alpha^a}$ of natural transformations
  $\alpha^a : \F \naturalTrans G$, if
  $\family{\alpha^a}$ is summable, then 
  $\sum_{a \in A} \alpha^a$ is a natural transformation 
  $F \naturalTrans G$. 
  If $A = \N$, then 
  $\Spairing{\alpha^i}$ is a
  natural transformation $\F \naturalTrans \S \G$.
\end{proposition}

\begin{proof} Assume first that $A = \N$.
  We want to show that
  $\Spairing{\alpha_X^i} \in \categoryAdd(\F X, \S \G X)$ is a natural
  transformation $F \naturalTrans \S \G$. That is, for all
  $f \in \catbis(X, Y)$,
  \[ \Spairing{\alpha_Y^i} \compl \F f = \S \G f \compl
  \Spairing{\alpha_X^i} \, .\]
  For all $i$,
  \[ \Sproj_i \compl \Spairing{\alpha_Y^i} \compl \F f = \alpha_Y^i
  \compl \F f = \G f \compl \alpha_X^i = \G f \compl \Sproj_i \compl
  \Spairing{\alpha_X^i} = \Sproj_i \compl \S \G f \compl
  \Spairing{\alpha_X^i} \, . \]
  By joint monicity of the $\Sproj_i$, we conclude that
  $(\Spairing{\alpha_X^i})_X$ is natural.
  Furthermore,
  $\sum_{i \in \N} \alpha^i_X = \Ssum_{\G X} \compl \Spairing{\alpha_X^i}$
  is natural since $\Ssum$ is natural.
  The case where $A$ is an arbitrary countable set 
  immediately follows by \cref{prop:reindexing} and naturality of $0$.
\end{proof}

\subsection{The bimonad $\S$}

\label{sec:bimonad}

In this section we endow $\S$ with a structure of bimonad in the 
sense of \cref{def:bimonad}. The corresponding natural transformations 
are essential in our approach to Taylor expansion and are similar to 
natural transformations introduced in tangent 
categories~\cite{Rosicky84} and in coherent differentiation~\cite{Walch23}.

We will see in \cref{sec:Taylor} that Taylor expansion\footnote{Or 
more precisely, the Faà~di~Bruno formula, but we prefer to stick 
to the concept of Taylor expansion which is our central object of 
study and whose ``functorialization'' requires its Faà~di~Bruno 
generalization.} consists in a functor $\D$ 
that extends $\S$ to 
$\category$ (in particular, $\D X = \S X$), and that the usual properties 
of differentiation correspond to the naturality with respect to 
$\D$ of the families of morphisms associated with the bimonad structure 
of $\S$. We provide an intuition on those 
natural transformations as operations on formal power series.

It is quite enlightening to write a sequence $\Spairing{f_i} 
\in \category(X, \S Y)$ as a
formal power series $\sum_{i=0}^{\infty} f_i t^i$ (this series is
called the \emph{generating function}).
The reason is that as we will see in \cref{sec:Taylor}, the element at
position $i$ in $\Spairing{f_i}$ can be seen as an order $i$ 
degree of differentiation. This is in fact the reason why we write 
the projections $\Sproj_i$ starting from index $0$ and not
from index $1$.
Then the projection $\Sproj_i$ maps $\sum_{i=0}^{\infty} f_i t^i$ 
to $f_i$, and the sum maps $\sum_{i=0}^{\infty} f_i t^i$ to 
$\sum_{i \in \N} f_i$ (it evaluates the formal power series in $t = 1$).
We introduce and recall the following families of 
morphisms on $\category$.
\begin{enumerate}
  \item A morphism $\Sinj_{i,X} \in \category(X, \S X)$ that 
  intuitively maps an element 
    $x \in X$ to the monomial $x t^i$, this morphism was used in the 
    proof of \cref{prop:S-preserves-sum};
  \item Remember that we have a morphism $\Ssum_X \in \category(\S X, X)$ 
  that intuitively maps a power series $\psum x_i t^i$ to $\psum x_i$ (it evaluates the 
  series in $t = 1$);
  \item A morphism $\SmonadSum_X \in \category(\S^2 X, \S X)$
  that intuitively maps a multivariate formal power series 
  $\sum x_{i,j} t^i u^j$ to the power series 
  $\psum[k] (\sum_{i+j = k} x_{i,j}) v^k$;
  \item A lift morphism $\Slift_X \in \category(\S X, \S^2 X)$ that 
  intuitively maps the power series $\psum x_i t^i$ to the power series with two 
  indeterminate $\psum x_i u^i v^i$;
  \item A flip operator $\Sswap_X \in \category(\S^2 X, \S^2 X)$ 
  that intuitively maps the power series
  $\sum_{i, j \in \N} x_{i,j} t^i u^j$ to 
  $\sum_{i, j \in \N} x_{i,j} t^j u^i$.
\end{enumerate} 
We show that those morphisms formally exists thanks to the axioms of (left) 
$\Sigma$-summability structure, that they are $\Sigma$-additive,
that they are natural transformations when considered as morphisms 
in $\categoryAdd$, and 
that they provide a bimonad structure to 
$\S$ on $\categoryAdd$ (we will recall 
the commutations that are involved on the fly).

\begin{notation} 
  For any $i, j \in \N$ and any morphism 
  $f \in \category(X, Y)$, we use $\Kronecker ij f$
  for the morphism which is equal to $f$ if $i = j$, and to 
  $0$ otherwise.
\end{notation}

\begin{theorem}\label{thm:bimonad-structure}
  \begin{enumerate}
    \item For all $i \in \N$, 
    there is a morphism \(\Sinj_i \in\category(X,\Sfun X)\) such
    that \(\Sproj_j\comp\Sinj_i=\Kronecker ij\id\) (and then
    $\Ssum \comp \Sinj_i = \id$).
    \item There is a morphism \(\Smonm \in\category(\Sfun^2X,\Sfun X)\)
  such that, for all \(i\in\Nat\),
  \(\Sproj_i\comp\Smonm=\sum_{j=0}^i\Sproj_{i-j}\comp\Sproj_j\).
    \item There is a morphism $\Slift \in \category(\S X, \S^2 X)$ called 
    lift such that $\Sproj_i \comp \Sproj_j \comp \Slift = \Kronecker ij \Sproj_i$.
  That is, 
  $\Slift = \Spairing{\Sinj_i \comp \Sproj_i}$.
    \item There is a morphism \(\Sflip\in\category(\Sfun^2X,\Sfun^2X)\) such that
  for all \(i,j\in\Nat\), one has
  \(\Sproj_i\comp\Sproj_j\comp\Sflip=\Sproj_j\comp\Sproj_i\).
  \end{enumerate}
\end{theorem}
\begin{proof}
  The existence of $\Sinj_i$ is a direct consequence of \ref{ax:unary} 
  and \ref{def:summability-structure-2}.
  The family $\family<\N^2>[(i,j)]{\Sproj_i \compl \Sproj_j}$ is 
  summable by double distributivity (\cref{prop:double-distributivity}).
  Thus, $\SmonadSum$ exists by \ref{ax:pa} and \ref{def:summability-structure-2}, 
  and $\Sswap$ exists by \ref{def:summability-structure-4}.
  Finally, the family $\sequence{\Sinj_i \comp \Sproj_i}$ is summable 
  if and only if $\family<\N^2>[(i,j)]{\Sproj_j \comp \Sinj_i \comp \Sproj_i}$ 
  is summable by \ref{def:summability-structure-4}.
  But 
  \[\Sproj_j \comp \Sinj_i \comp \Sproj_i = \begin{cases}
  \Sproj_i \text{ if } i = j \\ 
  0 \text{ otherwise.}
   \end{cases} \] 
   Thus, $\family<\N^2>[(i,j)]{\Sproj_j \comp \Sinj_i \comp \Sproj_i} 
   = \Famact \phi \sequence{\Sproj_i}$ where $\phi : \N \injection \N^2$ maps 
   $i$ to $(i,i)$. So this family is summable by \cref{prop:reindexing}, and 
   $\Slift$ exists by \ref{def:summability-structure-2}.
  \end{proof}

\begin{proposition} The families $\Sswap$, $\SmonadSum$, $\Sinj_i$ and $\Slift$ are
  $\Sigma$-additive, and they are natural transformations: $\Sswap : \S^2 \naturalTrans \S^2$,
  $\SmonadSum : \S^2 \naturalTrans \S$, $\Sinj_i : \idfun \naturalTrans \S$
  and $\Slift: \S \naturalTrans \S^2$.
\end{proposition}

\begin{proof} All the morphisms above are defined as tuplings of $0$,
   projections, and sums of projections. Those basic blocks are all 
   $\Sigma$-additive and
  natural transformations (the sums of projections are $\Sigma$-additive by 
  \cref{prop:sum-additive} and natural by \cref{prop:Spairing-natural}).
  So by \cref{prop:Spairing-natural} and \cref{prop:Spairing-additive},
  $\Sswap$, $\SmonadSum$, $\Sinj$ and $\Slift$ 
  are all $\Sigma$-additive and are natural transformations.
\end{proof}

\begin{proposition} The tuple $\monadS = (\S, \Sinj_0, \SmonadSum)$ is a monad on $\categoryAdd$, 
  that is the following diagrams of natural transformations commute.
  \[
\begin{tikzcd}
	\S & {\S^2} & \S \\
	& \S
	\arrow["{\Sinj_0 \S}", from=1-1, to=1-2]
	\arrow[Rightarrow, no head, from=1-1, to=2-2]
	\arrow["\SmonadSum"{description}, from=1-2, to=2-2]
	\arrow["{\S \Sinj_0}"', from=1-3, to=1-2]
	\arrow[Rightarrow, no head, from=1-3, to=2-2]
\end{tikzcd} \quad 
\begin{tikzcd}
	{\S^3} & {\S^2} \\
	{\S^2} & \S
	\arrow["{\S \SmonadSum}", from=1-1, to=1-2]
	\arrow["{\SmonadSum \S}"', from=1-1, to=2-1]
	\arrow["\SmonadSum", from=1-2, to=2-2]
	\arrow["\SmonadSum"', from=2-1, to=2-2]
\end{tikzcd} \]
\end{proposition}

\begin{proof} As usual, we use the joint monicity of the $\Sproj_i$.
  \begin{align*}
    \Sproj_i\compl\Smonm_X\compl\Sinj_{0, \Sfun X}
    = \sum_{l+r=i}\Sproj_l\compl\Sproj_r\compl\Sinj_{0, \Sfun X}
    = \sum_{l+r=i}\Sproj_l\compl (\Kronecker r 0 \id)
    = \Sproj_i\compl\Id = \Sproj_i
  \end{align*}
  \begin{align*}
    \Sproj_i\compl\Smonm\compl\Sfun\Sinj_{0, X}
    = \sum_{l+r=i}\Sproj_l\compl\Sproj_r\compl\Sfun\Sinj_{0, X}
    = \sum_{l+r=i}\Sproj_l\compl\Sinj_{0, X}\compl\Sproj_r
    = \sum_{l+r=i}(\Kronecker l 0 \id) \compl\Sproj_r
    = \Sproj_i
  \end{align*}
  using the naturality of \(\Sproj_r\).
  Hence, \(\Smonm\compl\Sinj_{0, \Sfun X}=\Smonm\compl\Sfun\Sinj_{0, X}=\Id\)
  by joint monicity of the \(\Sproj_i\)'s.

  Next we have
  \begin{align*}
    \Sproj_i\compl\Smonm_X\compl\Smonm_{\Sfun X}
    &= \sum_{l+r=i}\Sproj_l\compl\Sproj_r\compl\Smonm_{\Sfun X}\\
    &= \sum_{l+r=i}\Sproj_l\sum_{s+t=r}\Sproj_s\compl\Sproj_t\\
    &= \sum_{l+r+s=i}\Sproj_l\compl\Sproj_r\compl\Sproj_s
  \end{align*}
  and
  \begin{align*}
    \Sproj_i\compl\Smonm_X\compl\Sfun\Smonm_{X}    
    &= \sum_{l+r=i}\Sproj_l\compl\Sproj_r\compl\Sfun\Smonm_{X}\\
    &= \sum_{l+r=i}\Sproj_l\compl\Smonm_{X}\compl\Sproj_r\\
    &= \sum_{l+r=i}(\sum_{j+k=l}\Sproj_j\compl\Sproj_k)\compl\Sproj_r\\
    &= \sum_{j+k+r=i}\Sproj_j\compl\Sproj_k\compl\Sproj_k
  \end{align*}
  and hence
  \(\Smonm_X\compl\Smonm_{\Sfun X}=\Smonm_X\compl\Sfun\Smonm_{X}\).
\end{proof}

We use $\categoryAddS$ for the Kleisli category of the monad $\monadS$ on 
$\categoryAdd$. Observe that the composition of 
$f \in \categoryAddS(X, Y)$ with $g \in \categoryAddS(Y, Z)$, 
defined as $\SmonadSum \compl \S g \compl f$, 
is characterized by the equations for each $k\in\N$ 
\begin{equation}  
  \label{eq:composition-cauchy-product}
  \Sproj_k \compl \SmonadSum \compl \S g \compl f 
= \sum_{i+j = k} \Sproj_j \compl \Sproj_i \compl \S g \compl f 
= \sum_{i+j = k} (\Sproj_j \compl g) \compl (\Sproj_i \compl f)
\end{equation}
Coming back to formal power series, 
a morphism $f \in \categoryAddS(X, Y)$ can  
be equivalently written as a formal series $\psum f_i t^i$ 
where $f_i = \Sproj_i \compl f_i$.
Then \cref{eq:composition-cauchy-product} above 
means that the composition of two series 
corresponds to their Cauchy product
(which is the usual product of power series),
in which the 
multiplication of coefficients is composition of linear morphisms.
\[ \left(\psum[j] g_j t^j\right) 
\compl \left(\psum f_i t^i\right)
= \psum[k] \left(\sum_{i+j = k} g_j \compl f_i\right) \compl t^k \, . \] 
The functor $\kleisliLS : \categoryAdd \arrow \categoryAddS $ 
maps a morphism in $f \in \categoryAdd$ to the formal power 
series $f t^0$.

\begin{remark} The (left) summability structures of 
  \cite{Ehrhard23-cohdiff} and \cite{Walch23} follow a similar principle,
  except that $\S X$ is intuitively 
  the set of summable pairs $\Spair{x_0}{x_1}$ with 
  $x_0, x_1 \in X$.
  An element $\Spair{x_0}{x_1} \in \S X$ can be seen as 
  a polynomial $x_0 + x_1 t$ with the convention that the formal variable 
  $t$ follows the equation $t^2 = 0$.
  Then the natural transformations $\Sinj_i$ (with $i \in \{0,1\}$)
  $\Ssum, \Slift$ and $\Sswap$ 
  follow the same intuition.
  On the other hand, $\SmonadSum$ maps a polynomial 
  $x_{0,0} + x_{1,0} t + x_{0,1} u + x_{1,1} t u$
  to the polynomial $x_{0,0} + (x_{0,1} + x_{1,0}) v$, since $v^2 = 0$.

  In that case, a morphism $\categoryAddS(X, Y)$ can be seen as 
  a polynomial $f_0 + f_1 t$ with $f_i \in \categoryAdd(X,Y)$,
  and we have the following composition
  \[ (g_0 + g_1 t) \compl (f_0 + f_1 t) 
  = g_0 \compl f_0 + (g_0 \compl f_1 + g_1 \compl f_0) t \] 
  Which is similar to the multiplication on Clifford's \emph{dual numbers}.
  Similarly, an $n$-ary counterpart of (left) 
  summability structures for any 
  $n \in \N^*$ would induce a monad and a category of 
  polynomials quotiented by the equation $t^{n+1} = 0$.
\end{remark}

\begin{proposition}
  The tuple  $\comonadS = (\S, \Ssum, \Slift)$ is a comonad on $\categoryAdd$,
  meaning that the following diagrams of natural transformations commute.
  \[
  \begin{tikzcd}
    \S & {\S^2} & \S \\
    & \S
    \arrow["{\Ssum \S}"', from=1-2, to=1-1]
    \arrow["{\S \Ssum}", from=1-2, to=1-3]
    \arrow[Rightarrow, no head, from=2-2, to=1-1]
    \arrow["\Slift", from=2-2, to=1-2]
    \arrow[Rightarrow, no head, from=2-2, to=1-3]
  \end{tikzcd} \quad
\begin{tikzcd}
	{\S^3} & {\S^2} \\
	{\S^2} & \S
	\arrow["{\Slift \S}"', from=1-2, to=1-1]
	\arrow["{\S \Slift}", from=2-1, to=1-1]
	\arrow["\Slift"', from=2-2, to=1-2]
	\arrow["\Slift", from=2-2, to=2-1]
\end{tikzcd} \]
\end{proposition}

\begin{proof} As usual, we use the joint monicity of the $\Sproj_i$.
  Remember that $\Sproj_i \compl \Slift_X = \Sinj_i \compl \Sproj_i$.
  \[ \Sproj_i \compl \S \Ssum_X \compl \Slift_X = \Ssum_X \compl \Sproj_i \compl \Slift_X 
  = \Ssum_X \compl \Sinj_i \compl \Sproj_i = \Sproj_i \]
  \[ \Sproj_i \compl \Ssum_{\S X} \compl \Slift_X = 
  \Sproj_i \compl \left( \sum_{j \in \N} \Sproj_j \right) \compl \Slift_X
  = \sum_{j \in \N} \Sproj_i \compl \Sproj_j \compl \Slift_X
  = \sum_{j \in \N} \Kronecker i j \Sproj_i = \Sproj_i \]
so by joint monicity of the $\Sproj_i$,
$\S \Ssum_X \compl \Slift_X = \Ssum_{\S X} \compl \Slift_X = \id_X$.
Next, we have
\[ \Sproj_i \compl \Sproj_j \compl \S \Slift_X \Slift_X 
= \Sproj_i \compl \Slift_X \compl \Sproj_j \compl \Slift_X 
= \Sinj_i \compl \Sproj_i \compl \Sinj_j \compl \Sproj_j 
= \begin{cases} \Sinj_i \compl \Sproj_i \text{ if $i=j$} \\
  0 \text{ otherwise} \end{cases} \]
  \[ \Sproj_i \compl \Sproj_j \compl \Slift_{\S X} \compl \Slift_X 
  = \Sproj_i \compl \Sinj_j \compl \Sproj_j \compl \Slift 
  = \Sproj_i \compl \Sinj_j \compl \Sinj_j \compl \Sproj_j 
  = \begin{cases} \Sinj_i \compl \Sproj_i \text{ if $i=j$} \\
    0 \text{ otherwise} \end{cases} \]
So by joint monicity of the $\Sproj_i \compl \Sproj_j$ we have 
$\S \Slift_X \Slift_X = \Slift_{\S X} \compl \Slift_X$.
\end{proof}

\begin{remark}
  We put a strong emphasis on $\categoryAddS$, the Kleisli category of the monad $\S$, 
  as opposed to the Eilenberg-Moore category of the monad $\S$, or 
  the coKleisli and coEilenberg-Moore category of the comonad $\comonadS$.
  The reason is the $\categoryAddS$ is simple to describe in terms of power series
  (or dual numbers, in the binary case of \cite{Ehrhard23-cohdiff}), and 
  is deeply related to LL (see \cref{sec:summable-resource-category,sec:Taylor}).
  Still, those three other categories are interesting and deserve further
  studies, which are postponed to future work.  
\end{remark}

The monad structure and the comonad structure are compatible,
in the sense that they form a $\Sswap$-bimonad, see \cref{def:bimonad}.
Here are some useful observations on $\Sswap$.

\begin{lemma} \label{lemma:equation-Sswap}
  \begin{enumerate}
    \item $\Sswap$ is involutive;
    \item $\Sswap = \Spairing{\S \Sproj_i}$;
    \item $\S \Sproj_i \compl \Sswap = \Sproj_i$.
  \end{enumerate}
\end{lemma}

\begin{proof} First, 
  $\Sproj_i \compl \Sproj_j \compl \Sswap \compl \Sswap 
  = \Sproj_j \compl \Sproj_i \compl \Sswap 
  = \Sproj_i \compl \Sproj_i \compl \id$ so 
  by joint monicity of the $\Sproj_i \compl \Sproj_j$,
  $\Sswap \compl \Sswap = \id$. 
  Furthermore, 
  $\Sproj_i \compl \Sproj_j \compl \Sswap 
  = \Sproj_j \compl \Sproj_i = \Sproj_i \compl \S \Sproj_j$ so by joint monicity of
  the $\Sproj_i$, $\Sproj_j \compl \Sswap = \S \Sproj_j$. That is,
  $\Sswap = \Spairing{\S \Sproj_j}$.
  Finally, $\Sproj_i \compl \Sswap = \S \Sproj_i$ so using the fact that $\Sswap$ is involutive,
  $\Sproj_i = \S \Sproj_i \compl \Sswap$.
\end{proof}

\begin{proposition} \label{prop:Sswap-distributive}
  The natural transformation $\Sswap$ is a distributive law
  $\monadS \comonadS \naturalTrans \comonadS \monadS$ and 
  $\comonadS \monadS \naturalTrans \monadS \comonadS$.
\end{proposition}

\begin{proof}
  By \cref{rem:bimonad-involutive} it is only necessary to show that
  $\Sswap$ is a distributive law
  $\monadS \comonadS \naturalTrans \comonadS \monadS$ since $\Sswap$
  is involutive. The first condition is that $\Sswap$ is a
  distributive law $\monadS \S \naturalTrans \S \monadS$ which
  corresponds to the two diagrams below.  For the sake of readability,
  we write $\monadS$ instead of just $\S$ to make it clear which part
  is playing the role of the monad.
  \[ \begin{tikzcd}
    \S \\
    {\monadS \S} & {\S \monadS}
    \arrow["{\Sinj_0 \S}"', from=1-1, to=2-1]
    \arrow["\Sswap"', from=2-1, to=2-2]
    \arrow["{\S \Sinj_0}", from=1-1, to=2-2]
  \end{tikzcd} \quad 
  \begin{tikzcd}
    {\monadS^2 \S} & {\monadS \S \monadS} & {\S \monadS^2} \\
    {\monadS \S} && {\S \monadS}
    \arrow["{\monadS \Sswap}", from=1-1, to=1-2]
    \arrow["{\Sswap \monadS}", from=1-2, to=1-3]
    \arrow["{\SmonadSum \S}"', from=1-1, to=2-1]
    \arrow["\Sswap"', from=2-1, to=2-3]
    \arrow["{\S \SmonadSum}", from=1-3, to=2-3]
  \end{tikzcd} \]
  We use the joint monicity of the $\Sproj_i \compl \Sproj_j$. Observe that
  \[ \Sproj_i \compl \Sproj_j \compl \Sswap \compl \Sinj_0 = \Sproj_j \compl \Sproj_i \compl \Sinj_0
  = \begin{cases} \Sproj_j \text{ if $i=0$} \\ 0 \text{ otherwise} \end{cases} \]
  \[ \Sproj_i \compl \Sproj_j \compl \S \Sinj_0 = \Sproj_i \compl \Sinj_0 \compl \Sproj_j 
  = \begin{cases} \Sproj_j \text{ if $i=0$} \\ 0 \text{ otherwise} \end{cases} \]
  so by joint monicity of the $\Sproj_i \compl \Sproj_j$,
  $\Sswap \compl \Sinj_0 = \S \Sinj_0$ and the left diagram holds.
  Next,
  \[ \Sproj_i \compl \Sproj_j \compl \Sswap \compl \SmonadSum_{\S X}
  = \Sproj_j \compl \Sproj_i \compl  \SmonadSum_{\S X}
  = \Sproj_j \compl \left(\sum_{k = 0}^i \Sproj_k \compl \Sproj_{i-k} \right) \]
  \begin{align*}
    \Sproj_i \compl \Sproj_j \compl \S \SmonadSum_X \compl \Sswap_{\S X} \compl \S \Sswap_X
    &= \Sproj_i \compl \SmonadSum_X \compl \Sproj_j \compl \Sswap_{\S X} \compl \S \Sswap_X
    \tag*{ naturality of $\Sproj_j$} \\
    &= \Sproj_i \compl \SmonadSum_X \compl \S \Sproj_j \compl \S \Sswap_X 
    \tag*{\cref{lemma:equation-Sswap}} \\
    &= \Sproj_i \compl \SmonadSum_X \compl \S^2 \Sproj_j 
    \tag*{functoriality of $\S$ and \cref{lemma:equation-Sswap}} \\
    &= (\sum_{k=0}^i \Sproj_k \compl \Sproj_{i-k}) \compl \S^2 \Sproj_j \\
    &= \Sproj_j \compl (\sum_{k=0}^i \Sproj_k \compl \Sproj_{i-k})
    \tag*{naturality, thanks to \cref{prop:Spairing-natural}}
  \end{align*}
  So by joint monicity of the $\Sproj_i \compl \Sproj_j$, 
  $\Sswap \compl \SmonadSum_{\S X} = \S \SmonadSum_X \compl \Sswap_{\S X} \compl \S \Sswap_X$
  and the right diagram holds. Next, we show that $\Sswap$ is a
  distributive law $\S \comonadS \naturalTrans \comonadS \S$, that is,
  the two diagrams below commute.  For the sake of readability, we
  write $\comonadS$ instead of just $\S$ to make clear which part is
  playing the role of the comonad.
  \[ \begin{tikzcd}
    {\S \comonadS} & \comonadS \S \\
    & \S
    \arrow["{\Sswap}", from=1-1, to=1-2]
    \arrow["{\S \Ssum}"', from=1-1, to=2-2]
    \arrow["{\Ssum \S}", from=1-2, to=2-2]
  \end{tikzcd} \quad 
  \begin{tikzcd}
    {\S \comonadS} && {\comonadS \S} \\
    {\S \comonadS^2} & {\comonadS \S \comonadS} & {\comonadS^2 \S}
    \arrow["\Sswap", from=1-1, to=1-3]
    \arrow["{\S \Slift}"', from=1-1, to=2-1]
    \arrow["{\Sswap \comonadS}"', from=2-1, to=2-2]
    \arrow["{\comonadS \Sswap}"', from=2-2, to=2-3]
    \arrow["{\Slift \S}", from=1-3, to=2-3]
  \end{tikzcd}\]
  By naturality of $\Ssum$, of $\Sproj_i$, and using the fact 
  that $\S \Sproj_i \compl \Sswap = \Sproj_i$ (see \cref{lemma:equation-Sswap}), we have
  \[ \Sproj_i \compl \Ssum_{\S X} \compl \Sswap = \Ssum_{\S X} \compl 
  \S \Sproj_i \compl \Sswap = \Ssum_{\S X} \compl \Sproj_i 
  = \Sproj_i \compl \S \Ssum_{X} \]
  so by joint monicity of the $\Sproj_i$, $\Ssum_{\S X} \compl \Sswap_X = \S \Ssum_X$ 
  and the left diagram holds.
  Next, we have 
  \begin{align*}
    \Sproj_i \compl \Sproj_j \compl \S \Sswap_X \compl \Sswap_{\S X} \compl \S \Slift_X
    &= \Sproj_i \compl \Sswap_X \compl \Sproj_j \compl  \Sswap_{\S X} \compl \S \Slift_X \\
    &= \S \Sproj_i \compl \S \Sproj_j \compl \S \Slift_X \\
    &= \S (\Sproj_i \compl \Sproj_j \compl \Slift) \\
    &= \begin{cases} \S 0 = 0 \text{ if $i \neq j$} \\ \S \Sproj_i \text{ otherwise} \end{cases}
  \end{align*}
  \[ \Sproj_i \compl \Sproj_j \compl \Slift_{\S X} \compl \Sswap_X 
  = \begin{cases} 0 \text{ if $i \neq j$} \\ \Sproj_i \compl \Sswap = \S \Sproj_i 
    \text{ otherwise} \end{cases} \]
  so by joint monicity of the $\Sproj_i \compl \Sproj_j$ we have $\S \Sswap_X \compl 
  \Sswap_{\S X} \compl \S \Slift_X = \Slift_{\S X} \compl \Sswap_X$ and the right diagram 
  holds.
\end{proof}

\begin{theorem} The monad $\monadS $ and 
  the comonad $\comonadS$ form a $\Sswap$-bimonad on
  $\categoryAdd$.
\end{theorem}

\begin{proof} The first three diagrams expressing that $\monadS$ and
  $\comonadS$ are a $\Sswap$-bimonad are as follows.
\[ \begin{tikzcd}
	{\S^2} & \S \\
	{\S } & \idfun
	\arrow["\SmonadSum"', from=1-1, to=2-1]
	\arrow["\Ssum"', from=2-1, to=2-2]
	\arrow["{\Ssum \S}", from=1-1, to=1-2]
	\arrow["\Ssum", from=1-2, to=2-2]
\end{tikzcd} \quad
\begin{tikzcd}
	\idfun & \S \\
	\S & {\S^2}
	\arrow["{\Sinj_0}", from=1-1, to=1-2]
	\arrow["\Slift", from=1-2, to=2-2]
	\arrow["{\Sinj_0}"', from=1-1, to=2-1]
	\arrow["{\Sinj_0 \S}"', from=2-1, to=2-2]
\end{tikzcd} \quad 
\begin{tikzcd}
	\idfun & \S \\
	& \idfun
	\arrow["{\Sinj_0}", from=1-1, to=1-2]
	\arrow["\Ssum", from=1-2, to=2-2]
	\arrow[Rightarrow, no head, from=1-1, to=2-2]
\end{tikzcd} \]
The left diagram holds thanks to the computation below that relies on 
\ref{ax:pa} and \cref{prop:double-distributivity}.
\[ \Ssum_X \compl \SmonadSum_X = \sum_{n=0}^{\infty} \left( 
  \sum_{k = 0}^n \Sproj_k \compl \Sproj_{n-k} \right)
= \sum_{(i, j) \in \N^2} \Sproj_i \compl \Sproj_j = 
\left( \sum_{i \in \N} \Sproj_i \right) \compl \left( \sum_{j \in \N} \Sproj_j \right) 
= \Ssum \compl \Ssum \, . \]
The center diagram holds by a straightforward computation using the joint monicity
of the $\Sproj_i \compl \Sproj_j$ and some case analysis. The rightmost diagram 
holds by \ref{ax:unary}. The last diagram is the following, 
where $\Slift \ntcomph \Slift$ and $\SmonadSum \ntcomph \SmonadSum$
use the horizontal composition of natural transformations, see \cref{eq:horizontal-composition}.
\[ \begin{tikzcd}
  {\S \S} & \S & \S \S \\
  \S \S \S \S && \S \S \S \S
  \arrow["\SmonadSum", from=1-1, to=1-2]
  \arrow["\Slift", from=1-2, to=1-3]
  \arrow["{\Slift \ntcomph \Slift}"', from=1-1, to=2-1]
  \arrow["{\S \Sswap \S}"', from=2-1, to=2-3]
  \arrow["{\SmonadSum \ntcomph \SmonadSum}"', from=2-3, to=1-3]
\end{tikzcd} \]
We use the joint monicity of the $\Sproj_i \compl \Sproj_j$.
We check the top path first.
\[ \Sproj_i \compl \Sproj_j \compl \Slift_X \compl \SmonadSum_X
= \begin{cases} \Sproj_i \compl \SmonadSum_X = \sum_{k=0}^i \Sproj_k \compl \Sproj_{i-k}
\text{ if $i=j$} \\ 0 \text{ otherwise.} \end{cases} \]
For the bottom part, we first compute $\Sproj_i \compl \Sproj_j \compl 
(\SmonadSum \ntcomph \SmonadSum)_X$.
\begin{align*}
  \Sproj_i \compl \Sproj_j \compl (\SmonadSum \ntcomph \SmonadSum)_X
  &= \Sproj_i \compl \Sproj_j \compl \S \SmonadSum_X \compl \SmonadSum_{\S \S X} \\
  &= \Sproj_i \compl \SmonadSum_X \compl \Sproj_j \compl \SmonadSum_{\S \S X} \\
  &= \left(\sum_{k=0}^i \Sproj_k \compl \Sproj_{i-k} \right) \compl
  \left( \sum_{l=0}^j \Sproj_l \compl \Sproj_{j-l} \right) \\
  &\sumsub \sum_{k \in \interval{0}{i}, l \in \interval{0}{j}} \Sproj_k \compl \Sproj_{i-k}
  \compl \Sproj_l \compl \Sproj_{j-l}
\end{align*}
Now observe that 
\begin{align*}
  \Sproj_{i_4} \compl \Sproj_{i_3} \compl \Sproj_{i_2} \compl \Sproj_{i_1} \compl
  \S \Sswap_{\S X} \compl (\Slift \ntcomph \Slift)_X
  &= \Sproj_{i_4} \compl \Sproj_{i_3} \compl \Sproj_{i_2} \compl \Sproj_{i_1} \compl
  \S \Sswap_{\S X} \compl \S \S \Slift_X \compl \Slift_{\S X} \\
  &=  \Sproj_{i_4} \compl \Sproj_{i_3} \compl \Sproj_{i_2} \compl
  \Sswap_{\S X} \compl \Sproj_{i_1} \compl \S \S \Slift_X \compl \Slift_{\S X} \\
  &= \Sproj_{i_4} \compl \Sproj_{i_2} \compl \Sproj_{i_3} \compl  \Sproj_{i_1}  \compl
   \S \S \Slift_X  \compl\compl \Slift_{\S X} \\
   &= \Sproj_{i_4} \compl \Sproj_{i_2} \compl  \Slift_X 
   \compl \Sproj_{i_3} \compl  \Sproj_{i_1} \compl \Slift_{\S X} \\
   &= \begin{cases} \Sproj_{i_2} \compl \Sproj_{i_1} \text{ if $i_1 = i_3$ and $i_2 = i_4$}
    \\ 0 \text{ otherwise.} \end{cases}
\end{align*}
Thus,
\begin{align*}
\Sproj_i \compl \Sproj_j \compl (\SmonadSum \ntcomph \SmonadSum)_X
\compl \S \Sswap_{\S X} \compl (\Slift \ntcomph \Slift)_X 
&= \left( \sum_{k \in \interval{0}{i}, l \in \interval{0}{j}} \Sproj_k \compl \Sproj_{i-k}
\compl \Sproj_l \compl \Sproj_{j-l} \right) \compl 
\S \Sswap_{\S X} \compl (\Slift \ntcomph \Slift)_X \\
&\sumsub \sum_{k \in \interval{0}{i}, l \in \interval{0}{j}} \Sproj_k \compl \Sproj_{i-k}
\compl \Sproj_l \compl \Sproj_{j-l} \compl 
\S \Sswap_{\S X} \compl (\Slift \ntcomph \Slift)_X \\
&= \begin{cases} \sum_{k=0}^{i} \Sproj_k \compl \Sproj_{i-l} \text{ if $i = j$}
  \\ 0 \text{ otherwise.} \end{cases}
\end{align*}
We conclude that the diagram commutes.
\end{proof}

\begin{remark}
The following diagram, called the \emph{Yang-Baxter equation} 
commutes.
\[ 
\begin{tikzcd}
	{\S \S \S} & {\S \S \S} & {\S \S \S} \\
	{\S \S \S} & {\S \S \S} & {\S \S \S}
	\arrow["{\Sswap \S}", from=1-1, to=1-2]
	\arrow["{\S \Sswap}"', from=1-1, to=2-1]
	\arrow["{\S \Sswap}", from=1-2, to=1-3]
	\arrow["{\Sswap \S}", from=1-3, to=2-3]
	\arrow["{\Sswap \S}"', from=2-1, to=2-2]
	\arrow["{\S \Sswap}"', from=2-2, to=2-3]
\end{tikzcd} \] 
This is easy to check using the joint monicity of the 
$\Sproj_i \compl \Sproj_j \compl \Sproj_k$.
This means that $\Sswap$ is a \emph{local prebraiding}, following 
the terminology of \cite{Mesablishvili11}. In particular, this 
commutation imply that $\S \S$ can be equipped with a bimonad 
structure by doubling the bimonad $\S$. This 
process is described in section 6.8 of \cite{Mesablishvili11}.
This doubling is very similar to the commutation between 
two derivative operations observed in the coherent differential PCF of 
\cite{Ehrhard22-pcf}.
\end{remark}

\section{$\Sigma$-additivity in models of LL}
\label{sec:summable-resource-category}

We now assume that $\categoryLL$ is a $\Sigma$-additive category equipped with a
$\Sigma$-summability structure $(\S, \vect \Sproj, \Ssum)$. Since
$\categoryLLAdd = \categoryLL$, we write the composition
of $f \in \categoryLL(X, Y)$ with $g \in \categoryLL(Y, Z)$ as 
$g \compl f$. Then as seen in \cref{sec:bimonad}, $\S$ is a bimonad 
on $\categoryLL$. The category $\categoryLL$ is typically a model of LL, 
but needs not be a full-fledged model, so we will try to be very precise 
about the assumptions we use.

\begin{notation} We will use $\kleisliS$ for the Kleisli category of 
  the monad $\monadS$ on $\categoryLL$, and 
  $\kleisliLS$ for the functor $\categoryLL \arrow \kleisliS$ defined as in 
  \cref{sec:monad-dl}.
\end{notation}

This section describes the interaction between the 
$\Sigma$-summability structure and the LL structure of $\categoryLL$.
It explains why this interaction provides a symmetric monoidal 
closed structure to $\kleisliS$, as well as a 
categorical product (it is then shown in 
\cref{sec:cartesian-Taylor} that $\kleisliS$ is a model of LL).

\subsection{Interaction with the monoidal structure}
\label{sec:summability-tensor}
Assume that $\categoryLL$ is a
symmetric monoidal category, see \cref{sec:monoidal-functor}. We write 
$(\tensor, 1, \tensorUnitL, \tensorUnitR, \tensorAssoc, \tensorSym)$ for 
the symmetric monoidal structure. This section details how the sum interacts 
with the monoidal product.

\begin{definition} \label{def:sum-sm}
A symmetric monoidal $\Sigma$-additive category is a category that is both $\Sigma$-additive
and symmetric monoidal, and such that sum distributes over $\tensor$. That is, 
for all indexed families $\family{f_a \in \categoryLL(X_0, Y_0)}$
and $\family<B>[b]{g_b \in \categoryLL(X_1, Y_1)}$ and for all
$f \in \categoryLL(X_0, Y_0)$ and $g \in \categoryLL(X_1, Y_1)$,
  \[ \left( \sum_{a \in A} f_a \right) \sm g \sumsub \sum_{a \in A} f_a \sm g  
  \quad \quad \quad
  f \sm \left( \sum_{b \in B} g_b \right)  \sumsub \sum_{b \in B} f \sm g_i \, .\]
\end{definition}

Similarly to additivity in \cref{prop:additive}, the compatibility between 
the sum and the monoidal product $\tensor$ can be written either as a property of the $\Sigma$-monoid 
structure, or as a property of the $\Sigma$-summability structure.

\begin{definition} \labeltext{($\S \sm$-dist)}{ax:S-sm-dist}
  A $\Sigma$-summability structure satisfies \ref{ax:S-sm-dist} if 
  for all objects $X_0, X_1$ one has
  $X_0 \tensor 0 = 0$, $0 \tensor X_1 = 0$, and both 
  $\sequence{\Sproj_i \sm X_1 \in \categoryLL(\S X_0 \sm X_1, X_0 \sm X_1)}$
  and 
  $\sequence{X_0 \sm \Sproj_i \in \categoryLL(X_0 \sm \S X_1, X_0 \sm X_1)}$
  are summable,  with respective sums $\Ssum \sm X_1$ and $X_0 \sm \Ssum$.
\end{definition}

It follows from \ref{def:summability-structure-2} that  
\ref{ax:S-sm-dist} implies the existence of the following witnesses.
  \begin{equation}
    \begin{split}
      \SstrL_{X_0, X_1} = \Spairing{\Sproj_i \sm X_1} 
  \in \categoryLL(\S X_0 \sm X_1, \S (X_0 \sm X_1)) \text{ such that } 
  \Ssum \compl \SstrL = (\Ssum \sm X_1) \\ 
  \SstrR_{X_0, X_1} = \Spairing{X_0 \sm \Sproj_i} 
  \in \categoryLL(X_0 \sm \S X_1, \S (X_0 \sm X_1)) \text{ such that } 
  \Ssum \compl \SstrR = (X_0 \sm \Ssum)
    \end{split}
  \end{equation}
  They are natural transformations thanks to \cref{prop:Spairing-natural}.

\begin{remark} Because the monoidal product is symmetric, only one of the summability 
  assumption of \ref{ax:S-sm-dist} is necessary, and
  $\SstrR$ can be defined from $\SstrL$ as 
  $\SstrR = \S \smsym \compl \SstrL \compl \smsym$ and vice versa.
\end{remark}

\begin{proposition}  \label{prop:S-sm-dist-equivalent}
  For any $\Sigma$-additive category $\categoryLL$ equipped with a summability structure, 
  the following are equivalent. \begin{enumerate}
  \item $\categoryLL$ is a symmetric monoidal $\Sigma$-additive category;
  \item the left summability structure satisfies \ref{ax:S-sm-dist}.
  \end{enumerate}
\end{proposition}

\begin{proof} The proof is very similar to that of \cref{prop:additive}. The implication 
  $(1) \Rightarrow (2)$ is a direct consequence of the summability of the $\Sproj_i$.
  For $(2) \Rightarrow (1)$, assume that $\vect f = \family{f_a
  \in \category(X, Y)}$ is summable. 
  Then, there is an injection $\phi : A \arrow \N$ and 
  $\Sfamily{\vect f}{\phi} \in \category(X, \S Y)$
  such that $\sequence{\Sproj_i \compl \Sfamily{\vect f}{\phi}} 
  = \Famact{\phi} \vect f$.
  Then we can check that $\SstrL \compl (\Sfamily{\vect f}{\phi} \tensor X_1)$ 
  is a witness for $\Famact{\phi} \family{f_a \tensor X_1}$, using the 
  fact that $0 \tensor X_1 = 0$. 
  So by \cref{prop:summability-alt},
  $\family{f_a \tensor X_1}$ is summable with sum 
  \[ \Ssum \compl \SstrL \compl (\Sfamily{\vect f}{\phi} \tensor X_1)
  = (\Ssum \compl \Sfamily{\vect f}{\phi}) \tensor X_1 
  = \left( \sum_{a \in A} f_a \right) \tensor X_1 \, . \qedhere \]
\end{proof}

We now assume that $\categoryLL$ 
satisfies one of the equivalent assumptions 
of \cref{prop:S-sm-dist-equivalent}.
\begin{corollary} For all indexed families 
  $\family{f_a \in \categoryLL(X_0, Y_0)}$
  and $\family<B>[b]{g_b \in \categoryLL(X_1, Y_1)}$,
  \[ \left(\sum_{a \in A} f_a \right) \sm \left(\sum_{b \in B} g_b \right) \sumsub
  \sum_{(a, b) \in A \times B} (f_a \sm g_b) \, .
  \]
\end{corollary}

\begin{proof} It follows directly from
  double distributivity (\cref{prop:double-distributivity}).
\end{proof}

\begin{theorem} The natural transformations $\SstrR$ is
  a strength for the monad $\monadS$ (see \cref{def:strength}).
  Furthermore, the following diagram commutes.
  \[ \begin{tikzcd}
    {\S (\S X \tensor Y)} & {\S X \tensor \S Y} & {\S (X \tensor \S Y)} \\
    {\S^2 (X \tensor Y)} && {\S^2 (X \tensor Y)}
    \arrow["{\SstrR_{\S X, Y}}"', from=1-2, to=1-1]
    \arrow["{\S \SstrL_{X, Y}}"', from=1-1, to=2-1]
    \arrow["{\SstrL_{X, \S Y}}", from=1-2, to=1-3]
    \arrow["{\S \SstrR_{X, Y}}", from=1-3, to=2-3]
    \arrow["{\Sswap_{X \tensor Y}}"', from=2-1, to=2-3]
  \end{tikzcd} \]
  This implies that $\monadS$ equipped with $\SstrR$ is a commutative
  monad, see \cref{def:commutative-monad}.
\end{theorem}

\begin{proof} Let us check that $\SstrR$ is a left strength. We first show the 
  compatibility with the monoidal structure.
  \[ \begin{tikzcd}
    {1 \tensor \S X} & {\S(1 \tensor X)} \\
    & {\S X}
    \arrow["\SstrR", from=1-1, to=1-2]
    \arrow["{\S \tensorUnitL_X}", from=1-2, to=2-2]
    \arrow["{\tensorUnitL_{\S X}}"', from=1-1, to=2-2]
  \end{tikzcd} \quad 
  \begin{tikzcd}
    {(X \tensor Y) \tensor \S Z} && {\S ((X \tensor Y) \tensor Z)} \\
    {X \tensor (Y \tensor \S Z)} & {X \tensor \S (Y \tensor Z)} & {\S (X \tensor (Y \tensor Z))}
    \arrow["{\tensorAssoc_{X,Y,\S Z}}"', from=1-1, to=2-1]
    \arrow["{X \tensor \SstrR_{Y,Z}}"', from=2-1, to=2-2]
    \arrow["{\SstrR_{X, Y \tensor Z}}"', from=2-2, to=2-3]
    \arrow["{\SstrR_{X \tensor Y, Z}}", from=1-1, to=1-3]
    \arrow["{\S \tensorAssoc_{X, Y, Z}}", from=1-3, to=2-3]
  \end{tikzcd} \]
  The two diagrams above are just routinely checked by joint monicity of the 
  $\Sproj_i$, their naturality (recall that $\categoryLL = \categoryLLAdd$) 
  and the naturality of $\tensorUnitL$ and $\tensorAssoc$. Then we show the 
  compatibility with the monad structure:
  \[ \begin{tikzcd}
    {X \tensor Y} \\
    {X \tensor \S Y} & {\S (X \tensor Y)}
    \arrow["{X \tensor \Sinj_0}"', from=1-1, to=2-1]
    \arrow["\SstrR"', from=2-1, to=2-2]
    \arrow["{\Sinj_0}", from=1-1, to=2-2]
  \end{tikzcd} \quad 
  \begin{tikzcd}
    {X \tensor \S^2 Y} & {\S (X \tensor \S Y)} & {\S^2 (X \tensor Y)} \\
    {X \tensor \S Y} && {\S (X \tensor Y)}
    \arrow["{X \tensor \SmonadSum_Y}"', from=1-1, to=2-1]
    \arrow["{\SstrR_{X, \S Y}}", from=1-1, to=1-2]
    \arrow["{\S \SstrR_{X,Y}}", from=1-2, to=1-3]
    \arrow["{\SmonadSum_{X \tensor Y}}", from=1-3, to=2-3]
    \arrow["{\SstrR_{X, Y}}"', from=2-1, to=2-3]
  \end{tikzcd} \] 
  We also use the joint monicity of the $\Sproj_i$.
  The left diagram is a consequence of the equation $X \tensor 0 = 0$:
  \[ \Sproj_i \compl \SstrR_{X, Y} \compl (X \tensor \Sinj_0)
  = X \tensor (\Sproj_i \compl \Sinj_0)
  = \begin{cases} X \tensor 0 = 0 \text{ if $i \neq 0$} \\
    X \tensor Y \text{ otherwise}
  \end{cases} 
  = \Sproj_i \compl \Sinj_0 \, .\]
  The right diagram is a consequence of the distributivity of the sum over $\tensor$:
  \[ \Sproj_i \compl \SmonadSum_{X \tensor Y} \compl \S \SstrR_{X, Y}
  \compl \SstrR_{X, \S Y} 
  = \sum_{k=0}^i \Sproj_{k} \compl \Sproj_{i-k} \compl \S \SstrR_{X, Y}
  \compl \SstrR_{X, \S Y} \\
  = \sum_{k=0}^i \Sproj_{k} \compl \SstrR_{X, Y} \compl \Sproj_{i-k}
  \compl \SstrR_{X, \S Y} \\
  = \sum_{k=0}^i (X \tensor \Sproj_k \compl \Sproj_{i-k}) \]
  \[ \Sproj_i \compl \SstrR \compl (X \tensor \SmonadSum_Y) =
  (X \tensor \Sproj_i) \compl (X \tensor \SmonadSum_Y) =
  X \tensor (\sum_{k=0}^i \Sproj_k \compl \Sproj_{i-k}) \] 
  So the monad $\monadS$ is strong. Finally, saying that the monad is commutative 
  means that the following diagram commutes.
  \begin{equation} \label{eq:commutative-strength}
  \begin{tikzcd}
	  {\S (\S X \tensor Y)} & {\S X \tensor \S Y} & {\S (X \tensor \S Y)} \\
	  {\S^2 (X \tensor Y)} & {\S (X \tensor Y)} & {\S^2 (X \tensor Y)}
	  \arrow["{\SstrR_{\S X, Y}}"', from=1-2, to=1-1]
	  \arrow["{\S \SstrL_{X, Y}}"', from=1-1, to=2-1]
	  \arrow["{\SmonadSum_{X \tensor Y}}"', from=2-1, to=2-2]
	  \arrow["{\SstrL_{X, \S Y}}", from=1-2, to=1-3]
	  \arrow["{\S \SstrR_{X, Y}}", from=1-3, to=2-3]
	  \arrow["{\SmonadSum_{X \tensor Y}}", from=2-3, to=2-2]
  \end{tikzcd}
  \end{equation}
  This is a consequence of the fact that $\SmonadSum_X \compl \Sswap_X 
  = \SmonadSum_X$ and of a stronger property, which as announced is the commutation 
  of 
  \[ \begin{tikzcd}
    {\S (\S X \tensor Y)} & {\S X \tensor \S Y} & {\S (X \tensor \S Y)} \\
    {\S^2 (X \tensor Y)} && {\S^2 (X \tensor Y)}
    \arrow["{\SstrR_{\S X, Y}}"', from=1-2, to=1-1]
    \arrow["{\S \SstrL_{X, Y}}"', from=1-1, to=2-1]
    \arrow["{\SstrL_{X, \S Y}}", from=1-2, to=1-3]
    \arrow["{\S \SstrR_{X, Y}}", from=1-3, to=2-3]
    \arrow["{\Sswap_{X \tensor Y}}"', from=2-1, to=2-3]
  \end{tikzcd} \]
  This is proven by the joint monicity of the 
  $\Sproj_i \compl \Sproj_j$:
  \[ \Sproj_i \compl \Sproj_j \compl \Sswap \compl \S \SstrL \compl \SstrR
   = \Sproj_j \compl \Sproj_i \compl \S \SstrL \compl \SstrR
   = \Sproj_j \compl \SstrL \compl \Sproj_i \compl \SstrR 
   = (\Sproj_j \tensor Y) \compl (\S X \tensor \Sproj_j) 
   = \Sproj_j \tensor \Sproj_i \] 
   \[ \Sproj_i \compl \Sproj_j \compl \S \SstrR \compl \SstrL 
   = \Sproj_i \compl \SstrR \compl \Sproj_j \compl \SstrL 
   = (X \tensor \Sproj_i) \compl (\Sproj_j \tensor \S Y)
   = \Sproj_j \tensor \Sproj_i\qedhere \]
\end{proof}

\begin{remark} \label{rem:extension-sm-KleisliS}
As discussed in \cref{sec:commutative-monad},
the commutative monad $\monadS$ is then
canonically a lax symmetric monoidal monad in the sense of \cref{def:monoidal-monad}. 
The natural transformation
$\Sdist_{X_0,X_1} \in \categoryLL(\S X_0 \sm \S X_1, \S(X_0 \sm X_1))$ 
is defined as the diagonal of 
the diagram in \cref{eq:commutative-strength}.
As stated in \cref{thm:sm-extension}, this implies that $\kleisliS$
(the Kleisli category of the monad $\monadS$ on $\categoryLL$) 
is a symmetric monoidal category. The monoidal product of 
$f \in \kleisliS(X_0, Y_0)$ with $g \in \kleisliS(X_1, Y_1)$ is 
defined as the following composition.
\[
\begin{tikzcd}
  f \extension{\tensor} g \defEq 
	{X_0 \tensor X_1} & {\S X_0 \tensor \S X_1} & {\S (X_0 \tensor X_1)}
	\arrow["{f \tensor g}", from=1-1, to=1-2]
	\arrow["{\Sdist_{X_0, X_1}}", from=1-2, to=1-3]
\end{tikzcd} \] 
and $\extension{\tensor}$ is an extension of $\tensor$ to $\kleisliS$.
Observe that 
\[ \Sdist_{X_0,X_1} = \Spairing{\sum_{k=0}^{i} \Sproj_k \sm \Sproj_{i-k}} \, . \]
Let $f_i = \Sproj_i \compl f$ and $g_j = \Sproj_j \compl g$. Then 
the equation above implies that
\[ f \extension{\tensor} g
= \Spairing{\sum_{k=0}^i f_k \sm g_{i-k}} \, . \] 
Coming back to our power series notation
of \cref{sec:bimonad}, this equation reads
\[ \left(\psum f_i t^i\right) \extension{\tensor} \left(\psum[j] g_j t^j\right) 
= \psum[k] (\sum_{i+j = k} f_i \tensor g_j) t^k \] 
which is again a Cauchy product of power series, where the multiplication of scalars is now
tensor product instead of composition. 
\end{remark}

\subsection{Interaction with the closed structure}

\label{sec:summability-closure-LL}

We now assume that $\categoryLL$ is closed with respect to its monoidal
product.
This means that any pair of objects $X$, $Y$ of $\categoryLL$ has an
internal hom, which is a pair $(X \linarrow Y, \ev_X)$ where
$\ev_X \in \categoryLL((X \linarrow Y) \tensor X, Y)$, characterized
by the following universal property: for any
$f \in \categoryLL(Z \tensor X, Y)$, there is a unique
$\cur_X(f) \in \categoryLL(Z, X \linarrow Y)$ (called the Curry
transpose of $f$) such that
  \[ \ev_X (\cur_X(f) \tensor X) = f \, .\] 
 This assumption implies that
 $\cur_X: \categoryLL(Z \tensor X, Y) \arrow \categoryLL(Z, X \linarrow Y)$
  is a bijection whose inverse is given by
  \[ \uncur_X (g) = \ev_X \compl (g \tensor X) \, . \]
  We label $\cur_X$ and $\ev_X$ with the object $X$ (as opposed to the 
  objects $Z$ and $Y$ that are always kept implicit) because this choice makes 
  some situations clearer, but 
  we will often keep the object $X$ itself implicit and simply write $\cur$ and $\ev$.
  
  Remember that if \(f\in\cL(X_2,X_1)\) and \(g\in\cL(Y_1,Y_2)\), and if the
  pairs \((X_i,Y_i)\) have internal homs 
  $(X_i \linarrow Y_i, \Evlin_{X_i})$ for \(i=1,2\) then it is
  possible to define \(f\Limpl g\in\cL(X_1\Limpl Y_1,X_2\Limpl Y_2)\),
  turning \(\Limpl\) into a functor \(\Op\cL\times\cL\to\cL\).
  Explicitly,
  \begin{equation} 
    f\Limpl g =\Curlin_{X_2}(g\Compl\Evlin_{X_1}\Compl((X_1\Limpl Y_1)\Times f)) \, .
  \end{equation}
  Intuitively, $f \linarrow g$ maps an element $h \in X_1 \linarrow Y_1$ 
  to $g \compl h \compl f \in X_2 \linarrow Y_2$.

  If $X$ is such that an internal hom of $X$, $Y$ exists for all $Y$, then there 
  is an adjunction 
  $\_ \tensor X \dashv X \linarrow \_$,
  of unit $\cur_X(\id_{Y \tensor X}) \in \categoryLL(Y, X \linarrow (Y \tensor X))$ and 
  co-unit $\ev_X \in \categoryLL((X \linarrow Y) \tensor X, Y)$. Explicitly,
  the functor $X \linarrow \_$ maps a morphism $f \in \categoryLL(Y, Z)$ to 
  $\cur_X(f \compl \ev_X) \in \categoryLL(X \linarrow Y, X \linarrow Z)$.
  The bijections $\cur_X$ and $\uncur_X$ are natural with respect to this functor:
  \begin{equation*} \label{eq:cur-natural}
    \begin{split}
    \cur_X(f \compl g \compl (h \tensor X)) &= (X \linarrow f) \compl \cur_X(g) \compl h \\
    \uncur_X((X \linarrow f) \compl g \compl h) &= f \compl \uncur_X(g) \compl (h \tensor X)
    \end{split}
  \end{equation*}

\begin{definition} \label{def:sum-closure}
 A symmetric monoidal closed $\Sigma$-additive category $\categoryLL$ is a symmetric 
 monoidal $\Sigma$-additive category which is closed with respect to its symmetric monoidal 
 structure and such 
 that the sum is compatible with the internal homs: 
 for any family $\family{f_a \in \categoryLL(X, Y)}$, 
 \[ \sum_{a \in A} \cur(f_a) 
 \sumiff \cur \left(\sum_{a \in A} f_a \right) \, .\]
\end{definition}

This definition is quite intuitive but contains redundant information so 
and can be pruned out.
\begin{proposition} \label{prop:sum-closure}
  For any symmetric monoidal $\Sigma$-additive category $\categoryLL$ that is closed, 
  \[ \sum_{a \in A} \cur(f_a) \sumsub \cur \left(\sum_{a \in A} f_a \right) \, .\]
  Consequently, the following assertions are equivalent.
  \begin{enumerate}
  \item $\categoryLL$ is a symmetric monoidal closed $\Sigma$-additive category;
  \item for any summable family $\family{f_a \in \categoryLL(Z \tensor X, Y)}$, 
  $\family{\cur(f_a)}$ is summable.
  \end{enumerate}
\end{proposition}

\begin{proof}
 This is a consequence of additivity, distributivity of sums over $\tensor$, and of the 
 fact that $\cur$ is a bijection: we have
 \[ \cur^{-1} \left( \sum_{a \in A} \cur(f_a) \right) 
 = \ev \compl \left( \sum_{a \in A} \cur(f_a) \tensor X \right)
 \sumsub \sum_{a \in A} \ev \compl (\cur(f_a) \tensor X) 
 = \sum_{a \in A} f_a \] 
 and applying $\cur$ on this equality gives that 
 $\sum_{a \in A} \cur(f_a) 
 \sumsub \cur \left(\sum_{a \in A} f_a \right)$. The equivalence 
 directly follows.
\end{proof}

We now assume that $\categoryLL$ is equipped with a summability
structure that satisfies \ref{ax:S-sm-dist}. Again, the compatibility
between sums and the monoidal closedness can be written as a property of the
summability structure.

\begin{definition} \labeltext{($\S\tensor$-fun)}{ax:S-fun}
  The $\Sigma$-summability structure satisfies \ref{ax:S-fun} 
  if $\sequence{X \linarrow \Sproj_i}$ is summable. 
\end{definition}

Contrarily to \cref{prop:additive,prop:S-sm-dist-equivalent},
assumptions such as $X \linarrow 0 = 0$ and $\cur(0) = 0$ are not
needed here: they result from \cref{prop:sum-closure}.

\begin{proposition} \label{prop:S-fun-equivalent}
  The following conditions are equivalent: \begin{enumerate}
    \item the summability structure satisfies \ref{ax:S-fun};
    \item if $\family{f_a \in \categoryLL(X, Y)}$ is summable then
    $\family{X \linarrow f_a}$ is summable;
    \item if $\family{f_a \in \categoryLL(Z \tensor X, Y)}$ 
    is summable then $\family{\cur(f_a)}$ is summable;
    \item $\categoryLL$ is a symmetric monoidal closed $\Sigma$-additive category.
  \end{enumerate}
\end{proposition}

\begin{proof} The proof of $(1) \Leftrightarrow (2)$ is similar to that
  of \cref{prop:additive} and \cref{prop:S-sm-dist-equivalent}.
  The implication $(2) \Rightarrow (3)$ follows from left distributivity 
  and the fact that 
  $\cur(f_a) = (X \linarrow f_a) \compl \cur(\id_{Y \tensor X})$ by 
  naturality of $\cur$. The implication $(3) \Rightarrow (2)$ is a direct consequence 
  of the definition of $X \linarrow \_$. Finally, the equivalence 
  $(3) \Leftrightarrow (4)$ results from \cref{prop:sum-closure}.
\end{proof}

Axiom \ref{ax:S-fun} is an isomorphism property. 
We define a natural transformation as follows, only by assuming \ref{ax:S-sm-dist}.
\begin{equation} \label{eq:Sclos}
  \Sclos_{X} \defEq \cur((\S \ev) \compl \SstrL_{Y \linarrow X, Y}) \in 
\categoryLL(\S (Y \linarrow X), X \linarrow \S Y) \, .
\end{equation}
We will show in \cref{prop:Sclos-inverse} that \ref{ax:S-fun} holds
if and only if $\Sclos$ is a natural isomorphism.
This definition of the distributivity of $\S$ over 
$X \linarrow \_$ following from the distributivity over 
$\_ \tensor X$ is standard and is called a \emph{pointwise structure},
see~\cite{Kock71}.

\begin{remark} \label{rem:lin-closure-monic}
  Monicity is preserved by right adjoint functors, and hence the
morphisms $\sequence{X \linarrow \Sproj_i}$ are jointly monic. This fact
can also be checked by hand in a rather straightforward way using the fact
that $\cur$ is a bijection and is natural. Thus, $\Sclos$ is characterized 
by the equations given in \cref{prop:SstrL-prop} below.
\end{remark}

\begin{proposition} \label{prop:SstrL-prop} 
  We have $(X \linarrow \Sproj_i) \compl \Sclos = \Sproj_i$
  and $(X \linarrow \Ssum) \compl \Sclos = \Ssum$.
\end{proposition}

\begin{proof}
We have
\begin{equation*}
  \begin{split}
    (X \linarrow \Sproj_i) \compl \Sclos 
    &= \cur(\Sproj_i \compl (\S \ev) \compl \SstrL_{X \linarrow Y, X}) 
    \quad \text{by naturality of $\cur$}\\
    &= \cur(\ev \compl \Sproj_i \compl \SstrL_{X \linarrow Y, X}) 
    \quad \text{by naturality of $\Sproj_i$} \\
    &= \cur(\ev \compl (\Sproj_i \tensor X)) \\
    &= \cur(\uncur(\Sproj_i)) = \Sproj_i \, .
  \end{split}
\end{equation*}
The equality $(X \linarrow \Ssum) \compl \Sclos = \Ssum$ is proved
similarly, using
$\Ssum \compl \SstrL_{X_0, X_1} = (\Ssum \tensor X_1)$.
\end{proof}

\begin{proposition} \label{prop:Sclos-inverse} 
  The following are equivalent: \begin{enumerate}
    \item \ref{ax:S-fun} holds;
    \item $\Sclos$ is an isomorphism.
  \end{enumerate}
  And then $(\Sclos)^{-1} = \Spairing{X \linarrow \Sproj_i}$.
\end{proposition}
\begin{proof} 
  $(1) \Rightarrow (2)$: assume that $\sequence{X \linarrow \Sproj_i}$ 
  is summable. Then by \cref{prop:SstrL-prop}
  \[ \Sproj_i \comp \Spairing{X \linarrow \Sproj_i} \comp \Sclos
    = (X \linarrow \Sproj_i) \comp \Sclos = \Sproj_i \]
    \[ (X \linarrow \Sproj_i) \comp \Sclos \comp 
    \Spairing{X \linarrow \Sproj_i}  
    = \Sproj_i \comp \Spairing{X \linarrow \Sproj_i}
    = X \linarrow \Sproj_i \]
    It follows that
    $\Spairing{X \linarrow \Sproj_i}
    \comp \Sclos = \id_{\S (X \linarrow Y)}$ 
    by joint monicity of the $\Sproj_i$, and that
    $\Sclos \comp \Spairing{X \linarrow \Sproj_i} 
    = \id_{X \linarrow \S Y}$
    by joint monicity of the $X \linarrow \Sproj_i$ (see \cref{rem:lin-closure-monic}).

  (2) $\Rightarrow$ (1): assume that $\Sclos$ is an isomorphism.
  By \cref{prop:SstrL-prop},
  $(X \linarrow \Sproj_i) \comp \Sclos = \Sproj_i$.
  Thus, $X \linarrow \Sproj_i = \Sproj_i \comp (\Sclos)^{-1}$ and
  $\sequence{X \linarrow \Sproj_i}$ is summable with witness 
  $(\Sclos)^{-1}$.
\end{proof}

The invertibility of $\Sclos$ turns $\kleisliS$ into a symmetric monoidal closed 
  category. The internal hom of $(X, Y)$ is taken to be $(X \linarrow Y, \kleisliLS(\ev_X))$
  and the Curry transpose of $f \in \kleisliS(Y \tensor X, Z)$ is defined as 
  the following composition of morphisms:
  \[
  \begin{tikzcd}
    Y & {(X \linarrow \S Z)} & {\S (X \linarrow Z)}
    \arrow["{\cur(f)}", from=1-1, to=1-2]
    \arrow["{(\Sclos)^{-1}}", from=1-2, to=1-3]
  \end{tikzcd} \, . \] 
  The fact that $\kleisliS$ is closed can be checked by hand, but is also a consequence of 
  a more general categorical observation, as discussed in \cref{rem:extension-closure-KleisliS}
  below.

\begin{remark} \label{rem:extension-closure-KleisliS}
  We will review in \cref{sec:mate} the mate construction, a bijection between natural 
  transformations arising from adjunctions.
  It turns out that $\Sclos$ is the mate of $\SstrL$ through the 
  adjunction $\_ \tensor X \dashv X \linarrow \_$. 
  This is shown by unfolding the definition given in 
  \cref{rem:mate-simplified}, taking $\ladj = \radj = \S$,
  As we discussed in \cref{prop:partial-dl} and
  \cref{sec:commutative-monad}, $\SstrL$ is a distributive law 
  $\S \_ \tensor X \naturalTrans \S (\_ \tensor X)$.
  Then, as we will see in \cref{thm:extension-adjunction}, $\Sclos$
  being an isomorphism is  necessary and sufficient for the
  adjunction $\_ \tensor A \dashv A \linarrow \_$ to extend to
  $\kleisliS$, turning then $\kleisliS$ into a symmetric monoidal closed
  category.
\end{remark}

\subsection{Interaction with the cartesian product}
\label{sec:summability-product-LL}

Assume now that $\categoryLL$ has (finite, countable or all small)
categorical products.
For what remains of this section, the product indexing sets $I$ are
universally quantified over the corresponding class of sets (finite,
countable of arbitrary).
Following the notations of LL, we use $\top$ for the terminal object
and $\withFam X_i$ for the categorical product of the family
$(X_i)_{i \in I}$. We use
$(\prodProj_i \in \categoryLL(\withFam X_i, X_i))_{i\in I}$ for the
projections and $\prodPairing{f_i} \in \categoryLL(X, \withFam Y_i)$ 
for the tupling of the $(f_i \in \category(X, Y_i))_{i\in I}$.
In the special case where $I=\emptyset$, this tupling is the unique
element $\final_X$ of $\categoryLL(X, \top)$.  Notice that
$\final_X = 0^{X, \top}$.

Generalizing the product category $\category \times \category$ of
\cref{sec:product-category} one defines a category $\category^I$ for 
any set $I$ whose objects are families $\family<I>[i]{X_i}$ of objects of $\category$,
and whose morphisms from $\family<I>[i]{X_i}$ to $\family<I>[i]{Y_i}$ 
are the families of morphism 
$\family<I>[i]{f_i \in \category(X_i, Y_i)}$.
As a special case of limit, the $I$-indexed categorical product can be
described as an adjunction $\diagonal \dashv \withFam{\_}$ where
$\diagonal : \category \arrow \category^I$ is the diagonal functor
that maps an object $X$ to the family $(X)_{i \in I}$, and a morphism
$f \in \category(X, Y)$ to the family
$(f \in \category(X, Y))_{i \in I}$, and then
$\withFam \_ : \category^I \arrow \category$ is the functor that maps
a family of objects $(X_i)_{i \in I}$ to $\withFam X_i$, and a family
of morphisms $(f_i \in \categoryLL(X_i, Y_i))_{i \in I}$ to
$\withFam f_i \defEq \prodPairing{f_i \compl \prodProj_i} \in
\categoryLL(\withFam X_i, \withFam Y_i)$.

\begin{definition} \label{def:sum-cartesian} A cartesian
  $\Sigma$-additive category $\categoryLL$ is a $\Sigma$-additive
  category whose underlying category is cartesian, and is such that
  the sum is compatible with the categorical product in the sense that
  for any family $(f_a^i \in \categoryLL(X, Y_i))_{a\in A,\,i\in I}$,
 \[ \sum_{a \in A} \prodPairing{f_a^i} \sumiff \prodPairing{\sum f_a^i} \, . \]
\end{definition}

The results about the interaction between categorical products
and sums are strikingly similar to those about the interaction between
sums and the internal hom functor\footnote{The reason is that both the
  categorical product and the internal hom functor are right adjoint to
  a functor that distributes over sums.}, so this section mirrors the
structure of \cref{sec:summability-closure-LL}.
The definition above can be pruned out.

\begin{proposition} \label{prop:sum-cartesian}
  For any $\Sigma$-additive category $\category$ with a cartesian product,
  \[ \sum_{a \in A} \prodPairing{f_a^i} \sumsub \prodPairing{\sum f_a^i} \, . \] 
  It follows that the following assertions are equivalent. \begin{enumerate}
  \item $\categoryLL$ is cartesian $\Sigma$-additive;
  \item if, for all $i \in I$, $\family{f_a^i \in \categoryLL(X, Y_i)}$ is summable, 
  then $\family{\prodPairing{f_a^i}}$ is summable.
  \end{enumerate}
\end{proposition}

\begin{proof}
 By additivity of the projections and
 $\prodProj_i$ of the categorical product we have
 \[ \prodProj_i \compl \left( \sum_{a \in A} \prodPairing{f_a^i} \right) 
 \sumsub \sum_{a \in A} f_a^i \, . \]
 Hence, $\sum_{a \in A} \prodPairing{f_a^i} \sumsub \prodPairing{\sum f_a^i}$ by 
 uniqueness of the tupling.
 The announced equivalence directly follows.
\end{proof}

We now assume that $\categoryLL$ is equipped with a
$\Sigma$-summability structure.  Again compatibility between sums
and categorical products can be expressed equivalently as a
property of the summability structure.

\begin{definition} \labeltext{($\S$-$\with$)}{ax:S-with} The
  summability structure satisfies \ref{ax:S-with} if the family
  $$
  \sequence[j]{\withFam \Sproj_j\in\categoryLL\left(\withFam \S X_i,\withFam X_i\right)}
  $$ is summable.
\end{definition}

No assumptions such as $\withFam 0 = 0$ or 
$\prodPairing{0} = 0$ are required as these equations hold by 
\cref{prop:sum-cartesian}.

\begin{proposition} \label{prop:S-with-equivalent}
  The following assertions are equivalent
  \begin{enumerate}
      \item the summability structure satisfies \ref{ax:S-with};
      \item if for any $i \in I$, $\family{f_a^i \in \category(X_i, Y_i)}$ is summable,
      then $\family{\withFam f^i_a}$ is summable; 
      \item if for any $i \in I$, $\family{f_a^i \in \category(X, Y_i)}$ is summable, 
      then $\family{\prodPairing{f^i_a}}$ is summable;
      \item $\categoryLL$ is cartesian $\Sigma$-additive.
  \end{enumerate}
\end{proposition}

\begin{proof} The proof mirrors the proof of \cref{prop:S-fun-equivalent}.
\end{proof}

We can always define a natural transformation
\begin{equation} \label{eq:SprodDist}
  \SprodDist = \prodPairing{\S \prodProj_i} \in 
  \categoryLL(\S (\withFam X_i), \withFam \S X_i) 
\end{equation}
simply by functoriality of $\S$ and naturality of $\Sproj_i$.

\begin{proposition} \label{prop:prodSwap-inverse}
  The following are equivalent: \begin{enumerate}
  \item the $\Sigma$-summability structure satisfies \ref{ax:S-with};
  \item $\SprodDist$ is an isomorphism.
  \end{enumerate}
  And then, $\SprodDist^{-1} = \Spairing[j]{\withFam \Sproj_j}$.
\end{proposition}

\begin{proof}
  $(1) \Rightarrow (2)$. First, we show that 
  $\SprodDist \compl \Spairing[j]{\withFam \Sproj_j} = \id_{\withFam \S X_i}$
  by uniqueness of the tupling:
  \[ \prodProj_i \compl \SprodDist \compl \Spairing[j]{\withFam \Sproj_j} 
  = \S \prodProj_i \compl \Spairing[j]{\withFam \Sproj_j} 
  = \Spairing[j] {\prodProj_i \compl \withFam \Sproj_j} 
  = \Spairing[j]{\Sproj_j \compl \prodProj_i} 
  = \Spairing[j]{\Sproj_j} \compl \prodProj_i = \prodProj_i \, . \]
  Then we show that $\Spairing[j]{\withFam \Sproj_j} \compl \SprodDist = 
  \id_{\S \withFam X_i}$ by joint monicity of the $\Sproj_j$:

  \[ \Sproj_j \compl \Spairing[j]{\withFam \Sproj_j} \compl \SprodDist 
  = \left( \withFam \Sproj_j \right) \compl \prodPairing{\S \prodProj_i}
  = \prodPairing{\Sproj_j \compl \S \prodProj_i} 
  = \prodPairing{\prodProj_i \compl \Sproj_j}
  = \prodPairing{\prodProj_i} \compl \Sproj_j
  = \Sproj_j \, . \]
  
  $(2) \Rightarrow (1)$. We have 
  $\Sproj_j \compl \prodProj_i \compl \SprodDist = \prodProj_j \compl \Sproj_j$,
  thus $\Sproj_j \compl \prodProj_i 
  = \prodProj_j \compl \Sproj_j \compl \SprodDist^{-1}$. It follows by uniqueness 
  of the tupling that $\Sproj_j \compl \SprodDist^{-1} = \prodPairing{\Sproj_j 
  \compl \prodProj_i} = \withFam \Sproj_j$. 
  So by \ref{def:summability-structure-2}, $\sequence[j]{\withFam \Sproj_j}$ is 
  summable and the summability structure satisfies \ref{ax:S-with}.
\end{proof}

We can check that the invertibility of $\SprodDist$ (and 
of $\Sinj_{0, \top}$) turns $\kleisliS$ into a cartesian category. 
The cartesian product of $\family<I>[i]{X_i}$ is simply $\withFam X_i$ with projections  
$\kleisliLS(\prodProj_i) \in \kleisliS(\withFam X_i, \S X_i)$.
The tupling of $\family<I>[i]{f_i \in \kleisliS(X, Y)}$ is the following 
composition.
\[
\begin{tikzcd}
X & {\withFam \S X_i} & {\S \withFam X_i}
\arrow["{\prodPairing{f_i}}", from=1-1, to=1-2]
\arrow["{\SprodDist^{-1}}", from=1-2, to=1-3]
\end{tikzcd} \]
This definition is very similar to the symetric monoidal closed structure defined
from $(\Sclos)^{-1}$ in \cref{sec:summability-closure-LL}.
The reason is that the cartesian product follows a similar pattern,
as explained in \cref{rem:extension-product-KleisliS} bellow.

\begin{remark} \label{rem:extension-product-KleisliS}
  We will review in \cref{sec:mate} the mate construction, a bijection between natural 
  transformations arising from adjunctions.
  It turns out that $\SprodDist$ is the mate of the natural transformation
  $\id : \diagonal \S \naturalTrans \S^I \diagonal$ (where $\S^I : \categoryLL^I \arrow 
  \categoryLL^I$ is the functor that applies $\S$ on each index $i \in I$) through 
  the adjunction $\diagonal \dashv \withFam \_$.
  We explain in \cref{thm:extension-adjunction} that  
  the invertibility of $\SprodDist$ is a necessary and sufficient 
  condition to extend the adjunction $\diagonal \dashv \withFam \_$
  to the Kleisli category $\kleisliS$, thus providing $\kleisliS$ with a cartesian product
  ($\top$ is final in $\kleisliS$ since $\Sinj_0 \in \categoryLL(\top, \S \top)$ is 
  an isomorphism).
  \end{remark}
  
  We can check that the monad $\monadS$ on $\categoryLL$ equipped with
  the natural transformations $\Sinj_0$ and $\SprodDist^{-1}$ is a
  strong symmetric monoidal monad (see \cref{def:monoidal-monad}) with
  regard to the symmetric monoidal structure induced by the categorical
  product.
  The diagrams can be checked by hand, but this 
  is a consequence of a more general observation.
  As mentioned in~\cite{Aguiar18} in paragraph 2.3, any monad
  $\mtriple = (\monad, \munit, \msum)$ on a cartesian category can be
  endowed with the structure of an oplax symmetric monoidal
    monad (see \cref{rem:hopf-monad}) taking
  $\osmfzero$ to be the unique element of $\category(\monad \top, \top)$ and
  $\osmftwo_{X_0, X_1} \defEq \prodPair{\monad \prodProj_0}{\monad \prodProj_1} 
  \in \category(\monad (X_0 \with X_1), \monad X_0 \with \monad X_1)$\footnote{This
  oplax monoidal structure is the mate, as in \cref{rem:extension-closure-KleisliS},
  of the trivial lax monoidal structure on $\diagonal$.}.
  Then it turns out that if $\osmfzero = \munit_{\top}^{-1}$ 
  and if $\osmftwo$ is invertible, then $(\mtriple, \osmfzero, \osmftwo)$ 
  is a strong monoidal monad.

  As discussed
  in \cref{sec:commutative-monad} (with the symmetric monoidal structure
  generated by the cartesian structure $\with$), 
  the strong symmetric monoidal monad $\monadS$ is also a
  commutative monad.
  The left and right strengths are given by the following equations.
  \begin{equation} \label{eq:Sprod-strength}
    \begin{split}
    \SprodstrL = \SprodDist^{-1} \compl (\D X_0 \with \Sinj_0)
  = \Sbracket{\Sproj_0 \with X_1, \Sproj_1 \with 0, \Sproj_2 \with 0,
  \ldots}\in \category(\S X_0 \with X_1, \S (X_0 \with X_1)) \\
  \SprodstrR = \SprodDist^{-1} \compl (\Sinj_0 \with \D X_1)
  = \Sbracket{X_0 \with \Sproj_0, 0 \with \Sproj_1, 0 \with \Sproj_2,
  \ldots}\in \category(X_0 \with \S X_1, \S (X_0 \with X_1))
  \end{split}
\end{equation}

\subsection{Interaction with the resource comonad}
\label{sec:summability-comonad}

Here, we assume that $\categoryLL$ is equipped with a resource comonad, 
an exponential for LL based on the axiomatization of 
\emph{Seely categories} of~\cite{Bierman95}.
We show in this section that a $\Sigma$-summability structure on $\categoryLL$ induces a 
left $\Sigma$-summability structure on the coKleisli category of 
that comonad.

\begin{definition} \label{def:resource-comonad}
  A \emph{resource comonad} (also called resource modality) 
  consists of a tuple
  \((\Oc,\Deru,\Digg,\Seelyz,\Seelyt)\) where \((\Oc,\Deru,\Digg)\) is a
  comonad on \(\cL\) (with counit \(\Deru\), called \emph{dereliction},
  and comultiplication \(\Digg\) called \emph{digging}) and
  \((\Seelyz,\Seelyt)\) is a strong symmetric monoidal structure on the functor $\oc\_$
  from the symmetric monoidal category \((\cL,\With,\Top)\) to the symmetric monoidal
  category \((\cL,\Times,\Sone)\) (see \cref{def:smf}) 
  satisfying some coherence diagram that we will not recall here, 
  see for instance~\cite{Mellies09}. 
  In particular,
  \(\Seelyz\in\cL(\Sone,\Oc\Top)\) 
  is an isomorphism and
  \(\Seelyt_{X,Y}\in\cL(\Oc X\Times\Oc Y,\Oc(X\With Y))\) is a natural isomorphism.
  They are called the Seely isomorphisms.
\end{definition}

Let $\kleisliExp$ be the coKleisli category of this comonad (see \cref{sec:comonad-dl} 
for a definition). Let us recall that the composition of $f \in \kleisliExp(X, Y)$
with $g \in \kleisliExp(Y, Z)$ is defined as follows.
\[ g \comp f \defEq
\begin{tikzcd}
	{!X} & {!!X} & {!Y} & Z
	\arrow["\dig", from=1-1, to=1-2]
	\arrow["{!f}", from=1-2, to=1-3]
	\arrow["g", from=1-3, to=1-4]
\end{tikzcd} \]

We write $\Der$ the functor from 
$\categoryLL$ to $\kleisliExp$ defined as  
$\Der h = h \compl \der$.
Let us recall the following equations: for 
any $f \in \kleisliExp(A, B)$, $h \in \categoryLL(B, C)$ and 
$g \in \kleisliExp(C, D)$,
\begin{equation} \label{prop:Der-composition}
	\Der h \comp f = h \compl f \quad 
  g \comp \Der h = g \compl !h \, .
\end{equation}

We show that if $\categoryLL$ is a $\Sigma$-additive category, then
$\kleisliExp$ is a left $\Sigma$-additive category 
and any $\Sigma$-summability structure on $\categoryLL$ induces 
a left $\Sigma$-summability structure on $\kleisliExp$.

\begin{proposition}
If $\categoryLL$ is a $\Sigma$-additive category, then $\kleisliExp$ is a 
left $\Sigma$-additive category.
\end{proposition}

\begin{proof}
 The $\Sigma$-monoid structure on $\kleisliExp(X, Y)$ is simply 
 the $\Sigma$-monoid structure 
 on $\categoryLL(!X, Y)$.
 Left distributivity is a consequence of left distributivity in $\categoryLL$.
 \[ \left( \sum_{a \in A} g_a \right) \comp f 
 = \left( \sum_{a \in A} g_a \right) \compl \Oc f \compl \dig 
 \sumsub \sum_{a \in A} (g_a \compl \Oc f \compl \dig) 
 = \sum_{a \in A} (g_a \comp f) \, . \qedhere \]
\end{proof}

Note that $\kleisliExp$ is not $\Sigma$-additive, because there is no reason that 
$\Oc(\sum_{a \in A} f_a) = \sum_{a \in A} \Oc f_a$.
In fact, $\kleisliExp$ should be understood as a category of analytic maps, and 
such a map does not commute with sums in general.
Still, maps of the form $\Der h$ are $\Sigma$-additive.

\begin{lemma} \label{prop:additivity-der}
 For any $h \in \categoryLL(X, Y)$, $\Der h$ is $\Sigma$-additive.
\end{lemma}

\begin{proof}
This is a consequence of the additivity of $h$ in $\categoryLL$
and the fact that $\Der h \comp f = h \compl f$, see 
\cref{prop:Der-composition}.
\end{proof}

\begin{lemma} \label{prop:summability-der}
 For any family $\family{h_a \in \category(X, Y)}$, 
 \[ \Der \left(\sum_{a \in A} h_a \right) \sumsub \sum_{a \in A} \Der h_a
 \, . \]
\end{lemma}

\begin{proof}
  We use the fact that $\Der h \defEq h \compl \der$
   and left distributivity of sum 
  in $\categoryLL$.
\end{proof}

\begin{proposition} \label{prop:Kleisli-left-summability}
	The pair $(\S, \sequence{\Der \Sproj_i})$
is a left $\Sigma$-summability structure on $\kleisliExp$, and 
\[ \Der \Spairing{h_i} = \Spairing{\Der h_i} \, . \] 
\end{proposition}

\begin{proof}
  We first check \ref{def:summability-structure-1}.
  The $\Der \Sproj_i$ are jointly monic in $\kleisliExp$ because of the 
joint monicity of the $\Sproj_i$ 
in $\categoryLL$ and of the equation
$\Der \Sproj_i \comp f = \Sproj_j \compl$.

Then we check \ref{def:summability-structure-2}.
Assume that $\sequence{f_i \in \kleisliExp(X, Y)}$ is summable
in $\kleisliExp$.
Then $\sequence{f_i}$ is a summable family of $\categoryLL(\Oc X, Y)$.
By \ref{def:summability-structure-2} in $\categoryLL$, 
there exists $f \in \categoryLL(\Oc X, \S Y)$ such that 
$\Sproj_i \compl f = f_i$. Then $f$ is also a morphism
of $\kleisliExp(X, Y)$, and by
\cref{prop:Der-composition} we have $\Der \Sproj_i \comp f = f_i$. 
So $\sequence{f_i}$ has a witness in $\kleisliExp$. 
Conversely, if $\sequence{f_i}$ has a witness in $\kleisliExp$
then it is summable in $\kleisliExp$, the 
proof is very similar to the proof above. So we
conclude that \ref{def:summability-structure-2} 
holds in $\kleisliExp$.

The equation $\Der \Spairing{h_i} 
= \Spairing{\Der h_i}$ is a consequence of the functoriality of 
$\Der$ and of \cref{prop:Der-composition}.
Furthermore, \ref{def:summability-structure-3} holds by \cref{prop:additivity-der}.
Finally, \ref{def:summability-structure-4} holds in $\kleisliExp$ because 
it holds in $\categoryLL$ and because
of \cref{prop:Der-composition}. 
\end{proof}

\begin{remark} Let $\Sadd$ be the functor on $\kleisliExpAdd$ 
  induced by the left $\Sigma$-summability structure $(\S, \sequence{\Der \Sproj_i})$
  on $\kleisliExp$, so that there is no ambiguity with the functor $\S$ on $\categoryLL$.
  We can easily check that
  \[ \Sadd (\Der h) = \Der (\S h) \] 
  so $\Sadd$ is an extension of $\S$ to $\kleisliExpAdd$.
  In general, there is no reason that $\kleisliExpAdd = \categoryLL$.
\end{remark}

\begin{definition} \label{def:sigma-additive-resource-category}
A $\Sigma$-additive resource category $\categoryLL$ 
is a cartesian and a symmetric monoidal category endowed with a resource
comonad, as well as a $\Sigma$-summability structure following
\ref{ax:S-sm-dist} and \ref{ax:S-with}.
\end{definition}

We do not assume \ref{ax:S-fun} in the definition above as it is not crucial to define 
Taylor expansion, so any use of this condition will be made explicit.

\section{Taylor expansion as a distributive law in models of LL}

\label{sec:Taylor}

We have all the necessary tools to axiomatize Taylor expansion in
models of LL. Assume that $\categoryLL$ is a $\Sigma$-additive
resource category.  In particular, we do not need to assume
\ref{ax:S-fun} for this section to make sense.
As seen in \cref{sec:summability-comonad}, $\kleisliExp$ is then
endowed with a left $\Sigma$-summability structure
$(\S, \sequence{\Der \Sproj_i})$.
We provide the intuitions first on what the Taylor expansion operator
should look like.

\subsection{Motivation} \label{sec:Taylor-operator}

The idea behind differential LL is that a morphism
$f \in \kleisliExp(X, Y) = \categoryLL(!X, Y)$ can be seen as
some kind of analytic map between some kind of vector spaces
associated with $X$ and with $Y$.
Let us recall what is a differentiable map (see~\cite{Dieudonne69} 
for instance) and analytic map (see~\cite{Whittlesey65}).
A map $f : E \arrow F$ between two Banach spaces is differentiable in
a point $x$ if its variation around $x$ can be approximated by a
continuous linear map $\derive{f} \in \linbanach(E, F)$ called the
differential of $f$ at \(x\).
\begin{equation} \label{eq:differential}
  f(x + u) = f(x) + \derive{f}
  \cdot u + o(\norm{u})
\end{equation}
If $f$ is regular enough, the map $x \mapsto \derive{f}$ going from
$E$ to $\linbanach(E, F)$ is also differentiable so that for any $x$
there exists a map
$\deriven{f}{2} \in \linbanach(E, \linbanach(E, F))$ called the second
order differential.
Repeating the process yields an $n$-order differential
$\deriven{f}{n} \in \linbanach(E, \linbanach(E, (\ldots, \linbanach(E,
F)\ldots )))$
which can also be seen as an $n$ linear map $E \times \cdots \times E \arrow F$.
These iterated differentials allow approximating
$f$ around a point $x$
by a map which is polynomial of degree \(n\):
\[ f(x + u) = \sum_{k=0}^n \frac{1}{k!} \deriven{f}{k} (u, \ldots, u)
  + o(\norm{u}^n) \, .\]
A map is analytic if it is equal to the limit of
its successive approximations, that is
\begin{equation}
  \label{eq:Taylor}
  f(x + u) = \sum_{n=0}^{\infty} \frac{1}{n!} \deriven{f}{n} 
  (u, \ldots, u) \, .
\end{equation}
The series in the equation above is called the \emph{Taylor series} of
$f$ at $x$.

The point of coherent differentiation is to generalize the ideas of
differentiation to a setting where addition is not necessarily a total
operation.
In coherent differentiation, $\categoryLL$ is a 
category equipped with a binary counterpart of
our $\Sigma$-summability structure (see \cite{Ehrhard23-cohdiff}) that
induces a left summability structure on $\kleisliExp$ (see
\cite{Walch23}).
Then all the sums in \cref{eq:differential} must be sums in the sense of
this left summability structure.
This means that the differential of $f \in \kleisliExp(X, Y)$
should be introduced as a morphism $\cohdiff f \in \kleisliExp(\S X, \S Y)$ 
that intuitively maps a summable pair $\Spair{x}{u}$ to a summable
pair $\Spair{f(x)}{\derive{f} \cdot u}$.
We will use the notation $\d f \Spair{x}{u} \defEq \derive{f} \cdot u$
as this operator $\d$ is the counterpart of the differentiation operator 
in the (left-additive) cartesian differential categories 
of~\cite{Blute09}\footnote{Recall that coherent differentiation 
is a generalization of ordinary differentiation, see \cref{fig:map-of-concepts}.}.

It turns out that this point of view comes with real conceptual benefits, 
because the equations on $\d$ found in cartesian differential 
categories are equivalent to functoriality and naturality 
equations on $\cohdiff$.
More precisely, the chain rule of the differential calculus corresponds to the
functoriality of $\cohdiff$ on $\kleisliExp$, and the other rules
(Leibniz, Schwarz, linearity of the derivative) correspond to the
naturality of $\Der \Sinj_0$, $\Der \SmonadSum$, $\Der \Slift$ and
$\Der \Sswap$ with respect to $\cohdiff$.
An equationnal account of those observations can be found in the long
version of~\cite{Walch23}\footnote{In a setting where the category
  $\category$ considered can be any cartesian category, and not
  necessarily the coKleisli category of a model of LL.}.

\begin{remark}
Note that the functor $\cohdiff$ is very similar to the tangent 
bundle functor of the \emph{tangent categories} of~\cite{Cockett14}.
The naturality equations are also very similar. 
The key difference is the underlying notion of sum, that is, 
the action on object of the functor $\cohdiff$.
In tangent categories, the action of the tangent functor on object corresponds to a tangent 
bundle construction. In that setting, there is a clear distinction between points 
and vectors. It is not possible to sum a vector with a point, 
but vectors of the same tangent space can be added freely. 
Typically, objects in such categories should be seen as smooth manifolds. 
On the other hand, coherent differentiation does not make any distinction
between points and vectors, but the sum is restrained. One of the main example 
of coherent differentiation are the probabilistic coherence spaces, in which 
the objects can be seen as bounded convex sets in a cone.

A more formal difference is that tangent categories usually arise as the 
coEilenberg-Moore categories of the resource comonad of a differential 
category~\cite{Cockett17}, whereas the functor of coherent differentiation is in the 
coKleisli category of a coherent differential category.  
We conjecture the existence of a notion of \emph{coherent tangent category},
that should combine the idea of coherent differentiation and of tangent category.
They should arise as the coEilenberg-Moore category of a 
coherent differential category\footnote{In fact, the \emph{representable} 
theory of \cref{sec:elementary} feature some similarities with the content 
of~\cite{Cockett17}.}.
\end{remark}

The map $\cohdiff^2 f$ can be seen as the following map.
\begin{equation} \label{eq:D2}
	\cohdiff^2 f \Spair{\Spair{x}{u}}{\Spair{v}{w}} 
 = \Spair{\Spair{f(x)}{\derive{f} \cdot u}} 
 {\Spair{\derive{f} \cdot v}{\deriven{f}{2} (u, v) + 
 \derive{f} \cdot w}} \, .
 \end{equation}
 Note that the rightmost component %
 $\Sproj_1 \Sproj_1 \cohdiff^2 f = \d \d f$ does not only contain the
 second order derivative $\deriven{f}{2} (u, v)$, but also the term
 $\derive{f} \cdot w$.
 This happens because $\d \d f$ is the \emph{total} derivative of
 $\d f$, that is its derivative with regard to both of its coordinates
 at the same time, whereas $\deriven{f}{2}$ is only the partial
 derivative of $\d f$ with respect to its first argument. See~\cite{Lemay18}
 for more discussions on this total derivative.

 One specificity of coherent differentiation is that the second order
 derivative $\deriven{f}{2} (u, v)$ that appears \cref{eq:D2} requires
 $u$ and $v$ to be summable.
 This is in sharp contrast with what we
 want to do for Taylor expansion in \cref{eq:Taylor}:
 there is no reason for $u$ to be summable with itself, but
 $\frac{1}{n!} \deriven{f}{n} (u, \ldots, u)$ is well-defined
 nonetheless thanks to the sharply decreasing coefficient
 $\frac{1}{n!}$ in front of the derivative.
 This phenomenon does not seem to be taken easily into
 account by the coherent differential setting.
 So instead of defining $\cohdiff$ only as the first order development
 $\Spair{f(x)}{\derive{f} \cdot u}$,
 we develop a new coherent differential approach where a single application 
 of $\cohdiff$ takes into account the whole Taylor expansion operator. 
 
 Let us start with what a functor implementing a second order development
 would look like. Let us
 introduce a functor $\D$ as follows (using a notion of ternary
 summability structure that should exist, as mentioned in
 \cref{rem:n-ary-summability}).
\begin{equation}
	\D f \Striple{x}{u_1}{u_2} = \Striple{f(x)}{\derive{f} \cdot u_1}
	{\frac{1}{2} \deriven{f}{2} (u_1, u_1) + \derive{f} \cdot u_2 } \, .
\end{equation}
The term $u_1$ can be intuitively understood as a first order variation and the term
$u_2$ as a second order variation.
So $\D f \Striple{x}{u_1}{u_2}$ gives the components
(graded by orders)
of the best order 2 approximation of $f$ on the variation $u_1 + u_2$
(and guarantees that the sums involved are well-defined).
We can recover the usual order $2$ development of $f$ taking $u_2 = 0$
\[
  \D f \Striple{x}{u}{0} = \Striple{f(x)}{\derive{f} \cdot u}
  {\frac{1}{2} \deriven{f}{2} (u, u)}\,.
\]
The term $\derive{f} \cdot u_2$ is still crucial though, and comes
from the Faà di Bruno formula.
Indeed, one can check that
\begin{align*}
  &\D g (\D f \Striple{x}{u}{0}) \\
  &=\D g \Striple{f(x)}{\derive{f} \cdot u}
 {\frac{1}{2} \deriven{f}{2}
 (u, u)} \\
 &= \Striple{g(f(x))}{\derive{g} \cdot \derive{f} \cdot u} 
 {\frac{1}{2}\deriven[f(x)]{g}{2} \cdot \left(\derive{f} \cdot u, \derive{f} \cdot u\right)
 + \frac{1}{2} \derive{g} \cdot \deriven{f}{2} (u, u)} \, . 
\end{align*}
Then we can check that this corresponds to the Taylor expansion of $g \comp f$, in the sense 
that
\begin{equation} \label{eq:Taylor-compositional-ternary}
	\D g (\D f \Striple{x}{u}{0}) = \D (g \comp f) \Striple{x}{u}{0} \, .
\end{equation}
\Cref{eq:Taylor-compositional-ternary} above 
hints at the functoriality of $\D$ and is a consequence of 
the chain rule and the second order chain rule
(formalized in the Faà di Bruno formula, see \cite{Fraenkel78}):
\[
  \deriven{(g \comp f)}{2} (u, u) =
\deriven{g}{2} \left(\derive{f} \cdot u, \derive{f} \cdot u\right) 
+ \derive{g} \cdot \deriven{f}{2}  (u, u) \, .
\]

Similar computations can be performed for all finite orders \(n\)
instead of \(2\), and ultimately for an infinite sequence of terms,
possibly of all finite degrees.
Let \[\mpart{n} \defEq \{ m \in \mfin(\N^*) | \sum_{i \in \N^*}  m(i)\,i =
  n\} \]
where $\mfin(\N^*)$ is the set of finite multisets of elements of $\N^*$,
see \cref{sec:notation-multiset} for the notations.
Define $\D f$ as the map
\begin{equation} \label{eq:Taylor-functor-formula}
\D f \Spairing{x_i} = \Spairing[n]{\sum_{m \in \mpart{n}}
\frac{1}{m!} \deriven[x_0]{f}{\Mscard m} \cdot \vec{x}_m} 
\end{equation}
where $m! = \prod_{i \in \supp{m}} m(i)!$, $\Mscard m
= \sum_{i \in \supp{m}} m(i)$ and 
\[ \vec{x}_m = (\underbrace{x_1, \ldots, x_1}_{m(1) \ \mathrm{times}}, \ldots, 
\underbrace{x_i, \ldots, x_i}_{m(i) \ \mathrm{times}}, \ldots, 
\underbrace{x_n, \ldots, x_n}_{m(n) \ \mathrm{times}}) \, . \]
The term $x_i$ corresponds intuitively to an order $i$ infinitesimal 
$\varepsilon^i \, x_i$, so that 
$\D f \Spairing{x_i}$ contains the components sorted by order of the 
Taylor expansion of $f$ at $x_0 + \varepsilon \, x_1 + \varepsilon^2 \, x_2 + \cdots$.

The case $k = n$ and $m = [1, \ldots, 1]$ gives the value
$\frac{1}{n!} \deriven[x_0]{f}{n} (u_1, \ldots, u_1)$, so we can recover
all the terms
of the Taylor expansion of $f$ by erasing all the higher order
variations.
\begin{equation} \label{eq:recover-taylor}
  \D f \Sbracket{x, u, 0, \ldots} =
  \Spairing[n]{\frac{1}{n!} \deriven{f}{n} (u, \ldots, u)}\,.
\end{equation}
Again, the other cases are still very relevant as they allow to
recover the compositionality of the Taylor expansion. We can check that
\[ \D g (\D f \Sbracket{x, u, 0, \ldots}) = \Spairing[n]{\sum_{m \in
      \mpart{n}} \frac{1}{m!} \deriven{g}{\Mscard{m}} \cdot
    \vec{y}_m} \]
where
$\vec{y} = \sequence[n]{\frac{1}{n!} \deriven{f}{n} \cdot (u, \ldots,
  u)}$.
We recognize above the terms of the Taylor expansion of $g \comp f$, in 
the sense that 
\begin{equation} \label{eq:Taylor-compositional-infinitary}
  \D g (\D f \Sbracket{x, u, 0, \ldots}) = \D (g \comp f) 
  \Sbracket{x, u, 0, \ldots} \, .
\end{equation}
\Cref{eq:Taylor-compositional-infinitary} above 
is a consequence of the Faà di Bruno formula (see~\cite{Fraenkel78}), which
states that for any $n \in \N^*$,
\[
  \deriven{(g \comp f)}{n} \cdot (u, \ldots, u) =  \sum_{m \in \mpart{n}}
  \frac{n!}{m!} \deriven{g}{\Mscard m} \cdot \vec{y}_m\,.
\]
Once again, this means that the Faà di Bruno formula expresses a form of functoriality of the 
Taylor expansion $\D$.

Finally, if $f$ is linear then 
$\derive{f} \cdot u = f(u)$ and $\deriven{f}{n}  (u, \ldots, u) = 0$. 
So for any linear $f \in \kleisliExp(X, Y)$ (that is,
$f = \Der \overline{f}$ for some $\overline{f} \in \categoryLL(X, Y)$)
we should have
\[ \D f (\Spairing{x_i})= \Spairing{f(x_i)} \, . \]
That is, $\D f = \Der{\S \overline{f}}$.
So $\D$ should extend the functor $\S$ to 
$\kleisliExp$\footnote{Very much as in the coherent differential setting of 
\cite{Ehrhard23-cohdiff}.}, in the sense of 
\cref{def:extension}.

\subsection{The axioms of Taylor expansion}

Let $\categoryLL$ be a $\Sigma$-additive resource category.
As motivated above, Taylor expansion should be seen as a functor $\D$
on $\kleisliExp$ that extends $\S$ to $\kleisliExp$. 
It is showed in \cref{sec:dl} that this notion of extension is deeply tied to the notion of 
distributive laws. So Taylor expansion should be a natural transformation 
$\Sdl : \oc \S \naturalTrans \S \oc$ that follows 
commutations typical of distributive laws.
\begin{remark}
  These commutations are exactly the same as the commutations of coherent
  differentiation~\cite{Ehrhard23-cohdiff}, except that the
  summability structure is now infinitary.
  Their meaning in coherent differentiation is well understood, see
  also~\cite{Walch23}.
  They should have a similar meaning in this new setting of Taylor
  expansion, but the underlying combinatorics is much more complicated and
  still slightly unclear.
  What we know for now is that these axioms indeed hold in our
  examples from LL for the exact same reasons that the axioms of
  coherent differentiation hold\footnote{Note however that some 
  models admit a coherent differentiation but not such coherent Taylor expansion,
  see \cref{sec:nucs}.}, 
  and that the functor $\D$ involved
  indeed correspond to the intuitive formula given in
  \cref{eq:Taylor-functor-formula}, see
  \cref{thm:Taylor-expansion-wrel}.
\end{remark}

The first axiom is called \ref{ax:Sdl-chain} and 
states that $\Sdl$ is a distributive law (\cref{sec:dl}) between the functor $\S$ and
the comonad $\Oc\_$.

\begin{equation*} \labeltext{($\Sdl$-chain)}{ax:Sdl-chain}
    \text{\ref{ax:Sdl-chain}} \quad
 \begin{tikzcd}
	{! \S X} & {\S !X} \\
	& {\S X}
	\arrow["{\Sdl_X}", from=1-1, to=1-2]
	\arrow["\der"', from=1-1, to=2-2]
	\arrow["{\S \der}", from=1-2, to=2-2]
\end{tikzcd} \quad
\begin{tikzcd}
	{!\S X} && {\S ! X} \\
	{!! \S X} & {! \S ! X} & {\S !! X}
	\arrow["{\Sdl_X}", from=1-1, to=1-3]
	\arrow["{\dig_{\S X}}"', from=1-1, to=2-1]
	\arrow["{! \Sdl_X}"', from=2-1, to=2-2]
	\arrow["{\Sdl_{!X}}"', from=2-2, to=2-3]
	\arrow["{\S \dig_X}", from=1-3, to=2-3]
\end{tikzcd}
\end{equation*}
By \cref{thm:coextension-and-dl}, \ref{ax:Sdl-chain} also means that 
$\S$ can be extended to a functor $\D$ on $\kleisliExp$ defined as
$\D f = (\S f) \compl \Sdl$. This functor $\D$ 
corresponds to the operator motivated in \cref{sec:Taylor-operator}.
So \ref{ax:Sdl-chain} should be understood as the higher order chain rule (Faà
di Bruno formula).

Next, the axiom \ref{ax:Sdl-local} means that $\Sproj_0$ is a 
morphism of distributive law between $\Sdl$ and $\id$, see 
\cref{def:codl-morphism}.

\begin{equation*} \labeltext{($\Sdl$-local)}{ax:Sdl-local}
    \text{\ref{ax:Sdl-local}} \quad
\begin{tikzcd}
	{! \S X} & {\S ! X} \\
	& {!X}
	\arrow["\Sdl_X", from=1-1, to=1-2]
	\arrow["{\Sproj_0}", from=1-2, to=2-2]
	\arrow["{! \Sproj_0}"', from=1-1, to=2-2]
\end{tikzcd}
\end{equation*}
By \cref{thm:coextension-and-dl-morphism}, \ref{ax:Sdl-local}
also means that
$\Sproj_0$ extends to a natural transformation
$\D \naturalTrans \idfun$ on $\kleisliExp$.
Note that this axiom is requested for $\Sproj_0$ and not for the
$\Sproj_i$ with \(i>0\).
As noted in \cite{Ehrhard23-cohdiff}, differentiation breaks the
symmetry between the components of the functor $\S$.

The axiom \ref{ax:Sdl-add} states that $\Sinj_0$ is a morphism of distributive laws
between $\id$ and $\Sdl$, and that $\SmonadSum$ is a morphism of 
distributive laws between the composition of
$\Sdl$ with itself and $\Sdl$.
\begin{equation*} \labeltext{($\Sdl$-add)}{ax:Sdl-add}
    \text{\ref{ax:Sdl-add}} \quad 
\begin{tikzcd}
	{!X} \\
	{!\S X} & {\S ! X}
	\arrow["{!\Sinj_0}"', from=1-1, to=2-1]
	\arrow["{\Sdl_X}"', from=2-1, to=2-2]
	\arrow["{\Sinj_0}", from=1-1, to=2-2]
\end{tikzcd} \quad 
\begin{tikzcd}
	{! \S^2 X} & {\S ! \S X} & {\S^2 ! X} \\
	{! \S X} && {\S ! X}
	\arrow["{\Sdl_{\S X}}", from=1-1, to=1-2]
	\arrow["{\S \Sdl_X}", from=1-2, to=1-3]
	\arrow["{! \SmonadSum_X}"', from=1-1, to=2-1]
	\arrow["{\Sdl_X}"', from=2-1, to=2-3]
	\arrow["{\SmonadSum_{!X}}", from=1-3, to=2-3]
\end{tikzcd}
\end{equation*}
By \cref{thm:coextension-and-dl-morphism}, this axiom means that $\Sinj_0$ and 
$\SmonadSum$ extend to natural transformations
$\Der \Sinj_0 : \Id \naturalTrans \D$ and $\Der \SmonadSum : \D^2 \naturalTrans \D$
on $\kleisliExp$, meaning that $\D$ inherits the monad structure of $\S$.
This axiom should be seen as the additivity
of the iterated derivatives in 
each of their coordinates.

\begin{remark} \label{rem:extension-comonad-KleisliS}
Dually, this axiom also means that 
$\Sdl$ is a distributive law between the functor 
$\oc$ and the monad $\monadS$, so
$\oc\_$ extends to a functor 
$\extension{\oc}\_$ on the Kleisli category 
$\kleisliS$ of $\S$. Then, \ref{ax:Sdl-chain} means 
that $\der$ and $\dig$ are morphisms of distributive laws,
so by \cref{thm:extension-and-dl}
the comonadic structure of $\oc\_$ extends to 
a comonadic structure on $\extension{\oc}\_$.
\end{remark}

Then, the axiom \ref{ax:Sdl-Schwarz} means that $\Sswap$ 
is a morphism of distributive laws
(again for the composition of $\Sdl$ with itself).
\begin{equation*} \labeltext{($\Sdl$-Schwarz)}{ax:Sdl-Schwarz}
    \text{\ref{ax:Sdl-Schwarz}} \quad
    \begin{tikzcd}
        {! \S^2 X} & {\S ! \S X} & {\S^2 ! X} \\
        {! \S^2 X} & {\S ! \S X} & {\S^2 ! X}
        \arrow["{\Sdl_{\S X}}", from=1-1, to=1-2]
        \arrow["{\S \Sdl_X}", from=1-2, to=1-3]
        \arrow["{! \Sswap_X}"', from=1-1, to=2-1]
        \arrow["{\Sdl_{\S X}}"', from=2-1, to=2-2]
        \arrow["{\S \Sdl_X}"', from=2-2, to=2-3]
        \arrow["{\Sswap_{!X}}", from=1-3, to=2-3]
    \end{tikzcd}
\end{equation*}
By \cref{thm:coextension-and-dl-morphism}, \ref{ax:Sdl-Schwarz} 
also means that $\Sswap$ extends to
a natural transformation $\Der \Sswap : \D^2 \naturalTrans \D^2$ 
on $\kleisliExp$.
\ref{ax:Sdl-Schwarz} can be interpreted as the Schwarz theorem that states that 
the higher order derivatives are symmetric. 

The next axiom, is not among the axioms given in 
\cite{Ehrhard23-cohdiff}, but as discussed 
Section 5.1 of the long version of \cite{Walch23} it should have been a part of it.
This axiom called \ref{ax:Sdl-lin} means that $\Slift$ is a morphism of 
distributive laws between $\Sdl$ and the composition of $\Sdl$ with itself.
\begin{equation*} \labeltext{($\Sdl$-lin)}{ax:Sdl-lin}
    \text{\ref{ax:Sdl-lin}} \quad
\begin{tikzcd}
	{!\S X} && {\S ! X} \\
	{! \S^2 X} & {\S ! \S X} & {\S^2 ! X}
	\arrow["{!\Slift_X}"', from=1-1, to=2-1]
	\arrow["{\Sdl_X}", from=1-1, to=1-3]
	\arrow["{\Slift_{!X}}", from=1-3, to=2-3]
	\arrow["{\Sdl_{\S X}}"', from=2-1, to=2-2]
	\arrow["{\S \Sdl_X}"', from=2-2, to=2-3]
\end{tikzcd}
\end{equation*}
By \cref{thm:coextension-and-dl-morphism}, \ref{ax:Sdl-lin} 
also means that
$\Slift$ extends to a natural transformation
$\Der \Slift : \D \naturalTrans \D^2$ on $\kleisliExp$.
Together with \ref{ax:Sdl-add}, this axiom means that the derivatives
are not only additive in their individual coordinates, but also
$\D$-linear in the sense of \cref{def:D-linear}\footnote{This explains the clash of terminology with
  \cite{Ehrhard23-cohdiff}, in which the axiom ($\Sdl$-lin)
  corresponds to our axiom \ref{ax:Sdl-add}.}.

The last axiom, \ref{ax:Sdl-with}, means that $\seelyTwo$ is a morphism
of distributive laws between the composition of 
$\Sdl$ with $\Sdist$, and the composition of $\Sdl$ 
with $(\SprodDist)^{-1}$ (recall that the structure of a lax symmetric 
monoidal monad can be seen as a distributive law, as explained in 
\cref{sec:monoidal-functor}).

\begin{equation*} \labeltext{($\Sdl$-$\with$)}{ax:Sdl-with}
    \text{\ref{ax:Sdl-with}} \quad
\begin{tikzcd}
	{! \S X \tensor ! \S Y} & {\S ! X \tensor \S ! Y} & {\S (!X \tensor !Y)} \\
	{!(\S X \with \S Y)} & {! \S (X \with Y)} & {\S ! (X \with Y)}
	\arrow["{\Sdl_X \tensor \Sdl_Y}", from=1-1, to=1-2]
	\arrow["{\Sdist_{!X, !Y}}", from=1-2, to=1-3]
	\arrow["{\seelyTwo_{\S X,\S Y}}"', from=1-1, to=2-1]
	\arrow["{\S \seelyTwo_{X,Y}}", from=1-3, to=2-3]
	\arrow["{! \SprodDist^{-1}}"', from=2-1, to=2-2]
	\arrow["{\Sdl_{X \with Y}}"', from=2-2, to=2-3]
\end{tikzcd}
\end{equation*}
By \cref{thm:extension-and-dl-morphism}, \ref{ax:Sdl-with}
means that $\seelyTwo$ extends to a natural transformation 
$\exclS \_ \tensorS \exclS \_ \naturalTrans \exclS(\_ \withS \_)$
(where $\exclS\_$ is the extension of $\oc\_$ to $\kleisliS$, see 
\cref{rem:extension-comonad-KleisliS}).
This extension provides a resource comonad $\exclS\_$ on $\kleisliS$, 
turning $\kleisliS$ into a model 
of intuitionist LL\footnote{It is then possible to prove that if $\categoryLL$ is 
a model of classical LL, then so is $\kleisliS$.}.
\ref{ax:Sdl-with} also implies that any morphism that is multilinear in the 
sense of LL is also $\D$-multilinear (see \cref{def:D-multilinear}).
More on this in \cref{sec:Taylor-from-ll}.

\begin{remark} \label{rem:Sdl-with}
In \cite{Ehrhard23-cohdiff}, the diagram below was stated as part
of the axiom \ref{ax:Sdl-with}.
\[ \begin{tikzcd}
		{! \S \top} && {\S ! \top} \\
		{! \top} & 1 & {\S 1}
		\arrow["{\Sdl_{\top}}", from=1-1, to=1-3]
		\arrow["{!0}"', from=1-1, to=2-1]
		\arrow["{(\seelyOne)^{-1}}"', from=2-1, to=2-2]
		\arrow["{\Sinj_0}"', from=2-2, to=2-3]
		\arrow["{\S (\seelyOne)^{-1}}", from=1-3, to=2-3]
	\end{tikzcd} \]
It turns out that this diagram always holds (assuming \ref{ax:Sdl-add}). 
Indeed, we know that $!0 = ! \final_{\S \top}$ is invertible of inverse $! \Sinj_0$,
see \cref{sec:summability-product-LL}.
So we can prove the diagram with the diagram chase below. The commutation 
$(a)$ is \ref{ax:Sdl-add} and the commutation $(b)$ is the naturality of $\Sinj_0$.
\[ 
\begin{tikzcd}
	{! \S \top} && {\S ! \top} \\
	{! \top} & 1 & {\S 1}
	\arrow["{!\Sinj_0}", from=2-1, to=1-1]
	\arrow["{\Sdl_{\top}}", from=1-1, to=1-3]
	\arrow[""{name=0, anchor=center, inner sep=0}, "{\Sinj_0}"{description}, from=2-1, to=1-3]
	\arrow["{(\seelyOne)^{-1}}"', from=2-1, to=2-2]
	\arrow["{\Sinj_0}"', from=2-2, to=2-3]
	\arrow["{\S (\seelyOne)^{-1}}", from=1-3, to=2-3]
	\arrow["{(a)}"{description}, draw=none, from=1-1, to=0]
	\arrow["{(b)}"{description}, draw=none, from=2-3, to=0]
\end{tikzcd} \]
\end{remark}

Recall (see~\cite{Mellies09}) that we can derive from the Seely 
isomorphisms a canonical cocommutative comonoid structure on 
all objects $\oc X$, as follows.
\begin{equation} \label{eq:weakening}
	\weak_X =
\begin{tikzcd}
	{!X} & {!\top} & 1
	\arrow["{!0}", from=1-1, to=1-2]
	\arrow["{(\seelyOne)^{-1}}", from=1-2, to=1-3]
\end{tikzcd}
\end{equation}
\begin{equation} \label{eq:contraction} 
	\contr_X = 
\begin{tikzcd} 
	{!X} & {!(X \with X)} & {!X \tensor !X}
	\arrow["{!\prodPair{\id}{\id}}", from=1-1, to=1-2]
	\arrow["{(\seelyTwo)^{-1}}", from=1-2, to=1-3]
\end{tikzcd}
\end{equation}
Then the axiom~\ref{ax:Sdl-with} induces a compatibility condition 
between $\weak, \contr$ and the distributive law $\Sdl$. 
It suggests that coherent 
differentiation and Taylor expansion would also work in other axiomatizations of 
models of LL, such as the linear categories of \cite{Bierman95}, we 
refer to \cite{Mellies09} for a 
definition.

\begin{proposition} \label{prop:ctr-wk-morphism-dl}
	The left diagram always hold, and the right diagram 
	is a consequence of \ref{ax:Sdl-with}.
\[ \begin{tikzcd}
	{!\S X} & {\S ! X} \\
	1 & {\S 1}
	\arrow["{\Sdl_X}", from=1-1, to=1-2]
	\arrow["{\weak_{\S X}}"', from=1-1, to=2-1]
	\arrow["{\Sinj_0}"', from=2-1, to=2-2]
	\arrow["{\S \weak_X}", from=1-2, to=2-2]
\end{tikzcd} \quad 
	\begin{tikzcd}
		{!\S X} && {\S ! X} \\
		{!\S X \tensor ! \S X} & {\S ! X \tensor \S ! X} & {\S (!X \tensor !X)}
		\arrow["\Sdl_X", from=1-1, to=1-3]
		\arrow["{\contr_{\S X}}"', from=1-1, to=2-1]
		\arrow["{\Sdl_X \tensor \Sdl_X}"', from=2-1, to=2-2]
		\arrow["{\Sdist_{!X, !X}}"', from=2-2, to=2-3]
		\arrow["{\S \contr_X}", from=1-3, to=2-3]
	\end{tikzcd} \]
\end{proposition}

\begin{proof} This is a straightforward consequence of the diagram of
\cref{rem:Sdl-with} and \ref{ax:Sdl-with}, unfolding the definition of
$\weak$ and $\contr$ and using naturality.
\end{proof}

The left diagram means that $\weak$ is a morphism of distributive laws
($\Sinj_0 : ! 1 \naturalTrans \S 1$ is the distributive law associated to the extension
of the constant functor $1$ on $\categoryLL$ as the constant functor $1$ on $\kleisliS$). 
The right diagram means that $\contr$ is a morphism of distributive laws between 
$\Sdl$ and the composition of $\Sdl$ with $\Sdist$.
By \cref{thm:extension-and-dl-morphism}, those diagrams mean that the contraction and weakening
extends to $\kleisliS$. The result of \cref{prop:ctr-wk-morphism-dl} is not surprising then,
since the weakening and the contraction on $\kleisliS$ can also be defined directly
from $\seelyOne$ and the extension of $\seelyTwo$ to $\kleisliS$.

\begin{definition} A \emph{Taylor expansion} in
	a $\Sigma$-additive resource category is a natural transformation
	$\Sdl : ! \S \naturalTrans \S !$ following 
	\ref{ax:Sdl-chain}, \ref{ax:Sdl-local}, \ref{ax:Sdl-add},
	\ref{ax:Sdl-Schwarz}, \ref{ax:Sdl-lin}, \ref{ax:Sdl-local}.
	A \emph{Taylor category} is a $\Sigma$-additive resource category
	equipped with a Taylor expansion.
\end{definition}

\begin{remark} \label{rem:n-ary-summability}
	It should be possible to define in a uniform way a notion of 
	$n$-ary summability structure $\S_n$ for any $n \in \N \cup \{ \infty \}$,
	as a summability structure which has 
	only projections $\Sproj_i$ for $i \in \interval{0}{n}$. 
	Then a Taylor expansion in an $n$-ary summable resource category
	would simply be a distributive law $\oc \S_n \naturalTrans \S_n \oc$
	defined exactly in the same way as above. This operation 
	should be seen as an order $n$ Taylor approximation.
	Then the coherent differentiation of \cite{Ehrhard23-cohdiff} would be 
	a particular case in which $n=1$, and the Taylor expansion in our article a 
	particular case in which $n = \infty$. 
	It should be interesting to see if a Taylor expansion at order $n$ induces 
	a Taylor expansion at order $m$ for $m < n$. This will be studied in future 
	work.
\end{remark}

We have not stated yet the fact that the morphisms in the category are analytic, 
in the sense that they coincide with their Taylor expansion.
 This is the role of the axiom
\ref{ax:Sdl-Taylor} below.

\begin{equation*} \labeltext{($\Sdl$-analytic)}{ax:Sdl-Taylor}
    \text{\ref{ax:Sdl-Taylor}} \quad
	\begin{tikzcd}
		{!\S X} & {\S ! X} \\
		& {!X}
		\arrow["{\Sdl_X}", from=1-1, to=1-2]
		\arrow["{\Ssum_X}", from=1-2, to=2-2]
		\arrow["{! \Ssum_X}"', from=1-1, to=2-2]
	\end{tikzcd}
\end{equation*}
The axiom \ref{ax:Sdl-Taylor} means that $\Ssum$ is a morphism of distributive laws between
$\Sdl$ and $\id$. By \cref{thm:coextension-and-dl-morphism}, 
\ref{ax:Sdl-Taylor} means that 
$\Ssum$ extends to a natural transformation 
$\Der \Ssum : \D \naturalTrans \idfun$. 
The combination of \ref{ax:Sdl-Taylor} and \ref{ax:Sdl-lin} 
means that $\Ssum$ and $\Slift$ extends to 
$\kleisliExp$, so $\D$ inherits the comonadic structure
of $\comonadS$.
Dually, \ref{ax:Sdl-lin} and \ref{ax:Sdl-Taylor} also mean 
that $\Sdl$ is a distributive law 
between the functor $\oc\_$ and the comonad $\comonadS$.

Recall that in \cref{sec:Taylor-operator} 
the Taylor expansion of $f$ at $x$ on variation $u$ could be defined as 
$\Der (\Ssum) \comp \D f \comp \Sbracket{x, u, 0, \ldots}$. So
the naturality of $\Der(\Ssum)$ implies that 
\[ \Der(\Ssum) \comp \D f \comp \Sbracket{x, u, 0, \ldots} = 
f \comp \Der(\Ssum) \comp \Sbracket{x, u, 0, \ldots} = f \comp (x+u) \, . \]
This exactly corresponds to the property that $f$ is equal to its Taylor expansion.

\begin{definition} \label{def:Taylor-category}
	A Taylor expansion is analytic if it satisfies \ref{ax:Sdl-Taylor}. 
	An \emph{analytic category} is a $\Sigma$-additive resource category 
	equipped with an analytic Taylor expansion.
      \end{definition}  

Although the notion of order $n$ Taylor category makes perfect sense,
an order $n$ analytic category may not be very interesting, as it would mean that the
morphisms are all polynomials of degree lower than $n$.

\begin{remark}
	\label{rk:Taylor-analytic-terminology}
	The purpose of our choice of terminology is to make a clear
	distinction between the infinitary setting of the present article and
	the finitary settings of~\cite{Ehrhard23-cohdiff,Walch23}. This is why we
	prefer to speak directly of a Taylor category: it is a category where
	any morphism has a Taylor expansion (involving all of its higher
	derivatives). 
	This Taylor expansion is provided by the endofunctor 
	\(\Tayfun\), which is much richer than a mere Taylor expansion and is
	a categoryfication of the Faà di Bruno formula, by need of
	functoriality.
	In such categories, the morphisms however are not necessarily equal
	to the infinite sum of all the terms of their Taylor expansion.

	We use the adjective ``analytic'' for the situation where any
	morphism is equal to the sum of all the terms of its Taylor
	expansion, following the standard mathematical terminology, with the
	slight difference that, in analysis, analyticity is a local concept
	whereas here, it is a global condition (it should be noted that 
	we have nothing like a topology which would allow to make it local).
  \end{remark}   

\begin{remark}  
	Except for \ref{ax:Sdl-local} (which is about the structure of the functor 
	$\D$) and \ref{ax:Sdl-with} (which is discussed in \cref{sec:Taylor-from-ll}), 
	the axioms of analytic categories are exactly the necessary and 
	sufficient conditions to extend the 
	structure of the bimonad $\S$ to $\kleisliExp$.
	This means that Taylor expansion essentially
        acts as a bimonad and can be framed in a very algebraic way,
        which is an interesting observation \emph{per se}, even if one
        does not care about the partiality of sums.

	In \cite{Kerjean23}, Taylor expansion is framed in some models
        of differential LL as a monad structure on $\oc\_$. The article conjectures 
		that this monadic structure is compatible with the resource 
		comonad structure of $\oc\_$, turning $\oc\_$ intro a bimonad. 
		This seems to be quite different from
        our bimonad $\S$, but a closer comparison is required.
\end{remark}

\section{Cartesian analytic categories}

\label{sec:cartesian-Taylor}

We provide in this section a direct axiomatization of Taylor 
expansion in any category $\category$.
We show how this expansion should interact with the cartesian 
closed structure of $\category$, whenever there is one.
Typically, $\category = \kleisliExp$ for some model $\categoryLL$
of LL, but the point of this axiomatization is that it is more general,
more compact, and does not depend at all on LL.
In particular, this axiomatization should provide direct models
for a $\lambda$-calculus featuring an internal operation of Taylor 
expansion.

The interplay between the LL axiomatization of Taylor expansion and the direct 
axiomatization is similar to the interplay between differential 
categories \cite{Blute06} and cartesian differential categories \cite{Blute09}.
On one hand, the direct axiomatization is simpler and more general.
On the other hand, the rich structure of models of LL 
brings many insights on the structure of Taylor expansion, as well as 
many examples. 
For example, \cref{sec:elementary} shows
that Taylor expansions often boils down to a simple $\oc$-coalgebra, and this 
provides a substantial source of concrete examples described in 
\cref{sec:examples-taylor,sec:elementary-Taylor-examples}.

\subsection{Taylor expansion in a left $\Sigma$-additive category}

The axiomatization of Taylor expansion on $\category$ is very similar
to the cartesian coherent differentiation of~\cite{Walch23}, except that the 
left summability structure is now infinitary.
We assume that $\category$ is a left $\Sigma$-additive category,
equipped with a left $\Sigma$-summability structure $(\S, \vect \Sproj)$.

\begin{definition} \label{def:D-linear}
  Let $\D$ be a map on morphisms such that for any
  $f \in \category(X, Y)$, $\D f \in \category(\S X, \S Y)$.  A
  morphism $h$ is $\D$-linear if it is $\Sigma$-additive and if
  $\Sproj_i \comp \D h = h \comp \Sproj_i$. That is, $\D h = \S h$.
\end{definition}

\begin{definition} 
    \labeltext{($\D$-chain)}{ax:D-chain}
    \labeltext{($\D$-add)}{ax:D-add} 
    \labeltext{($\D$-lin)}{ax:D-lin}
    \labeltext{($\D$-Schwarz)}{ax:D-Schwarz}
    \labeltext{($\D$-local)}{ax:D-local}
    \labeltext{($\D$-proj-lin)}{ax:D-proj-lin}
    \labeltext{($\D$-sum-lin)}{ax:D-sum-lin}
    An (infinitary) Taylor expansion on $\category$ is a map on morphisms $\D$
    such that for any $f \in \category(X, Y)$, $\D f \in \category(\S X, \S Y)$ and such that:
    \begin{enumerate}
        \item \ref{ax:D-chain} $\D$ is a functor
        \item \ref{ax:D-local} $\Sproj_0$ is a natural transformation
        \item \ref{ax:D-proj-lin} The projections $\Sproj_i$ are $\D$-linear
        \item \ref{ax:D-sum-lin} $\Ssum$ and $0$ are $\D$-linear 
        \item \ref{ax:D-add} $\Sinj_0 \in \category(X, \D X)$ 
        and $\SmonadSum \in \category(\D^2 X, \D X)$ are natural transformations
        \item \ref{ax:D-lin} $\Slift \in \category(\D X, \D^2 X)$ is a natural transformation
        \item \ref{ax:D-Schwarz} $\Sswap \in \category(\D^2 X, \D^2 X)$ is a natural transformation
    \end{enumerate}
    We recall that $\Sinj_0, \SmonadSum, \Slift, \Sswap$ are defined in 
    \cref{thm:bimonad-structure}, and $\Ssum = \sum_{i \in \N} \Sproj_i$.
\end{definition}

Again, assuming a suitable notion of $n$-ary summability structure, it should be possible 
to define an order $n$ Taylor expansion. The operator $\D f$ would perform the order 
$n$ Taylor approximation of $f$. 

\begin{definition} \labeltext{($\D$-analytic)}{ax:D-Taylor}
An analytic structure on $\category$ is a Taylor expansion 
such that $\Ssum \in \category(\D X, X)$ is natural. We call this property \ref{ax:D-Taylor}.
\end{definition}

We only assume in what follows that $\D$ is a map on morphism 
such that for any $f \in \category(X, Y)$, $\D f \in \category(\S X, \S Y)$.
Any use of the axioms of Taylor expansion will be made explicit.

\begin{proposition} \label{prop:differential-Spairing}
    Assuming \ref{ax:D-chain} and \ref{ax:D-proj-lin}, if $\sequence{f_i}$ 
    is summable then $\sequence{\D f_i}$ is summable and 
    $\Spairing{\D f_i} = \Sswap \comp \D \Spairing{f_i}$.
\end{proposition}
   
\begin{proof} We have 
  \[ \Sproj_i \comp \Sswap \comp \D \Spairing{f_i}
    = \S \Sproj_i \comp \D \Spairing{f_i} = \D \Sproj_i \comp \D \Spairing{f_i}
    = \D (\Sproj_i \comp \Spairing{f_i}) = \D f_i \] 
    using \ref{ax:D-proj-lin} and \ref{ax:D-chain}.
\end{proof}

\begin{corollary} \label{prop:Spairing-linear}
  Assuming \ref{ax:D-chain} and \ref{ax:D-proj-lin}, if $\sequence{h_i}$ is a 
  family of $\D$-linear morphisms, then 
  $\Spairing{h_i}$ is $\D$-linear.
\end{corollary}

\begin{proof}
  We know that $\Spairing{h_i}$ is $\Sigma$-additive thanks to 
  \cref{prop:Spairing-additive}, and 
  \[ \D \Spairing{h_i} = \Sswap \comp \Spairing{\D h_i}
  = \Sswap \comp \Spairing{\S h_i}
  = \S \Spairing{h_i} \, .\] 
  The last equality can be easily checked by 
  joint monicity of the $\Sproj_i \comp \Sproj_j$.
  So $\Spairing{h_i}$ is $\D$-linear. 
\end{proof}

As usual, the axiom \ref{ax:D-sum-lin} can be framed either as a property of the summability
structure, or as an interaction between $\D$ and the sum.

\begin{proposition} \label{prop:D-preserves-sum}
  Assuming \ref{ax:D-chain} and \ref{ax:D-proj-lin}, \ref{ax:D-sum-lin} holds 
  if and only if for any family $\family{f_a \in \category(X, Y)}$, 
        $\D (\sum_{a \in A} f_a) \sumiff \sum_{a \in A} \D f_a$
\end{proposition}

\begin{proof} 
  Assume that \ref{ax:D-sum-lin} holds.
  Let $\phi : A \injection \N$ be an injection, and let 
  $\sequence{f_i'} = \Famact{\phi}{\vect f}$.
  Observe that $\sequence{\D f_i'} = \Famact{\phi}{\family{\D f_a}}$ 
  because $\D 0 = 0$ by assumption.
  Then \begin{align*}
  \D \left( \sum_{a \in A} f_a \right) 
  & \sumiff \D \left(\sum_{i \in \N} f_i' \right) \tag*{by \cref{prop:reindexing}}\\
  &= \D (\Ssum \comp \Spairing{f_i'}) \\
  &= \D \Ssum \comp \D \Spairing{f_i'} \tag*{by \ref{ax:D-chain}} \\
  &= \S \Ssum \comp \Sswap \comp \Spairing{\D f_i'} \tag*{by assumption} \\
  &= \Ssum \comp \Spairing{\D f_i'} \tag*{by \cref{prop:Sswap-distributive}} \\
  &= \sum_{i \in \N} (\D f_i') \\
  & \sumiff \sum_{a \in A} (\D f_a) \tag*{by \cref{prop:reindexing}.}
  \end{align*}

  Conversely, assume that for all family 
  $\family{f_a}$, $\D (\sum_{a \in A} f_a) \sumiff \sum_{a \in A} (\D f_a)$.
  Then $\D 0 = 0$ follows from the fact that $0$ is the sum over the empty family,
  and
  \[ \D \sigma = \D \left(\sum_{i \in \N} \Sproj_i\right) 
  \sumiff \sum_{i \in \N} \D \Sproj_i 
  = \sum_{i \in \N} \S \Sproj_i 
  \sumiff \S \left(\sum_{i \in \N} \Sproj_i \right)
  = \S \Ssum \]  
  where the equivalence $\sum_{i \in \N} \S \Sproj_i 
  \sumiff \S \sum_{i \in \N} \Sproj_i$ follows from \cref{prop:S-preserves-sum}.
  So $\Ssum$ is $\D$-linear.
\end{proof}

\begin{corollary} \label{prop:sum-linear}
  Assuming \ref{ax:D-chain}, \ref{ax:D-proj-lin} and \ref{ax:D-sum-lin},
  if $\family{h_a}$ is a summable family of $\D$-linear morphisms, then 
  $\sum_{a \in A} h_a$ is $\D$-linear.
\end{corollary}

\begin{proof} This is a direct consequence of \cref{prop:D-preserves-sum,prop:S-preserves-sum}.
\[ \D \sum_{a \in A} h_a \sumiff \sum_{a \in A} \D h_a 
= \sum_{a \in A} \S h_a \sumiff \S \sum_{a \in A} h_a \, . \qedhere \]
\end{proof}

\begin{proposition} \label{prop:composition-linear}
Assuming \ref{ax:D-chain}, the composition of two $\D$-linear 
morphism is also $\D$-linear and $\id$ is $\D$-linear.
\end{proposition}

\begin{proof} The identity is $\Sigma$-additive, and 
  by \ref{ax:D-chain}, $\Sproj_i \comp \D \id 
  = \Sproj_i \comp \id = \Sproj_i = \id \comp \Sproj_i$ 
  so $\id$ is $\D$-linear.
  If $h \in \category(X, Y)$ and $h' \in \category(Y, Z)$ are 
    $\D$-linear, then $h' \comp h$ is $\Sigma$-additive by \cref{prop:category-additive}.
    Furthermore, 
    \[ \Sproj_i \comp \D (h' \comp h)
    = \Sproj_i \comp \D h' \comp \D h 
    = h' \comp \Sproj_i \comp \D h 
    = h' \comp h \comp \Sproj_i \] 
    so $h' \comp h$ is $\D$-linear.
\end{proof}

\begin{definition} Let $\categoryLin$ be the category with the same objects 
as $\category$ and whose morphisms are the $\D$-linear morphisms. The identity 
and the composition are the same as in $\category$. 
\end{definition}
Observe that
$\categoryLin$ is a subcategory of $\categoryAdd$.
Then the axioms \ref{ax:D-chain}, \ref{ax:D-proj-lin}, \ref{ax:D-sum-lin}, 
and \cref{prop:Spairing-additive,prop:sum-linear}
ensure that $\categoryLin$ is a $\Sigma$-additive category whose sum coincides 
with the sum in $\category$, and that 
$(\S, \vect \Sproj)$ is a $\Sigma$-summability structure 
on $\categoryLin$ whose witnesses coincide with the witnesses in $\category$.

\begin{proposition}
    Assuming \ref{ax:D-chain}, \ref{ax:D-proj-lin} and \ref{ax:D-sum-lin},
    $\Sinj_i$, $\SmonadSum$, $\Sswap$, $\Slift$ and $\Ssum$
    defined in~\cref{thm:bimonad-structure} are all $\D$-linear
\end{proposition}

\begin{proof} These morphisms are all tuplings of sums and composition 
  of $\Sproj_i$'s and $0$, which are $\D$-linear by assumption, 
  so they are $\D$-linear by 
  \cref{prop:Spairing-linear,prop:sum-linear,prop:composition-linear} above. 
\end{proof}

As a result, $\D \Sinj_i = \S \Sinj_i$, $\D \SmonadSum = \S \SmonadSum$,
$\D \Sswap = \S \Sswap$, $\D \Slift = \S \Slift$ and $\D \Ssum = \S \Ssum$.
In particular, all the diagrams turning $\S$ into a bimonad also hold when 
replacing $\S$ with $\D$. 
So the axioms of an analytic structure except for \ref{ax:D-local} 
are exactly the conditions allowing to turn $\D$ into a bimonad 
on $\category$ that extends the bimonad $\S$ in $\categoryLin$.

\subsection{Interaction with the cartesian structure}

\label{sec:compatibility-product}

We assume that $\category$ is a cartesian category.  The notations on
the cartesian product will be the same as those of
\cref{sec:summable-resource-category}. This section is a
straightforward adaption of \cite{Walch23} to infinitary summability
structures.
For the rest of this section, all indexing sets $I$ are universally
quantified over the sets for which the categorical product $\withFam$
is well-defined. In particular, the category may have countable
products or not.

\subsubsection{Sums and cartesian product}
\label{sec:product-summability}

This section is a generalization of \cref{sec:summability-product-LL}
to a situation in which the category is only left $\Sigma$-additive.

\begin{definition}
  A cartesian left $\Sigma$-additive category $\categoryLL$ is a left 
  $\Sigma$-additive category 
  that is cartesian, and such that the sum is compatible with the cartesian product in 
  the sense that for any families $\family{f_a^i \in \categoryLL(X, Y_i)}$,
  \begin{equation} \label{eq:sum-pairing}
    \sum_{a \in A} \prodPairing{f_a^i} \sumiff \prodPairing{\sum f_a^i} \, .
  \end{equation}
 \end{definition}

 \begin{proposition} \label{prop:left-sum-cartesian}
  For any left $\Sigma$-additive category $\category$ with a cartesian product,
  the projections $\prodProj_i$ are $\Sigma$-additive if and only if
  \begin{equation} \label{eq:proj-additive}
    \sum_{a \in A} \prodPairing{f_a^i} \sumsub \prodPairing{\sum f_a^i} \, .
  \end{equation}
  As a result, the following assertions are equivalent. \begin{enumerate}
  \item $\categoryLL$ is cartesian left $\Sigma$-additive;
  \item the projections $\prodProj_i$ are $\Sigma$-additive, 
  and if for all $i \in I$, $\family{f_a^i \in \categoryLL(X, Y_i)}$ is summable, 
  then $\family{\prodPairing{f_a^i}}$ is summable.
  \end{enumerate}
\end{proposition}

\begin{proof}
The forward implication has the same proof as \cref{prop:sum-additive}. The only additional 
thing to check is that the additivity of the projections follows from 
\cref{eq:proj-additive}. We have
\[ \prodProj_i \comp \left( \sum_{a \in A} f_a \right) 
= \prodProj_i \comp \left(\sum_{a \in A} \prodPairing{\prodProj_i \comp f_a} \right) 
\sumsub \prodProj_i \comp \prodPairing{\sum_{a \in A} \prodProj_i \comp f_a} 
= \sum_{a \in A} \prodProj_i \comp f_a \, . \qedhere \]
\end{proof}

Again, the compatibility between sum and product can also be 
written as a property of the summability structure. 

\begin{definition}  \labeltext{($\S$-$\with$)}{ax:S-with-left}
Assume that $\category$ is equipped with a left $\Sigma$-summability
structure.
This left summability structure satisfies \ref{ax:S-with-left} if the projections 
$\prodProj_i$ 
are $\Sigma$-additive and if
$\sequence[j]{\withFam \Sproj_j}$ is summable.
\end{definition}

When the left summability structure is a summability structure, 
\ref{ax:S-with-left} coincides with the property \ref{ax:S-with} seen in 
\cref{sec:summability-product-LL}, this is why the names are the same.

\begin{proposition} \label{prop:S-with-equivalent-left}
  Assume that $\category$ is a left $\Sigma$-additive category, is cartesian, and 
  is equipped with a left $\Sigma$ summability structure.
  The following assertions are equivalent
  \begin{enumerate}
      \item the summability structure satisfies \ref{ax:S-with-left};
      \item $\category$ is cartesian left $\Sigma$-additive.
  \end{enumerate}
\end{proposition}

\begin{proof}
  Same as \cref{prop:S-with-equivalent}
\end{proof}

Assuming that the projections are $\Sigma$-additive, we can define as in 
\cref{eq:SprodDist} the morphism
\begin{equation} \label{eq:DprodDist}
  \DprodDist = \prodPairing{\S \prodProj_i} \in 
  \category(\S (\withFam X_i), \withFam \S X_i) \, .
\end{equation}

\begin{proposition} \label{prop:prodSwap-inverse-left}
  Assume that $\category$ is equipped with a left $\Sigma$-summability
  structure and that the projections $\prodProj_i$ are $\Sigma$-additive.
The following are equivalent. \begin{enumerate}
\item $\sequence[j]{\withFam \Sproj_j}$ is summable;
\item $\DprodDist$ is an isomorphism.
\end{enumerate}
Then, $\DprodDist^{-1} = \Spairing[j]{\withFam \Sproj_j}$.
\end{proposition}

\begin{proof}
Same as \cref{prop:prodSwap-inverse}.
\end{proof}

\begin{proposition} \label{prop:pairing-additive}
  Assume that $\category$ is a cartesian left 
  $\Sigma$-additive category.
  If for all $i \in I$, $h_i \in \category(X, Y_i)$ is $\Sigma$-additive, then 
  $\prodPairing{h_i}$ is also $\Sigma$-additive.
  Thus, $\categoryAdd$ is a cartesian $\Sigma$-additive category whose cartesian product 
  coincides with the one of $\category$.
\end{proposition}

\begin{proof} Assume that $\family{f_a \in \category(Z, X)}$ is summable.
  We need to show that $\family{\prodPairing{h_i} \comp f_a}
  = \family{\prodPairing{h_i \comp f_a}}$ is summable and that
  \[ \sum_{a \in A} \prodPairing{h_i \comp f_a}
  = \prodPairing{h_i} \comp (\sum_{a \in A} f_a) \, . \] 
But
$\prodPairing{h_i} \comp (\sum_{a \in A} f_a) = \prodPairing{h_i \comp
  (\sum_{a \in A} f_a)}$ so the equality directly follows from
\cref{eq:sum-pairing} and from the additivity of each $h_i$.
  Furthermore, the projections are $\Sigma$-additive by \cref{prop:left-sum-cartesian},
  and $\top$ is also a terminal object in $\categoryAdd$ because the 
  unique morphism in $\category(X, \top)$ is necessarily $0$, which is $\Sigma$-additive.
\end{proof}

So if $\category$ is a cartesian left $\Sigma$-additive equipped with a left 
$\Sigma$-summability structure, $\categoryAdd$ is a cartesian
$\Sigma$-additive category equipped with a $\Sigma$-summability structure.
Notice then that both $\DprodDist$ and $\DprodDist^{-1}$ are a morphism in $\categoryAdd$, and 
they coincide with their counterpart given in \cref{sec:summability-product-LL}.
Recall then that the monad $\monadS$ on $\categoryAdd$
 equipped with $\Sinj_0$ and $\DprodDist^{-1}$ is a lax symmetric monoidal monad.
 This canonically induces two strengths $\SprodstrL$ and $\SprodstrR$
 as defined in \cref{eq:Sprod-strength}, turning $\monadS$ into a
 commutative monad (with respect to the cartesian product).

  \subsubsection{Differential structure and cartesian product}

  \label{sec:product-differential}

  We assume that $\category$ is equipped with a left $\Sigma$-summability structure
  $(\S, \vect \Sproj)$ that satisfies \ref{ax:S-with-left}.

  \begin{definition} A Taylor expansion is compatible with
    the cartesian product if the projections $\prodProj_i$ are $\D$-linear.
    A \emph{cartesian Taylor category} is a cartesian category equipped 
    with a Taylor expansion that is compatible with the cartesian product.
    A \emph{cartesian analytic category} is a cartesian category equipped 
    with an analytic Taylor expansion that is compatible with the cartesian product.
  \end{definition}

  We now assume that $\category$ is a cartesian Taylor category (or a cartesian analytic 
  category)

  \begin{proposition} \label{prop:differential-pairing}
    $\prodPairing{\D f_i} = \DprodDist \comp \D \prodPairing{f_i}$.
  \end{proposition}
  
  \begin{proof} We have 
    $\prodProj_i \comp \DprodDist \comp \D \prodPairing{f_i}
    = \S \prodProj_i \comp \D \prodPairing{f_i}
    = \D \prodProj_i \comp \D \prodPairing{f_i}
    = \D (\prodProj_i \comp \prodPairing{f_i}) = \D f_i$,
    and we conclude by uniqueness of the tupling.
  \end{proof}

  \begin{proposition} \label{prop:pairing-linear}
    If $(h_i \in \category(X, Y_i))_{i\in I}$ is a collection of $\D$-linear 
    morphisms, then $\prodPairing{h_i}$ is $\D$-linear. If 
    $(h_i' \in \category(X_i, Y_i))_{i\in I}$ is a collection of $\D$-linear 
    morphisms, then $\withFam h_i'$ is $\D$-linear. In particular, 
    $\categoryLin$ is a cartesian category.
  \end{proposition}

  \begin{proof} First, we know that $\prodPairing{h_i}$ is $\Sigma$-additive thanks to
    \cref{prop:pairing-additive}. Then,
    \[ \D \prodPairing{h_i} = \DprodDist^{-1} \comp \prodPairing{\D h_i}
    = \DprodDist^{-1} \comp \prodPairing{\S h_i} \]
    by \cref{prop:differential-pairing}.
    So $\Sproj_j \comp \D \prodPairing{h_i}
    = (\withFam \Sproj_j) \comp \prodPairing{\S h_i}$
    using that $\DprodDist^{-1} = \Spairing[j]{\withFam \Sproj_j}$ 
    (\cref{prop:prodSwap-inverse}). Then,
    \[ \left( \withFam \Sproj_j \right) \comp \prodPairing{\S h_i}
    = \prodPairing{\Sproj_j \comp \S h_i}
    = \prodPairing{h_i \comp \Sproj_j}
    = \prodPairing{h_i} \comp \Sproj_j \, .\] 
    So $\Sproj_j \comp \D \prodPairing{h_i} = \prodPairing{h_i} \comp \Sproj_j$. 
    We conclude that $\D \prodPairing{h_i} = \S \prodPairing{h_i}$ by joint 
    monicity of the $\Sproj_j$.
    The $\D$-linearity of $\withFam h_i'$ then follows directly from above thanks
    to the fact that $\withFam h_i' = \prodPairing{h_i' \comp \prodProj_i}$ and that 
    the composition of two $\D$-linear morphism is also $\D$-linear. We 
    conclude that $\categoryLin$ is cartesian because 
    $\final_X = 0$ is also $\D$-linear by \ref{ax:D-sum-lin}, so $\top$ is also
    final in $\categoryLin$.
  \end{proof}

  The $\D$-linearity of the projections imply that 
  \[ \DprodDist = \prodPairing{\S \prodProj_i} 
  = \prodPairing{\D \prodProj_i} \in \category(\D \withFam X_i, \withFam \D X_i) \]
  is a natural transformation. 
  Since $\DprodDist$ is invertible, it is then a natural isomorphism 
  and $\DprodDist^{-1} \in \category(\withFam \D X_i, \D \withFam X_i)$
  is natural. Furthermore, both $\DprodDist$ and 
  $\DprodDist^{-1}$ are $\D$-linear, thanks to 
  \cref{prop:prodSwap-inverse,prop:pairing-linear}.
  Then for the same reasons that $(\monadS, \Sinj_0, \DprodDist^{-1})$
  is a lax symmetric monoidal monad in $\categoryAdd$
  (and on $\categoryLin$),
  $(\monadD, \Sinj_0, \DprodDist^{-1})$
  is a lax symmetric monoidal monad on $\category$. 
  
  \begin{remark} Similarly to \cref{sec:summability-product-LL},
    the invertibility of $\DprodDist$ and 
    $\Sinj_0$ ensure that $\kleisliD$, the Kleisli category of the 
    monad $\monadD$, is cartesian.
    This can be shown either by hand, or by following the reasoning of
    \cref{rem:extension-product-KleisliS}.
  \end{remark}
  
  As seen 
  in \cref{sec:commutative-monad} (taking the symmetric monoidal structure
  to be the one generated by the cartesian structure $\with$), 
  the lax monoidality of $\monadD$ implies that $\monadD$ is a commutative monad,
  whose left and right strengths are given by the following equations.
  \todo{Index from 1}
  \begin{equation} \label{eq:D-strength}
    \begin{split} 
      \DprodstrL = \DprodDist^{-1} \comp (\D X_0 \with \Sinj_0)
  = \Sbracket{\Sproj_0 \with X_1, \Sproj_1 \with 0, \Sproj_2 \with 0,
  \ldots}\in \category(\D X_0 \with X_1, \D (X_0 \with X_1)) \\ 
  \DprodstrR = \DprodDist^{-1} \comp (\Sinj_0 \with \D X_1)
  = \Sbracket{X_0 \with \Sproj_0, 0 \with \Sproj_1, 0 \with \Sproj_2,
  \ldots}\in \category(X_0 \with \D X_1, \D (X_0 \with X_1))
    \end{split}
  \end{equation}
  Those strengths are $\D$-linear, and coincides with the strengths
  $\SprodstrL$ and $\SprodstrR$
  associated to the lax monoidal 
  monad $\monadS$ on $\categoryLin$ and $\categoryAdd$.

  \begin{definition} Let $f \in \category(X_0 \with X_1, Y)$. 
    Define $\D_0 f = \D f \comp \DprodstrL \in \category(\D X_0 \with X_1, \D Y)$ and 
    $\D_1 f = \D f \comp \DprodstrR \in \category(X_0 \with \D X_1, \D Y)$.
  \end{definition}

  Intuitively, the strength maps a summable family
  $\Spairing{x_i}$ and an element $y$ to the family
  $\Sbracket{(x_0, y), (x_1, 0), (x_2, 0), \ldots}$. 
  So $\D_0 f$ performs the Taylor expansion 
  of $f$ on this family. In particular, the coefficient 
  at position $1$ should be seen as 
  $\derive[x_0, y]{f} \cdot (x_1, 0) = \partial_0 f(x_0, y) \cdot x_1$ where 
  $\partial_0$ is the partial derivative of $f$ with regard to its first argument. 
  So $\D_0, \D_1$ are the infinitary counterpart of the notion of partial derivatives:
  the Taylor expansion of $f$ is computed only with regard to a variation on
  one input.

  This theory of partial Taylor
  expansions behaves very nicely, and is crucial for the
  development of a syntax. Let us stress that combining Taylor
  expansions with regard to different arguments would be
  combinatorially very heavy, but our categorical point of view on
  Taylor expansion turns all those computations into simple naturality
  equations.
  We refer the reader to \cite{Walch23}
  for a development of this theory in the framework of coherent 
  differentiation.

  \begin{definition}
    A morphism $f \in \category(Y_0 \with Y_1, Z)$ is $\Sigma$-additive in its first argument
    if for all family $\family{h_a \in \category(X, Y_0)}$, 
    \[  f \comp ((\sum_{a \in A} h_a) \with Y_1) 
    \sumsub \sum_{a \in A} f \comp (h_a \with Y_1) \, .\]
  Similarly, we define the notion of morphisms $\Sigma$-additive in their
  second argument. A morphism is $\Sigma$-biadditive if it is separately 
  $\Sigma$-additive in both of its arguments.
  \end{definition} 

  As usual, additivity in an argument can be expressed 
  as an interaction with regard to 
  the left $\Sigma$-summability structure.

  \begin{proposition}
    A morphism $f \in \category(Y_0 \with Y_1, Z)$ is $\Sigma$-additive in its first argument
    if and only if $f \comp (0 \with Y_1) = 0$
    and $\sequence{f \comp (\Sproj_i \with Y_1)}$ 
    is summable of sum $f \comp (\Ssum \with Y_1)$. A similar result 
    hold for morphisms $\Sigma$-additive in their second argument.
  \end{proposition}

   \begin{proof} Very similar to the proof of \cref{prop:additive}.
   \end{proof}

  \begin{definition} \label{def:D-multilinear}
    A morphism $h \in \category(X_0 \with X_1, Y)$ is 
    $\D$-linear in its first argument if it is $\Sigma$-additive in that 
    argument and if 
    $\Sproj_i \comp \D_0 h = h \comp (\Sproj_i \with X_1)$.
    It is $\D$-linear in its second argument if it is $\Sigma$-additive in that 
    argument and if 
    $\Sproj_i \comp \D_1 h = h \comp (X_0 \with \Sproj_i)$.
    A morphism is $\D$-bilinear if it is separately 
    linear in both of its arguments.
  \end{definition}

  Those notions can be generalized to arbitrary finite product, defining 
  a strength 
  \[ \Dprodstr{i} \defEq \DprodDist \comp (\id_{X_0} \with \cdots \with 
  \Sinj_0 \with \cdots \with \id_{X_n})
  \in \category(X_0 \with \cdots \with \D X_i \with \cdots \with X_n,
  \D (X_0 \with \cdots \with X_n)). \]
\todo{Index from $1$ to $n$}  
This induces a Taylor expansion with respect to only one parameter, for 
  any $f \in \category(X, Y)$,
  \[ \D_i f \defEq f \comp \Dprodstr{i} \in \category(X_0 \with \cdots \with \D X_i
    \with \cdots \with X_n, \D Y) \, . \]
  It is then possible to define a notion 
  of $\Sigma$-multiadditive morphism and a notion of $\D$-multilinear morphism.
  The latter plays an important role in \cite{Walch23}.

  \subsection{Compatibility with the cartesian closed structure}

  We now assume that $\category$ is a cartesian left $\Sigma$-additive that 
  is closed with respect to its categorical product.
  That is, for all objects $X$ and $A$ 
  there is an object $A \closure X$ and a morphism $\Ev \in \category((A \closure X) \with A, X)$
  such that for any morphism $f \in \category(X \with A, Y)$, 
  there exists a unique  
  morphism $\Cur(f) \in \category(X, A \closure Y)$ such that 
  \[ \Ev (\Cur(f) \with A) = f \, . \] 
  This assumption implies that
  $\Cur: \category(X \with A, Y) \arrow \category(X, A \closure Y)$
  is a bijection whose inverse is given by
  \[ \Uncur (g) = \Ev \comp (g \with A) \, .\]
  Closedness can be seen as an adjunction 
  $\_ \with A \dashv A \closure \_$ for any object $A$,
  of unit $\Cur(\id_{X \with A}) \in \category(X, A \closure (X \with A))$ and 
  co-unit $\Ev \in \category((A \closure X) \with A, X)$. 
  $A \closure \_$ maps a morphism $f \in \category(X, Y)$ to a morphism
  $A \closure f \in \category(A \closure X, A \closure Y)$ defined as 
  $\Cur(f \comp \Ev)$.
  Then $\Cur$ and $\Uncur$ are natural bijections, that is, \cref{eq:Cur-natural}
  below holds
  \begin{equation} \label{eq:Cur-natural}
    \begin{split}
    \Cur(f \comp g \comp (h \with A)) &= (A \closure f) \comp \Cur(g) \comp h \\
    \Uncur((A \closure f) \comp g \comp h) &= f \comp \Uncur(g) \comp (h \with A)
    \end{split}
  \end{equation}
  This section is very similar to \cref{sec:summability-closure-LL}. The difference is 
  that the category is closed with regard to the cartesian product 
  and not to the tensor, and that 
  linearity assumptions need to be explicit.

  \subsubsection{Sums and closedness}

  \begin{definition}
    A cartesian closed left $\Sigma$-additive category $\category$ is a cartesian 
    $\Sigma$-additive category that is closed with regard to the cartesian product 
    and such that for any family $\family{f_a \in \category(X, Y)}$, 
    \begin{equation} \label{eq:sum-closure}
      \sum_{a \in A} \Cur(f_a) 
    \sumiff \Cur \left(\sum_{a \in A} f_a \right) \, .
    \end{equation}
   \end{definition}

   \begin{proposition} \label{prop:sum-closure-left}
    For any cartesian left $\Sigma$-additive category $\category$ that is closed,
    $\Ev \in \category((A \closure X) \with A, X)$ is $\Sigma$-additive in 
    its first argument if and only if
    \begin{equation} \label{eq:ev-left-additive}
      \sum_{a \in A} \Cur(f_a) \sumsub \Cur \left(\sum_{a \in A} f_a \right) \, .
    \end{equation}
    As a result, the following assertions are equivalent.
    \begin{enumerate}
    \item $\category$ is a cartesian closed left $\Sigma$-additive category;
    \item $\Ev$ is $\Sigma$-additive in its first argument and 
    for any summable family $\family{f_a \in \category(X \with A, Y)}$, 
    $\family{\Cur(f_a)}$ is summable.
    \end{enumerate}
  \end{proposition}

  \begin{proof}
  The forward implication is the same proof as \cref{prop:sum-closure}, 
  using the additivity of $\Ev$ in its first 
  argument instead of the distributivity of the sum over the tensor product. 
  The only additional thing to prove is that the left additivity of $\Ev$ follows from
  \cref{eq:ev-left-additive}: 
  \begin{align*} 
    \Ev \comp ((\sum_{a \in A} h_a) \with A) 
  &= \Ev \comp \left( \left(\sum_{a \in A} \Cur(\Uncur(h_a)) \right) \with A\right) \\
  &\sumsub \Ev \comp \left(\Cur\left(\sum_{a \in A} \Uncur(h_a)\right) \with A \right) 
  \tag*{by \cref{eq:ev-left-additive}}\\
  &= \sum_{a \in A} \Uncur(h_a) = \sum_{a \in A} \Ev \comp (h_a \with A) \, . \tag*{\qedhere}
  \end{align*}
  \end{proof}

  As usual, the compatibility between the sum and the closedness can 
  be rephrased as a property on the left $\Sigma$-summability structure.

  \begin{definition}  \labeltext{($\S\with$-fun)}{ax:S-with-fun}
    Assume that $\category$ is equipped with a left $\Sigma$-summability
    structure that satisfies \ref{ax:S-with-left}. Then this 
    structure satisfies \ref{ax:S-with-fun} if $\Ev$ is $\Sigma$-additive in its first 
    argument and if $\sequence{A \closure \Sproj_i}$ is summable.
  \end{definition}

  \begin{proposition} \label{prop:S-fun-with-equivalent}
    Assume that $\category$ is a cartesian left $\Sigma$-additive category 
    and is equipped with a left $\Sigma$-summability structure.
    The following assertions are equivalent \begin{enumerate}
      \item The summability structure satisfies \ref{ax:S-with-fun};
      \item $\categoryLL$ is a cartesian closed left $\Sigma$-additive category.
    \end{enumerate}
  \end{proposition}

  \begin{proof}
  The proof is very similar to the proof of \cref{prop:S-fun-equivalent}.
  \end{proof}

  The property \ref{ax:S-with-fun} consists in an isomorphism
  property. Assume $\Ev$ is $\Sigma$-additive in its first argument.
  Then we can define $\DclosDist$ by
    \begin{equation}  \label{eq:Dclos}
      \DclosDist \defEq \Cur \Spairing{\Ev \comp (\Sproj_i \with A)} 
    \in \category(\S (A \closure X), A \closure \S X) \, .
    \end{equation}

    As observed in \cref{rem:lin-closure-monic}, the morphisms 
    $\sequence{A \closure \Sproj_i}$ are jointly monic, so 
    $\DclosDist$ is characterized by the equation of \cref{prop:Dclos-equation}
    below.

    \begin{lemma} \label{prop:Dclos-equation}
      $(A \closure \Sproj_i) \comp \DclosDist = \Sproj_i$.
    \end{lemma}
    
    \begin{proof} We use the naturality of $\Cur$.
      \begin{align*}
      (A \closure \Sproj_i) \comp \DclosDist &= 
        (A \closure \Sproj_i) \comp \Cur
        \Spairing{\Ev \comp (\Sproj_i \with A)} \\
        &= \Cur(\Sproj_i \comp 
          \Spairing{\Ev \comp (\Sproj_i \with A)}) 
        & \quad \text{By naturality of $\Cur$}\\
          &= \Cur(\Ev \comp (\Sproj_i \with A)) \\
          &= \Cur(\Uncur(\Sproj_i)) = \Sproj_i \tag*{\qedhere}
      \end{align*}
    \end{proof}

\begin{proposition} \label{prop:Dclos-inverse} %
  Assume that $\Ev$ is $\Sigma$-additive in its first argument. 
  The following assertions are equivalent:
  \begin{enumerate}
    \item $\sequence{A \closure \Sproj_i}$ is summable;
    \item $\DclosDist$ is an isomorphism.
  \end{enumerate}
  And then,
  \(\Spairing{A \closure \Sproj_i} =\DclosDist^{-1}\).
\end{proposition}

\begin{proof}
  Same proof as \cref{prop:Sclos-inverse}.
\end{proof}

\subsubsection{Differential structure and closedness}
\label{section:differential-and-closure}

We now assume that $\category$ is a cartesian closed left
$\Sigma$-additive category, equipped with a left $\Sigma$-summability
structure.

\begin{definition} A Taylor expansion that is compatible 
  with the cartesian product is compatible with
  the closedness if $\Ev \in \category((A \closure X) \with A, X)$
  is $\D$-linear in its first argument.
  
  A cartesian closed Taylor category is a cartesian closed category equipped 
  with a Taylor expansion compatible with the closedness.
  A cartesian closed analytic category is a cartesian closed category equipped 
  with an analytic Taylor expansion compatible with the closedness.
\end{definition}

We now assume that the Taylor expansion is compatible with the closedness.
This implies that $\Sproj_i \comp \D_0 \Ev
= \Ev \comp (\Sproj_i \with A)$ so
\[ \DclosDist \defEq \Cur \Spairing{\Ev \comp (\Sproj_i \with \id_A)} 
 = \Cur (\D \Ev \comp \DprodstrL_{A \closure X, A})
 \in \category(\D (A \closure X), A \closure \D X) \]
and hence $\DclosDist$ is a natural transformation $\D (A \closure \_) \naturalTrans 
(A \closure \D \_)$. 
The definition of this morphism is very similar to that of $\Sclos$, see \cref{eq:Sclos}.
The natural transformation $\DclosDist$ 
is the pointwise structure of $\D$ on $A \closure \_$ 
defined from the strength $\DprodstrL$, as it is standard, see \cite{Kock71}.

As in, \cref{sec:summability-closure-LL},
the invertibility of $\DclosDist$ turns $\kleisliD$ into a cartesian closed 
  category. The internal hom of $(A, X)$ is taken to be $(A \closure X, \kleisliLD(\Ev_A))$
  (where $\kleisliLD$ is the canonical functor $\category \arrow \kleisliD$).
  The Curry transpose of $f \in \kleisliD(X \with A, Y)$ is defined as 
  the following composition.
  \[
  \begin{tikzcd}
    X & {(A \closure \D Y)} & {\D (A \closure Y)}
    \arrow["{\Cur(f)}", from=1-1, to=1-2]
    \arrow["{(\DclosDist)^{-1}}", from=1-2, to=1-3]
  \end{tikzcd} \, . \] 
  The fact that $\kleisliD$ is closed can be checked by hand, but is also a consequence of 
  a more general categorical observation very similar to the one of  
  \cref{rem:extension-closure-KleisliS},
  using the mate construction reviewed in \cref{sec:mate}.

\subsection{Cartesian closed analytic categories arising from LL}

\label{sec:Taylor-from-ll}

We show in this section that the coKleisli category of a resource
comonad on an analytic 
category is a cartesian (closed) analytic category.
Let $\categoryLL$ be a Taylor category, see \cref{def:Taylor-category}.
Then as seen in \cref{sec:summability-comonad}, 
$(\S, \sequence{\Der \Sproj_i})$
is a left $\Sigma$-summability structure on $\kleisliExp$ and 
\begin{equation} \label{eq:der-witness-sum}
  \Der \Spairing{h_i} = \Spairing{\Der h_i} \quad \quad 
\Der \psum h_i = \psum \Der h_i \, .
\end{equation}
Let $\Sinj_0, \SmonadSum, \Slift, \Sswap$ be the morphisms 
in $\categoryLL$ defined in \cref{thm:bimonad-structure} 
from the $\Sigma$-summability structure 
$(\S, \sequence{\Sproj_i})$. 
Then, \cref{eq:der-witness-sum} ensures that their 
counterpart in $\kleisliExp$, defined in 
\cref{thm:bimonad-structure} from the left $\Sigma$-summability 
structure $(\S, \sequence{\Der \Sproj_i})$,
are respectively equal to $\Der \Sinj_0$,  $\Der \SmonadSum$,
$\Der \Slift$ and $\Der \Sswap$.
Then as discussed in \cref{sec:Taylor}, the Taylor structure 
induces a functor $\D$ on $\kleisliExp$ for which 
$\Der \Sproj_0$,
$\Der \Sswap$, $\Der \SmonadSum$,
$\Der \Sinj_0$ and $\Der \Slift$ are natural transformations.
If the Taylor structure is analytic, $\Der \Ssum$ is also natural.
Thus, it only suffices to show that $\Der \Sproj_0$,
$\Der \Sswap$ and $\Der \Ssum$ are $\D$-linear to conclude that 
$\kleisliExp$ is a Taylor category (respectively, an analytic category).
This is a consequence of the fact that linearity in the sense of LL 
implies $\D$-linearity.

\begin{definition} \label{def:linear}
  Following the standard terminology of LL,
  a morphism $f \in \kleisliExp(X, Y)$ is 
  \emph{linear} if $f = \Der(h)$ for some $h \in \categoryLL(X, Y)$. 
\end{definition}

\begin{proposition} \label{prop:linear-is-D-linear}
  Every linear morphism is also 
  $\D$-linear.
\end{proposition}

\begin{proof} Let $f = \Der h \in \kleisliExp(X, Y)$. 
  By \cref{prop:additivity-der}, $\Der h$ is $\Sigma$-additive. 
  Furthermore, \ref{ax:Sdl-chain} ensures that $\D$ extends $\S$,
  meaning that $\D (\Der h) = \Der (\S h)$. Then
  \[ \Der \Sproj_i \comp \D (\Der h)
  = \Der \Sproj_i \comp \Der (\S h) 
  = \Der (\Sproj_i \compl \S h) 
  = \Der(h \compl \Sproj_i)
  = \Der h \comp \Der \Sproj_i \] so $\Der h$ is $\D$-linear.
\end{proof}
Then by \cref{prop:linear-is-D-linear},
$\Der \Sproj_i$, $\Der \Ssum$, $\Der 0 = 0$ are all $\D$-linears, 
and we just proved the following result.

\begin{proposition} \label{prop:kleisli-analytic}
  If $\categoryLL$ is a Taylor category, then $\kleisliExp$ has a Taylor expansion.
  If $\categoryLL$ is an analytic category, then $\kleisliExp$ has an analytic structure.
\end{proposition}

\begin{remark} There are three layers of linearity: 
  additivity, $\D$-linearity and linearity, with 
  the following string of implications.
  \[
\begin{tikzcd}[column sep = huge]
  \text{linearity} \arrow[r, "\text{\cref{prop:linear-is-D-linear}}", Rightarrow] 
  & \text{$\D$-linearity} \arrow[r, "\text{by definition}", Rightarrow] 
  & \text{additivity.}
  \end{tikzcd}
    \]
  As discussed in \cite{Blute09}, additivity does not necessarily imply 
  $\D$-linearity. The link between $\D$-linearity and linearity 
  should be investigated further, drawing inspirations from 
  ~\cite{Blute14,Garner21} that study under which conditions 
  a cartesian differential category is the coKleisli category 
  of its category of linear morphisms.
\end{remark}

We now show that the Taylor expansion given in \cref{prop:kleisli-analytic} 
is compatible with the cartesian closedness of $\kleisliExp$.
We crucially rely on \cref{prop:linear-is-D-linear} and on 
a generalization of this property to multilinear morphisms.

Remember that the category $\kleisliExp$ is cartesian. The cartesian product 
$\withFam X_i$ is the same as the one in $\categoryLL$, and the projections
are $\Der \prodProj_i$. 
By \cref{prop:linear-is-D-linear} above, these projections
are $\D$-linear. In particular, they are $\Sigma$-additive. 

Let $\SprodDist \in \categoryLL(\S (\withFam X_i), \withFam \S X_i)$ 
be the natural isomorphism of $\categoryLL$ 
defined in \cref{eq:SprodDist} from the $\Sigma$-summability 
structure $(\S, \sequence{\Sproj_i})$.
Then, the distributivity of the left $\Sigma$-summability structure 
$(\S, \sequence{\Der(\Sproj_i)})$ 
over the categorical product of $\kleisliExp$ given 
in \cref{eq:DprodDist} is equal to
\[ \prodPairing{\Sadd (\Der \prodProj_i)}
= \prodPairing{\Der (\S \prodProj_i)}
= \Der \prodPairing{\S \prodProj_i} = \Der \DprodDist \, .\]
This morphism is an isomorphism, with inverse $\Der \SprodDist^{-1}$.
We just proved the following result.
\begin{proposition}
  The Taylor expansion in $\kleisliExp$ given in \cref{prop:kleisli-analytic}
  is compatible with the categorical product of $\kleisliExp$.
\end{proposition}
As seen in \cref{sec:product-differential}, the monad $\monadD$
on $\kleisliExp$ is a strong monad, whose left and right strengths
$\DprodstrL$ and $\DprodstrR$ are given by 
\cref{eq:D-strength}. Explicitly,
\[ \DprodstrL \defEq \Der \SprodDist^{-1} \comp (\D X_0 \with \Der\Sinj_0)
= \Der \SprodDist^{-1} \comp \Der (\S X_0 \with \Sinj_0) 
= \Der (\SprodDist^{-1} \compl (\S X_0 \with \Sinj_0)) \, .\]
So $\DprodstrL = \Der \SprodstrL$ 
and similarly $\DprodstrR = \Der \SprodstrR$, where 
$\SprodstrL$ and $\SprodstrR$ are the strength 
of the monad $\monadS$ on $\categoryLL$, given in 
\cref{eq:Sprod-strength}.

The axiom 
\ref{ax:Sdl-with} that has not been used so far essentially ensures that 
a bilinear morphism is $\D$-bilinear.
As shown in \cite{Ehrhard22-pcf}, the axiom \ref{ax:Sdl-with}
implies the commutativity of the following diagrams:
\begin{equation} \label{eq:Sdl-with-partial-L}
\begin{tikzcd}
	{!\S X \tensor !Y } & {S !X \tensor !Y} & {S (!X \tensor !Y)} \\
	{!(\S X \with Y)} & {!\S(X \with Y)} & {\S !(X \with Y)}
	\arrow["{\seelyTwo_{\S X, Y}}"', from=1-1, to=2-1]
	\arrow["{\Sdl_X \tensor !Y}", from=1-1, to=1-2]
	\arrow["{\SstrL_{!X, !Y}}", from=1-2, to=1-3]
	\arrow["{! \SprodstrL}"', from=2-1, to=2-2]
	\arrow["{\Sdl_{X \with Y}}"', from=2-2, to=2-3]
	\arrow["{\S \seelyTwo_{X, Y}}", from=1-3, to=2-3]
\end{tikzcd}
\end{equation}
\begin{equation} \label{eq:Sdl-with-partial-R}
\begin{tikzcd}	{!X \tensor !\S Y} & {!X \tensor \S !Y} & {\S (!X \tensor !Y)} \\
	{!(X \with \S Y)} & {!\S(X \with Y)} & {\S ! (X \with Y)}
	\arrow["{!X \tensor \Sdl_Y}", from=1-1, to=1-2]
	\arrow["\SstrR", from=1-2, to=1-3]
	\arrow["{\seelyTwo_{X, \S Y}}"', from=1-1, to=2-1]
	\arrow["{! \SprodstrR_{X,Y}}"', from=2-1, to=2-2]
	\arrow["{\Sdl_{X \with Y}}"', from=2-2, to=2-3]
	\arrow["{\S \seelyTwo_{X, Y}}", from=1-3, to=2-3]
\end{tikzcd}
\end{equation}
Recall that $\SstrL$ and $\SstrR$ are the strengths of the monad $\monadS$ over 
the monoidal product $\tensor$ of $\categoryLL$, see \cref{sec:summability-tensor}.
Those commutations are not surprising. For example, \cref{eq:Sdl-with-partial-L}
means that $\seelyTwo$ is a morphism
between the distributive law 
$!\S \_ \tensor !Y \naturalTrans \S(!\_ \tensor !Y)$ 
and the distributive law 
$!(\S \_ \with Y) \naturalTrans \S !(\_ \with Y)$. Those 
are similar to the distributive laws involved in \ref{ax:Sdl-with}
except that they use the strengths $\SprodstrL$ and 
$\SstrL$ instead of the lax monoidalities
$\SprodDist^{-1}$ and $\Sdist$.

\begin{definition}
  Following the standard notation of LL, a morphism 
  $f \in \kleisliExp(X \with Y, Z)$ is \emph{bilinear} if 
  there is $h \in \categoryLL(X \tensor Y, Z)$ such that 
  \[ f =
  \begin{tikzcd}
    {!(X \with Y)} & {!X \tensor !Y} & {X \tensor Y} & Z
    \arrow["{(\seelyTwo)^{-1}}", from=1-1, to=1-2]
    \arrow["{\der \tensor \der }", from=1-2, to=1-3]
    \arrow["h", from=1-3, to=1-4]
  \end{tikzcd} \, .\]
  A morphism $f \kleisliExp(X \with Y, Z)$ is \emph{linear in its first argument} 
  if there exists 
  $h \in \categoryLL(X \tensor !Y, Z)$ such that 
  \[ f =
  \begin{tikzcd}
    {!(X \with Y)} & {!X \tensor !Y} & {X \tensor !Y} & Z
    \arrow["{(\seelyTwo)^{-1}}", from=1-1, to=1-2]
    \arrow["{\der \tensor ! Y }", from=1-2, to=1-3]
    \arrow["h", from=1-3, to=1-4]
  \end{tikzcd} \, . \]
  We can define similarly what is a morphism linear in its second argument.
\end{definition}
Observe that a bilinear morphism is linear in both of its arguments.

\begin{proposition} \label{prop:multilinear-is-D-multilinear}
  A morphism which is linear in an argument 
  is $\D$-linear in that argument. A bilinear morphism is $\D$-bilinear.
\end{proposition}

\begin{proof} Assume that $f$ is linear in its first argument:
$f = h \compl (\der \tensor !Y) \compl (\seelyTwo_{X, Y})^{-1}$. 
Then for any $g \in \categoryLL(A, X)$
\begin{align*}
  f \comp (\Der g \with Y) &= f \comp \Der(g \with Y) \\
  &= h \compl (\der \tensor !Y) \compl (\seelyTwo_{X, Y})^{-1} \compl !(g \with Y)
  \tag*{by \cref{prop:Der-composition}} \\
  &= h \compl (\der \tensor !Y) \compl (!g \tensor !Y) \compl (\seelyTwo_{X, Y})^{-1} 
  \tag*{by naturality of $\seelyTwo$} \\
  &= h  \compl (g \tensor !Y) \compl (\der_{\S X} \tensor !Y) \compl (\seelyTwo_{X, Y})^{-1} 
  \tag*{by naturality of $\der$.} \\
\end{align*}
So $f \comp (0 \with Y) = f \comp (\Der 0 \with Y) = 
h  \compl (0 \tensor !Y) \compl (\der_{\S X} \tensor !Y) \compl (\seelyTwo_{X, Y})^{-1}
= 0$ by \ref{ax:S-sm-dist}, by additivity in the morphisms in $\categoryLL$, and 
by left additivity.
Moreover,
\[ f \comp (\Der \Sproj_i \with Y) 
= h  \compl (\Sproj_i \tensor !Y) \compl (\der_{\S X} \tensor !Y) 
\compl (\seelyTwo_{X, Y})^{-1} \] so by \ref{ax:S-sm-dist}, 
by additivity in the morphisms in $\categoryLL$, and by left additivity, the 
sequence $\sequence{f \comp (\Der \Sproj_i \with Y)}$ is summable 
with sum 
\[ h \compl (\Ssum \tensor !Y) \compl (\der_{\S X} \tensor !Y) 
\compl (\seelyTwo_{X, Y})^{-1} = f \comp (\Der \Ssum \with Y) \, . \]
Thus, $f$ is $\Sigma$-additive in its first argument. Furthermore,
\begin{align*}
  \D_0 f &= \D f \comp \Der(\SprodstrL_{X, Y}) \\
  &= \S f \compl \Sdl_{X \with Y} \compl !\SprodstrL_{X, Y} 
  \tag*{by definition and \cref{prop:Der-composition}} \\
  &= \S h \compl \S(\der_X \tensor !Y) \compl \S \seelyTwo_{X, Y}^{-1} 
  \compl \Sdl_{X \with Y} \compl !\SprodstrL_{X, Y} \\
  &=  \S h \compl \S(\der_X \tensor !Y) \compl \SstrL_{!X, !Y} 
  \compl (\Sdl_X \tensor !Y) \compl (\seelyTwo_{\S X, Y})^{-1}
  \tag*{by \cref{eq:Sdl-with-partial-L}} \\
  &=  \S h \compl \SstrL_{X, !Y} \compl (\S \der_X \tensor !Y) 
  \compl (\Sdl_X \tensor !Y) \compl (\seelyTwo_{\S X, Y})^{-1}
  \tag*{by naturality of $\SstrL$} \\
  &= \S h \compl \SstrL_{X, !Y} \compl (\der_{\S X} \tensor !Y) 
  \compl (\seelyTwo_{\S X, Y})^{-1} 
  \tag*{by \ref{ax:Sdl-chain}.} 
\end{align*}
In particular, $\D_0 f$ is also linear in its first argument, and 
\begin{align*}
  \Der \Sproj_i \comp \D_0 f 
  &= \Sproj_i \compl \S h \compl \SstrL_{X, !Y} \compl
  (\der_{\S X} \tensor !Y) \compl (\seelyTwo_{\S X, Y})^{-1}
  \tag*{by what precedes and \cref{prop:Der-composition}} \\
  &= h \compl \Sproj_i \compl \SstrL_{X, !Y} \compl
  (\der_{\S X} \tensor !Y) \compl (\seelyTwo_{\S X, Y})^{-1} 
  \tag*{by naturality of $\Sproj_i$} \\
  &= h \compl (\Sproj_i \tensor !Y) \compl
  (\der_{\S X} \tensor !Y) \compl (\seelyTwo_{\S X, Y})^{-1} \\
  &= h \compl (\der_X \tensor !Y) \compl (!\Sproj_i \tensor !Y)
  \compl (\seelyTwo_{\S X, Y})^{-1}
  \tag*{by naturality of $\der$} \\
  &= h \compl (\der_X \tensor !Y) \compl (\seelyTwo_{X, Y})^{-1}
  \compl !(\Sproj_i \with Y)
  \tag*{by naturality of $\seelyTwo$} \\
  &= f \comp \Der(\Sproj_i \with Y) 
  \tag*{by \cref{prop:Der-composition}} \\
  &= f \comp (\Der \Sproj_i \with Y)  \, .
\end{align*}
So $f$ is $\D$-linear in its first argument. A similar proof based
on \cref{eq:Sdl-with-partial-R} show that if $f$ is linear in 
its second argument then it is $\D$-linear in that argument. 
Finally, applying both results on a bilinear morphism shows 
that any bilinear morphism is $\D$-bilinear.
\end{proof}

Assume that $\categoryLL$ is closed with regard to $\tensor$,
and that \ref{ax:S-fun} holds.
The category $\kleisliExp$ is closed with regard to its cartesian product. 
The internal hom of $(A, X)$ 
is given by $(A \closure X, \Ev)$ where $A \closure X = !A \linarrow X$
and 
\[ \Ev = \begin{tikzcd}
	{!((!A \linarrow X) \with A)} & & {!(!A \linarrow X) \tensor !A} & & {(!A \linarrow X) \tensor !A} & X
	\arrow["{(\seelyTwo_{!A \linarrow X, A})^{-1}}", from=1-1, to=1-3]
	\arrow["{\der_{!A \linarrow X} \tensor !A}", from=1-3, to=1-5]
	\arrow["\ev", from=1-5, to=1-6]
      \end{tikzcd} \, . \]
If $f \in \kleisliExp(X \with A, Y)$ then 
$f \compl \seelyTwo_{X, A} \in \categoryLL(!X \tensor !A, Y)$ and
$\Cur(f) = \cur(f \compl \seelyTwo_{X}) \in \kleisliExp(X, !A \linarrow Y)$.

Observe that $\Ev$ is linear in its first argument, so
$\Ev$ is $\Sigma$-additive and $\D$-linear in that argument thanks to
\cref{prop:multilinear-is-D-multilinear}.
Furthermore, for any sequence 
$\sequence{f_i \in \kleisliExp(X \with A, Y)}$ that is summable, 
$\sequence{f_i \compl \seelyTwo}$ is summable by left additivity, 
so $\sequence{\Cur(f_i)}$ is summable by \ref{ax:S-fun}
and \cref{prop:Sclos-inverse}. By \cref{prop:sum-closure-left},
this implies that $\kleisliExp$ is a cartesian closed left $\Sigma$-additive category, 
and that the Taylor expansion given in \cref{prop:kleisli-analytic} 
is compatible with the closedness.
To summarize everything, we have proved the 
following result.

\begin{theorem} For any Taylor category $\categoryLL$ that is closed and that 
  follows \ref{ax:S-fun}, $\kleisliExp$ is a cartesian closed Taylor category.
  For any analytic category $\categoryLL$ that is closed and that 
  follows \ref{ax:S-fun}, $\kleisliExp$ is a cartesian closed analytic category.
\end{theorem}

\section{Categorical complements: adjunction, mates and distributive law}
\label{sec:mate}

We review in this section the mate 
construction, a bijection between natural transformations involving 
left and right adjoint functors. 
The main application of this section is in 
\cref{sec:elementary}, in which the bimonad $\S$ developed in 
\cref{sec:summability-structure} is equal to $\Dbimon \linarrow \_$
for some object $\Dbimon$, so 
the adjunction $\_ \tensor \Dbimon \dashv \Dbimon \linarrow \_$ induces 
a mate construction that relates the bimonad structure of $\S$ to a 
bimonad structure on $\_ \tensor \Dbimon$.
We suggest the reader to read \cref{sec:elementary} first,
and to only refer to this background section when necessary. 

First, we introduce in \cref{sec:mate-construction} the mate bijection
and motivate that this bijection is compositional.
We review in \cref{sec:mate-monade-comonade}
how the mate induces a bijection between monads 
and comonad structures on the adjoint pair.
Then we use the compositionality of the mate construction 
to prove in \cref{sec:compatibility-adjunction} that 
this bijection preserves distributive laws and morphisms of 
distributive laws.
As such, the mate construction 
also relates lax and oplax monoidal structures, as shown in 
\cref{sec:mate-oplax-lax}.
Finally, we use those result in \cref{sec:extension-adjunction}
to give a direct proof of a fairly known result on the extension 
of adjunctions.

\subsection{Adjunctions and the mate construction}

\label{sec:mate-construction}

We recall definitions on adjunctions and on the mate construction. We refer 
the reader to~\cite{Kelly74} for more details on the content of this section.
An adjunction between two functors $\ladj : \cat \arrow \catbis$ and 
$\radj : \catbis \arrow \cat$ 
consists of two natural
transformations $\unit : \idfun_{\cat} \naturalTrans \radj \ladj$ 
and $\counit : \ladj \radj \naturalTrans \idfun_{\catbis}$ such that the 
following diagrams of natural transformations commute.
\[
	\begin{tikzcd}[column sep = small]
		& {\ladj \radj \ladj} \\
		\ladj && \ladj
		\arrow["{\counit \ladj}", from=1-2, to=2-3]
		\arrow["{\ladj \unit}", from=2-1, to=1-2]
		\arrow[Rightarrow, no head, from=2-1, to=2-3]
	\end{tikzcd} \quad \quad
	\begin{tikzcd}[column sep = small]
	& {\radj \ladj \radj} \\
	\radj && \radj
	\arrow["{\radj \counit}", from=1-2, to=2-3]
	\arrow["{\unit \radj}", from=2-1, to=1-2]
	\arrow[Rightarrow, no head, from=2-1, to=2-3]
	\end{tikzcd} \]
Those two equations 
are called the triangle identities. The functor $\ladj$ is called 
the left adjoint, $\radj$ is called the right adjoint, 
$\unit$ is called the unit, and $\counit$ is called the co-unit. We write 
an adjunction
$\unit, \counit : \ladj \dashv \radj$.
	
	The unit and co-unit induce 
	two natural bijections
	$\isoadjL_{X, Y} : \catbis(\ladj X, Y) \arrow \cat(X, \radj Y)$ and 
	$\isoadjR_{X,Y} : \cat(X, \radj Y) \arrow \catbis(\ladj X, Y)$ inverse of 
	each other, defined as follows: 
	\[ \isoadjL_{X, Y}(f) = \radj f \comp \unit_X \quad \quad 
	\isoadjR_{X, Y}(g) = \counit_Y \comp \ladj g \, .\] 
	In fact, an adjunction can also be given by two natural bijections
	$\isoadjL_{X, Y} : \catbis(\ladj X, Y) \arrow \cat(X, \radj Y)$ and 
	$\isoadjR_{X,Y} : \cat(X, \radj Y) \arrow \catbis(\ladj X, Y)$ inverse of 
	each other. The unit and the co-unit are defined from those bijections
	taking $\unit_X = \isoadjL_{X, \ladj X}(\id_{\ladj X})$ and 
	$\counit = \isoadjR_{\radj X, X}(\id_{\radj X})$. 
	
	There is a category whose objects are the
	categories, and 
	whose morphisms are the adjoint pairs. 
	The identity is the trivial adjoint pair $\id, \id : \idfun \dashv \idfun$.
	The composition of an adjoint pair $\unit[1], \counit[1] : \ladj[1] \dashv \radj[1]$ 
	with an adjoint pair $\unit[2], \counit[2] : \ladj[2] \dashv \radj[2]$
	is defined as the adjoint pair
	$\unit, \counit : 
	\ladj[2] \ladj[1] \dashv \radj[1] \radj[2]$ where
	\begin{center} 
	$\unit \defEq 
	\begin{tikzcd}
		\idfun & {\radj[1] \ladj[1]} & {\radj[1] \radj[2] \ladj[2] \ladj[1].}
		\arrow["{\unit[1]}", from=1-1, to=1-2]
		\arrow["{\radj[1] \unit[2] \ladj[1]}", from=1-2, to=1-3]
	\end{tikzcd}$
	
	$\counit \defEq 
	\begin{tikzcd}
		{\ladj[2] \ladj[1] \radj[2] \radj[1]} & {\ladj[2] \radj[2]} & \idfun .
		\arrow["{\counit[2]}", from=1-2, to=1-3]
		\arrow["{\ladj[2] \counit[1] \radj[2]}", from=1-1, to=1-2]
	\end{tikzcd}$
	\end{center}
	It is straightforward to check that $\unit$ and $\counit$ follow the triangle equalities.
	If $\isoadjL^i, \isoadjR^i$ are the natural bijections associated with 
	$\unit_i, \counit_i : \ladj_i \dashv \radj_i$, then the natural bijections 
	associated to $\unit, \counit : \ladj[2] \ladj[1] \dashv \radj[1] \radj[2]$ are
	$\isoadjL(f) = \isoadjL^1 (\isoadjL^2(f))$ and 
	$\isoadjR(g) = \isoadjR^2(\isoadjR^1(g))$.
	
	We assume throughout this section that we have the following adjunctions, of respective units
	$\unit, \unit['], \unit['']$ and counits $\counit, \counit['], \counit['']$.
\[\begin{tikzcd}
	\cat & \catbis
	\arrow[""{name=0, anchor=center, inner sep=0}, "\ladj", curve={height=-12pt}, from=1-1, to=1-2]
	\arrow[""{name=1, anchor=center, inner sep=0}, "\radj", curve={height=-12pt}, from=1-2, to=1-1]
	\arrow["\dashv"{anchor=center, rotate=-90}, draw=none, from=0, to=1]
\end{tikzcd} \ \ 
\begin{tikzcd}
	\cat['] & \catbis[']
	\arrow[""{name=0, anchor=center, inner sep=0}, "{\ladj[']}", curve={height=-12pt}, from=1-1, to=1-2]
	\arrow[""{name=1, anchor=center, inner sep=0}, "{\radj[']}", curve={height=-12pt}, from=1-2, to=1-1]
	\arrow["\dashv"{anchor=center, rotate=-90}, draw=none, from=0, to=1]
\end{tikzcd}
\begin{tikzcd}
	\cat[''] & \catbis['']
	\arrow[""{name=0, anchor=center, inner sep=0}, "{\ladj['']}", curve={height=-12pt}, from=1-1, to=1-2]
	\arrow[""{name=1, anchor=center, inner sep=0}, "{\radj['']}", curve={height=-12pt}, from=1-2, to=1-1]
	\arrow["\dashv"{anchor=center, rotate=-90}, draw=none, from=0, to=1]
\end{tikzcd}.\]

	\begin{proposition} (Proposition 2.1 of~\cite{Kelly74}) \label{def:mate}
		Let $\lfun : \lcat \arrow \lcat[']$ and 
		$\rfun : \rcat \arrow \rcat[']$ two functors.
		 There exists a bijection called the \emph{mate construction} between
		the natural transformations $\lmate : \ladj['] \lfun \naturalTrans \rfun \ladj$ and 
		the natural transformations $\rmate: \lfun \radj \naturalTrans \radj['] \rfun$.
		This bijection is given by the following pasting 
		diagrams.
		\[\begin{tikzcd}
			\lcat & {\lcat[']} \\
			\rcat & {\rcat[']}
			\arrow["\lfun", from=1-1, to=1-2]
			\arrow["\ladj"', from=1-1, to=2-1]
			\arrow["{\ladj[']}", from=1-2, to=2-2]
			\arrow["\rfun"', from=2-1, to=2-2]
			\arrow["\lmate"', shorten <=6pt, shorten >=6pt, Rightarrow, from=1-2, to=2-1]
		\end{tikzcd} \mapsto 
		\begin{tikzcd}
			\rcat & \lcat & {\lcat[']} & {\lcat[']} \\
			\rcat & \rcat & {\rcat[']} & {\lcat[']}
			\arrow["\lfun", from=1-2, to=1-3]
			\arrow["\ladj"{description}, from=1-2, to=2-2]
			\arrow["{\ladj[']}"{description}, from=1-3, to=2-3]
			\arrow["\rfun"', from=2-2, to=2-3]
			\arrow["\lmate"', shorten <=6pt, shorten >=6pt, Rightarrow, from=1-3, to=2-2]
			\arrow["\radj", from=1-1, to=1-2]
			\arrow[Rightarrow, no head, from=2-1, to=1-1]
			\arrow[Rightarrow, no head, from=2-1, to=2-2]
			\arrow["\counit"', shorten <=6pt, shorten >=6pt, Rightarrow, from=1-2, to=2-1]
			\arrow[Rightarrow, no head, from=1-3, to=1-4]
			\arrow[Rightarrow, no head, from=1-4, to=2-4]
			\arrow["{\radj[']}"', from=2-3, to=2-4]
			\arrow["{\unit[']}"', shorten <=6pt, shorten >=6pt, Rightarrow, from=1-4, to=2-3]
		\end{tikzcd}\]
		\[\begin{tikzcd}
			\rcat & {\rcat[']} \\
			\lcat & {\lcat[']}
			\arrow["\rfun", from=1-1, to=1-2]
			\arrow["\radj"', from=1-1, to=2-1]
			\arrow["\lfun"', from=2-1, to=2-2]
			\arrow["{\radj[']}", from=1-2, to=2-2]
			\arrow["\rmate"', shorten <=6pt, shorten >=6pt, Rightarrow, from=2-1, to=1-2]
		\end{tikzcd} \mapsto
		\begin{tikzcd}
			\lcat & \rcat & {\rcat[']} & {\rcat[']} \\
			\lcat & \lcat & {\lcat[']} & {\rcat[']}
			\arrow["\rfun", from=1-2, to=1-3]
			\arrow["\radj"{description}, from=1-2, to=2-2]
			\arrow["\lfun"', from=2-2, to=2-3]
			\arrow["{\radj[']}"{description}, from=1-3, to=2-3]
			\arrow["\rmate"', shorten <=6pt, shorten >=6pt, Rightarrow, from=2-2, to=1-3]
			\arrow["\ladj", from=1-1, to=1-2]
			\arrow[Rightarrow, no head, from=1-1, to=2-1]
			\arrow[Rightarrow, no head, from=2-1, to=2-2]
			\arrow["\unit"', shorten <=6pt, shorten >=6pt, Rightarrow, from=2-1, to=1-2]
			\arrow[Rightarrow, no head, from=1-3, to=1-4]
			\arrow[Rightarrow, no head, from=1-4, to=2-4]
			\arrow["{\ladj[']}"', from=2-3, to=2-4]
			\arrow["{\counit[']}"', shorten <=6pt, shorten >=6pt, Rightarrow, from=2-3, to=1-4]
		\end{tikzcd}\]
		In other words, $\lmate$ and $\rmate$ are mates if they are defined from each other by
		\begin{center}
		$\rmate = 
		\begin{tikzcd}
			{\lfun \radj} & {\radj['] \ladj['] \lfun \radj} & {\radj['] \rfun \ladj \radj} & {\radj['] \rfun}
			\arrow["{\unit['] \lfun \radj}", from=1-1, to=1-2]
			\arrow["{\radj['] \lmate \radj}", from=1-2, to=1-3]
			\arrow["{\radj['] \rfun \counit}", from=1-3, to=1-4]
		\end{tikzcd}$
		
		$\lmate = 
		\begin{tikzcd}
			{\ladj['] \lfun} & {\ladj['] \lfun \radj \ladj} & {\ladj['] \radj['] \rfun \ladj} & {\rfun \ladj}
			\arrow["{\ladj['] \lfun \unit }", from=1-1, to=1-2]
			\arrow["{\ladj['] \rmate \ladj}", from=1-2, to=1-3]
			\arrow["{\counit['] \rfun \ladj}", from=1-3, to=1-4]
		\end{tikzcd}$
		\end{center}
		\end{proposition}
		
		\begin{remark} \label{rem:mate-simplified}
			Observe that if $\isoadjL_{X, Y}' : \catbis['](\ladj['] X, Y) \arrow 
			\cat['](X, \radj['] Y)$ 
			is the natural bijection associated with
			$\unit['], \counit['] : \ladj['] \dashv \radj[']$, then 
			$\rmate_X = \isoadjL'_{\lfun \radj X, X}(\rfun \counit_X \comp \lmate_{\radj X})$.
			This alternative definition of the right mate will be useful later on when 
			considering the mate construction applied through the adjunction 
			$\_ \tensor A \dashv A \linarrow \_$.
		
		One particular case
		is when $\cat = \cat[']$, $\catbis = \catbis[']$ and $\lfun = \rfun = 
		\idfun$ so that there is a bijection between the natural transformations 
		$\lmate : \ladj['] \naturalTrans \ladj$ and $\rmate : \radj \naturalTrans 
		\radj[']$ given by
		\begin{center}
		$\rmate = \begin{tikzcd}
			\radj & {\radj['] \ladj['] \radj} & {\radj['] \ladj \radj} & {\radj[']}
			\arrow["{\unit['] \radj}", from=1-1, to=1-2]
			\arrow["{\radj['] \lmate \radj}", from=1-2, to=1-3]
			\arrow["{\radj['] \counit}", from=1-3, to=1-4]
		\end{tikzcd}$
		
		$\lmate = \begin{tikzcd}
			{\ladj[']} & {\ladj['] \radj \ladj} & {\ladj['] \radj['] \ladj} & \ladj
			\arrow["{\ladj['] \unit}", from=1-1, to=1-2]
			\arrow["{\ladj['] \rmate \ladj}", from=1-2, to=1-3]
			\arrow["{\counit['] \ladj}", from=1-3, to=1-4]
		\end{tikzcd}$
		\end{center}
		Then $\rmate_X = \isoadjL'_{\radj X, X}(\counit_X \comp \lmate_{\radj X})$.
		\end{remark}

One can define a double category associated to those
constructions, see~\cite{Kelly74} for a definition of
this construction that we spell out in our setting.
We can define a category of \emph{horizontal morphisms} 
by taking the category $\categorycat$ whose objects are the (small) 
categories and whose morphisms are the functor. 
We can also define a category of \emph{vertical morphisms} by taking the category
of adjoint pairs. Note that those two categories have 
the same objects, but different morphisms.
Then we can define a \emph{square} as the data of two vertical morphisms
$\ladj \dashv \radj$ and $\ladj['] \dashv \radj[']$, two horizontal morphisms $\lfun, \rfun$ and
a natural transformation $\lmate : \ladj['] \lfun \naturalTrans \rfun \ladj$.
We write squares using pasting diagrams.
\[ \begin{tikzcd}
	\cat & {\cat[']} \\
	\catbis & {\catbis[']} 
	\arrow["\ladj"', from=1-1, to=2-1]
	\arrow["\lfun", from=1-1, to=1-2]
	\arrow["\rfun"', from=2-1, to=2-2]
	\arrow["{\ladj[']}", from=1-2, to=2-2]
	\arrow["\lmate", shorten <=7pt, shorten >=7pt, Rightarrow, from=1-2, to=2-1]
\end{tikzcd} \]
There are identity squares defined below, both in the horizontal and vertical directions.
One can check that the left square instantiated in $\ladj = \radj = \idfun$ 
is the same as the right square instantiated in $\lfun = \idfun$.
\[\begin{tikzcd}
	\cat & \cat \\
	\catbis & \catbis
	\arrow[Rightarrow, no head, from=1-1, to=1-2]
	\arrow[Rightarrow, no head, from=2-1, to=2-2]
	\arrow["\ladj"{description}, from=1-1, to=2-1]
	\arrow["\ladj"{description}, from=1-2, to=2-2]
	\arrow["\id", shorten <=7pt, shorten >=7pt, Rightarrow, from=1-2, to=2-1]
\end{tikzcd} \ \ \ \ \
\begin{tikzcd}
	\cat & {\cat[']} \\
	\cat & {\cat[']}
	\arrow[Rightarrow, no head, from=1-1, to=2-1]
	\arrow[Rightarrow, no head, from=1-2, to=2-2]
	\arrow["\lfun", from=1-1, to=1-2]
	\arrow["\lfun"', from=2-1, to=2-2]
	\arrow["\id", shorten <=6pt, shorten >=6pt, Rightarrow, from=1-2, to=2-1]
      \end{tikzcd}\]
It is then possible to compose two squares horizontally or vertically, 
using pasting diagrams. 
\[ \begin{tikzcd}
	\cat & {\cat[']} & {\cat['']} \\
	\catbis & {\catbis[']} & {\catbis['']}
	\arrow["\ladj"', from=1-1, to=2-1]
	\arrow["\lfun", from=1-1, to=1-2]
	\arrow["\rfun"', from=2-1, to=2-2]
	\arrow["{\ladj[']}"{description}, from=1-2, to=2-2]
	\arrow["{\rfun[']}"', from=2-2, to=2-3]
	\arrow["{\lfun[']}", from=1-2, to=1-3]
	\arrow["{\ladj['']}", from=1-3, to=2-3]
	\arrow["\lmate", shorten <=7pt, shorten >=7pt, Rightarrow, from=1-2, to=2-1]
	\arrow["{\lmate[']}", shorten <=7pt, shorten >=7pt, Rightarrow, from=1-3, to=2-2]
\end{tikzcd} \ \ \ \ \
\begin{tikzcd}
	\cat & {\cat[']} \\
	\catbis & {\catbis[']} \\
	\catter & {\catter[']}
	\arrow["\lfun", from=1-1, to=1-2]
	\arrow["\rfun"{description}, from=2-1, to=2-2]
	\arrow["\rrfun"', from=3-1, to=3-2]
	\arrow["{\ladj[1]}"', from=1-1, to=2-1]
	\arrow["{\ladj[2]}"', from=2-1, to=3-1]
	\arrow["{\ladj[1']}", from=1-2, to=2-2]
	\arrow["{\ladj[2']}", from=2-2, to=3-2]
	\arrow["{\lmate[1]}", shorten <=7pt, shorten >=7pt, Rightarrow, from=1-2, to=2-1]
	\arrow["{\lmate[2]}", shorten <=7pt, shorten >=7pt, Rightarrow, from=2-2, to=3-1]
\end{tikzcd}\]
That is, the horizontal composition of the two squares is given by the natural transformation
\[ \rfun['] \lmate \ntcomp \lmate['] \lfun :
 \begin{tikzcd}
	{\ladj[''] \lfun['] \lfun} & {\rfun['] \ladj['] \lfun} & {\rfun['] \rfun \ladj}
	\arrow["{ \lmate['] \lfun}", from=1-1, to=1-2]
	\arrow["{\rfun['] \lmate.}", from=1-2, to=1-3]
\end{tikzcd} \]
and the vertical composition of the two squares is given by the natural transformation
\[ \lmate[2] \ladj[1] \ntcomp \ladj[2'] \lmate[1] :
 \begin{tikzcd}
	{\ladj[2'] \ladj[1'] \lfun} & {\ladj[2'] \rfun \ladj[1]} & {\rrfun \ladj[2] \ladj[1]}
	\arrow["{\ladj[2'] \lmate[1]}", from=1-1, to=1-2]
	\arrow["{\lmate[2] \ladj[1].}", from=1-2, to=1-3]
\end{tikzcd} \]
The horizontal/vertical identity squares 
are neutral with regard to the horizontal/vertical composition, 
the horizontal and vertical compositions are associative, 
and the squares follow the interchange law of double categories (that is, all the 
possible ways to compose blocks of multiple squares are equivalent).

Similarly, there is a double category with the same horizontal and vertical 
morphisms as above, but where the squares are natural transformations
$\rmate: \lfun \radj \naturalTrans \radj['] \rfun$. 
\[ \begin{tikzcd}
	\cat & {\cat[']} \\
	\catbis & {\catbis[']} 
	\arrow["\radj", from=2-1, to=1-1]
	\arrow["\lfun", from=1-1, to=1-2]
	\arrow["\rfun"', from=2-1, to=2-2]
	\arrow["{\radj[']}"', from=2-2, to=1-2]
	\arrow["\rmate", shorten <=7pt, shorten >=7pt, Rightarrow, from=1-1, to=2-2]
\end{tikzcd} \]
The composition of squares is also given by the composition of the pasting diagrams, 
and the identity squares are similar.

\begin{theorem}[\cite{Kelly74}] \label{thm:mate-iso}
    The bijection of mate is an isomorphism of double categories between the 
	two double categories
    described above.
\end{theorem}

The fact that the mate construction is an isomorphism means 
that if the following squares are mate
\[ \begin{tikzcd}
	\cat & {\cat[']} \\
	\catbis & {\catbis[']} 
	\arrow["\ladj"', from=1-1, to=2-1]
	\arrow["\lfun", from=1-1, to=1-2]
	\arrow["\rfun"', from=2-1, to=2-2]
	\arrow["{\ladj[']}", from=1-2, to=2-2]
	\arrow["\lmate", shorten <=7pt, shorten >=7pt, Rightarrow, from=1-2, to=2-1]
\end{tikzcd} \matesequiv 
\begin{tikzcd}
	\cat & {\cat[']} \\
	\catbis & {\catbis[']} 
	\arrow["\radj", from=2-1, to=1-1]
	\arrow["\lfun", from=1-1, to=1-2]
	\arrow["\rfun"', from=2-1, to=2-2]
	\arrow["{\radj[']}"', from=2-2, to=1-2]
	\arrow["\rmate", shorten <=7pt, shorten >=7pt, Rightarrow, from=1-1, to=2-2]
\end{tikzcd} \ \ \ \ \ 
\begin{tikzcd}
	\cat['] & {\cat['']} \\
	\catbis['] & {\catbis['']} 
	\arrow["{\ladj[']}"', from=1-1, to=2-1]
	\arrow["{\lfun[']}", from=1-1, to=1-2]
	\arrow["{\rfun[']}"', from=2-1, to=2-2]
	\arrow["{\ladj['']}", from=1-2, to=2-2]
	\arrow["{\lmate[']}", shorten <=7pt, shorten >=7pt, Rightarrow, from=1-2, to=2-1]
\end{tikzcd} \matesequiv 
\begin{tikzcd}
	\cat['] & {\cat['']} \\
	\catbis['] & {\catbis['']} 
	\arrow["{\radj[']}", from=2-1, to=1-1]
	\arrow["{\lfun[']}", from=1-1, to=1-2]
	\arrow["{\rfun[']}"', from=2-1, to=2-2]
	\arrow["{\radj['']}"', from=2-2, to=1-2]
	\arrow["{\rmate[']}", shorten <=7pt, shorten >=7pt, Rightarrow, from=1-1, to=2-2]
\end{tikzcd}  \]
then their horizontal compositions are also mates.
\[ \begin{tikzcd}
	\cat & {\cat[']} & {\cat['']} \\
	\catbis & {\catbis[']} & {\catbis['']}
	\arrow["\ladj"', from=1-1, to=2-1]
	\arrow["\lfun", from=1-1, to=1-2]
	\arrow["\rfun"', from=2-1, to=2-2]
	\arrow["{\ladj[']}"{description}, from=1-2, to=2-2]
	\arrow["{\rfun[']}"', from=2-2, to=2-3]
	\arrow["{\lfun[']}", from=1-2, to=1-3]
	\arrow["{\ladj['']}", from=1-3, to=2-3]
	\arrow["\lmate", shorten <=7pt, shorten >=7pt, Rightarrow, from=1-2, to=2-1]
	\arrow["{\lmate[']}", shorten <=7pt, shorten >=7pt, Rightarrow, from=1-3, to=2-2]
\end{tikzcd} \matesequiv
\begin{tikzcd}
	\cat & {\cat[']} & {\cat['']} \\
	\catbis & {\catbis[']} & {\catbis['']}
	\arrow["\radj", from=2-1, to=1-1]
	\arrow["\lfun", from=1-1, to=1-2]
	\arrow["\rfun"', from=2-1, to=2-2]
	\arrow["{\radj[']}"{description}, from=2-2, to=1-2]
	\arrow["{\rfun[']}"', from=2-2, to=2-3]
	\arrow["{\lfun[']}", from=1-2, to=1-3]
	\arrow["{\radj['']}"', from=2-3, to=1-3]
	\arrow["\rmate", shorten <=7pt, shorten >=7pt, Rightarrow, from=1-1, to=2-2]
	\arrow["{\rmate[']}", shorten <=7pt, shorten >=7pt, Rightarrow, from=1-2, to=2-3]
\end{tikzcd}  \]
Furthermore, if the following squares are mate
\[ \begin{tikzcd}
	\cat & {\cat[']} \\
	\catbis & {\catbis[']} 
	\arrow["{\ladj[1]}"', from=1-1, to=2-1]
	\arrow["\lfun", from=1-1, to=1-2]
	\arrow["\rfun"', from=2-1, to=2-2]
	\arrow["{\ladj[1']}", from=1-2, to=2-2]
	\arrow["{\lmate[1]}", shorten <=7pt, shorten >=7pt, Rightarrow, from=1-2, to=2-1]
\end{tikzcd} \matesequiv 
\begin{tikzcd}
	\cat & {\cat[']} \\
	\catbis & {\catbis[']} 
	\arrow["{\radj[1]}", from=2-1, to=1-1]
	\arrow["\lfun", from=1-1, to=1-2]
	\arrow["\rfun"', from=2-1, to=2-2]
	\arrow["{\radj[1']}"', from=2-2, to=1-2]
	\arrow["{\rmate[1]}", shorten <=7pt, shorten >=7pt, Rightarrow, from=1-1, to=2-2]
\end{tikzcd}  \ \ \ \ \
\begin{tikzcd}
	\catbis & {\catbis[']} \\
	\catter & {\catter[']} 
	\arrow["{\ladj[2]}"', from=1-1, to=2-1]
	\arrow["\rfun", from=1-1, to=1-2]
	\arrow["\rrfun"', from=2-1, to=2-2]
	\arrow["{\ladj[2']}", from=1-2, to=2-2]
	\arrow["{\lmate[2]}", shorten <=7pt, shorten >=7pt, Rightarrow, from=1-2, to=2-1]
\end{tikzcd} \matesequiv 
\begin{tikzcd}
	\catbis & {\catbis[']} \\
	\catter & {\catter[']} 
	\arrow["{\radj[2]}", from=2-1, to=1-1]
	\arrow["\rfun", from=1-1, to=1-2]
	\arrow["\rrfun"', from=2-1, to=2-2]
	\arrow["{\radj[2']}"', from=2-2, to=1-2]
	\arrow["{\rmate[2]}", shorten <=7pt, shorten >=7pt, Rightarrow, from=1-1, to=2-2]
\end{tikzcd} \]
then their vertical compositions are also mates.
\[ \begin{tikzcd}
	\cat & {\cat[']} \\
	\catbis & {\catbis[']} \\
	\catter & {\catter[']}
	\arrow["\lfun", from=1-1, to=1-2]
	\arrow["\rfun"{description}, from=2-1, to=2-2]
	\arrow["\rrfun"', from=3-1, to=3-2]
	\arrow["{\ladj[1]}"', from=1-1, to=2-1]
	\arrow["{\ladj[2]}"', from=2-1, to=3-1]
	\arrow["{\ladj[1']}", from=1-2, to=2-2]
	\arrow["{\ladj[2']}", from=2-2, to=3-2]
	\arrow["{\lmate[1]}", shorten <=7pt, shorten >=7pt, Rightarrow, from=1-2, to=2-1]
	\arrow["{\lmate[2]}", shorten <=7pt, shorten >=7pt, Rightarrow, from=2-2, to=3-1]
\end{tikzcd} \matesequiv 
\begin{tikzcd}
	\cat & {\cat[']} \\
	\catbis & {\catbis[']} \\
	\catter & {\catter[']}
	\arrow["\lfun", from=1-1, to=1-2]
	\arrow["\rfun"{description}, from=2-1, to=2-2]
	\arrow["\rrfun"', from=3-1, to=3-2]
	\arrow["{\radj[1]}", from=2-1, to=1-1]
	\arrow["{\radj[2]}", from=3-1, to=2-1]
	\arrow["{\radj[1']}"', from=2-2, to=1-2]
	\arrow["{\radj[2']}"', from=3-2, to=2-2]
	\arrow["{\rmate[1]}", shorten <=7pt, shorten >=7pt, Rightarrow, from=1-1, to=2-2]
	\arrow["{\rmate[2]}", shorten <=7pt, shorten >=7pt, Rightarrow, from=2-1, to=3-2]
\end{tikzcd} \]
Here is a reformulation of the results that does not use pasting diagrams.

\begin{proposition} \label{prop:mate-horizontal-composite} If
$\lmate : \ladj['] \lfun \naturalTrans \rfun \ladj$,
$\rmate : \lfun \radj \naturalTrans \radj['] \rfun$ are mates and if 
$\lmate['] : \ladj[''] \lfun['] \naturalTrans \rfun['] \ladj[']$,
$\rmate['] : \lfun['] \radj['] \naturalTrans \radj[''] \rfun[']$ are mates,
then their horizontal compositions
\begin{center}
$ \rfun['] \lmate \ntcomp \lmate['] \lfun:
\begin{tikzcd}
	{\ladj[''] \lfun['] \lfun} & {\rfun['] \ladj['] \lfun} & {\rfun['] \rfun \ladj}
	\arrow["{ \lmate['] \lfun}", from=1-1, to=1-2]
	\arrow["{\rfun['] \lmate}", from=1-2, to=1-3]
\end{tikzcd}$

$\rmate['] \rfun \ntcomp \lfun['] \rmate:
\begin{tikzcd}
	{\lfun['] \lfun \radj} & {\lfun['] \radj['] \rfun} & {\radj[''] \rfun['] \rfun}
	\arrow["{\lfun['] \rmate}", from=1-1, to=1-2]
	\arrow["{\rmate['] \rfun}", from=1-2, to=1-3]
\end{tikzcd}$
\end{center} 
are also mates.
\end{proposition}

\begin{proposition} If \label{prop:mate-vertical-composite}
	$\lmate[1] : \ladj[1'] \lfun \naturalTrans \rfun \ladj[1]$,
	$\rmate[1] : \lfun \radj[1] \naturalTrans \radj[1'] \rfun$ are mates and if 
	$\lmate[2] : \ladj[2'] \rfun \naturalTrans \rrfun \ladj[2]$,
	$\rmate[2] : \rfun \radj[2] \naturalTrans \radj[2'] \rrfun$ are mates,
	then their vertical compositions
	\begin{center}
	$\lmate[2] \ladj[1] \ntcomp \ladj[2'] \lmate[1] :
	\begin{tikzcd}
		{\ladj[2'] \ladj[1'] \lfun} & {\ladj[2'] \rfun \ladj[1]} & {\rrfun \ladj[2] \ladj[1]}
		\arrow["{\ladj[2'] \lmate[1]}", from=1-1, to=1-2]
		\arrow["{\lmate[2] \ladj[1].}", from=1-2, to=1-3]
	\end{tikzcd}$

	$\radj[1'] \rmate[2] \ntcomp \rmate[1] \radj[2] :
	\begin{tikzcd}
		{\lfun \radj[1] \radj[2]} & {\radj[1'] \rfun \radj[2]} & {\radj[1'] \radj[2'] \rrfun}
		\arrow["{\rmate[1] \radj[2]}", from=1-1, to=1-2]
		\arrow["{\radj[1'] \rmate[2]}", from=1-2, to=1-3]
	\end{tikzcd}$
\end{center}
	are also mates.
\end{proposition}

\begin{definition} We call $\categoryadj$ the double category where 
vertical morphisms are adjunction, horizontal morphisms are functors, 
and squares are pairs $(\lmate, \rmate)$ of mates.
\end{definition}
By \cref{thm:mate-iso},
$\categoryadj$ is isomorphic to 
the two double categories described above.

\subsection{Mate construction between a monad and a comonad}
\label{sec:mate-monade-comonade}

Assume that $\unit, \counit : \ladj \dashv \radj$ where $\ladj$ and $\radj$
are endofunctors on the same category. 
Recall that there is 
an identity adjunction $\id, \id : \idfun \dashv \idfun$ and 
an adjunction $\ladj \ladj \dashv \radj \radj$ of unit 
$\radj \unit \ladj \ntcomp \unit$ and counit $\counit \ntcomp \ladj \counit \radj$
given by the composition of $\unit, \counit : \ladj \dashv \radj$
with itself.
Then the mate construction induces a bijection between the natural transformations
$\radjunit : \idfun \naturalTrans \radj$ and the natural transformations
$\ladjcounit : \ladj \naturalTrans \idfun$, given by
\[ \ladjcounit = \begin{tikzcd}
	\ladj & {\ladj \radj} & \idfun
	\arrow["{\ladj \radjunit}", from=1-1, to=1-2]
	\arrow["\counit", from=1-2, to=1-3]
\end{tikzcd} \quad \quad
\radjunit = \begin{tikzcd}
	\idfun & {\radj \ladj} & \radj
	\arrow["\unit", from=1-1, to=1-2]
	\arrow["{\radj \ladjcounit}", from=1-2, to=1-3]
\end{tikzcd} \, . \]
The mate construction also induces a bijection between the natural transformations 
$\radjsum : \radj^2 \naturalTrans \radj$ and the natural transformations 
$\ladjcosum : \ladj \naturalTrans \ladj^2$ given by 
\[ \ladjcosum = \begin{tikzcd}
	\ladj & {\ladj \radj \ladj} & {\ladj \radj \radj \ladj \ladj} & {\ladj \radj \ladj \ladj} & {\ladj \ladj}
	\arrow["{\ladj \unit}", from=1-1, to=1-2]
	\arrow["{\ladj \radj \unit \ladj}", from=1-2, to=1-3]
	\arrow["{\ladj \radjsum \ladj \ladj}", from=1-3, to=1-4]
	\arrow["{\counit \ladj \ladj}", from=1-4, to=1-5]
\end{tikzcd} \]
\[ \radjsum = \begin{tikzcd}
	{\radj \radj} & {\radj \ladj \radj \radj} & {\radj \ladj \ladj \radj \radj} & {\radj \ladj \radj} & \radj
	\arrow["{\unit \radj \radj}", from=1-1, to=1-2]
	\arrow["{\radj \ladjcosum \radj \radj}", from=1-2, to=1-3]
	\arrow["{\radj \ladj \counit \radj}", from=1-3, to=1-4]
	\arrow["{\radj \counit}", from=1-4, to=1-5]
\end{tikzcd} \]
The compositionality of the mate construction ensures that $(\radj, \radjunit, \radjsum)$ is
a monad if and only if $(\ladj, \ladjcounit, \ladjcosum)$ is a comonad, see 
7.3 of~\cite{Mesablishvili11}.
This means that the mate construction induces a bijection between the monad structures 
on $\radj$ and the comonad structures on $\ladj$.

\begin{definition} A monad $\radjtriple$ and a comonad $\ladjcomtriple$ such that 
$\ladj \dashv \radj$ are called \emph{mates} if their structure are related through 
the bijections defined above.
\end{definition}

Similarly, the mate construction induces a bijection between the comonad structures 
on $\radj$ and the monad structures on $\ladj$.
A comonad $\radjcomtriple$ and a monad $\ladjtriple$ are called mate
if their structure are related through these bijections.

  \begin{theorem}[7.7 of \cite{Mesablishvili11}] \label{thm:mate-bimonad}
	The mate construction induces a bijection between 
	bimonad structure on $\ladj$ and bimonad structures on $\radj$.
  \end{theorem}

We give a sketch of the proof of \cref{thm:mate-bimonad} above.  
We already know that the mate construction relates 
	monads structure on $\ladj$ with comonad structures on $\radj$, 
	and comonad structures on $\ladj$ with monads structures on $\radj$.
	The mate construction also relates distributive law 
	$\ladjcomtriple \ladjtriple \naturalTrans  \ladjtriple \ladjcomtriple$
	with distributive laws 
	$\radjtriple \radjcomtriple \naturalTrans \radjcomtriple \radjtriple$,
	and distributive laws 
	$\ladjtriple \ladjcomtriple \naturalTrans  \ladjcomtriple \ladjtriple$
	with distributive laws
	$\radjcomtriple \radjtriple \naturalTrans \radjtriple \radjcomtriple$,
	as proved in section 7.5 of 
	\cite{Mesablishvili11}. This result relies on 
	the compositionality of the mate construction, and is also a consequence of
	to the development of \cref{sec:compatibility-adjunction}.
	Finally, the mate construction preserves the commutation of the 
	bimonad diagrams for similar reasons.

\subsection{Application to distributive laws}
\label{sec:compatibility-adjunction}

The mate construction relates distributive laws and morphism of
distributive laws because it preserves compatibility, in the sense
explained below. It is possible that the results of this section and
of \cref{sec:extension-adjunction} are particular instances of doctrinal
adjunction, see \cite{Kelly74-doctrinal}.

\begin{definition} \label{def:mate-vertical-compatible}
A pair of natural transformation 
$\alpha : \lfun[1] \naturalTrans \lfun[2]$, $\beta : \rfun[1] \naturalTrans \rfun[2]$ is
compatible with the squares 
$\lmate[1] : \ladj['] \lfun[1] \naturalTrans \rfun[1] \ladj$ and
$\lmate[2] : \ladj['] \lfun[2] \naturalTrans \rfun[2] \ladj$ if the diagram
of \cref{eq:compatibility-vertical-1} commutes.
Similarly, such a pair is compatible with the squares
$\rmate[1]: \lfun[1] \radj \naturalTrans \radj['] \rfun[1]$ and
$\rmate[2]: \lfun[2] \radj \naturalTrans \radj['] \rfun[2]$ if the diagram 
of \cref{eq:compatibility-vertical-2} below commutes.

\begin{center}
\begin{subequations} 
	\begin{minipage}{0.4\textwidth}
		\begin{equation} \label{eq:compatibility-vertical-1}
			\begin{tikzcd}
				{\ladj['] \lfun[1]} & {\rfun[1] \ladj} \\
				{\ladj['] \lfun[2]} & {\rfun[2] \ladj}
				\arrow["{\lmate[1]}", from=1-1, to=1-2]
				\arrow["{\ladj['] \alpha}"', from=1-1, to=2-1]
				\arrow["{\lmate[2]}"', from=2-1, to=2-2]
				\arrow["{\beta \ladj}", from=1-2, to=2-2]
			\end{tikzcd}
		\end{equation}
	\end{minipage}
	\begin{minipage}{0.2 \textwidth}
		
	\end{minipage}
	\begin{minipage}{0.4\textwidth}
		\begin{equation} \label{eq:compatibility-vertical-2}
			\begin{tikzcd}
				{\lfun[1] \radj} & {\radj['] \rfun[1]} \\
				{\lfun[2] \radj} & {\radj['] \rfun[2]}
				\arrow["{\alpha \radj}"', from=1-1, to=2-1]
				\arrow["{\rmate[1]}", from=1-1, to=1-2]
				\arrow["{\rmate[2]}"', from=2-1, to=2-2]
				\arrow["{\radj['] \beta}", from=1-2, to=2-2]
			\end{tikzcd}
		\end{equation}
	\end{minipage}
\end{subequations}
\end{center}

\cref{eq:compatibility-vertical-1,eq:compatibility-vertical-2}
respectively corresponds to the equality of pasting diagrams 
given in 
\cref{eq:compatibility-vertical-pasting-1,eq:compatibility-vertical-pasting-2}
below.
\begin{center}
\begin{subequations} 
	\begin{minipage}{0.4\linewidth}
		\begin{equation} \label{eq:compatibility-vertical-pasting-1}
			\begin{tikzcd}
				\cat & {\cat[']} \\
				\catbis & {\catbis[']} \\
				\catbis & {\catbis[']}
				\arrow["\ladj"', from=1-1, to=2-1]
				\arrow[Rightarrow, no head, from=2-1, to=3-1]
				\arrow["{\lfun[1]}", from=1-1, to=1-2]
				\arrow["{\ladj[']}", from=1-2, to=2-2]
				\arrow[Rightarrow, no head, from=2-2, to=3-2]
				\arrow["{\rfun[1]}"{description}, from=2-1, to=2-2]
				\arrow["{\rfun[2]}"', from=3-1, to=3-2]
				\arrow["{\lmate[1]}", shorten <=7pt, shorten >=7pt, Rightarrow, from=1-2, to=2-1]
				\arrow["\beta", shorten <=7pt, shorten >=7pt, Rightarrow, from=2-2, to=3-1]
			\end{tikzcd} 
			\begin{tikzcd}
				\cat & {\cat[']} \\
				\cat & {\cat[']} \\
				\catbis & {\catbis[']}
				\arrow[Rightarrow, no head, from=1-1, to=2-1]
				\arrow[Rightarrow, no head, from=1-2, to=2-2]
				\arrow["{\lfun[1]}", from=1-1, to=1-2]
				\arrow["{\lfun[2]}"{description}, from=2-1, to=2-2]
				\arrow["\ladj"', from=2-1, to=3-1]
				\arrow["{\ladj[']}", from=2-2, to=3-2]
				\arrow["{\rfun[2]}"', from=3-1, to=3-2]
				\arrow["\alpha", shorten <=8pt, shorten >=8pt, Rightarrow, from=1-2, to=2-1]
				\arrow["{\lmate[2]}", shorten <=7pt, shorten >=7pt, Rightarrow, from=2-2, to=3-1]
			\end{tikzcd}
		\end{equation}
	\end{minipage}
	\begin{minipage}{0.2 \textwidth}

	\end{minipage}
	\begin{minipage}{0.4\linewidth}
		\begin{equation} \label{eq:compatibility-vertical-pasting-2}
			\begin{tikzcd}
				\cat & {\cat[']} \\
				\catbis & {\catbis[']} \\
				\catbis & {\catbis[']}
				\arrow["\radj", from=2-1, to=1-1]
				\arrow[Rightarrow, no head, from=2-1, to=3-1]
				\arrow["{\lfun[1]}", from=1-1, to=1-2]
				\arrow["{\radj[']}"', from=2-2, to=1-2]
				\arrow[Rightarrow, no head, from=2-2, to=3-2]
				\arrow["{\rfun[1]}"{description}, from=2-1, to=2-2]
				\arrow["{\rfun[2]}"', from=3-1, to=3-2]
				\arrow["{\rmate[1]}", shorten <=7pt, shorten >=7pt, Rightarrow, from=1-1, to=2-2]
				\arrow["\beta", shorten <=7pt, shorten >=7pt, Rightarrow, from=2-1, to=3-2]
			\end{tikzcd} = 
			\begin{tikzcd}
				\cat & {\cat[']} \\
				\cat & {\cat[']} \\
				\catbis & {\catbis[']}
				\arrow[Rightarrow, no head, from=1-1, to=2-1]
				\arrow[Rightarrow, no head, from=1-2, to=2-2]
				\arrow["{\lfun[1]}", from=1-1, to=1-2]
				\arrow["{\lfun[2]}"{description}, from=2-1, to=2-2]
				\arrow["\radj", from=3-1, to=2-1]
				\arrow["{\radj[']}"', from=3-2, to=2-2]
				\arrow["{\rfun[2]}"', from=3-1, to=3-2]
				\arrow["\alpha", shorten <=8pt, shorten >=8pt, Rightarrow, from=1-1, to=2-2]
				\arrow["{\rmate[2]}", shorten <=7pt, shorten >=7pt, Rightarrow, from=2-1, to=3-2]
			\end{tikzcd}
		\end{equation}
	\end{minipage}
\end{subequations}
\end{center}
\end{definition}

\begin{proposition} \label{prop:mate-vertical-compatible}
If $\lmate[1], \rmate[1]$ are mate and $\lmate[2], \rmate[2]$ are mate,
then $(\alpha, \beta)$ is compatible with $\lmate[1]$ and $\lmate[2]$ if and only if 
$(\alpha, \beta)$ is compatible with $\rmate[1]$ and $\rmate[2]$.
\end{proposition}

\begin{proof} The proof is a straightforward computation using the definition
of mate, but we can also give a more generic argument making full use of the isomorphism 
of \cref{thm:mate-iso}.
By assumption, $\lmate[1]$ and $\rmate[1]$ are mates. Furthermore, $\beta$ 
is a mate with itself through the identity adjunction. Thus, 
by \cref{thm:mate-iso}, the two squares on the left-hand side of 
\cref{eq:compatibility-vertical-pasting-1,eq:compatibility-vertical-pasting-2}
are mates of each other. Similarly, the two squares on the right-hand side of
\cref{eq:compatibility-vertical-pasting-1,eq:compatibility-vertical-pasting-2} 
are mates of each other. Because the mate construction is a bijection, 
this implies \cref{eq:compatibility-vertical-pasting-1} is equivalent 
to \cref{eq:compatibility-vertical-pasting-2}.
\end{proof}

\begin{corollary} \label{prop:mate-dl-morphism-vertical} 
	Assume that $\ladj$ can be endowed with a comonad structure 
	$\ladjcomtriple$ and $\radj$ with a monad structure
	$\radjtriple$. Let $F, G$ be two endofunctors on $\cat$.
	Assume that $\lmate^F : \ladjcomtriple F \naturalTrans F \ladjcomtriple$
	and $\rmate^F : F \radjtriple \naturalTrans \radjtriple F$ are mates and 
	distributive laws, and that 
	$\lmate^G : \ladjcomtriple G \naturalTrans G \ladjcomtriple$
	and $\rmate^G : G \radjtriple \naturalTrans \radjtriple G$ are mates and 
	distributive laws. Then $\alpha : F \naturalTrans G$ is a morphism from 
	$\lmate^F$ to $\lmate^G$ if and only if $\alpha$ is a morphism from 
	$\rmate^F$ to $\rmate^G$. The same holds when $\ladj$ is endowed with 
	the structure of a monad $\ladjtriple$ and $\radj$ with the structure of a comonad 
	$\radjcomtriple$ instead.
\end{corollary}

\begin{proof} The diagrams corresponding to the fact that $\alpha$ is 
	a morphism of distributive law are the following.
	\[ \begin{tikzcd}
		{\ladj \F} & {\F \ladj} \\
		{\ladj \G} & {\G \ladj}
		\arrow["{\lmate^F}", from=1-1, to=1-2]
		\arrow["{\ladj \alpha}"', from=1-1, to=2-1]
		\arrow["{\lmate^G}"', from=2-1, to=2-2]
		\arrow["{\alpha \ladj}", from=1-2, to=2-2]
	\end{tikzcd}  \quad
	\begin{tikzcd}
		\F \radj \arrow[r, "\rmate^\F"] \arrow[d, "{\alpha \radj}"'] & \radj \F  \arrow[d, "{\radj \alpha}"]  \\
		\G \radj \arrow[r, "\rmate^\G"']                           & \radj \G                            
		\end{tikzcd} \]
	They correspond respectively to the fact that $(\alpha, \alpha)$ 
	is compatible with $\lmate^F$ and $\lmate^G$, and to the fact that 
	$(\alpha, \alpha)$ is compatible with $\rmate^F$ and $\rmate^G$, so
	they are equivalent by \cref{prop:mate-vertical-compatible}.
\end{proof}

\begin{remark} We will see in \cref{prop:mate-dl-comonad-monad} and 
\cref{prop:mate-dl-monad-comonad} that whenever $\ladjcomtriple$ 
(respectively $\ladjtriple$) is the mate of 
$\radjtriple$ (respectively $\radjcomtriple$), 
then $\lmate^F$ is a distributive law if and 
only if $\rmate^F$ is a distributive law, and
$\lmate^G$ is a distributive law if and 
only if $\rmate^G$ is a distributive law, so the assumption
of the result above is not strong.
\end{remark}

\begin{corollary} \label{prop:mate-dl-comonad}
    Assume that $\ladj['] = \ladj$ and $\radj['] = \radj$, 
    that $\lfun, \rfun$ can be equipped with a comonad structure
    $\lcomtriple = (\lfun, \lcomunit, \lcomsum)$ and 
    $\rcomtriple = (\rfun, \rcomunit, \rcomsum)$,
    and that $\lmate : \ladj \lfun \naturalTrans \rfun \ladj$ and 
    $\rmate : \lfun \radj \naturalTrans \radj \rfun$  are mates.
	Then $\lmate : \ladj \lcomtriple \naturalTrans \rcomtriple \ladj$ 
	is a distributive law (in the sense of \cref{def:dl-em})
    if and only if $\rmate : \lcomtriple \radj \naturalTrans \radj \rcomtriple$
	is a distributive law (in the sense of \cref{def:codl}).
\end{corollary}

\begin{proof}
Notice that 
$\rmate : \lcomtriple \radj \naturalTrans \radj \rcomtriple$ is a distributive law 
if and only if \begin{enumerate}
\item $(\lcomunit, \rcomunit)$ is compatible with
$\rmate$ and $\id$;
\item $(\lcomsum, \rcomsum)$ is compatible with
$\rmate$ and $\rmate \rfun \ntcomp \lfun \rmate$.
\end{enumerate}
Similarly, $\lmate : \ladj \lcomtriple \naturalTrans \rcomtriple \ladj$ is
a distributive law if and only if: \begin{enumerate}
\item $(\lcomunit, \rcomunit)$ is compatible with
$\lmate$ and $\id$;
\item $(\lcomsum, \rcomsum)$ is compatible with
$\rmate$ and $\rfun \lmate \ntcomp \lmate \lfun$.
\end{enumerate}

But $\id$ and $\id$ are mates, $\lmate$ and $\rmate$ are mates, and
$\rfun \lmate \ntcomp \lmate \lfun$ and
$\rmate \rfun \ntcomp \lfun \rmate$ are mates by compositionality of the 
mate construction (\cref{prop:mate-horizontal-composite}).
So by \cref{prop:mate-vertical-compatible},
$\rmate$ is a distributive law if and only if $\lmate$
is a distributive law.
\end{proof}

\begin{corollary} \label{prop:mate-dl-monad}
    Assume that $\lfun, \rfun$ are equipped with a monad structure
    $\ltriple$ and 
    $\rtriple$.
    Assume that $\lmate : \ladj \lfun \naturalTrans \rfun \ladj$ and 
    $\rmate : \lfun \radj \naturalTrans \radj \rfun$  are mates.
    Then $\lmate : \ladj \ltriple \naturalTrans \rtriple \ladj$ 
	is a distributive law (in the sense of \cref{def:dl})
    if and only if $\rmate : \ltriple \radj \naturalTrans \radj \rtriple$
	is a distributive law (in the sense of \cref{def:dl-em}).
\end{corollary}

\begin{proof} Same proof as in \cref{prop:mate-dl-comonad}.
\end{proof}

\begin{definition}  \label{def:mate-horizontal-compatible}
	A pair of natural transformations 
$\alpha : \ladj[2] \naturalTrans \ladj[1]$ and 
$\alpha' : \ladj[2'] \naturalTrans \ladj[1']$ is compatible with
the squares 
$\lmate[1] : \ladj[1'] \lfun \naturalTrans \rfun \ladj[1]$ and
$\lmate[2] : \ladj[2'] \lfun \naturalTrans \rfun \ladj[2]$ if
the following diagram commutes:
\begin{subequations}
\begin{equation} \label{eq:compatibility-horizontal-1}
	\begin{tikzcd}
	{\ladj[1'] \lfun} & {\rfun \ladj[1]} \\
	{\ladj[2'] \lfun} & {\rfun \ladj[2]}
	\arrow["{\lmate[1]}", from=1-1, to=1-2]
	\arrow["{\lmate[2]}"', from=2-1, to=2-2]
	\arrow["\alpha' \lfun", from=2-1, to=1-1]
	\arrow["{\rfun \alpha}"', from=2-2, to=1-2]
\end{tikzcd}
\end{equation}
A pair of natural transformations 
$\beta : \radj[1] \naturalTrans \radj[2]$ and 
$\beta' : \radj[1'] \naturalTrans \radj[2']$ is compatible with the squares
$\rmate[1]: \lfun \radj[1] \naturalTrans \radj[1'] \rfun$ and
$\rmate[2]: \lfun \radj[2] \naturalTrans \radj[2'] \rfun$ if the 
diagram below commutes.
\begin{equation} \label{eq:compatibility-horizontal-2}
	\begin{tikzcd}
	{\lfun \radj[1]} & {\radj[1'] \rfun} \\
	{\lfun \radj[2]} & {\radj[2'] \rfun}
	\arrow["{\lfun \beta}"', from=1-1, to=2-1]
	\arrow["{\rmate[1]}", from=1-1, to=1-2]
	\arrow["{\beta' \rfun}", from=1-2, to=2-2]
	\arrow["{\rmate[2]}"', from=2-1, to=2-2]
\end{tikzcd} 
\end{equation}
\end{subequations}
\cref{eq:compatibility-horizontal-1,eq:compatibility-horizontal-2}
respectively corresponds to the equality of pasting diagrams 
given in 
\cref{eq:compatibility-horizontal-pasting-1,eq:compatibility-horizontal-pasting-2}
below.

\begin{subequations}
\begin{equation} \label{eq:compatibility-horizontal-pasting-1}
	\begin{tikzcd}
	\cat & {\cat[']} & {\cat[']} \\
	\catbis & {\catbis[']} & {\catbis[']}
	\arrow["{\ladj[1]}"', from=1-1, to=2-1]
	\arrow["\lfun", from=1-1, to=1-2]
	\arrow["\rfun"', from=2-1, to=2-2]
	\arrow["{\ladj[1']}"{description}, from=1-2, to=2-2]
	\arrow["{\lmate[1]}", shorten <=7pt, shorten >=7pt, Rightarrow, from=1-2, to=2-1]
	\arrow[Rightarrow, no head, from=1-3, to=1-2]
	\arrow[Rightarrow, no head, from=2-3, to=2-2]
	\arrow["{\ladj[2']}", from=1-3, to=2-3]
	\arrow["{\alpha'}", shorten <=7pt, shorten >=7pt, Rightarrow, from=1-3, to=2-2]
\end{tikzcd} =  
\begin{tikzcd}
	\cat & \cat & {\cat[']} \\
	\catbis & \catbis & {\catbis[']}
	\arrow["{\ladj[2]}"{description}, from=1-2, to=2-2]
	\arrow["\lfun", from=1-2, to=1-3]
	\arrow["\rfun"', from=2-2, to=2-3]
	\arrow["{\ladj[2']}", from=1-3, to=2-3]
	\arrow["{\lmate[2]}", shorten <=7pt, shorten >=7pt, Rightarrow, from=1-3, to=2-2]
	\arrow[Rightarrow, no head, from=1-1, to=1-2]
	\arrow[Rightarrow, no head, from=2-1, to=2-2]
	\arrow["{\ladj[1]}"', from=1-1, to=2-1]
	\arrow["\alpha", shorten <=7pt, shorten >=7pt, Rightarrow, from=1-2, to=2-1]
\end{tikzcd} \end{equation}
\begin{equation} \label{eq:compatibility-horizontal-pasting-2}
\begin{tikzcd}
	\cat & {\cat[']} & {\cat[']} \\
	\catbis & {\catbis[']} & {\catbis[']}
	\arrow["{\radj[1]}", from=2-1, to=1-1]
	\arrow["\lfun", from=1-1, to=1-2]
	\arrow["\rfun"', from=2-1, to=2-2]
	\arrow["{\radj[1']}"{description}, from=2-2, to=1-2]
	\arrow["{\rmate[1]}", shorten <=7pt, shorten >=7pt, Rightarrow, from=1-1, to=2-2]
	\arrow[Rightarrow, no head, from=1-3, to=1-2]
	\arrow[Rightarrow, no head, from=2-3, to=2-2]
	\arrow["{\radj[2']}"', from=2-3, to=1-3]
	\arrow["{\beta'}", shorten <=7pt, shorten >=7pt, Rightarrow, from=1-2, to=2-3]
\end{tikzcd} = 
\begin{tikzcd}
	\cat & \cat & {\cat[']} \\
	\catbis & \catbis & {\catbis[']}
	\arrow["{\radj[2]}"{description}, from=2-2, to=1-2]
	\arrow["\lfun", from=1-2, to=1-3]
	\arrow["\rfun"', from=2-2, to=2-3]
	\arrow["{\radj[2']}"', from=2-3, to=1-3]
	\arrow["{\rmate[2]}", shorten <=7pt, shorten >=7pt, Rightarrow, from=1-2, to=2-3]
	\arrow[Rightarrow, no head, from=1-1, to=1-2]
	\arrow[Rightarrow, no head, from=2-1, to=2-2]
	\arrow["{\radj[1]}", from=2-1, to=1-1]
	\arrow["\beta", shorten <=7pt, shorten >=7pt, Rightarrow, from=1-1, to=2-2]
\end{tikzcd}
\end{equation}
\end{subequations}
\end{definition}

\begin{proposition}  \label{prop:mate-horizontal-compatible}
	Let $\alpha, \alpha', \beta, \beta', \lmate[1], \lmate[2], \rmate[1], \rmate[2]$
	be natural transformations whose type is given in \cref{def:mate-horizontal-compatible}.
	Assume that $\alpha, \beta$ are mates, that $\alpha', \beta'$ are
mates, that $\lmate[1], \rmate[1]$ are mates and that $\lmate[2], \rmate[2]$ are mates.
Then $(\alpha, \alpha')$ is compatible with $\lmate[1]$ and $\lmate[2]$
if and only if $(\beta, \beta')$ is compatible with $\rmate[1]$ and $\rmate[2]$.
\end{proposition}

\begin{proof} By \cref{thm:mate-iso} and assumption, the two pasting diagrams 
on the left-hand side of 
\cref{eq:compatibility-horizontal-pasting-1,eq:compatibility-horizontal-pasting-2}
are mates, and the two pasting 
diagrams on the right-hand side of 
\cref{eq:compatibility-horizontal-pasting-1,eq:compatibility-horizontal-pasting-2}
are mates. 
Because the mate construction is a bijection, the two equalities 
are equivalent.
\end{proof}

\begin{remark} Using pasting diagrams sheds light on why 
\cref{def:mate-vertical-compatible} and \cref{def:mate-horizontal-compatible}
are very similar yet different. They consist of the same kind of equation 
except that the first one is vertical, and the second one is horizontal.
This also explains why an alternative proof consisting of
unfolding the definition of the mate construction is very 
straightforward in \cref{prop:mate-vertical-compatible}, but not straightforward
at all in the case of \cref{prop:mate-horizontal-compatible}.
The reason is that the functoriality of the mate construction for the vertical 
composition holds almost by definition of the mates, whereas the functoriality of the mate
construction for the horizontal composition involves a non-trivial computation.
A direct proof of \cref{prop:mate-horizontal-compatible} would 
duplicate this computation multiple times.
\end{remark}

\begin{corollary} \label{prop:mate-dl-morphism}
    Assume that $\ladj[']_i = \ladj_i$ and $\radj[']_i = \radj_i$,
    that $\lfun, \rfun$ are equipped with a comonad structure
    $\lcomtriple$ and 
    $\rcomtriple$,
    and that $\lmate_i : \ladj_i \lcomtriple \naturalTrans \rcomtriple \ladj_i$ and 
    $\rmate_i : \lcomtriple \radj_i \naturalTrans \radj_i \rcomtriple$ are distributive laws
	and are mates.
    Assume that $\alpha : \ladj_2 \naturalTrans \ladj_1$ and 
	$\beta: \radj_1 \naturalTrans \radj_2$ are mates. Then 
	$\alpha$ is a morphism of distributive law (from $\lmate_2$ to 
	$\lmate_1$) if and only if $\beta$ is a morphism of distributive 
	law (from $\rmate_1$ to $\rmate_2$).
	The same property hold when taking a comonad structure $\ltriple$ and 
    $\rtriple$ instead.
\end{corollary}

\begin{proof} The diagram making $\alpha$ a morphism of distributive law 
is the same as the one expressing that the pair $(\alpha, \alpha)$ is 
compatible with the squares $\lmate_1 : \ladj_1 \lcomtriple \naturalTrans \rcomtriple \ladj_1$
and $\lmate_2 : \ladj_2 \lcomtriple \naturalTrans \rcomtriple \ladj_2$.
The diagram making $\beta$ a morphism of distributive law 
is the same as the one expressing that the pair $(\beta, \beta)$ is 
compatible with the squares $\rmate_1 : \lcomtriple \radj_1 \naturalTrans \radj_1 \rcomtriple$
and $\rmate_2 : \lcomtriple \radj_2 \naturalTrans \radj_2 \rcomtriple$.
We then conclude by \cref{prop:mate-horizontal-compatible}.
\end{proof}

\begin{corollary} \label{prop:mate-dl-comonad-monad}
	Assume that $\ladjcomtriple = (\ladj, \comunit, \comsum)$ 
	is a comonad and $\radjtriple = (\radj, \munit, \msum)$ is a monad 
	on the same category $\cat$ and that they are mates.
	Assume that $\ladjcomtriple' = (\ladj', \comunit', \comsum')$ 
	is a comonad and $\radjtriple' = (\radj', \munit', \msum')$ is a monad 
	on the same category $\cat[']$ and that they are mate.
	Let $F : \cat \arrow \cat[']$.
	Assume that $\lmate : \ladj' \F \naturalTrans \F \ladj$ and 
	$\rmate: \F \radj \naturalTrans \radj['] \F$ are mate. Then 
	$\lmate : \ladjcomtriple['] \F \naturalTrans \F \ladjcomtriple$
	is a distributive law if and only if 
	$\rmate: \F \radjtriple \naturalTrans \radjtriple['] \F$
	is a distributive law.
\end{corollary}

\begin{proof} Notice that 
	$\lmate : \ladjcomtriple['] \F \naturalTrans \F \ladjcomtriple$
	is a distributive law if and only if 
	\begin{enumerate}
	\item $(\comunit, \comunit')$ is compatible with $\lmate$ and $\id$;
	\item $(\comsum, \comsum')$ is compatible with 
	$\lmate \ladj \ntcomp \ladj['] \lmate$ and $\lmate$.
	\end{enumerate}
	Similarly, $\rmate: \F \radjtriple \naturalTrans \radjtriple['] \F$ 
	is a distributive law if and only if
	\begin{enumerate}
		\item $(\munit, \munit')$ is compatible with $\rmate$ and $\id$.
		\item $(\msum, \msum')$ is compatible with 
		$\radj['] \rmate \ntcomp \rmate \radj$ and $\rmate$. 
	\end{enumerate}
	But $\id$ and $\id$ are mate, $\lmate$ and $\rmate$ are mate, and
	by \cref{prop:mate-vertical-composite}, 
	$\lmate \ladj['] \ntcomp \ladj['] \lmate$
	and $\radj['] \rmate \ntcomp \rmate \radj$ are mate.
	Furthermore, the monad and the comonad structures are mate
	by assumption.
	So by \cref{prop:mate-horizontal-compatible}, 
	$\lmate$ is a distributive law if and only if $\rmate$ is a distributive 
	law.
\end{proof}

\begin{corollary} \label{prop:mate-dl-monad-comonad}
	Under the same assumptions as in \cref{prop:mate-dl-comonad-monad} except 
	that now $\ladjtriple, \ladjtriple[']$ are monads and $\radjcomtriple, \radjcomtriple[']$ 
	are comonads, 
	$\lmate : \ladjtriple['] \lfun \naturalTrans \rfun \ladjtriple$
	is a distributive law if and only if 
	$\rmate: \lfun \radjcomtriple \naturalTrans \radjcomtriple['] \rfun$
	is a distributive law.
\end{corollary}

\begin{proof} The proof is the same as for \cref{prop:mate-dl-comonad}.
\end{proof}

\subsection{Application to lax and oplax monoidal structures}
\label{sec:mate-oplax-lax}

Any adjunctions $\unit[1], \counit[1] : \ladj[1] \dashv \radj[1]$
and $\unit[2], \counit[2] : \ladj[2] \dashv \radj[2]$,
induces an adjunction $(\unit[1], \unit[2]), (\counit[1], \counit[2]) :
\ladj[1] \times \ladj[2] \dashv \radj[1] \times \radj[2]$.

\begin{proposition} \label{prop:mate-product}
	If $\lmate[1] : \ladj[1'] \lfun[1] \naturalTrans \rfun[1] \ladj[1]$
	and $\rmate[1] : \lfun[1] \radj[1] \naturalTrans \radj[1'] \rfun[1]$ are mates,
	and if $\lmate[2] : \ladj[2'] \lfun[2] \naturalTrans \rfun[2] \ladj[2]$
	and $\rmate[2] : \lfun[2] \radj[2] \naturalTrans \radj[2'] \rfun[2]$ are mates
	then 
	\[ \begin{split}
	(\lmate[1], \lmate[2]) : (\ladj[1'] \times \ladj[2']) 
	(\lfun[1] \times \lfun[2]) \naturalTrans (\rfun[1] \times \rfun[2])
	(\ladj[1] \times \ladj[2]) \\
	(\rmate[1], \rmate[2]) : (\lfun[1] \times \lfun[2])
	(\radj[1] \times \radj[2]) \naturalTrans (\radj[1'] \times \radj[2']) 
	(\rfun[1] \times \rfun[2])
	\end{split} \]
	are mates.
\end{proposition}

\begin{proof} Straightforward computation.
\end{proof}

\begin{remark} It is very likely that the observation above means 
that $\categoryadj$ is monoidal, for a suitable notion of monoidal 
double category.
\end{remark}

Assume that $(\cat, \sm, \smone)$ and 
$(\catbis, \sm[2], \smone[2])$ are (symmetric) monoidal categories.
The adjunction
$\ladj \dashv \radj$ induces a bijection between the morphisms 
$\smfzero \in \cat(\smone, \radj \smone[2])$ and the morphisms 
$\osmfzero \in \catbis(\ladj \smone[2], \smone)$. This can also be 
seen as an instance of the mate construction when taking 
$H = K$ to be the constant functor $\smone$ on $\cat$.

The mate construction applied to the bifunctor
$\_ \sm \_ : \cat \times \cat \arrow \cat$ and taking the adjunctions to 
be $\ladj \dashv \radj$ and 
$\ladj \times \ladj \dashv \radj \times \radj$ induces a bijection 
between natural transformations
$\osmftwo_{X, Y} \in \catbis(\ladj(X \sm Y), \ladj X \sm[2] \ladj Y)$
and $\smftwo_{X, Y} \in \cat(\radj X \sm \radj Y, \radj (X \sm[2] Y))$.
The result below is well known, and is also a result of doctrinal adjunction.
We provide a direct proof.

\begin{theorem} (\cite{Kelly74-doctrinal}) 
\label{thm:mate-oplax-lax}
The bijection above induces a bijection between 
lax monoidal structures on $\radj$ and oplax monoidal 
structures on $\ladj$. Furthermore, if
$\radj$ is endowed with a monad structure 
$\radjtriple = (\radj, \munit, \msum)$ and $\ladj$ with a comonad structure 
$\ladjcomtriple = (\ladj, \comunit, \comsum)$ that are mates, 
the bijection induces a bijection between lax monoidal structures 
on $\radjtriple$ and $\ladjcomtriple$.
\end{theorem}

\begin{proof} The two diagrams below express 
	the compatibility conditions of lax and oplax 
	monoidal structures with respect to the right 
	unitors. 
	\[\begin{tikzcd}
		{\ladj X \sm[2] \smone[2]} & {\ladj X \sm[2] \ladj \smone[1] } & {\ladj(X \sm[1] \smone[1])} \\
		{\ladj X} & & {\ladj X} 
		\arrow["{\ladj X \sm[2] \smfzero}", from=1-1, to=1-2]
		\arrow["\smftwo", from=1-2, to=1-3]
		\arrow["{\smunitR[2]}"', from=1-1, to=2-1]
		\arrow["{\ladj \smunitR[1]}", from=1-3, to=2-3]
		\arrow[Rightarrow, no head, from=2-1, to=2-3]
	\end{tikzcd} \quad 
	\begin{tikzcd}
		{\radj (X \sm[2] \smone[2])} & {\radj X \sm \radj \smone[2]} & {\radj X \sm \smone} \\
		{\radj X} && {\radj X}
		\arrow["\osmftwo", from=1-1, to=1-2]
		\arrow["\radj X \sm \osmfzero", from=1-2, to=1-3]
		\arrow["\smunitR", from=1-3, to=2-3]
		\arrow["{\radj \smunitR[2]}"', from=1-1, to=2-1]
		\arrow[Rightarrow, no head, from=2-1, to=2-3]
	\end{tikzcd} \]
	They correspond to the fact that $(\smunitR[2], \smunitR)$ is compatible 
	with $\smftwo \comp (\ladj X \sm[2] \smfzero)$ and $\id$, and that 
	$(\smunitR[2], \smunitR)$ is compatible with 
	$(\radj X \sm \osmfzero) \comp \osmftwo$ and $\id$ in 
	the sense of \cref{def:mate-vertical-compatible}.
	But $\smftwo \comp (\ladj X \sm[2] \smfzero)$ and 
	$(\radj X \sm \osmfzero) \comp \osmftwo$ are mates 
	by compositionality of the mate construction and by \cref{prop:mate-product}.
	So by \cref{prop:mate-vertical-compatible}, those two diagrams are equivalent.
	The other equivalences of commutations follow a similar 
	principle, and we conclude that 
	$(\radj, \smfzero, \smftwo)$ is a lax monoidal structure
	if and only if $(\ladj, \osmfzero, \osmftwo)$ is an oplax monoidal
	structure.
	
	Furthermore, if $\radjtriple$ is a monad and $\ladjcomtriple$ is a comonad, 
	$\osmfzero = \comunit$ if and only if 
	$\smfzero = \munit$ because $\comunit$ and $\munit$ are 
	mates. Finally, the two diagrams of \cref{def:monoidal-monad}
	mean that $\smftwo$ is a distributive law. By 
	\cref{prop:mate-dl-comonad-monad}, $\smftwo$
	is a distributive law if and only if 
	$\osmftwo$ is a distributive law, and the associated diagrams 
	are exactly the ones turning $\ladjcomtriple$ into an
	oplax monoidal comonad.
\end{proof}

\subsection{Extension of an adjunction}
\label{sec:extension-adjunction}

In this section, we provide a direct proof
of a well known results on extensions (and
liftings) of adjunctions by relying on the theory of distributive laws and on the 
mate construction.
Let $\unit, \counit : \ladj \dashv \radj$ be an adjunction, with 
$\ladj : \cat \arrow \catbis$ and $\radj : \catbis \arrow \cat$. 
Let $\mtriple[1] = (\monad[1], \munit[1], \msum[1])$ be a monad 
on $\cat$, and $\mtriple[2] = (\monad[2], \munit[2], \msum[2])$ be a monad 
on $\catbis$. 

\begin{definition} \label{def:extension-adjunction}
	An extension of an adjunction $\unit, \counit : \ladj \dashv \radj$
 	to the Kleisli categories $\kleisli[1]$ and 
	$\kleislibis[2]$ is an adjunction
	$\extension{\unit}, \extension{\counit} : \ladjext \dashv \radjext$
	such that 
	$\ladjext : \kleisli[1] \arrow \kleislibis[2]$ and 
	$\radjext : \kleislibis[2] \arrow \kleisli[1]$ are extensions of 
	$\ladj$ and $\radj$ respectively, and such that 
	$\extension{\unit}, \extension{\counit}$ are extensions of 
	$\unit, \counit$.
\end{definition}

A similar definition can be found in~\cite{Keigher75} for the lifting to algebras, 
along with a counterpart of \cref{prop:extension-adjunction} below.

\begin{proposition}
	\label{prop:extension-adjunction}
	There is a bijection between extensions of 
	the adjunction $\unit, \counit : \ladj \dashv \radj$
	to $\kleisli[1]$ and $\kleislibis[2]$, and pairs 
	$(\dlL, \dlR)$ of distributive 
	laws $\dlL : \ladj \mtriple[1] \naturalTrans 
	\mtriple[2] \ladj$ and $\dlR : \radj \mtriple[2]
	\naturalTrans \mtriple[1] \radj$ such that the following diagrams commute.
	\begin{equation} \label{eq:extension-adjunction}
	\begin{tikzcd}
		{\monad[1]} && {\monad[1]} \\
		{\radj \ladj \monad[1]} & {\radj \monad[2] \ladj} & {\monad[1] \radj \ladj}
		\arrow["{\unit \monad[1]}"', from=1-1, to=2-1]
		\arrow["{\radj \dlL}"', from=2-1, to=2-2]
		\arrow["{\dlR \ladj}"', from=2-2, to=2-3]
		\arrow[Rightarrow, no head, from=1-1, to=1-3]
		\arrow["{\monad[1] \unit}", from=1-3, to=2-3]
	\end{tikzcd} \quad 
	\begin{tikzcd}
		{\ladj \radj \monad[2]} & {\ladj \monad[1] \radj} & {\monad[2] \ladj \radj} \\
		{\monad[2]} && {\monad[2]}
		\arrow["{\ladj \dlR}", from=1-1, to=1-2]
		\arrow["{\dlL \radj}", from=1-2, to=1-3]
		\arrow["{\monad[2] \counit}", from=1-3, to=2-3]
		\arrow["{\counit \monad[2]}"', from=1-1, to=2-1]
		\arrow[Rightarrow, no head, from=2-1, to=2-3]
	\end{tikzcd}
      \end{equation}
\end{proposition}

\begin{proof} 
	Direct application of \cref{thm:extension-and-dl,thm:extension-and-dl-morphism}.
\end{proof}

\begin{proposition} \label{prop:mate-inverse}
	Let $\dlL : \ladj \monad[1] \naturalTrans 
	\monad[2] \ladj$ and $\dlR : \radj \monad[2]
	\naturalTrans \monad[1] \radj$ be two natural transformations (we 
	do not assume that they are distributive laws). Then 
	the diagrams of \cref{eq:extension-adjunction} commute if and only if 
	$\dlR$ is the inverse of the mate $\dlLmate$ of $\dlL$:
	\[ \dlLmate = \begin{tikzcd}
			{\monad[1] \radj} & {\radj \ladj \monad[1] \radj} & {\radj \monad[2] \ladj \radj} & {\radj \monad[2]}
			\arrow["{\unit \monad[1] \radj}", from=1-1, to=1-2]
			\arrow["{\radj \dlL \radj}", from=1-2, to=1-3]
			\arrow["{\radj \monad[2] \counit}", from=1-3, to=1-4]
		\end{tikzcd} \]
\end{proposition}

\begin{proof} Assume that the diagrams of \cref{eq:extension-adjunction}
	commute. We show that $\dlR$ is the inverse of 
	$\dlLmate$ with the following diagram chase.
	\[ \begin{tikzcd}
		{\radj \monad[2]} & {} & {\radj \monad[2]} \\
		& {\radj \ladj \radj \monad[2]} \\
		{\monad[1] \radj} & {\radj \ladj \monad[1] \radj} & {\radj \monad[2] \ladj \radj}
		\arrow[Rightarrow, no head, from=1-1, to=1-3]
		\arrow["{\unit \radj \monad[2]}"', from=1-1, to=2-2]
		\arrow["{\radj \counit \monad[2]}"', from=2-2, to=1-3]
		\arrow["{\radj \ladj \dlR}"{description}, from=2-2, to=3-2]
		\arrow["\dlR"', from=1-1, to=3-1]
		\arrow["{\unit \monad[1] \radj}"', from=3-1, to=3-2]
		\arrow["{\radj \dlL \radj}"', from=3-2, to=3-3]
		\arrow["{\radj \monad[2] \counit}"', from=3-3, to=1-3]
		\arrow["{(a)}"{description}, draw=none, from=2-2, to=1-2]
		\arrow["{(b)}"{description}, draw=none, from=3-1, to=2-2]
		\arrow["{(d)}"{description}, draw=none, from=3-3, to=2-2]
	\end{tikzcd} \quad 
	\begin{tikzcd}
		{\monad[1] \radj} & {} & {\monad[1] \radj} \\
		& {\monad[1] \radj \ladj \radj} \\
		{\radj \ladj \monad[1] \radj} & {\radj \monad[2] \ladj \radj} & {\radj \monad[2]}
		\arrow["{\unit \monad[1] \radj}"', from=1-1, to=3-1]
		\arrow["{\radj \dlL \radj}"', from=3-1, to=3-2]
		\arrow["{\radj \monad[2] \counit}"', from=3-2, to=3-3]
		\arrow["\dlR"', from=3-3, to=1-3]
		\arrow[Rightarrow, no head, from=1-1, to=1-3]
		\arrow["{\monad[1] \unit \radj}"', from=1-1, to=2-2]
		\arrow["{\monad[1] \radj \counit}"', from=2-2, to=1-3]
		\arrow["{(a)}"{description}, draw=none, from=2-2, to=1-2]
		\arrow["{\dlR \ladj \radj}"{description}, from=3-2, to=2-2]
		\arrow["{(e)}"{description}, draw=none, from=3-1, to=2-2]
		\arrow["{(c)}"{description}, draw=none, from=3-3, to=2-2]
	\end{tikzcd} \]

Commutation $(a)$ is the triangle identities of the adjunction, 
$(b)$ is the naturality of $\unit$, $(c)$ is the naturality
of $\dlR$, and $(d)$ and $(e)$ are the diagram of 
\cref{eq:extension-adjunction}. 
The converse direction is a similar computation.
\end{proof}

We conclude this section with the following result. 
By compositionality of the mate 
construction, this result corresponds to a well known result on 
extension of adjunction, see Theorem 4 of~\cite{Johnstone75} for example.

\begin{theorem} \label{thm:extension-adjunction} 
	The adjunction $\unit, \counit : \ladj \dashv \radj$
	extends to the Kleisli categories $\kleisli[1]$ and 
	$\kleislibis[2]$ if and only if there exists 
	$\dlL : \ladj \monad[1] \naturalTrans 
	\monad[2] \ladj$ and $\dlR : \radj \monad[2]
	\naturalTrans \monad[1] \radj$ two natural transformations 
	such that $\dlL : \ladj \mtriple[1] \naturalTrans 
	\mtriple[2] \ladj$ is a distributive law and 
	$\dlR$ is the inverse of its mate.
\end{theorem}

\begin{proof} By \cref{prop:mate-dl-monad}, the mate 
	$\dlLmate$ of $\dlL$ is a distributive law 
	$\mtriple[1] \radj \naturalTrans \radj \mtriple[2]$. 
	Since $\dlR$ is the inverse of $\dlLmate$, it must be a distributive 
	law $\dlR : \radj \mtriple[2]
	\naturalTrans \mtriple[1] \radj$. 
	Furthermore, by \cref{prop:mate-inverse} the diagrams of \cref{eq:extension-adjunction}
	hold, so by \cref{prop:extension-adjunction} the adjunction extends 
	to the Kleisli categories.	
\end{proof}

\section{The representable theory}

\label{sec:elementary}

We focus now on a more specific situation which is quite common in
models of LL, and where the $\Sigma$-summability structure \(\Sfun\) is equal 
to $\Dbimon \linarrow \_$, where $\Dbimon$ is a specific object, 
namely the \(\Nat\)-indexed cartesian
product of the unit \(\Sone\) of the tensor product, 
$\Dbimon = 1 \with 1 \with \cdots$.
We call the situation \emph{representable}, because of the similarity 
with the notion of representable tangent functor found in tangent categories,
see~\cite{Cockett14}. 
In this case, the  structure can be described very simply as a
\(\Oc\)-coalgebra structure on \(\Dbimon\).
We also provide three examples and one non-example of this
situation.

\subsection{Representably $\Sigma$-additive categories}

\begin{definition} \labeltext{($\exists \Dbimon$)}{ax:D-defined}
  A symmetric monoidal category $\categoryLL$ satisfies \ref{ax:D-defined} 
  if the countable cartesian product 
  \[ \Dbimon=\With_{i\in\Nat}\Sone
  =\underbrace{\Sone\With\Sone\With\cdots}_{\Nat} \] 
  exists in $\categoryLL$, and if 
  the internal hom $(\Dbimon \linarrow X, \ev_{\Dbimon})$ 
  exists in $\categoryLL$ for all objects $X$.
  The object \(\Dbimon\) is called \emph{the object of degrees} of \(\cL\),
  we motivate this terminology in \cref{rem:object-of-degrees}.
\end{definition}

The goal of this section is to endow 
$\categoryLL$ with a $\Sigma$-additive structure for which  
$\S = \Dbimon \linarrow \_$ is a $\Sigma$-summability structure.
We can define \emph{injections} into
\(\Dbimon\) by %
\(\Win_i=\Tuple{\Kronecker
  ji\Id}_{j\in\Nat}\in\cL(\Sone,\Dbimon)\) for $i \in \N$. 
  In other words \(\Win_i\) is characterized by
\begin{align*}
  \Wproj_j\Compl\Win_i=
  \begin{cases}
    \Id & \text{if }i=j\\
    0 & \text{otherwise.}
  \end{cases}
\end{align*}
We also define the diagonal 
$\Wdiag = \prodPairing<\N>{\id_1} \in \cL(1, \Dbimon)$.

Remember from \cref{sec:summability-closure-LL} 
that if \(f\in\cL(X_2,X_1)\) and if the
internal homs $(X_i \linarrow Y, \Evlin_{X_i})$ exist for \(i=1,2\), then it is
possible to define \(f\Limpl Y \in\cL(X_1\Limpl Y,X_2 \Limpl Y)\),
defining a functor $\_ \linarrow Y: \Op\cL \to\cL$.
Explicitly,
$f \linarrow Y = \cur_{X_2}(\ev_{X_1} ((X_1 \linarrow Y) \tensor f)) \in 
\categoryLL(X_1 \linarrow Y, X_2 \linarrow Y)$. 
Intuitively,
this functor maps a morphism 
$h \in X_1 \linarrow Y$ to $h \compl f \in X_2 \linarrow Y$.
This intuition corresponds to the following equation: for all
$h \in \categoryLL(Z \tensor X_1, Y)$,
\begin{equation} \label{eq:left-closure}
  (f \linarrow Z) \compl \cur_{X_1}(h) = \cur_{X_2}(h \compl (Z \tensor f)) \, .
\end{equation}

Notice that \((\Sone\Limpl X,\Evlin_1)\) always exists: we can take
\(\Sone\Limpl X=X\), \(\Evlin_1=\Rightu\) and 
$\cur_1(f) = f \compl \tensorUnitR^{-1}$.
So we can define the natural
transformations 
\begin{equation} 
  \Sproj_i=(\Win_i\Limpl X)\in\cL(\Sfun X,X) 
  \quad \quad
  \Ssum = (\Wdiag \Limpl X)\in\cL(\Sfun X,X) \, .
\end{equation}

\begin{remark} \label{rem:projection-sum-alternative}
  Unfolding the definitions of $\Sproj_i$ and $\Ssum$, we get that 
  \begin{equation} \label{eq:projection-elementary-def}
  \Sproj_i = \begin{tikzcd}[column sep = large]
    {\Dbimon \linarrow X} & {(\Dbimon \linarrow X) \tensor 1} & {(\Dbimon \linarrow X) \tensor \Dbimon} & X
    \arrow["{\tensorUnitR^{-1}}", from=1-1, to=1-2]
    \arrow["{(\Dbimon \linarrow X) \tensor \Sprojl_i}", from=1-2, to=1-3]
    \arrow["{\ev_{\Dbimon}}", from=1-3, to=1-4]
  \end{tikzcd}
  \end{equation}
  \begin{equation} 
  \Ssum = \begin{tikzcd}[column sep = large]
    {\Dbimon \linarrow X} & {(\Dbimon \linarrow X) \tensor 1} & {(\Dbimon \linarrow X) \tensor \Dbimon} & X
    \arrow["{\tensorUnitR^{-1}}", from=1-1, to=1-2]
    \arrow["{(\Dbimon \linarrow X) \tensor \Wdiag}", from=1-2, to=1-3]
    \arrow["{\ev_{\Dbimon}}", from=1-3, to=1-4]
  \end{tikzcd}
  \end{equation}
\end{remark}

It follows from \cref{eq:left-closure} and definition of $\cur_1$ that for all 
$f \in \categoryLL(X \tensor \Dbimon, Y)$, 
\begin{align}
  \Sproj_i \compl \cur_{\Dbimon}(f) &= f \compl (X \tensor \Sprojl_i) \compl \tensorUnitR^{-1}
  \label{eq:projection-elementary} \\
  \Ssum \compl \cur_{\Dbimon}(f) &= f \compl (X \tensor \Wdiag) \compl \tensorUnitR^{-1} \, .
  \label{eq:sum-elementary}
\end{align}

\begin{example} \label{ex:summability-wrel}
  Let us give some intuition on this structure using the weighted relational 
  model $\WREL[\Rbar]$ as an example.
  The objects of $\WREL[\Rbar]$ are countable sets, and a morphism of 
  $\WREL[\Rbar](E, F)$ is 
  a linear map from $\Rbar^E$ to $\Rbar^F$, see \cref{sec:wrel}
  for more details.
  The object $1$ is the singleton set $\{ * \}$. 
  There is a bijection between linear maps 
  $\Rbar \arrow \Rbar^E$ and elements of $\Rbar^E$, so a morphism 
  in $\WREL[\Rbar](1, \Rbar^E)$ is the same as an element of $\Rbar^E$.
  
  The object $\Dbimon = 1 \with 1 \with \cdots$ is the set $\N$.
  Thus, an element $\phi \in \Dbimon \linarrow X$
  corresponds to a linear map $\phi : \Rbar^\N \arrow X$.
  Observe that $\Sprojl_i = e_i$ where $e_i \in \Rbar^{\N}$ is such that 
  $e_{i, i} = 1$ and $e_{i, j} = 0$ if $j \neq i$. 
  By linearity, 
   \[ \phi \sequence{\lambda_i} = \psum \lambda_i \phi(e_i) \, . \]
  It means that $\phi$ is completely determined by its values on 
  the $\{e_i\}_{i \in \N}$. Thus, the applications 
  $\Sproj_i = \Sprojl_i \linarrow X$ that map 
  $\phi$ to $\phi(e_i)$ are jointly monic,
  and any $\phi \in \Dbimon \linarrow X$ uniquely corresponds to a family 
  $\sequence{x_i} \in \R^\N$ given by $x_i = \phi(e_i)$. Then 
  \[ \phi(1, 1, 1, \ldots) = \psum \phi(e_i) = \psum x_i\] 
  so the sum of the $x_i$ is obtained by evaluating $\phi$ 
  at $(1, 1, 1, \ldots) = \Wdiag$. This is precisely the role of 
  the morphism $\Ssum = \Wdiag \linarrow X$.
  The sum on morphisms is then defined as the point-wise sum.
  This example is a degenerate one in which all families
  are summable.
\end{example}
\begin{example} \label{ex:summability-pcoh}
  The same reasoning as the one above works in the category 
  $\PCOH$ of probabilistic coherent spaces.
  The objects of this category are called probabilistic coherent 
  spaces, and they consist in pairs 
  $E = (\Web E, \Pcoh E)$ where $\Web E$ is a countable set called the web and 
  $\Pcoh(E)$ is a subset of $\Rbar^{\Web E}$ 
  satisfying some properties, we refer the reader to \cref{sec:pcoh}. 
  A morphism in $u \in \PCOH(E, F)$ is a linear map $\Rbar^{\Web E} \arrow \Rbar^{\Web F}$ 
  such that for all $x \in \Pcoh(E)$, $u \cdot x \in \Pcoh(F)$. 

  Probabilistic coherence spaces notably include the spaces of subprobability
  distributions over countable sets.
  Given any countable set $A$, there exists a PCS $\probapcs A$ 
  such that $\Web {\probapcs A} = A$ and 
  $\Pcoh {\probapcs A} = \Proba(A)$, where 
  \[ \Proba(A) = \{x \in \Rbar^A \St \sum_{a \in A} x_a \leq 1\} \]
  is the set of sub-probability distributions on $A$. In fact, probabilistic 
  coherence spaces are a fully abstract model of a PCF language 
  extended with a probabilistic construct, see~\cite{Ehrhard14}.
  
  The object $1$ is the PCS $\probapcs {\{*\}}$ of subprobability distributions 
  over the singleton set, and the object $\Dbimon$ verifies
  $\Web \Dbimon = \N$ and $\Pcoh \Dbimon = [0,1]^{\N}$. 
  The same reasoning as the previous example applies, and
  there is a bijection between linear maps $\phi \in \Dbimon \linarrow 
  E$ and families $\sequence{x_i \in \Rbar^E}$ such that 
  $\sum_{i \in \N} x_i \in \Pcoh E$.
  For example, a family $\sequence{x_i} \in \Proba(A)^{\N}$
  of elements of $\probapcs A$ is summable if 
  $\sum_{i \in \N} x_i$ is still a subprobability distributions on $A$.
  Then, the sum of morphisms is the point-wise sum: 
  a family of morphism $\sequence{f_i} \in \PCOH(E, F)^{\N}$ is summable 
  if for all $x \in \Pcoh(E)$, $\sum_{i \in \N} f_i \cdot x \in \Pcoh(F)$.
  \end{example}

  \begin{remark} \label{rem:object-of-degrees}
    The representable theory pushes the analogy between $\S$ and 
    formal power series (see \cref{sec:bimonad}) further.
    Any formal power series $\sum_{n \in \N} f_n t^n$ with coefficients 
    in $\categoryLL(X, Y)$ can also been seen as a function 
    $P : \N \arrow \categoryLL(X, Y)$, that maps 
    $n \in \N$ to the coefficient $P(n) = f_n$.
    
    As we saw in \cref{ex:summability-wrel,ex:summability-pcoh},
    $\Dbimon$ behaves very similarly to $\N$. 
    Any integer $n \in \N$ can be seen as the element
    $\Sprojl_n \in \categoryLL(1, \Dbimon)$.
    Then, any $f \in \categoryLL(X, \Dbimon \linarrow Y)$
    can be seen as a power series $P : \N \arrow \categoryLL(X, Y)$,
    with a clear analogue between the evaluation of $P$ on 
    $n \in \N$, and the projection $\Sproj_n \compl f$
    that evaluates $f$ on the element
    $\Sprojl_n$.
    This explains why $\Dbimon$ is called the \emph{object of degrees}
    in \cref{ax:D-defined}.
    This analogy is further motivated by the comonoid structure of 
    $\Dbimon$ introduced in \cref{sec:bimonoid}. 

    Notice however that $\Dbimon$ is not the datatype of natural 
    numbers, which is traditionally given by the object $\Bplus_{i \in \N} 1$. 
    The difference is crucial and captures the fact that all the power 
    series considered must have a summable support.
    For any subset $I \subseteq \N$, the sum $\sum_{i \in I} f_i$ is given by 
    evaluating $f$ on the morphism $\Sprojl_I \in \categoryLL(1, \Dbimon)$ 
    such that $\prodProj_i \compl \Sprojl_I = \id$ if $i \in I$, and 
    $\prodProj_j \compl \Sprojl_I = 0$ if $j \notin I$.
    The existence of such $\Sprojl_I$ for $\size{I} > 1$ crucially relies on 
    the fact that $\Dbimon = \withFam<\N> 1$. 
  \end{remark}

\begin{lemma}\label{lemma:curry-epic-monic}
  Assume that \((X,Y)\) and \((X',Y)\) have internal homs and let %
  \(\Vect f\in\cL(X',X)^I\).
  Then the morphisms \((Z \Times f_i)_{i\in I}\) are jointly epic for all 
  objects $Z$ if and only if the morphisms
  \((f_i\Limpl Y\in\cL(X\Limpl Y,X'\Limpl Y))_{i\in I}\) are jointly
  monic.
\end{lemma}
\begin{proof}
  Recall from \cref{eq:left-closure} that for any $g \in \categoryLL(Z \tensor X, Y)$, 
  \[ (f_i \linarrow Y) \cur_X(g) = \cur_{X'} (g \compl (Z \tensor f_i)) \, . \]
  The equivalence immediately follows from the fact that $\cur$ is a bijection. 
\end{proof}

\begin{corollary} \label{prop:monic-epic}
  The $\Sproj_i$ are jointly monic if and only if for all object $X$, the 
  $X \tensor \Sprojl_i$ are jointly epic.
\end{corollary}

We now assume that one of the equivalent assumption of \cref{prop:monic-epic}
hold.
For any \(\Vect f=(f_i\in\cL(X,Y))_{i\in\Nat}\), the 
existence of $\Stuple{\Vect f}$ such that $\Sproj_i \compl \Stuple{\Vect f}
= f_i$ is equivalent by \cref{eq:projection-elementary} to the existence of 
an \(h\in\cL(X\Times\Dbimon,Y)\) such that, for all \(i\in\Nat\), the
following diagram commutes
\begin{equation}
  \begin{tikzcd}
    X\Times\Sone\ar[r,"X\Times\Win_i"]\ar[d,swap,"\Rightu"]
    & X\Times\Dbimon\ar[d,"h"]\\
    X\ar[r,"f_i"]
    &Y
  \end{tikzcd}
\end{equation}
\begin{notation} 
We set \(\Stuplet{\Vect f}=h\) since this \(h\) is unique when it
exists by the fact that the \(X \tensor \Win_i\)'s are jointly epic, so that
\(\Stuple{\Vect f}=\Curlin\Stuplet{\Vect f}\).
\end{notation}
By \cref{eq:left-closure} again, we have the equality 
\begin{equation} \label{eq:sum-alternative}
  \Ssum \compl \Stuple{\Vect f}
  =\Stuplet{\Vect f}\Compl(X\Times\Wdiag)\Compl\Inv{\Rightu} \, .
\end{equation}
So $\Stuplet{\Vect f}\Compl(X\Times\Wdiag)\Compl\Inv{\Rightu}$
provides a nice candidate for a suitable notion of sum.

\begin{proposition} \label{lemma:unique-sum-representable}
  Assume that the $X \tensor \Sprojl_i$ are jointly epic.
  For all objects $X, Y$, there exists at most one $\Sigma$-monoid 
  structure on $\categoryLL(X, Y)$ such that a $\N$-indexed family
  $\vect f = \sequence{f_i \in \categoryLL(X, Y)}$ is summable if and only if 
  $\Stuplet{f}$ exists, and such that 
  \[ \sum_{i \in \N} f_i = \Stuplet{\vect f} \compl (X \tensor \Wdiag) \compl 
  \tensorUnitR_X^{-1} \, . \]
\end{proposition}

\begin{proof}
  By \cref{prop:reindexing}, a $\Sigma$-monoid
  is completely determined by the summability over the 
  $\N$-indexed family, since any countable set admits an injection 
  into $\N$.
\end{proof}

\Cref{lemma:unique-sum-representable} above motivates the following definition.

\begin{definition} \label{def:RS}
A symmetric monoidal category $\categoryLL$ is 
\emph{representably $\Sigma$-additive} if it satisfies 
\ref{ax:D-defined} and if:
\labeltext{(RS-mon)}{ax:RS-mon}
\labeltext{(RS-epi)}{ax:RS-epi}
\labeltext{(RS-sum)}{ax:RS-sum}
\labeltext{(RS-witness)}{ax:RS-witness}
\begin{itemize}[align=center]
  \item[\ref{ax:RS-epi}] The morphisms $(X \tensor \Sprojl_i)$ are jointly epic for all $X$.
  \item[\ref{ax:RS-mon}] For all objects $X, Y$,
  there exists a $\Sigma$-monoid on $\categoryLL(X, Y)$,
  and the $0$'s of the $\Sigma$-monoids are zero morphisms and 
  are absorbing for the monoidal product ($f \compl 0 = 0$, $0 \compl g = 0$,
  $0 \tensor g = 0$, $f \tensor 0 = 0$ for all $f, g$).
  \item[\ref{ax:RS-sum}] A family $\vect f = \sequence{f_i}$ is summable if and only if 
  $\Stuplet{\vect f}$ exists, and then $\sum_{i \in \N} f_i = \Stuplet{\vect f}
  \compl (X \tensor \Wdiag) \compl \tensorUnitR^{-1}$.
  \item[\ref{ax:RS-witness}]
  If $\family{h_a \in \categoryLL(X \tensor \Dbimon, Y)}$ is such that 
  $\family<A \times \N>[(a, j)]{h_a \compl (X \tensor \Sprojl_j) \compl \tensorUnitR^{-1}}$
  is summable, then $\family{h_a}$ is summable.
\end{itemize}
\end{definition}

\Cref{def:RS} above may seem a bit unintuitive, but is very useful 
in practice.
Observe that this definition consists in a property on 
symmetric monoidal categories, not an additional 
structure: the $\Sigma$-monoid structure given by 
\ref{ax:RS-mon} is uniquely characterized by \ref{ax:RS-sum} and 
\cref{lemma:unique-sum-representable}.
Thus, it is direct to check whenever a category is representably
$\Sigma$-additive, even if the $\Sigma$-additive structure is not known
beforehand. It only suffices to unfold the definition of $\Dbimon$,
and to check if the induced notion of sum on $\N$-indexed family
yields a $\Sigma$-monoid.
For example, the $\Sigma$-additive structure of coherence spaces and 
non-uniform coherence spaces
is not straightforward and was discovered in that way.

\begin{remark} 
  \label{rk:elementary-closed-win-epicity}
  Upon taking \(X=\Sone\), the condition \ref{ax:RS-epi} implies that the
  \((\Win_i)_{i\in\Nat}\) are jointly epic. Conversely, if 
  the symetric monoidal category \(\cL\) is closed, the joint epicity of the
  \((\Win_i)_{i\in\Nat}\) implies \ref{ax:RS-epi}.
  Assume indeed that \(\cL\) is closed and that the
  \((\Win_i)_{i\in\Nat}\) are jointly epic.
  Let \(f,g\in\cL(\Tens X\Dbimon,Y)\) be such that %
  \((f\Compl\Tensp X{\Win_i})=g\Compl\Tensp X{\Win_i}_{i\in\Nat}\).
  By naturality of \(\Sym\), we get %
  \((f\Compl\Sym\Compl\Tensp {\Win_i}X)
  =g\Compl\Sym\Compl\Tensp {\Win_i}X)_{i\in\Nat}\) and hence %
  \((\Curlin(f\Compl\Sym)\Compl\Win_i
  =\Curlin(g \Compl\Sym)\Compl\Win_i)_{i\in\Nat}\) %
  so that \(\Curlin(f\Compl\Sym)=\Curlin(g\Compl\Sym)\) and hence %
  \(f=g\).
\end{remark}

We show that representably $\Sigma$-additive categories 
are always $\Sigma$-additive symmetric monoidal categories, 
with a summabibility structure given by $(\Dbimon \linarrow \_, \vect \Sproj)$, 
and that the sum is always compatible with the closedness and the cartesian 
product whenever those exist.
We rely on the following lemma that gives sufficient conditions to 
build a $\Sigma$-additive structure from a functor (this lemma will be applied 
on the functor $\Dbimon \linarrow \_$).

\begin{lemma} \label{prop:summability-structure-characterization} 
  Let $\categoryLL$ be a category,
  $\S$ an endofunctor, and $\sequence{\Sproj_i \in \category(\S X, X)}$ 
  a $\N$-indexed family of morphisms for all objects $X$
  such that: 
  \begin{enumerate}
    \item $\categoryLL(X, Y)$ is a $\Sigma$-monoid for all objects $X, Y$ and 
    the $0$'s of the $\Sigma$-monoids are zero morphisms ($f \compl 0 = 0$
    and $0 \compl g = 0$ for all $f, g$);
    \item $(\S, \vect \Sproj)$ satisfies
    \ref{def:summability-structure-1},\ref{def:summability-structure-2} and
    additionally, $\sum_{i \in \N} f_i = \left( \sum_{i \in \N}{\Sproj_i} \right) 
    \compl \Spairing{f_i}$;
    \label{prop:summability-structure-characterization-witness} 
    \item $\S$ has a functorial action such that the $\Sproj_i$'s and 
    $\sum_{i \in \N} \Sproj_i$ are natural transformations.
   \end{enumerate}
   Then $\categoryLL$ is a $\Sigma$-additive category.
   Furthermore, if $(\S, \vect \Sproj)$ satisfies 
   \ref{def:summability-structure-4} then $(\S, \vect \Sproj)$  
   is a $\Sigma$-summability structure.  
\end{lemma}

\begin{proof}
  We first show left distributivity. Assume that $\family{g_a \in \categoryLL(Y, Z)}$ 
is summable. Let $\phi : A \arrow \N$ be an injection. By 
\cref{prop:reindexing}, $\sequence{g_i'} \defEq 
\Famact \phi {(\family{g_a})}$ is summable so 
\cref{prop:summability-structure-characterization-witness} ensures that $\Spairing{g_i'}$ exists.
Furthermore, for any $f \in \categoryLL(X, Y)$,
\[ \left(\sum_{a \in A} g_a \right) \compl f = 
\left( \sum_{i \in \N} g_i' \right) \compl f 
= \left(\sum_{j \in \N} \Sproj_j \right) \compl \Spairing{g_i'} \compl f 
= \left(\sum_{j \in \N} \Sproj_j \right) \compl \Spairing{g_i' \compl f} \]
so by \cref{prop:summability-structure-characterization-witness}, 
$\sequence{g_i' \compl f}$ is summable with sum 
$\left(\sum_{a \in A} g_a \right) \compl f$.
But $\sequence{g_i' \compl f} = \phi^* (\family{g_a \compl f})$ using the 
fact that $0 \compl f = 0$. Thus, by \cref{prop:reindexing}, 
$\family{g_a \compl f}$ is summable with sum 
$\left(\sum_{a \in A} g_a \right) \compl f$.

Then we show that every morphism $g \in \categoryLL(Y, Z)$ is $\Sigma$-additive. 
Let $\family{f_a \in \categoryLL(X, Y)}$ be a summable family.
By \cref{prop:reindexing}, $\sequence{f_i'} \defEq 
\phi^* (\family{f_a})$ is summable so 
\cref{prop:summability-structure-characterization-witness} ensures that $\Spairing{f_i'}$ exists.
Then, 
\[ \Sproj_i \compl \S g \compl \Spairing{f_i'} 
= g \compl \Sproj_i \compl \Spairing{f_i'}
= g \compl f_i' \] 
so $\S g \compl \Spairing{f_i'} = \Spairing{g \compl f_i'}$ and 
by \cref{prop:summability-structure-characterization-witness}, 
$\sequence{g \compl f_i'}$ is summable with sum 
\[ \left(\sum_{j \in \N} \Sproj_j \right) \compl \S g \compl \Spairing{f_i'}
= g \compl \left(\sum_{j \in \N} \Sproj_j \right) \compl \Spairing{f_i'} 
= g \compl \left(\sum_{i \in \N} f_i' \right)
= g \compl \left(\sum_{a \in A} f_a \right) \, . \]
But $\family{g \compl f_i'} = \phi^* \family{g \compl f_a}$, using the 
fact that $g \compl 0 = 0$. Thus, by \cref{prop:reindexing},
$\family{g \compl f_a}$ is summable with sum
$g \compl \left(\sum_{a \in A} f_a \right)$, so $g$ is $\Sigma$-additive.
So $\categoryLL$ is $\Sigma$-additive, and it directly follows 
from the assumptions that $(\S, \vect \Sproj)$ is a summability structure
if it satisfies \ref{def:summability-structure-4}.
\end{proof}

\begin{lemma}\label{lemma:representable-projection-summable}
  If $\categoryLL$ is representably $\Sigma$-additive, then $\sequence{\Sprojl_i}$ is 
  summable with sum $\Wdiag$, and $\sequence{\Sproj_i}$ is summable with sum 
  $\Ssum$.
\end{lemma}
  
\begin{proof}
    Observe that $\tensorUnitL_{\Dbimon} \compl (1 \tensor \Sprojl_i) 
    \compl \tensorUnitR_{\Dbimon}^{-1} 
    = \Sprojl_i \compl \tensorUnitL_{1} \compl \tensorUnitR_{\Dbimon}^{-1}
    = \Sprojl_i$, so by \cref{ax:RS-sum} $\family{\Sprojl_i}$ is summable 
    with sum $\tensorUnitL_{\Dbimon} \compl (1 \tensor \Wdiag) 
    \compl \tensorUnitR_{\Dbimon}^{-1} 
    = \Wdiag \compl \tensorUnitL_{1} \compl \tensorUnitR_{\Dbimon}^{-1}
    = \Wdiag$.
  
    Then, observe that 
    $\uncur(\id_{\Dbimon \linarrow X}) \compl (X \tensor \Sprojl_i) \compl 
    \tensorUnitR^{-1} = \Sproj_i$ by \cref{eq:projection-elementary} so 
    by \ref{ax:RS-sum}, 
    $\sequence{\Sproj_i}$ is summable with sum 
    $\uncur(\id_{\Dbimon \linarrow X}) \compl (X \tensor \Wdiag) \compl 
    \tensorUnitR^{-1} = \Ssum$ by \cref{eq:sum-elementary}.
\end{proof}

\begin{lemma}\label{lemma:elem-win-epic-tens-gen}
  If \(\cL\) is a representably $\Sigma$-additive category, then
  for any \(n\in\Nat\) the family of morphisms %
  \((X\Times\Win_{i_1}\Times\cdots\Times\Win_{i_n})_{\Vect i\in\Nat^n}\) %
  is jointly epic.
\end{lemma}
\begin{proof}
  By induction on \(n\geq 1\).
  The base case is just our assumption that \(\cL\) is a representably
  $\Sigma$-additive category.
  Assume that the property holds for \(n\) and let %
  \(f,g\in\cL(X\Times\Dbimon^{\Times(n+1)})\) be such that %
  \(f\Compl(X\Times\Win_{i_1}\Times\cdots\Times\Win_{i_n}\Times\Win_{i_{n+1}})
  =g\Compl(X\Times\Win_{i_1}\Times\cdots\Times\Win_{i_n}\Times\Win_{i_{n+1}})\)
  for all \(\Vect i\in\Nat^{n+1}\).
  For \(i\in\Nat\) let \(f_i\in\cL(X\Times\Dbimon^{\Times n})\) be defined as
  \[ 
    \begin{tikzcd}
      X\Times\Dbimon^{\Times n}
      \ar[r,"\Invp{\Rightu}"]
      &[1em]
      X\Times\Dbimon^{\Times n}\Times\Sone
      \ar[r,"X\Times\Dbimon^{\Times n}\Times\Win_i"]
      &[2.4em]
      X\Times\Dbimon^{\Times n}\Times\Dbimon
      \ar[r,"f"]
      &[-1em]
      Y
    \end{tikzcd}
  \]
  and similarly for \(g_i\).
  By inductive hypothesis we have \(f_i=g_i\) for all \(i\in\Nat\) and
  hence \(f=g\) since \(\cL\) is a representably $\Sigma$-additive
  category.
\end{proof}

\begin{lemma} \label{prop:elementary-double-witness}
  For all $f \in \categoryLL(X, \Dbimon \linarrow (\Dbimon \linarrow Y))$ and
  $f' \in \categoryLL(X \tensor \Dbimon \tensor \Dbimon, Y)$ such that 
  $f = \cur(\cur(f'))$, 
  \begin{align*}
    \Sproj_i \compl \Sproj_j \compl f &= f' (X \tensor \Sprojl_j \tensor \Sprojl_i) 
  \compl \tensorUnitR_{X \tensor 1}^{-1} \compl \tensorUnitR_X^{-1} \\
  \Ssum \compl \Ssum \compl f &= f' (X \tensor \Wdiag \tensor \Wdiag) 
  \compl \tensorUnitR_{X \tensor 1}^{-1} \compl \tensorUnitR_X^{-1} \, .
  \end{align*} 
\end{lemma}

\begin{proof} It suffices to apply 
  \cref{eq:projection-elementary} and \cref{eq:sum-elementary} twice.
\end{proof}

\begin{theorem}  \label{thm:representable-summabibility-structure} 
  \label{thm:representably-sigma-additive-monoidal}
  If $\categoryLL$ is representably $\Sigma$-additive, then
  $\categoryLL$ is a $\Sigma$-additive symmetric monoidal category,
  the sum is compatible with the and 
    $(\Dbimon \linarrow \_, \vect \Sproj)$ is a $\Sigma$-summabibility 
    structure.
\end{theorem}

\begin{proof}
We first prove that $(\Dbimon \linarrow \_, \vect \Sproj)$ satisfies
the conditions of \cref{prop:summability-structure-characterization}.
The functorial action is simply $\Dbimon \linarrow \_$, the projections 
$\Sproj_{i,X} = \Sprojl_i \linarrow X$ are natural by construction, and 
by \cref{lemma:representable-projection-summable},
$\sum_{i \in \N} \Sproj_{i, X} = \Ssum = \Wdiag \linarrow X$ is also natural.
Then, $(\Dbimon \linarrow \_, \vect \Sproj)$  satisfies 
\ref{def:summability-structure-1} (the joint monicity of the $\Sproj_i$ 
is a consequence of \ref{ax:RS-epi} and \cref{prop:monic-epic}). 
The tuple $(\Dbimon \linarrow \_, \vect \Sproj)$ also satisfies 
\ref{def:summability-structure-2} by \ref{ax:RS-sum}, because the 
existence of $\Stuple{\vect f}$ is equivalent to the existence 
of $\Stuplet{\vect f}$.

We now prove that $(\Dbimon \linarrow \_, \vect \Sproj)$ 
satisfies \ref{def:summability-structure-4}.
It suffices to prove \ref{def:summability-structure-4} for 
$\N$-indexed families, the general case of 
$A$-indexed families immediately follows from \cref{prop:reindexing}.
Assume that $\sequence[j]{f_j \in \categoryLL(X, \Dbimon \linarrow X)}$ 
is a $\N$-indexed family such that 
$\family<\N \times \N>[(j, i)]{\Sproj_i \compl f_j}$ is summable
Then,
\[ \Sproj_i \compl f_j = \Sproj_i \compl \cur(\uncur (f_j)) 
= \uncur(f_j) \compl (X \tensor \Sprojl_i) \compl \tensorUnitR^{-1} \]
by \cref{eq:left-closure}. Thus, by \ref{ax:RS-witness},
$\sequence[j]{\uncur(f_j)}$ is summable, and by \ref{ax:RS-sum}
there exists $f' \in \categoryLL(X \tensor \Dbimon \tensor \Dbimon, Y)$
such that \[ f' \compl (X \tensor \Dbimon \tensor \Sprojl_j) \compl \tensorUnitR_{X \tensor \Dbimon}^{-1} 
= \uncur{f_j} \, . \]
Then, by \cref{eq:projection-elementary} and naturality, 
\[ f' \compl (X \tensor \tensorSym_{\Dbimon, \Dbimon})
\compl (X \tensor \Sprojl_j \tensor \Sprojl_i) \compl \tensorUnitR_{X \tensor \Dbimon}^{-1} 
\compl \tensorUnitR_X^{-1} 
=  f' \compl (X \tensor \Sprojl_i \tensor \Sprojl_j) \compl \tensorUnitR_{X \tensor \Dbimon}^{-1} 
\compl \tensorUnitR_X^{-1}
= \Sproj_i \compl f_j \, . \]
So by \cref{prop:elementary-double-witness}, taking $f = \cur(\cur(f'))$ 
we have $\Sproj_i \compl \Sproj_j \compl f = \Sproj_i \compl f_j$ 
so $\Sproj_j \compl f = f_j$ by joint monicity of the $\Sproj_j$.
Thus, $f$ is a witness for $\sequence[j]{f_j}$, and $\sequence[j]{f_j}$ is summable 
by \ref{def:summability-structure-2}.

Thus, by \cref{prop:summability-structure-characterization},
$\categoryLL$ is a $\Sigma$-additive category and $(\Dbimon \linarrow \_, \vect \Sproj)$
is a $\Sigma$-summability structure.
We now prove that the sum distributes over the monoidal product. It suffices to prove that 
$(\Dbimon \linarrow \_, \vect \Sproj)$ satisfies \ref{ax:S-sm-dist}.
  We already know from \ref{ax:RS-mon} that $0$ is absorbing for the 
  monoidal product. 
  We now show that $\sequence{\Sproj_i \tensor X_1}$ is summable 
  with sum $\Ssum \tensor X_1$.

\[
    f = \begin{tikzcd}[column sep = large]
      {\Dbimon \linarrow X_0 \tensor X_1 \tensor \Dbimon} & {\Dbimon \linarrow X_0 \tensor \Dbimon \tensor X_1} & {X_0 \tensor X_1}
      \arrow["{(\Dbimon \linarrow X_0) \tensor \tensorSym}", from=1-1, to=1-2]
      \arrow["{\ev \tensor X_1}", from=1-2, to=1-3]
    \end{tikzcd} . \]
    Then we can check that 
    \begin{align*}
      \Sproj_i \compl 
    \cur_{\Dbimon}(f) 
    &= f \compl ((\Dbimon \linarrow X_0) \tensor X_1 \tensor \Sprojl_i) \compl 
    \tensorUnitR^{-1}_{\Dbimon \linarrow X_0 \tensor X_1}
    \tag*{by \cref{eq:projection-elementary}} \\
    &= (\ev \tensor X_1) \compl ((\Dbimon \linarrow X_0) \tensor \Sprojl_i \tensor X_1)
    \compl (\tensorUnitR^{-1}_{\Dbimon \linarrow X_0} \tensor X_1) 
    \tag*{by naturality and axioms of symmetric monoidal categories }\\
    &= \Sproj_i \tensor X_1 \tag*{by \cref{eq:projection-elementary-def}.}
    \end{align*}
    So $\sequence{\Sproj_i \tensor X_1}$ is summable with witness $\SstrL = \cur(f)$, 
    and it follows from a similar computation that 
    $\left(\sum_{i \in \N} \Sproj_i\right) \compl \cur(f) 
    = \Ssum \compl \cur(f) = \Ssum \tensor X_1 = 
    \left(\sum_{i \in \N} \Sproj_i\right) \tensor X_1$.
    The existence of $\SstrR$ can be proved similarly (this morphism is actually 
    simpler, because it does not involve $\tensorSym$).
\end{proof}

  \begin{theorem} \label{thm:representable-compatible}
    For any representably $\Sigma$-additive category $\categoryLL$,
    \begin{enumerate}
      \item if $\categoryLL$ is cartesian, then $\categoryLL$ 
      is a cartesian $\Sigma$-additive category;
      \item if $\categoryLL$ is symmetric monoidal closed, then 
      $\categoryLL$ is a symmetric monoidal closed $\Sigma$-additive category.
    \end{enumerate}
  \end{theorem}

  \begin{proof}
    Assume that $\categoryLL$ is cartesian. We show that $\categoryLL$ is cartesian 
    $\Sigma$-additive. It suffices to show that $(\Dbimon \linarrow \_, \vect \Sproj)$
    satisfies \ref{ax:S-with}.
   Observe that 
   \[ \prodPairing{\ev \compl (\prodProj_i \tensor \Dbimon)} \in 
   \categoryLL\left(\left(\withFam (\Dbimon \linarrow X_i) \right) \tensor \Dbimon,  
     \withFam X_i\right) \]
   so we can define $f = \cur_{\Dbimon}
   (\prodPairing{\ev \compl (\prodProj_i \tensor \Dbimon)}) \in 
   \categoryLL(\withFam (\Dbimon \linarrow X_i), \Dbimon \linarrow \withFam X_i)$.
   Then, 
   \begin{align*}
   \Sproj_j \compl f &= 
   \prodPairing{\ev \compl (\prodProj_i \tensor \Dbimon)} \compl (\withFam X_i \tensor 
   \Sprojl_i) \compl \tensorUnitR^{-1} \tag*{by \cref{eq:projection-elementary}} \\
   &= \prodPairing{\ev \compl (\prodProj_i \tensor \Sprojl_j) \compl \tensorUnitR^{-1}} \\
   &= \prodPairing{\ev \compl (1 \tensor \Sprojl_j) \compl \tensorUnitR^{-1} \compl 
   \prodProj_i} \tag*{by functoriality and naturality} \\
   &= \prodPairing{\Sproj_j \compl \prodProj_i} \tag*{by \cref{eq:projection-elementary-def}} \\
   &= \withFam \Sproj_j
   \end{align*}
   which concludes the proof.
   Then we show that $\categoryLL$ is a symmetric monoidal closed $\Sigma$-additive 
    category.
    First, we know from the proof of 
    \cref{thm:representable-summabibility-structure} that
    $\SstrL_{X_0, X_1} \in 
    \categoryLL\left(\Dbimon \linarrow X_1, \Dbimon \linarrow (X_0 \tensor X_1)\right)$
    is the Curry transpose of 
    \[
    f = \begin{tikzcd}
      {(\Dbimon \linarrow X_0) \tensor X_1 \tensor \Dbimon} & & {(\Dbimon \linarrow X_0) \tensor \Dbimon \tensor X_1} & {X_0 \tensor X_1.}
      \arrow["{(\Dbimon \linarrow X_0) \tensor \tensorSym}", from=1-1, to=1-3]
      \arrow["{\ev \tensor X_1}", from=1-3, to=1-4]
    \end{tikzcd} \]
    It follows that for any object $A$,
    \[ \Sclos_{X} = \cur((\S \ev) \compl \SstrL_{A \linarrow X, A}) \]
    is equal to the double Curry transpose of the following morphism
    \[
    \begin{tikzcd}[column sep = small]
      {(\Dbimon \linarrow (A \linarrow X)) \tensor A \tensor \Dbimon} & & {(\Dbimon \linarrow (A \linarrow X)) \tensor \Dbimon \tensor A} & {(A \linarrow X) \tensor A} & X.
      \arrow["{\id \tensor \tensorSym}", from=1-1, to=1-3]
      \arrow["\ev", from=1-3, to=1-4]
      \arrow["\ev", from=1-4, to=1-5]
    \end{tikzcd} \] 
    It is straightforward to check that this is an isomorphism 
    (whose inverse is given by the double 
    Curry transpose of a similar morphism), so by \cref{prop:Sclos-inverse},
    $(\Dbimon \linarrow \_, \vect \Sproj)$ satisfies \ref{ax:S-fun}.
  \end{proof}

\subsection{Characterization of representably $\Sigma$-additive categories}

We want to characterize under which condition a cartesian symmetric monoidal 
$\Sigma$-additive category is representably $\Sigma$-additive.
This actually boils down to the following condition.

\begin{definition} \labeltext{(Comb-lin)}{ax:comb-lin}
A symmetric monoidal $\Sigma$-additive category that satisfies \ref{ax:D-defined}
satisfies \ref{ax:comb-lin} if for all summable 
family $\sequence{f_i \in \categoryLL(X, Y)}$, the family 
$\sequence{\tensorUnitR_Y \compl (f_i \tensor \prodProj_i) \in 
\categoryLL(X \tensor \Dbimon, Y)}$ is summable.
\end{definition}

\begin{example}
Let us consider the example of probabilistic coherence spaces, 
see \cref{ex:summability-pcoh}.
Assume that $\sequence{x_i} \in \Proba(E)^{\N}$ is summable, that is,
$\sum_{i \in \N} x_i \in \Proba(E)$.
Then for any family $\sequence{\lambda_i \in [0, 1]}$, 
\[ \sum_{i \in \N} \lambda_i x_i \leq \sum_{i \in \N} x_i \in \Proba(E) \] 
so $\sequence{\lambda_i x_i}$ is summable.
This observation corresponds
to an instance of
\ref{ax:comb-lin} where the $x_i$ are considered as morphisms of $\PCOH(1, \probapcs E)$, 
using the fact that $\Dbimon$ is the PCS such that $\Web \Dbimon = \N$ and 
$\Pcoh \Dbimon = {\intcc 0 1}^{\N}$.
\end{example}

\begin{theorem} \label{thm:representably-additive-characterization}
  Assume that $\categoryLL$ is a cartesian symmetric monoidal 
  category that satisfies \ref{ax:D-defined}. 
  The following are equivalent: \begin{enumerate}
  \item $\categoryLL$ is representably $\Sigma$-additive;
  \item $\categoryLL$ is a cartesian symmetric monoidal 
  closed $\Sigma$-additive category (that is, the sum is compatible with the 
  cartesian product, the monoidal product and the closedness) which
  satisfies \ref{ax:comb-lin}.
\end{enumerate}
\end{theorem}

\subsubsection{Proof of the forward implication of 
\cref{thm:representably-additive-characterization}}

Assume that $\categoryLL$ is a cartesian symmetric monoidal 
category that is representably $\Sigma$-additive. 
Then we know by \cref{thm:representable-summabibility-structure,thm:representable-compatible} 
that $\categoryLL$ is a cartesian symmetric monoidal $\Sigma$-additive category.
It remains to show \ref{ax:comb-lin}. Assume that $\vect f = 
\sequence{f_i \in \categoryLL(X, Y)}$
is summable. We need to prove that $\sequence{\tensorUnitR_Y \compl (f_i \tensor 
\prodProj_i)}$ is summable.
By \ref{ax:RS-sum}, it suffices to find 
$h \in \categoryLL(X \tensor \Dbimon \tensor \Dbimon, Y)$ 
such that 
\[ h \compl (X \tensor \Dbimon \tensor \Sprojl_i) \compl 
\tensorUnitR_{X \tensor \Dbimon}^{-1} 
= \tensorUnitR_Y \compl (f_i \tensor \prodProj_i) \, . \] 
Observe that for each \(i\in\Nat\) we have
\(\Leftu\Compl(\Wproj_i\Times\Wproj_i)
\in\cL(\Dbimon\Times\Dbimon,\Sone)\). We set %
\(\Dbimonm
=\Tuple{\tensorUnitR_1\Compl(\Wproj_i\Times\Wproj_i)}_{i\in\Nat}
\in\cL(\Dbimon\Times\Dbimon,\Dbimon)\) and we define 
\[ h =
\begin{tikzcd}
	{X \tensor \Dbimon \tensor \Dbimon} & {X \tensor \Dbimon} & Y
	\arrow["{X \tensor \Dbimonm}", from=1-1, to=1-2]
	\arrow["{\Stuplet{\vect f}}", from=1-2, to=1-3]
\end{tikzcd} \, . \] 
Observe that 
\[ \prodProj_j \compl \Dbimonm \compl (\Dbimon \tensor \Sprojl_i) \compl 
\tensorUnitR_{\Dbimon}^{-1} 
= \tensorUnitR_1 \compl (\prodProj_j \tensor \prodProj_j)
\compl (\Dbimon \tensor \Sprojl_i) \compl \tensorUnitR_{\Dbimon}^{-1} 
= \begin{cases} 
  \tensorUnitR_1 \compl (\prodProj_j \tensor 1) \compl \tensorUnitR_{\Dbimon}^{-1} 
  = \prodProj_i \text{ if } i = j \\
0 \text{ otherwise}
\end{cases} \]
\[ \prodProj_j \compl \Sprojl_i \compl \prodProj_i = \begin{cases} 
  \prodProj_i \text{ if } i = j \\ 
  0 \text{ otherwise} \end{cases} \]
Thus $\Dbimonm \compl (\Dbimon \tensor \Sprojl_i) \compl 
\tensorUnitR_{\Dbimon}^{-1} = \Sprojl_i \compl \prodProj_i$
by joint monicity of the $\prodProj_j$.
It follows that 
\begin{align*} 
  h \compl (X \tensor \Dbimon \tensor \Sprojl_i) \compl 
\tensorUnitR_{X \tensor \Dbimon}^{-1}
&= \Stuplet{\vect f} 
\compl (X \tensor (\Dbimonm \compl (\Dbimon \tensor \Sprojl_i) 
\compl \tensorUnitR_{\Dbimon}^{-1})) \\
&= \Stuplet{\vect f} \compl (X \tensor \Sprojl_i) 
\compl (X \tensor \prodProj_i) \\ 
&= f_i \compl \tensorUnitR_{X} \compl (X \tensor \prodProj_i) \\
&= \tensorUnitR_Y \compl (f_i \tensor \prodProj_i) \, .
\end{align*}
This ends the proof.

\subsubsection{Proof of the backward implication of 
\cref{thm:representably-additive-characterization}}

We assume now that $\categoryLL$ is a cartesian symmetric monoidal 
$\Sigma$-additive category which
satisfies \ref{ax:comb-lin}. Remember that by assumption,
 $\categoryLL$ satisfies \ref{ax:RS-mon}.

\begin{proposition} \label{prop:injections-summable}
  The family $\sequence{\Sprojl_i}$ is summable with sum $\Wdiag$. 
  The family $\sequence{\Sprojl_i \compl \prodProj_i \in \cL(\Dbimon, \Dbimon)}$ 
  is summable with sum $\id_{\Dbimon}$.
 \end{proposition}
 
 \begin{proof}
 For all $i \in \N$, the family $\sequence{\prodProj_i \compl \Sprojl_i}$ is 
 summable with sum $\id_{1}$ by \ref{ax:unary}. Thus, $\sequence{\Sprojl_i}$ is 
 summable with sum $\prodPairing<\N>{\id_1} = \Wdiag$, by compatibility between the sum and 
 the cartesian product.
 Similarly, for all $i \in \N$, the family $\sequence{\prodProj_i \compl 
 \Sprojl_i \compl \prodProj_i}$ is summable with sum $\prodProj_i$ by \ref{ax:unary}.
 Thus, $\sequence{\Sprojl_i \compl \prodProj_i}$ is 
 summable with sum $\prodPairing<\N>{\prodProj_i} = \id_{\Dbimon}$.
 \end{proof}

\begin{lemma}
  The morphisms $X \tensor \Sprojl_i$ are jointly epic, so 
  $\categoryLL$ satisfies \ref{ax:RS-epi}.
  \end{lemma}
  
  \begin{proof}
  Assume that $f \compl (X \tensor \Sprojl_i) = g \compl (X \tensor \Sprojl_i)$ for all 
  $i \in \N$. Remember that by \cref{prop:injections-summable}, $\sequence{\Sprojl_i \compl 
  \prodProj_i}$ is summable with sum $\id_{\Dbimon}$. Thus,
  \begin{align*}
  f &= f \compl \left(X \tensor \psum (\Sprojl_i \compl \prodProj_i) \right) \sumsub 
  \psum (f \compl (X \tensor \Sprojl_i \prodProj_i)) \\
  g &= g \compl \left(X \tensor \psum (\Sprojl_i \compl \prodProj_i) \right) \sumsub 
  \psum (g \compl (X \tensor \Sprojl_i \prodProj_i))
  \end{align*}
  So $f = g$, since $f \compl (X \tensor \Sprojl_i \prodProj_i) = 
  g \compl (X \tensor \Sprojl_i \prodProj_i)$ by assumption.
  \end{proof}

  \begin{lemma} \label{lemma:elementary-witness}
  Assume that $\vect f = \sequence{f_i \in \categoryLL(X, Y)}$ is summable. 
  Then $\vect f$ has a witness
  \[ \Stuplet{\vect f} = 
  \sum_{i \in \N} \tensorUnitR_Y \compl (f_i \tensor \prodProj_i) \, . \] 
  \end{lemma}

  \begin{proof} 
    Assume that $\vect f = \sequence{f_i}$ is summable. 
    Let $h = \sum_{i \in \N} \tensorUnitR_Y \compl (f_i \tensor \prodProj_i) 
    \in \categoryLL(X \tensor \Dbimon, Y)$, which is defined by \ref{ax:comb-lin}.
    Then, thanks to \ref{ax:unary},
    \[ h \compl (X \tensor \Sprojl_j) \compl \tensorUnitR^{-1}_X 
    \sumsub \sum_{i \in \N} \left(\tensorUnitR_Y \compl 
    (f_i \tensor \prodProj_i) \compl (X \tensor \Sprojl_j) 
    \compl \tensorUnitR_X^{-1}\right)
    = \tensorUnitR_Y \compl (f_j \tensor 1) \compl \tensorUnitR_X^{-1} = f_j \, . \qedhere \] 
  \end{proof}

\begin{lemma} \label{lemma:rs-sum}
The category $\categoryLL$ satisfies \ref{ax:RS-sum}.
\end{lemma}

\begin{proof} 
Any summable family $\vect f = \sequence{f_i}$ has a witness by \cref{lemma:elementary-witness}.
Conversely, assume that $\Stuplet{\vect f}$ exists.
Using the fact that $\sequence{\Sprojl_i}$ is summable with sum 
$\Wdiag$, and the distributivity of the sum over composition and over the tensor product,
we have that 
\[ \Stuplet{\vect f} \compl (X \tensor \Wdiag) 
\compl \tensorUnitR^{-1}_X 
= \Stuplet{\vect f} \left(X \tensor \left(\sum_{i \in \N} \Sprojl_i \right) \right) 
\compl \tensorUnitR^{-1}_X 
\sumsub \sum_{i \in \N} \left( \Stuplet{\vect f}
\compl (X \tensor \Sprojl_i) \compl \tensorUnitR^{-1}_X \right)
= \sum_{i \in \N} f_i \, . \]
It follows that $\sequence{f_i}$ is summable with sum 
$\Stuplet{\vect f} \compl (X \tensor \Wdiag) 
\compl \tensorUnitR^{-1}_X$.
\end{proof}

It only remains to show \ref{ax:RS-witness}. First, we prove a 
stronger version of \ref{ax:comb-lin} that directly follows from 
\ref{ax:comb-lin}.

\begin{lemma} \label{lemma:comb-lin}
For any summable family $\family<\N \times A>[(i,a)]{f_{i, a}}$, 
$\family<\N \times A>[(i,a)]{\tensorUnitR_Y \compl (f_{i, a}
\tensor \prodProj_i) \in \categoryLL(X \tensor \Dbimon, Y)}$ is summable.
\end{lemma}

\begin{proof}
By \ref{ax:pa}, $\sequence{\sum_{a \in A} f_{i,a}}$ is summable, so by 
\ref{ax:comb-lin} 
$\sequence{\tensorUnitR_X \compl \left((\sum_{a \in A} f_{i,a}) 
\tensor \prodProj_i \right) }$ is also summable. 
But $\tensorUnitR_X \compl \left((\sum_{a \in A} f_{i,a}) 
\tensor \prodProj_i \right) \sumsub 
\sum_{a \in A} \left( \tensorUnitR_X \compl (f_{i,a} \tensor \prodProj_i) \right)$.
So by \ref{ax:pa} again, $\family<\N \times A>[(i,a)]
{\tensorUnitR_Y \compl (f_{i, a}\tensor \prodProj_i)}$ is summable.
\end{proof}

\begin{lemma} \label{lemma:rs-witness}
The category $\categoryLL$ satisfies \ref{ax:RS-witness}.
\end{lemma}

\begin{proof}
Assume that $\family<I>[i]{h_i \in \categoryLL(X \tensor \Dbimon, Y)}$
is such that 
$\family<I \times \N>[(i, j)]{h_i \compl (X \tensor \Sprojl_j) 
\compl \tensorUnitR^{-1}}$ is summable.
Let $f_{i, j} = h_i \compl (X \tensor \Sprojl_j) 
\compl \tensorUnitR^{-1}$.
Then by \cref{lemma:comb-lin}, 
$\family<I \times \N>[(i, j)]{\tensorUnitR_Y \compl (f_{i, j} \tensor \prodProj_j)}$
is summable.
But it follows from \cref{lemma:elementary-witness} that for all $i$,
\[ h_i = \sum_{j \in \N} \tensorUnitR_Y \compl (f_{i, j} \tensor \prodProj_j) \, .\]
Thus, $\family<I>[i]{h_i}$ is summable by \ref{ax:pa}.
\end{proof}

This concludes the proof of the reverse implication of 
\cref{thm:representably-additive-characterization}.

\subsection{The bimonoid structure of \(\Dbimon\)}

\label{sec:bimonoid}

In this section, $\categoryLL$ is a representably $\Sigma$-additive category.
Our goal is to equip $\Dbimon$ with a bimonoid
structure which, as we will see in \cref{sec:bimonoid-bimonad-mate}, 
is deeply related to the bimonad structure on $\S
= \Dbimon \linarrow \_$
introduced in~\cref{sec:bimonad}.

\begin{definition}
A comonoid is a triple %
\(C=(\Comonca C,\Comonu_C,\Comonm_C)\) where the counit %
\(\Comonu_C\in\cL(\Comonca C,\Sone)\) and the comultiplication %
\(\Comonm_C\in\cL(\Comonca C,\Comonca C\Times\Comonca C)\) satisfy the
following commutations
\begin{equation} \label{eq:commonoid}
  \begin{tikzcd}
    \Comonca C \ar[r,"\Comonm_C"] \ar[dr,swap,"\Invp{\Leftu_{\Comonca
        C}}"] & {\Comonca C}\Times{\Comonca C}
    \ar[d,"{\Comonu_C}\Times{\Comonca C}"]
    \\
    & \Tens\Sone{\Comonca C}
  \end{tikzcd}
  \Diagsep
  \begin{tikzcd}
    \Comonca C \ar[rr,"\Comonm_C"] \ar[d,swap,"\Comonm_C"]
    &[1em]&[1em] \Tens{\Comonca C}{\Comonca C} \ar[d,"\Tens{\Comonca
      C}{\Comonm_C}"]
    \\
    \Tens{\Comonca C}{\Comonca C} \ar[r,"\Tens{\Comonm_C}{\Comonca
      C}"] & \Tens{\Tensp{\Comonca C}{\Comonca C}}{\Comonca C}
    \ar[r,"\Assoc_{\Comonca C,\Comonca C,\Comonca C}"] &
    \Tens{\Comonca C}{\Tensp{\Comonca C}{\Comonca C}}
  \end{tikzcd}
\end{equation}
A comonoid is cocommutative if the following diagram commutes.
\begin{equation} \label{eq:cocommutative}
\begin{tikzcd}
  \Comonca C \ar[r,"\Comonm_C"] \ar[dr,swap,"\Comonm_C"] &
  \Tens{\Comonca C}{\Comonca C} \ar[d,"\Sym_{\Comonca C,\Comonca
    C}"]
  \\
  & \Tens{\Comonca C}{\Comonca C}
\end{tikzcd}
\end{equation}
A comonoid morphism from a comonoid \(C\) to a comonoid \(D\) is an %
\(f\in\cL(\Comonca C,\Comonca D)\) such that the two following
diagrams commute
\begin{equation} \label{eq:comonoid-morphism}
  \begin{tikzcd}
    \Comonca C
    \ar[r,"f"]
    \ar[rd,swap,"\Comonu_C"]
    &
    \Comonca D
    \ar[d,"\Comonu_D"]\\
    &
    \Sone
  \end{tikzcd}
  \Diagsep
  \begin{tikzcd}
    \Comonca C
    \ar[d,swap,"\Comonm_C"]
    \ar[r,"f"]
    &
    \Comonca D
    \ar[d,"\Comonm_D"]\\
    \Tens{\Comonca C}{\Comonca C}
    \ar[r,"\Tens ff"']
    &
    \Tens{\Comonca D}{\Comonca D}
  \end{tikzcd}
\end{equation}
\end{definition}
\begin{definition} 
Dually, a monoid is a triple \(M=(\Monca M,\Monu_M,\Monm_M)\) where the unit %
\(\Monu_M\in\cL(\Monca M,\Sone)\) and the comultiplication %
\(\Monm_M\in\cL(\Monca M,\Monca M\Times\Monca M)\) satisfy commutations 
dual to \cref{eq:commonoid}. 
A monoid is commutative if it satisfies a commutation dual to
\cref{eq:cocommutative}. 
A morphism from a monoid $M$ to 
a monoid $N$ is an $f \in \cL(\Comonca M, \Comonca N)$ which satisfies
commutations dual to \cref{eq:comonoid-morphism}.
\end{definition}

\begin{definition}
A bicommutative bimonoid is a tuple $(A, \Monu_A,\Monm_A, \Comonu_A,\Comonm_A)$ 
such that $(A, \Monu_A, \Monm_A)$ is a commutative monoid, $(A, \Comonu_A, \Comonm_A)$ is 
a commutative comonoid, and 
such that one of the following equivalent assertion holds.
\begin{enumerate}
  \item $\Monu_A$ and $\Monm_A$ are morphism of comonoids;
  \item $\Comonu_A$ and $\Comonm_A$ are morphisms of monoids.
\end{enumerate}
\end{definition}
The corresponding diagrams will be given in the proof of 
\cref{thm:Dbimon-bimonoid}.

We equip \(\Dbimon\) with a co commutative comonoid structure.
The counit is the projection \(\Dbimoncu \in\cL(\Dbimon,\Sone)\).
We define a comultiplication
\(\Dbimoncm\in\cL(\Dbimon,\Dbimon\Times\Dbimon)\) which satisfies
\(\Dbimoncm\Compl\Win_n\Compl\Leftu_{\Sone}
=\sum_{i=0}^n\Win_i\Times\Win_{n-i}\).
As we saw in \cref{lemma:representable-projection-summable},
$\sequence[n]{\Sprojl_n}$ is summable of sum $\Dbimonu$.
Then, 
\[ \Dbimonu \tensor \Dbimonu = \left( \sum_{i \in \N} \Sprojl_i \right) 
\tensor \left(\sum_{j \in \N} \Sprojl_j \right) 
\sumsub \sum_{i,j \in \N}(\Sprojl_i \tensor \Sprojl_j)\]
by \ref{ax:S-sm-dist}. Thus, by \ref{ax:pa} 
$(\sum_{i=0}^n \Sprojl_i \tensor \Sprojl_{n-i})_{n \in \N}$ is 
defined and summable.
This means that there exists $h \in \categoryLL(1 \tensor 1 \tensor \Dbimon, 
\Dbimon \tensor \Dbimon)$ such that 
\[ h \compl (1 \tensor 1 \tensor \Sprojl_n) \compl \tensorUnitR_{1 \tensor 1}^{-1} = 
\left( \sum_{i=0}^n \Sprojl_i \tensor \Sprojl_{n-i} \right) . \] 
Let  \[ \Dbimoncm =
\begin{tikzcd}
	\Dbimon & {1 \tensor \Dbimon} & {1 \tensor 1 \tensor \Dbimon} & {\Dbimon \tensor \Dbimon .}
	\arrow["{\tensorUnitL_{\Dbimon}^{-1}}", from=1-1, to=1-2]
	\arrow["{\tensorUnitL_{1 \tensor \Dbimon}^{-1}}", from=1-2, to=1-3]
	\arrow["h", from=1-3, to=1-4]
\end{tikzcd} \]
Then it is straightforward to check that 
\(\Dbimoncm\Compl\Win_n
=(\sum_{i=0}^n\Win_i\Times\Win_{n-i}) \Compl\Leftu_{\Sone}^{-1}\).
The joint epicity of the \(\Win_n\)'s implies that
\(\Dbimoncm\) is uniquely characterized by these equations.

\begin{lemma}
  The triple \((\Dbimon,\Dbimoncu,\Dbimoncm)\) is a cocommutative comonoid.
\end{lemma}
\begin{proof}
  We first prove the counit axiom, which corresponds to the following 
  diagram:
  \[
    \begin{tikzcd}
      \Dbimon
      \ar[r,"\Dbimoncm"]
      \ar[dr,swap,"\Inv{\Rightu}"]
      &
      \Dbimon\Times\Dbimon
      \ar[d,"\Dbimon\Times\Dbimoncu"]\\
      &
      \Dbimon\Times\Sone
    \end{tikzcd}
  \]
 Let \(n\in\Nat\), we have
  \begin{align*}
    (\Dbimon\Times\Dbimoncu)\Compl\Dbimoncm\Compl\Win_n
    &=(\Dbimon\Times\Dbimoncu)\Compl
    \sum_{i=0}^n\left(\Win_i\Times\Win_{n-i}\right)\Inv{\Leftu}\\
    &=\sum_{i=0}^n(\Win_i\Times(\Dbimoncu\Compl\Win_{n-i}))\Inv{\Leftu}
  \end{align*}
  but
  \(\Dbimoncu\Compl\Win_{n-i}=\Wproj_0\Compl\Win_{n-i}
  =\Kronecker{n-i}{0}\Id_{\Sone}\) %
  and hence %
  \((\Dbimon\Times\Dbimoncu)\Compl\Dbimoncm\Compl\Win_n
  =(\Win_n\Times\Id_{\Sone})\Compl\Inv{\Leftu_{\Sone}}
  =\Inv{\Leftu_{\Dbimon}}\).
  So the diagram commutes by joint epicity of the \(\Win_n\)'s.

  Next we prove the associativity axiom, which corresponds to the 
  following diagram:
  \[
    \begin{tikzcd}
      \Dbimon
      \ar[r,"\Dbimoncm"]
      \ar[d,swap,"\Dbimoncm"]
      &[-2em]
      \Dbimon\Times\Dbimon
      \ar[r,"\Dbimoncm\Times\Dbimon"]
      &
      (\Dbimon\Times\Dbimon)\Times\Dbimon
      \ar[d,"\Assoc"]
      \\
      \Dbimon\Times\Dbimon
      \ar[rr,"\Dbimon\Times\Dbimoncm"]
      &&
      \Dbimon\Times(\Dbimon\Times\Dbimon)
    \end{tikzcd}
  \]
  Let \(n\in\Nat\), we have
  \begin{align*}
    \Assoc\Compl(\Dbimoncm\Times\Dbimon)\Compl\Dbimoncm\Compl\Win_n
    &=\Assoc
      \Compl(\Dbimoncm\Times\Dbimon)
      \Compl\left(\sum_{i=0}^n\Win_i\Times\Win_{n-i}\right)
      \Compl\Inv{\Leftu}\\
    &=\Assoc
      \Compl\left(\sum_{i=0}^n(\Dbimoncm\Compl\Win_i)\Times\Win_{n-i}\right)
      \Compl\Inv{\Leftu}\\
    &=\Assoc
      \Compl\left(\sum_{i=0}^n\left(\sum_{j=0}^i\Win_j\Times\Win_{i-j}\right)\Times\Win_{n-i}\right)
      \Compl(\Inv{\Leftu}\Times\Sone)\Compl\Inv{\Leftu}\\
    &\sumsub\Assoc
      \Compl\left(\sum_{\Biind{i,j,k\in\Nat}{i+j+k=n}}
      (\Win_i\Times\Win_j)\Times\Win_k\right)
      \Compl(\Inv{\Leftu}\Times\Sone)\Compl\Inv{\Leftu} \\
    &=\left(\sum_{\Biind{i,j,k\in\Nat}{i+j+k=n}}
      \Win_i\Times(\Win_j\Times\Win_k)\right)
      \Compl(\Sone\Times\Inv{\Leftu})\Compl\Inv{\Leftu}\\
    &=(\Dbimon\Times\Dbimoncm)\Compl\Dbimoncm\Compl\Win_n
      \text{\quad by a similar computation}
  \end{align*}
  and the announced diagram commutes by joint epicity of the
  \(\Win_n\)'s.

  Finally, we prove the commutativity of the comonoid, which 
  corresponds to the following diagram:
  \[
    \begin{tikzcd}
      \Dbimon
      \ar[r,"\Dbimoncm"]
      \ar[rd,swap,"\Dbimoncm"]
      &
      \Dbimon\Times\Dbimon
      \ar[d,"\Sym"]
      \\
      &
      \Dbimon\Times\Dbimon
    \end{tikzcd}
  \]
  This is done similarly using the fact that %
  \(\Sym\Compl\left(\sum_{i=0}^n\Win_i\Times\Win_{n-i}\right)\Inv{\Leftu}
  =\left(\sum_{i=0}^n\Win_{n-i}\Times\Win_i\right)\Sym\Compl\Inv{\Leftu}\) and
  that %
  \(\Sym\Compl\Inv{\Leftu}=\Inv{\Leftu}\) because
  \(\Leftu_{\Sone}=\Rightu_{\Sone}\).
\end{proof}

Next we equip \(\Dbimon\) with a commutative monoid structure.
The unit of the monoid is the diagonal \(\Wdiag\in\cL(\Sone,\Dbimon)\).
Next for each \(i\in\Nat\) we have
\(\Leftu\Compl(\Wproj_i\Times\Wproj_i)
\in\cL(\Dbimon\Times\Dbimon,\Sone)\), and we set %
\(\Dbimonm
=\Tuple{\Leftu\Compl(\Wproj_i\Times\Wproj_i)}_{i\in\Nat}
\in\cL(\Dbimon\Times\Dbimon,\Dbimon)\).
Notice that %
\(\Dbimonm\Compl(\Win_i\Times\Win_j)
=\Kronecker ij\Win_i\Leftu_{\Sone}
=\Kronecker ij\Win_i\Rightu_{\Sone}\) %
for all \(i,j\in\Nat\), which fully characterizes \(\Dbimonm\) by
\Cref{lemma:elem-win-epic-tens-gen}.

\begin{lemma}
  The triple \((\Dbimon,\Dbimonu,\Dbimonm)\) is a commutative monoid in \(\cL\).
\end{lemma}
\begin{proof}
  We first prove the unit axiom, which corresponds to the following 
  diagram:
  \[
    \begin{tikzcd}
      \Dbimon\Times\Sone
      \ar[r,"\Dbimon\Times\Dbimonu"]
      \ar[dr,swap,"\Rightu"]
      &
      \Dbimon\Times\Dbimon
      \ar[d,"\Dbimonm"]
      \\
      &
      \Dbimon
    \end{tikzcd}
  \]
  so let \(i\in\Nat\), we have
  \begin{align*}
    \Dbimonm\Compl(\Dbimon\Times\Dbimonu)\Compl(\Win_i\Times\Sone)
    &=\Dbimonm\Compl(\Win_i\Times\Wdiag)\\
    &=\Dbimonm\Compl(\Win_i\Times(\sum_{j\in\Nat}\Win_j))\\
    &=\Dbimonm\Compl\sum_{j\in\Nat}\Win_i\Times\Win_j\\
    &=\sum_{j\in\Nat}\Kronecker ij\Win_i\Rightu_{\Sone}\\
    &=\Rightu_{\Dbimon}\Compl(\Win_i\Times\Sone)
      \text{\quad by naturality of }\Rightu
  \end{align*}
  which proves the commutation by \Cref{lemma:elem-win-epic-tens-gen}.

  Next we prove the associativity axiom, which corresponds to the 
  following diagram. 
  \[
    \begin{tikzcd}
      (\Dbimon\Times\Dbimon)\Times\Dbimon
      \ar[r,"\Assoc"]
      \ar[d,swap,"\Dbimonm\Times\Dbimon"]
      &
      \Dbimon\Times(\Dbimon\Times\Dbimon)
      \ar[r,"\Dbimon\Times\Dbimonm"]
      &
      \Dbimon\Times\Dbimon
      \ar[d,"\Dbimonm"]
      \\
      \Dbimon\Times\Dbimon
      \ar[rr,"\Dbimonm"]
      &&
      \Dbimon
    \end{tikzcd}
  \]
  so let \(i,j,k\in\Nat\), we have
  \begin{align*}
    \Dbimonm
    \Compl(\Dbimon\Times\Dbimonm)
    \Compl\Assoc
    \Compl((\Win_i\Times\Win_j)\Times\Win_k)
    &=\Dbimonm
    \Compl(\Dbimon\Times\Dbimonm)
      \Compl(\Win_i\Times(\Win_j\Times\Win_k))\\
    &=\Dbimonm
      \Compl(\Win_i\Times(\Kronecker jk\Win_j\Compl\Leftu))\\
    &=\Kronecker jk\Dbimonm
      \Compl(\Win_i\Times\Win_j)
      \Compl(\Sone\Times\Leftu)\\
    &=\Kronecker jk\Kronecker ij\Win_i
      \Compl\Leftu\Compl(\Sone\Times\Leftu)\\
    &=\Dbimonm
      \Compl(\Dbimonm\Times\Dbimon\Compl)
      \Compl((\Win_i\Times\Win_j)\Times\Win_k)
  \end{align*}
  by a similar computation, so the diagram commutes by
  \Cref{lemma:elem-win-epic-tens-gen}.
  Finally, we prove the commutativity of the monoid, which 
  corresponds to the following diagram.
  \[
    \begin{tikzcd}
      \Dbimon\Times\Dbimon
      \ar[r,"\Sym"]
      \ar[rd,swap,"\Dbimonm"]
      &
      \Dbimon\Times\Dbimon
      \ar[d,"\Dbimonm"]\\
      &
      \Dbimon
    \end{tikzcd}
  \]
  using again the fact that \(\Leftu_{\Sone}=\Rightu_{\Sone}\).
\end{proof}

\begin{theorem} \label{thm:Dbimon-bimonoid}
  The tuple \((\Dbimon,\Dbimonu,\Dbimonm,\Dbimoncu,\Dbimoncm)\) is a
  bicommutative bimonoid.
\end{theorem}
\begin{proof}
  First, we prove the compatibility between the unit and the counit:
  \[
    \begin{tikzcd}
      \Sone
      \ar[r,"\Dbimonu"]
      \ar[dr,swap,"\Id"]
      &
      \Dbimon
      \ar[d,"\Dbimoncu"]
      \\
      &
      \Sone
    \end{tikzcd}
  \]
  This holds because \(\Dbimoncu\Compl\Dbimonu=\Id_{\Sone}\).
  Next we prove the compatibility between the unit and the comultiplication
  (we keep the unitors of the tensor $\tensorUnitL$ and $\tensorUnitR$
  implicit):
\[ \begin{tikzcd}
	1 & \Dbimon \\
	& {\Dbimon \tensor \Dbimon}
	\arrow["\Dbimonu", from=1-1, to=1-2]
	\arrow["{\Dbimonu \tensor \Dbimonu}"', from=1-1, to=2-2]
	\arrow["\Dbimoncm", from=1-2, to=2-2]
\end{tikzcd} \]
Remember that 
\[ \Dbimonu \tensor \Dbimonu = \left( \sum_{i \in \N} \Sprojl_i \right) 
\tensor \left(\sum_{j \in \N} \Sprojl_j \right) \sumsub \sum_{i, j \in \N}œ(\Sprojl_i \tensor \Sprojl_j)\]
by \ref{ax:S-sm-dist}, and
\[ \Dbimoncm \compl \Dbimonu = \Dbimoncm \compl \left( \sum_{n \in \N} \Sprojl_n \right) 
= \sum_{n \in \N} (\Dbimoncm \compl \Sprojl_n) 
= \sum_{n \in \N} (\sum_{i \in \N} \Sprojl_i \tensor \Sprojl_{n-i}) \] 
by characterization of $\Dbimoncm$. Those two sums are equal, by \ref{ax:pa}.

Next we prove the compatibility between the counit and the multiplication (again,
the unitors are left implicit):
\[
\begin{tikzcd}
	{\Dbimon \tensor \Dbimon} & \Dbimon \\
	& 1
	\arrow["\Dbimonm", from=1-1, to=1-2]
	\arrow["{\Dbimoncu \tensor \Dbimoncu}"', from=1-1, to=2-2]
	\arrow["\Dbimoncu", from=1-2, to=2-2]
\end{tikzcd} \]
The diagram commutes by the following computation and the joint epicity of the 
$\Sprojl_i \tensor \Sprojl_j$.
\[ \Dbimoncu \compl \Dbimonm \compl (\Sprojl_i \tensor \Sprojl_j)  
= \Dbimoncu \compl (\Kronecker ij\Win_i\Rightu_{\Sone}) 
= \begin{cases} \id \text{ if } i = j = 0 \\
  0 \text{ otherwise}
\end{cases} = (\Dbimoncu \tensor \Dbimoncu) \compl (\Sprojl_i \tensor \Sprojl_j)
\, . \]

  Finally, we prove the compatibility between the multiplication and 
  the comultiplication:
  \[ 
    \begin{tikzcd}
      \Dbimon\Times\Dbimon
      \ar[r,"\Dbimonm"]
      \ar[d,swap,"\Dbimoncm\Times\Dbimoncm"]
      &
      \Dbimon
      \ar[r,"\Dbimoncm"]
      &
      \Dbimon\Times\Dbimon
      \\
      (\Dbimon\Times\Dbimon)\Times(\Dbimon\Times\Dbimon)
      \ar[rr,"\Sym_{2,3}"]
      &&
      (\Dbimon\Times\Dbimon)\Times(\Dbimon\Times\Dbimon)
      \ar[u,swap,"\Dbimonm\Times\Dbimonm"]
    \end{tikzcd}
  \]
  where \(\Sym_{2,3}\) is defined using the canonical isomorphisms
  \(\Assoc\) and \(\Sym\) of the symmetric monoidal structure of \(\cL\), and is
  characterized by
  \(\Sym_{2,3}
  \Compl((\Win_{i_1}\Times\Win_{i_2})\Times(\Win_{i_3}\Times\Win_{i_4}))
  =(\Win_{i_1}\Times\Win_{i_3})\Times(\Win_{i_2}\Times\Win_{i_4})\).
  We have
  \begin{align*}
    \Dbimoncm
    \Compl\Dbimonm
    \Compl(\Win_i\Times\Win_j)
    =\Kronecker ij\Dbimoncm\Compl\Win_i\Compl\Leftu
    =\Kronecker ij\left(\sum_{k=0}^i\Win_k\Times\Win_{i-k}\right)\Compl\Leftu
  \end{align*}
  and
  \begin{align*}
    (\Dbimonm\Times\Dbimonm)
    \Compl\Sym_{2,3}
    &
    \Compl(\Dbimoncm\Times\Dbimoncm)
    \Compl(\Win_i\Times\Win_j)\\
    &=(\Dbimonm\Times\Dbimonm)
    \Compl\Sym_{2,3}
      \Compl\left(\left(\sum_{i_1+i_2=i}\Win_{i_1}\Times\Win_{i_2}\right)
      \Times\left(\sum_{j_1+j_2=j}\Win_{j_1}\Times\Win_{j_2}\right)\right)
      \Compl(\Inv\Leftu\Times\Inv\Leftu)\\
    &=(\Dbimonm\Times\Dbimonm)
      \Compl\Sym_{2,3}
      \Compl\left(\sum_{\Biind{i_1+i_2=i}{j_1+j_2=j}}
      (\Win_{i_1}\Times\Win_{i_2})\Times(\Win_{j_1}\Times\Win_{j_2})\right)
      \Compl(\Inv\Leftu\Times\Inv\Leftu)\\
    &=(\Dbimonm\Times\Dbimonm)
      \Compl\left(\sum_{\Biind{i_1+i_2=i}{j_1+j_2=j}}
      (\Win_{i_1}\Times\Win_{j_1})\Times(\Win_{i_2}\Times\Win_{j_2})\right)
      \Compl\Sym_{2,3}
      \Compl(\Inv\Leftu\Times\Inv\Leftu)\\
    &=\left(\sum_{\Biind{i_1+i_2=i}{j_1+j_2=j}}
      \Kronecker{i_1}{j_1}\Kronecker{i_2}{j_2}
      \Win_{i_1}\Times\Win_{i_2}\right)
      \Compl(\Leftu\Times\Leftu)
      \Compl\Sym_{2,3}
      \Compl(\Inv\Leftu\Times\Inv\Leftu)
  \end{align*}
  which ends the proof that
  \(\Dbimoncm \Compl\Dbimonm
  \Compl(\Win_i\Times\Win_j)=(\Dbimonm\Times\Dbimonm) \Compl\Sym_{2,3}
  \Compl(\Dbimoncm\Times\Dbimoncm) \Compl(\Win_i\Times\Win_j)\) %
  for all \(i,j\in\Nat\) upon observing that if %
  \(\Kronecker{i_1}{j_1}\Kronecker{i_2}{j_2}=1\) then
  \(i=i_1+i_2=j_1+j_2=j\) and
  \((\Leftu_{\Sone}\Times\Leftu_{\Sone}) \Compl\Sym_{2,3}
  \Compl(\Inv\Leftu_{\Sone}\Times\Inv\Leftu_{\Sone})
  =\Id_{(\Sone\Times\Sone)\Times(\Sone\Times\Sone)}\).
  The diagram then commutes by \Cref{lemma:elem-win-epic-tens-gen}.
\end{proof}

\subsection{Correspondence between the bimonad $\S$ and the bimonoid structure}
\label{sec:bimonoid-bimonad-mate}
We can define a functor %
\(\Sfuntens:\cL\to\cL\) by \(\Sfuntens X=X\Times\Dbimon\) %
and similarly on morphisms:
\(\Sfuntens f=f\Times\Dbimon\in\cL(\Sfuntens X,\Sfuntens Y)\) %
if \(f\in\cL(X,Y)\).
It is folklore that the bimonoid structure of \(\Dbimon\) induces straightforwardly a
bimonad structure on this functor, see for instance example 2.10 of~\cite{Bruguieres11}
in the setting of braided categories.

The comonad structure $\comonadSl$ is given by
\[
  \begin{tikzcd}
    \Sfuntens X=X\Times\Dbimon
    \ar[r,"X\Times\Dbimoncu"]
    &
    X\Times\Sone
    \ar[r,"\Rightu"]
    &
    X
  \end{tikzcd} \]
  \[ 
  \begin{tikzcd}
    {\Sl X = X \tensor \Dbimon} & {X \tensor (\Dbimon \tensor \Dbimon)} 
    & {(X \tensor \Dbimon) \tensor \Dbimon = \Sl^2 X}
    \arrow["{X \tensor \Dbimoncm}", from=1-1, to=1-2]
    \arrow["{\tensorAssoc}", from=1-2, to=1-3]
  \end{tikzcd} \]
The monad structure $\monadSl$ is given by 
\begin{center}
  \begin{tikzcd}
    X & {X \tensor 1} & {X \tensor \Dbimon = \Sl X}
    \arrow["{\tensorUnitR^{-1}}", from=1-1, to=1-2]
    \arrow["{X \tensor \Dbimonu}", from=1-2, to=1-3]
  \end{tikzcd}

  \begin{tikzcd}
    \Sfuntens^2X=(X\Times\Dbimon)\Times\Dbimon
    \ar[r,"\Assoc"]
    &
    X\Times(\Dbimon\Times\Dbimon)
    \ar[r,"X\Times\Dbimonm"]
    &
    X\Times\Dbimon
  \end{tikzcd}
\end{center}
and the distributive law is (keeping the associativity isomorphisms implicit)
\begin{center}
  \begin{tikzcd}
    \Sfuntens^2X=X\Times\Dbimon\Times\Dbimon
    \ar[r,"X\Times\Sym"]
    &[4em]
    X\Times\Dbimon\Times\Dbimon.
  \end{tikzcd}
\end{center}

Now the functor \(\Sfun=(\Dbimon\Limpl\_):\cL\to\cL\) is the right
adjoint of \(\Sfuntens\) and hence, as shown in
\cref{sec:mate-monade-comonade}, \(\Sfun\) inherits 
from the mate construction a bimonad
structure which is exactly the same as the one described in
Section~\ref{sec:bimonad}. The different constructions are summarized in 
\cref{fig:bimonoid-to-bimonad}. Let us write down the details for one of those 
constructions.

\begin{figure}
  \begin{center}
  \begin{tabular}{c c c c c}
    Bimonoid $\Dbimon$ & & Bimonad $\Sl$ & & Bimonad $\S$ \\ \\

    Projection $\prodProj_i \in \categoryLL(\Dbimon, 1)$ & & 
    $\tensorUnitR \compl (X \tensor \prodProj_i) \in \categoryLL(X \tensor \Dbimon, X)$ 
    & & Injection $\Sinj_i$ \\
    Injection $\Sprojl_i \in \categoryLL(1, \Dbimon)$ & & 
    $(X \tensor \Sprojl_i) \compl \tensorUnitR^{-1} \in \categoryLL(X, X \tensor \Dbimon)$
    & & Projection $\Sproj_i$ \\ \\ 

    Monoid unit $\Dbimonu$ & & Monad unit $(X \tensor \Dbimonu) \compl \tensorUnitR^{-1}$
    & & Comonad unit $\Ssum$ \\
    Monoid multiplication $\Dbimonm$ & & Monad sum $(X \tensor \Dbimonm) \compl \tensorAssoc$
    & & Comonad sum $\Slift$ \\
    Comonoid unit $\Dbimoncu$ & $\Longleftrightarrow$ & 
    Comonad unit $\tensorUnitR \compl (X \tensor \Dbimoncu)$ 
    & $\overset{\text{mates}}{\Longleftrightarrow}$& Monad unit $\Sinj_0$ \\
    Comonoid multiplication $\Dbimoncm$ & & Comonad sum $\tensorAssoc \compl (X \tensor \Dbimoncm)$
    & & Monad sum $\SmonadSum$ \\
    Commutativity $\tensorSym$ & & Distributive law $X \tensor \tensorSym$
    & & Distributive law $\Sswap$ \\ \\
  \end{tabular}
\end{center}
\caption{Bimonoid and bimonad relations}
\label{fig:bimonoid-to-bimonad}
\end{figure}

By \cref{rem:mate-simplified} instantiated in  
$\ladj \dashv \radj = \idfun \dashv \idfun$ and 
$\ladj['] \dashv \radj['] = \_ \tensor \Dbimon \dashv \Dbimon \linarrow \_$, 
the mate construction maps a natural transformation
$l_X \in \categoryLL(X \tensor \Dbimon, X)$ to a 
natural transformation 
$r_X \in \categoryLL(X, \Dbimon \linarrow X)$ defined as 
$r_X = \cur(l_X)$. 
Recall from \cref{eq:projection-elementary} that
\[ \Sproj_i \compl r_X
= l_X \compl (X \tensor \Sprojl_i) \compl \tensorUnitR^{-1} \, .\]
This equation actually corresponds to the 
compositionality of the mate construction, using the fact that 
$\Sproj_i$ is itself the mate of $(X \tensor \Sprojl_i) \compl \tensorUnitR^{-1}$.
Then the joint monicity of the $\Sproj_i$ and the joint epicity of the $
X \tensor \Sprojl_i$ 
implies that the equation above provides a complete characterization of 
$l$ from $r$, and of $r$ from $l$.
In particular, the mate of $\tensorUnitR \compl (X \tensor \prodProj_i)
\in \categoryLL(X \tensor \Dbimon, X)$ must be
$\Sinj_i \in \categoryLL(X, \Dbimon \linarrow X)$.

The same kind of argument applies for the rest of the bimonad structure.
The mate construction also relates the lax monoidal structures on 
$\S$ with oplax monoidal  
structures on $\Sl$, let us provide some details.

\begin{definition} We define a natural transformation (up to associativity)
\begin{equation} \label{eq:Sdistl}
  \Sdistl_{X_0, X_1} = 
\begin{tikzcd}
	{X_0 \tensor X_1 \tensor \Dbimon} &[2em] {X_0 \tensor X_1 \tensor \Dbimon \tensor \Dbimon} &[2em] {X_0 \tensor \Dbimon \tensor X_1 \tensor \Dbimon .}
	\arrow["{X_0 \tensor X_1 \tensor \Dbimoncm}", from=1-1, to=1-2]
	\arrow["{X_0 \tensor \tensorSym \tensor \Dbimon}", from=1-2, to=1-3]
\end{tikzcd}
\end{equation}
That is, 
$\Sdistl_{X_0, X_1} \in \categoryLL(\Sl (X_0 \tensor X_1), \Sl X_0 \tensor \Sl X_1)$
\end{definition}
The natural transformation $\Sdistl$ is characterized by the equation
$\Sdistl_{X_0, X_1} \compl (X_0 \tensor X_1 \tensor \Sprojl_n)
= \sum_{k=0}^{n} (X_0 \tensor \Sprojl_k \tensor X_1 \tensor \Sprojl_{n-k})$.
Applying the mate construction of \cref{sec:mate} 
with $\lfun = \rfun = \_ \tensor \_$ and taking
$\ladj['] \dashv \radj['] =  \_ \tensor \Dbimon \dashv \Dbimon \linarrow \Dbimon \linarrow \_$
and $\ladj \dashv \radj = (\_ \tensor \Dbimon) \times (\_ \tensor \Dbimon)
\dashv (\Dbimon \linarrow \_) \times (\Dbimon \linarrow \_)$ (this is the product 
of the adjunction $\_ \tensor \Dbimon \dashv \Dbimon \linarrow \_$ with itself, 
see \cref{sec:mate-oplax-lax}) yields a natural transformation 
\[ r = \cur((\ev \tensor \ev) \compl \Sdistl_{\Dbimon \linarrow X_0, \Dbimon \linarrow X_1}) 
\in \categoryLL((\Dbimon \linarrow X_0) \tensor (\Dbimon \linarrow X_1),
\Dbimon \linarrow (X_0 \tensor X_1)) . \]
The compositionality of the mate construction implies that
$\Sproj_i \compl r
= \sum_{k=0}^i \Sproj_i \tensor \Sproj_{i-k}$ so $r = \Sdist$ 
(where $\Sdist$ is defined in \cref{sec:summability-tensor}). Recall that 
$(\monadS, \Sinj_0, \Sdist)$ is a lax monoidal monad. By
\cref{thm:mate-oplax-lax}, monoidality is preserved through 
the mate construction so $(\comonadSl, 
\tensorUnitR \compl (\id \tensor \Dbimoncu), \Sdistl)$ is
an oplax monoidal comonad.
This oplax structure is the one associated with the 
following costrengths (we keep the use of the associator $\tensorAssoc$ 
of the monoidal product implicit) 
\[ \begin{split} 
  X_0 \tensor \tensorSym_{X_1, \Dbimon} &\in \categoryLL(\Sl (X_0 \tensor X_1), \Sl X_0 \tensor X_1) \\
\id &\in \categoryLL(\Sl(X_0 \tensor X_1), X_0 \tensor \Sl X_1).
\end{split} \]
Those costrengths were implicitly used when defining $\Sdistl$ in \cref{eq:Sdistl}.
The mates of those two natural transformations are 
$\SstrL$ and $\SstrR$ respectively\footnote{This explains the 
definition of $\SstrL$ in the
proof of the forward implication of \cref{thm:representable-compatible}}.

Finally, assume that $\cL$ is cartesian. 
There is a natural transformation 
\[ \prodPairing{\prodProj_i \tensor \Dbimon}
\in \categoryLL((\withFam X_i) \tensor \Dbimon, \withFam (X_i \tensor \Dbimon)) \]
whose mate is precisely $\SprodDist^{-1}$. This was implicitly shown in the proof of 
\cref{thm:representable-compatible} when proving that 
\[ \SprodDist^{-1} 
= \Spairing{\withFam \Sproj_i} 
= \cur_{\Dbimon} (\prodPairing{\ev \compl (\prodProj_i \tensor \Dbimon)}) \in 
   \categoryLL(\withFam (\Dbimon \linarrow X_i), \Dbimon \linarrow \withFam X_i) \, . \]

\subsection{A mate to the distributive law $\Sdl$}

Let \(\cL\) be a representably $\Sigma$-additive category.
We assume moreover that \(\cL\) is a $\Sigma$-additive resource category (see \cref{def:sigma-additive-resource-category}).
Applying the results of \cref{sec:compatibility-adjunction}, we can 
show that the mate construction 
induces a bijection between the Taylor expansions 
$\Sdl : \oc \S \naturalTrans \S \oc$ described in \cref{sec:Taylor}
and distributive laws
$\Sdlmate : \Sl \oc \naturalTrans \oc \Sl$
subject to commutations that we now detail.

The first commutation, \ref{ax:Sdlmate-chain}, 
means that $\Sdlmate$ is a distributive law
between the functor $\Sl$ and the comonad $\oc\_$.
\begin{equation*} \labeltext{($\Sdlmate$-chain)}{ax:Sdlmate-chain}
  \text{\ref{ax:Sdlmate-chain}} \quad
  \begin{tikzcd}
    {!X \tensor \Dbimon} & {\oc (X \tensor \Dbimon)} \\
    & {X \tensor \Dbimon}
    \arrow["{\Sdlmate_X}", from=1-1, to=1-2]
    \arrow["\der", from=1-2, to=2-2]
    \arrow["{\der \tensor \Dbimon}"', from=1-1, to=2-2]
  \end{tikzcd} \quad 
  \begin{tikzcd}
	  {\oc X \tensor \Dbimon} && {\oc(X \tensor \Dbimon)} \\
	  {\oc \oc X \tensor \Dbimon} & {\oc (\oc X \tensor \Dbimon)} & {\oc \oc (X \tensor \Dbimon)}
	  \arrow["{\dig \tensor \Dbimon}"', from=1-1, to=2-1]
	  \arrow["{\Sdlmate_{\oc X}}"', from=2-1, to=2-2]
	  \arrow["{\oc \Sdlmate_X}"', from=2-2, to=2-3]
	  \arrow["{\Sdlmate_X}", from=1-1, to=1-3]
	  \arrow["\dig", from=1-3, to=2-3]
  \end{tikzcd}
\end{equation*}
 By \cref{prop:mate-dl-comonad},
\ref{ax:Sdlmate-chain} holds if and only if 
$\Sdl$ is a distributive law between the functor $\S$ and the comonad $\oc\_$,
that is, if \ref{ax:Sdl-chain} holds.

The commutation of \ref{ax:Sdlmate-local} means that the natural transformation
$(X \tensor \Sprojl_0) \compl \tensorUnitR^{-1} \in \categoryLL(X, X \tensor \Dbimon)$
is a morphism of distributive laws.
\begin{equation*} \labeltext{($\Sdlmate$-local)}{ax:Sdlmate-local}
  \text{\ref{ax:Sdlmate-local}} \quad
\begin{tikzcd}
	{!X} \\
	{!X \tensor 1} & {!(X \tensor 1)} \\
	{!X \tensor \Dbimon} && {!(X \tensor \Dbimon)}
	\arrow["{\Sdlmate_X}", from=3-1, to=3-3]
	\arrow["\tensorUnitR"', from=1-1, to=2-1]
	\arrow["{!X \tensor \Sprojl_0}"', from=2-1, to=3-1]
	\arrow["{! \tensorUnitR}", from=1-1, to=2-2]
	\arrow["{!(X \tensor \Sprojl_0)}", from=2-2, to=3-3]
\end{tikzcd}
\end{equation*} By \cref{prop:mate-dl-morphism}, this 
is a morphism of distributive laws 
if and only if $\Sproj_0$ (its mate) is a morphism of distributive laws, that is, 
if \ref{ax:Sdl-local} commutes.

The commutation of \ref{ax:Sdlmate-add} means that $\Sdlmate$ is a distributive 
law between the functor $\oc\_$ and the comonad structure on $\Sl$
described in \cref{sec:bimonoid-bimonad-mate}.
\begin{equation*} \labeltext{($\Sdlmate$-add)}{ax:Sdlmate-add}
  \text{\ref{ax:Sdlmate-add}} \allowbreak
\begin{tikzcd}
	{!X \tensor \Dbimon} & {!(X \tensor \Dbimon)} \\
	& {!(X \tensor 1)} \\
	{!X \tensor 1} & {!X}
	\arrow["{\Sdlmate_X}", from=1-1, to=1-2]
	\arrow["{!(X \tensor \prodProj_0)}", from=1-2, to=2-2]
	\arrow["{!\tensorUnitR}", from=2-2, to=3-2]
	\arrow["{!X \tensor \prodProj_0}"', from=1-1, to=3-1]
	\arrow["\tensorUnitR"', from=3-1, to=3-2]
\end{tikzcd}
\begin{tikzcd}
	{!X \tensor \Dbimon} && {!(X \tensor \Dbimon)} \\
	{!X \tensor (\Dbimon \tensor \Dbimon)} && {!(X \tensor (\Dbimon \tensor \Dbimon))} \\
	{(!X \tensor \Dbimon) \tensor \Dbimon} & {!(X \tensor \Dbimon) \tensor \Dbimon} & {!((X \tensor \Dbimon) \tensor \Dbimon)}
	\arrow["{!X \tensor \Dbimoncm}"', from=1-1, to=2-1]
	\arrow["\tensorAssoc"', from=2-1, to=3-1]
	\arrow["{\Sdlmate_X}", from=1-1, to=1-3]
	\arrow["{\Sdlmate_X \tensor \Dbimon}"', from=3-1, to=3-2]
	\arrow["{\Sdlmate_{X \tensor \Dbimon}}"', from=3-2, to=3-3]
	\arrow["{!(X \tensor \Dbimoncm)}", from=1-3, to=2-3]
	\arrow["{! \tensorAssoc}", from=2-3, to=3-3]
\end{tikzcd}
\end{equation*} 
By \cref{prop:mate-dl-comonad-monad},
\ref{ax:Sdlmate-add} holds if and only if $\Sdl$ is a distributive law 
between the functor $\oc\_$ and the comonad $\comonadS$, that is if 
\ref{ax:Sdl-add} holds.

In the next diagrams, the use of the associator $\tensorAssoc$ 
of the symmetric monoidal structure is kept implicit.
The commutation of \ref{ax:Sdlmate-Schwarz} means that 
$X \tensor \tensorSym$ is a morphism of distributive laws. 
\begin{equation*} \labeltext{($\Sdlmate$-Schwarz)}{ax:Sdlmate-Schwarz}
  \text{\ref{ax:Sdlmate-Schwarz}} \quad
\begin{tikzcd}
	{!X \tensor \Dbimon \tensor \Dbimon} & {!(X \tensor \Dbimon) \tensor \Dbimon} & {!(X \tensor \Dbimon \tensor \Dbimon)} \\
	{!X \tensor \Dbimon \tensor \Dbimon} & {!(X \tensor \Dbimon) \tensor \Dbimon} & {!(X \tensor \Dbimon \tensor \Dbimon)}
	\arrow["{\Sdlmate_X \tensor \Dbimon}", from=1-1, to=1-2]
	\arrow["{\Sdlmate_{X \tensor \Dbimon}}", from=1-2, to=1-3]
	\arrow["{!X \tensor \tensorSym}"', from=1-1, to=2-1]
	\arrow["{\Sdlmate_X \tensor \Dbimon}"', from=2-1, to=2-2]
	\arrow["{\Sdlmate_{X \tensor \Dbimon}}"', from=2-2, to=2-3]
	\arrow["{!(X \tensor \tensorSym)}", from=1-3, to=2-3]
\end{tikzcd}
\end{equation*}
By \cref{prop:mate-dl-morphism},\ref{ax:Sdlmate-Schwarz} holds 
if and only if $\Sswap$ (its mate) is a morphism of distributive laws, that is,
if \ref{ax:Sdl-Schwarz} hold.

The commutation of \ref{ax:Sdlmate-with} means that 
$\seelyTwo$ is a morphism of distributive laws between the composition 
of $\Sdlmate$ with $\Sdistl$, and the composition of 
$\Sdlmate$ with $\prodPair{\prodProj_0 \tensor \Dbimon}
{\prodProj_1 \tensor \Dbimon}$.
\begin{equation*} \labeltext{($\Sdlmate$-$\with$)}{ax:Sdlmate-with}
  \text{\ref{ax:Sdlmate-with}} \quad
\begin{tikzcd}[column sep = large]
	{!X_0 \tensor !X_1 \tensor \Dbimon} & {!X_0 \tensor \Dbimon \tensor !X_1 \tensor \Dbimon} & {!(X_0 \tensor \Dbimon) \tensor !(X_1 \tensor \Dbimon)} \\
	{!(X_1 \with X_2) \tensor \Dbimon} & {!((X_0 \with X_1) \tensor \Dbimon)} & {!((X_0 \tensor \Dbimon) \with (X_1 \tensor \Dbimon))}
	\arrow["{\Sdlmate \tensor \Sdlmate}", from=1-2, to=1-3]
	\arrow["{\seelyTwo \tensor \Dbimon}"', from=1-1, to=2-1]
	\arrow["\seelyTwo", from=1-3, to=2-3]
	\arrow["{\Sdistl}", from=1-1, to=1-2]
	\arrow["\Sdlmate"', from=2-1, to=2-2]
	\arrow["{!\prodPair{\prodProj_0 \tensor \Dbimon}{\prodProj_1 \tensor \Dbimon}}"', from=2-2, to=2-3]
\end{tikzcd}
\end{equation*}
 But as we saw in \cref{sec:bimonoid-bimonad-mate},
the mate of $\Sdistl$ is $\Sdist$ and the mate of 
$\prodPair{\prodProj_0 \tensor \Dbimon}
{\prodProj_1 \tensor \Dbimon}$ is $\SprodDist^{-1}$. 
By compositionality of the mate construction and 
\cref{prop:mate-dl-morphism-vertical}, \ref{ax:Sdlmate-with}
holds if and only if 
\ref{ax:Sdl-with} holds.

Finally, the commutations \ref{ax:Sdlmate-lin} and \ref{ax:Sdlmate-Taylor} 
means that $\Sdlmate$ is a distributive law between the functor $\oc\_$ and 
the monad structure on $\Sl$ described in \cref{sec:bimonoid-bimonad-mate}.

\begin{equation*} \labeltext{($\Sdlmate$-lin)}{ax:Sdlmate-lin}
  \text{\ref{ax:Sdlmate-lin}} \quad
  \begin{tikzcd}
	{(!X \tensor \Dbimon) \tensor \Dbimon} & {!(X \tensor \Dbimon) \tensor \Dbimon} & {!((X \tensor \Dbimon) \tensor \Dbimon)} \\
	{!X \tensor (\Dbimon \tensor \Dbimon)} && {!(X \tensor (\Dbimon \tensor \Dbimon))} \\
	{!X \tensor \Dbimon} && {!(X \tensor \Dbimon)}
	\arrow["{\Sdlmate_X \tensor \Dbimon}", from=1-1, to=1-2]
	\arrow["{\Sdlmate_{X \tensor \Dbimon}}", from=1-2, to=1-3]
	\arrow["\tensorAssoc"', from=1-1, to=2-1]
	\arrow["{!X \tensor \Dbimonm}"', from=2-1, to=3-1]
	\arrow["{\Sdlmate_X}"', from=3-1, to=3-3]
	\arrow["{! \tensorAssoc}", from=1-3, to=2-3]
	\arrow["{!(X \tensor \Dbimonm)}", from=2-3, to=3-3]
\end{tikzcd}
\end{equation*}
\begin{equation*} \labeltext{($\Sdlmate$-analytic)}{ax:Sdlmate-Taylor}
  \text{\ref{ax:Sdlmate-Taylor}} \quad
  \begin{tikzcd}
	{!X} \\
	{!X \tensor 1} & {!(X \tensor 1)} \\
	{!X \tensor \Dbimon} && {!(X \tensor \Dbimon)}
	\arrow["{\tensorUnitR^{-1}}"', from=1-1, to=2-1]
	\arrow["{!X \tensor \Dbimonu}"', from=2-1, to=3-1]
	\arrow["{\Sdlmate_X}"', from=3-1, to=3-3]
	\arrow["{! \tensorUnitR^{-1}}", from=1-1, to=2-2]
	\arrow["{!(X \tensor \Dbimonu)}", from=2-2, to=3-3]
\end{tikzcd}
\end{equation*}
By \cref{prop:mate-dl-monad-comonad}, these diagrams commute if and only if 
$\Sdl$ is a distributive law between the functor $\oc\_$ and the comonad
$\comonadS$, that is, if and only if \ref{ax:Sdl-lin} and \ref{ax:Sdl-Taylor}
hold. 
\begin{definition} A left-sided Taylor expansion is a natural transformation 
  $\Sdlmate : \Sl \oc \naturalTrans \oc \Sl$
  that satisfies \ref{ax:Sdlmate-chain}, \ref{ax:Sdlmate-local},
  \ref{ax:Sdlmate-add}, \ref{ax:Sdlmate-Schwarz}, \ref{ax:Sdlmate-with},
  \ref{ax:Sdlmate-lin}. A left-sided Taylor expansion is analytic 
  if it also follows \ref{ax:Sdlmate-Taylor}.
\end{definition}

The following result summarizes what we have proved in this section.

\begin{theorem} \label{thm:Sdl-mate}
  Let $\Sdl : \oc \S \naturalTrans \S \oc$ and 
  $\Sdlmate : \Sl \oc \naturalTrans \oc \Sl$ be mates. Then 
  $\Sdl$ is an analytic Taylor expansion if and only if 
  $\Sdlmate$ is a left-sided analytic Taylor expansion.
\end{theorem}

\subsection{The Taylor coalgebra structure of \(\Dbimon\)}

An important structure can be derived from the cartesian structure 
of $\cL$ and its resource comonad: a lax
symmetric monoidality of the comonad \(\Oc\_\), 
from the symetric monoidal category \((\cL,\Times,\Sone)\) to 
itself (see \cref{rem:hopf-monad} for a definition
of a lax symmetric monoidal comonad). 
This monoidality 
turns $\categoryLL$ into a \emph{linear category}, 
see~\cite{Bierman95}.

More precisely we have a morphism \(\Ocmonz\in\cL(\Sone,\Oc\Sone)\)
and a natural transformation
\(\Ocmont_{X,Y}\in\cL(\Oc X\Times\Oc Y,\Oc{(X\Times Y)})\) which
satisfy the coherence diagrams of 
\cref{def:smf} and coherence diagrams similar to 
\cref{def:monoidal-monad}. They can be defined as the following
compositions of morphisms
\[ \Ocmonz = \begin{tikzcd}
    \Sone
    \ar[r,"\Seelyz"]
    &[-1em]
    \Oc\Top
    \ar[r,"\Digg_\Top"]
    &[1.4em]
    \Occ\Top
    \ar[r,"\Oc{\Invp{\Seelyz}}"]
    &[1.4em]
    \Oc\Sone
  \end{tikzcd}
\]
\[ \Ocmont_{X, Y} = \begin{tikzcd}
    \Oc X\Times\Oc Y
    \ar[r,"\Seelyt_{X,Y}"]
    &
    \Oc{(X\With Y)}
    \ar[r,"\Digg_{X\With Y}"]
    &[1em]
    \Occ{(X\With Y)}
    \ar[r,"\Oc{(\Invp{\Seelyt_{X,Y}})}"]
    &[2.4em]
    \Oc{(\Oc X\Times\Oc Y)}
    \ar[r,"\Oc{(\Digg_X\Times\Digg_Y)}"]
    &[3.2em]
    \Oc{(X\Times Y)}
  \end{tikzcd} \]

Particularly important is the associated coEilenberg-Moore category
\(\Emoc\cL\) whose objects are the coalgebras of \(\Oc\_\), that is, the
pairs \(P=(\Coalgca P,\Coalgst P)\) where \(\Coalgca P\) is an object
of \(P\) and \(\Coalgst P\in\cL(\Coalgca P,\Oc{\Coalgca P})\) makes
the two following diagrams commute
\[
  \begin{tikzcd}
    \Coalgca P
    \ar[r,"\Coalgst P"]
    \ar[rd,swap,"\Id"]
    &
    \Oc{\Coalgca P}
    \ar[d,"\Deru_{\Coalgca P}"]
    \\
    &
    \Coalgca P
  \end{tikzcd}
  \Diagsep
  \begin{tikzcd}
    \Coalgca P
    \ar[r,"\Coalgst P"]
    \ar[d,swap,"\Coalgst P"]
    &
    \Oc{\Coalgca P}
    \ar[d,"\Digg_{\Coalgca P}"]
    \\
    \Oc{\Coalgca P}
    \ar[r,"\Oc{\Coalgst P}"']
    &
    \Occ{\Coalgca P}    
  \end{tikzcd}
\]
In this category, an element of \(\Emoc\cL(P,Q)\) is called a $\Oc$-coalgebra 
morphism and is an %
\(f\in\cL(\Coalgca P,\Coalgca Q)\) such that
\[
  \begin{tikzcd}
    \Coalgca P
    \ar[r,"f"]
    \ar[d,swap,"\Coalgst P"]
    &
    \Coalgca Q
    \ar[d,"\Coalgst Q"]
    \\
    \Oc{\Coalgca P}
    \ar[r,"\Oc{f}"']
    &
    \Oc{\Coalgca Q}
  \end{tikzcd}
\]

It is well known that \((\Sone,\Ocmonz)\) is a $!$-colagebra 
that we simply denote as \(\Sone\) (so that
\(\Coalgst{\Sone}=\Ocmonz\)) and that, given objects \(P_1\) and
\(P_2\) of \(\Emoc\cL\), the object
\(\Coalgca{P_1}\Times\Coalgca{P_2}\) can be equipped with a
\(\oc\)-coalgebra structure
\begin{center}
  \begin{tikzcd}
    \Coalgca{P_1}\Times\Coalgca{P_2}
    \ar[r,"\Coalgst{P_1}\Times\Coalgst{P_2}"] %
    &[2em]
    \Oc{\Coalgca{P_1}}\Times\Oc{\Coalgca{P_2}}
    \ar[r,"\Ocmont"]
    &[-1em]
    \Oc{(\Coalgca{P_1}\Times\Coalgca{P_2}).}
  \end{tikzcd}
\end{center}
We use \(P_1\Times P_2\) to denote this coalgebra, so that %
\(\Coalgst{P_1\Times P_2}=\Ocmont\Compl(\Coalgst{P_1}\Times\Coalgst{P_2})\).
This coalgebra structure corresponds the lifting of the 
symmetric monoidal structure of $\categoryLL$ to 
$\Emoc\cL$ mentioned in \cref{rem:hopf-monad}.

Each coalgebra \(P\) can be equipped with a weakening
morphism \(\Coalgw_P\in\Emoc\cL(P,\Sone)\) and a contraction morphism
\(\Coalgc_P\in\Emoc\cL(P,P\Times P)\) which can be defined 
from $\weak$ and $\contr$ (see \cref{eq:weakening,eq:contraction})
as follows.
\[
\Coalgw_P = \begin{tikzcd}
	{\Coalgca P} & {\Oc{\Coalgca P}} & 1
	\arrow["{\Coalgst P}", from=1-1, to=1-2]
	\arrow["{\weak_{\Coalgca P}}", from=1-2, to=1-3]
\end{tikzcd} \] 
\[  \Coalgc_P = 
\begin{tikzcd}
	{\Coalgca P} & {\Oc{\Coalgca P}} & {\Oc{\Coalgca P} \tensor \Oc{\Coalgca P}} & {\Coalgca P \tensor \Coalgca P}
	\arrow["{\Coalgst P}", from=1-1, to=1-2]
	\arrow["{\contr_{\Coalgca P}}", from=1-2, to=1-3]
	\arrow["{\der \tensor \der}", from=1-3, to=1-4]
\end{tikzcd} \]

\Cref{th:resource-EM-comonoid,th:resource-EM-cartesian} that follow 
are non-trivial results, we refer to Proposition 28 of~\cite{Mellies09}
for a proof.
\begin{theorem} \label{th:resource-EM-comonoid}
  For any object \(P\) of \(\Emoc\cL\), the triple %
  \((P,\Coalgw_P,\Coalgc_P)\) is a commutative comonoid in \(\Emoc\cL\).
\end{theorem}

Given \((f_i\in\Emoc\cL(P_i,Q_i))_{i=1,2}\), it is easy to check that %
\(f_1\Times f_2\in\Emoc\cL(P_1\Times P_2,Q_1\Times Q_2)\).
And therefore if \((f_i\in\Emoc\cL(P,Q_i))_{i=1,2}\), one can define %
\(\Tupleoc{f_1,f_2}
=(f_1\Times f_2)\Compl\Coalgc_Q\in\Emoc\cL(Q,P_1\Times P_2)\).
We can also define projections
\(\Coalgpr_i\in\Emoc\cL(P_1\Times P_2,P_i)\), for instance the first
projection is simply
\begin{center}
  \begin{tikzcd}
    \Coalgca{P_1}\Times\Coalgca{P_2}
    \ar[r,"\Coalgca{P_1}\Times\Coalgw_{P_2}"]
    &[2em]
    \Coalgca{P_1}\Times\Sone
    \ar[r,"\Rightu"]
    &[-1em]
    \Coalgca{P_1}.
  \end{tikzcd}
\end{center}

\begin{theorem}\label{th:resource-EM-cartesian}
  The category \(\Emoc\cL\) is cartesian, with \(\Sone\) as terminal
  object and \((P_1,\Times P_2,\Coalgpr_1,\Coalgpr_2)\) as cartesian
  product of \(P_1\) and \(P_2\).
  Given \((f_i\in\Emoc\cL(P,Q_i))_{i=1,2}\),
  \(\Tupleoc{f_1,f_2} =(f_1\Times
  f_2)\Compl\Coalgc_Q\in\Emoc\cL(Q,P_1\Times P_2)\) is the unique
  morphism such that \(\Coalgpr_i\Compl\Tupleoc{f_1,f_2}=f_i\) for
  \(i=1,2\).
\end{theorem}

For any object \(X\) of \(\cL\), the pair
\(\Coalgfree X=(\Oc X,\Digg_X)\) is an object of \(\Emoc\cL\).
This defines a functor \(\Coalgfree:\cL\to\Emoc\cL\) which maps
\(f\in\cL(X,Y)\) to \(\Oc f\in\Emoc\cL(\Coalgfree X,\Coalgfree Y)\) as
easily checked.
The coalgebra \(\Coalgfree X\) is the cofree coalgebra generated by
\(X\) in the sense that, for any object \(P\) of \(\Emoc\cL\) and any
\(f\in\cL(\Coalgca P,X)\), there is exactly one morphism
\(\Prom f\in\Emoc\cL(P,\Coalgfree X)\) such that
\(\Deru_X\Compl\Prom f=f\). 

The ``image'' of this functor is a full subcategory which can be
described, up to equivalence, as the coKleisli category of the \(\oc\_\)
comonad.

\begin{definition}\label{def:elementary-Taylor-structure}
  An analytic coalgebra on \(\cL\) is a morphism %
  \(\Dbimonca\in\cL(\Dbimon,\Oc\Dbimon)\) such that
  \((\Dbimon,\Dbimonca)\) is a $\Oc$-coalgebra, the four
  structure maps of the bimonoid
  \((\Dbimon,\Dbimonu,\Dbimonm,\Dbimoncu,\Dbimoncm)\) are $\Oc$-coalgebra 
  morphisms, and such that 
  $\Sprojl_0$ is a $\Oc$-colagebra morphism.
  When such an analytic coalgebra is given, and when there are no
  possible ambiguities, we simply use \(\Dbimon\) to denote the
  coalgebra \((\Dbimon,\Dbimonca)\).
\end{definition}

Let us make these conditions more explicit.
The fact that \((\Dbimon,\Dbimonca)\) is a $\Oc$-colagebra
means that the two following diagrams commute
\[
  \begin{tikzcd}
    \Dbimon
    \ar[r,"\Dbimonca"]
    \ar[rd,swap,"\Id_{\Dbimon}"]
    &
    \Oc\Dbimon
    \ar[d,"\Deru_{\Dbimon}"]
    \\
    &
    \Dbimon
  \end{tikzcd}
  \Diagsep
  \begin{tikzcd}
    \Dbimon
    \ar[r,"\Dbimonca"]
    \ar[d,swap,"\Dbimonca"]
    &
    \Oc\Dbimon
    \ar[d,"\Digg_{\Dbimon}"]
    \\
    \Oc\Dbimon
    \ar[r,"\Oc{\Dbimonca}"']
    &
    \Occ\Dbimon
  \end{tikzcd}
\]

The fact that \(\Dbimoncu\in\Emoc\cL(\Dbimon,\Sone)\) means that
\[
  \begin{tikzcd}
    \Dbimon
    \ar[r,"\Dbimonca"]
    \ar[d,swap,"\Wproj_0"]
    &
    \Oc\Dbimon
    \ar[d,"\Oc\Termm"]
    \\
    \Sone
    \ar[r,"\Seelyz"']
    &
    \Oc\Top
  \end{tikzcd}
\]
and the fact that \(\Dbimoncm\in\Emoc\cL(\Dbimon,\Dbimon\Times\Dbimon)\)
means that
\[
  \begin{tikzcd}
    \Dbimon
    \ar[rr,"\Dbimonca"]
    \ar[d,swap,"\Dbimoncm"]
    &&
    \Oc\Dbimon
    \ar[d,"\Oc{\Dbimoncm}"]
    \\
    \Dbimon\Times\Dbimon
    \ar[r,"\Dbimonca\Times\Dbimonca"']
    &
    \Oc{\Dbimon}\Times\Oc{\Dbimon}
    \ar[r,"\Ocmont"']
    &[-1em]
    \Oc{(\Dbimon\Times\Dbimon)}
  \end{tikzcd}
\]
The facts that \(\Dbimonu\in\Emoc\cL(\Sone,\Dbimon)\) and 
$\Sprojl_0 \in \Emoc\cL(\Sone, \Dbimon)$ mean that the following diagrams commute
\[
  \begin{tikzcd}
    \Sone
    \ar[r,"\Ocmonz"]
    \ar[d,swap,"\Win_0"]
    &
    \Oc\Sone
    \ar[d,"\Oc{\Win_0}"]
    \\
    \Dbimon
    \ar[r,"\Dbimonca"']
    &
    \Oc\Dbimon
  \end{tikzcd} \Diagsep
  \begin{tikzcd}
    \Sone
    \ar[r,"\Ocmonz"]
    \ar[d,swap,"\Wdiag"]
    &
    \Oc\Sone
    \ar[d,"\Oc\Wdiag"]
    \\
    \Dbimon
    \ar[r,"\Dbimonca"']
    &
    \Oc\Dbimon
  \end{tikzcd}
\]
and the fact that
\(\Dbimonm\in\Emoc\cL(\Dbimon\Times\Dbimon,\Dbimon)\) means
\[
  \begin{tikzcd}
    \Dbimon
    \ar[rr,"\Dbimonca"]
    &&
    \Oc\Dbimon
    \\
    \Dbimon\Times\Dbimon
    \ar[u,swap,swap,"\Dbimonm"]
    \ar[r,"\Dbimonca\Times\Dbimonca"']
    &
    \Oc{\Dbimon}\Times\Oc{\Dbimon}
    \ar[r,"\Ocmont"']
    &[-1em]
    \Oc{(\Dbimon\Times\Dbimon)}
    \ar[u,swap,"\Oc{\Dbimonm}"]
  \end{tikzcd}
\]

\begin{theorem}\label{th:dbimon-comon-contr}
  We have
  \(\Dbimoncu=\Coalgw_\Dbimon\in\Emoc\cL(\Dbimon,\Sone)\) and
  \(\Dbimoncm=\Coalgc_\Dbimon\in\Emoc\cL(\Dbimon,\Dbimon\Times\Dbimon)\).
\end{theorem}
\begin{proof}
  It is well known that in any cartesian category, 
  the comonoid structure of any object
  is unique and is given by the unique map to the terminal object and the 
  diagonal maps. So the comonoid 
  \((P,\Coalgw_P,\Coalgc_P)\) is necessarily equal to the comonoid 
  \((P,\Dbimoncu,\Dbimoncm)\).
\end{proof}

\begin{theorem} \label{thm:coalgebra-to-Sdlmate}
  Any analytic coalgebra $\Dbimonca$
  induces a left-sided analytic Taylor expansion, defined as 
  \[ \Sdlmate = 
  \begin{tikzcd}
    {!X \tensor \Dbimon} & {!X \tensor !\Dbimon} & {!(X \tensor \Dbimon) .}
    \arrow["{!X \tensor \Dbimonca}", from=1-1, to=1-2]
    \arrow["\Ocmont", from=1-2, to=1-3]
  \end{tikzcd} \]
\end{theorem}

\begin{proof} 
We show that $\Sdlmate$ satisfies \ref{ax:Sdlmate-Taylor} with the diagram chase below.
The only crucial argument involved is the fact that \(\Dbimonu\in\Emoc\cL(\Sone,\Dbimon)\).
\[
\begin{tikzcd}
	{!X } & {!(X \tensor 1)} \\
	{!X \tensor 1} & {!X \tensor !1} \\
	{!X \tensor \Dbimon} & {!X \tensor !\Dbimon} & {!(X \tensor \Dbimon)}
	\arrow["{!X \tensor 1}"', from=1-1, to=2-1]
	\arrow["{!X \tensor \Ocmonz}"', from=2-1, to=2-2]
	\arrow["\Ocmont"', from=2-2, to=1-2]
	\arrow["{!\tensorUnitR^{-1}}", from=1-1, to=1-2]
	\arrow["{!X \tensor \Dbimonu}"', from=2-1, to=3-1]
	\arrow["{!X \tensor \Sdlmate}"', from=3-1, to=3-2]
	\arrow["{!X \tensor !\Dbimonu}", from=2-2, to=3-2]
	\arrow["{!(X \tensor \Dbimonu)}", from=1-2, to=3-3]
	\arrow["\Ocmont"', from=3-2, to=3-3]
\end{tikzcd} \]
We show that $\Sdlmate$ satisfies
  \ref{ax:Sdlmate-lin} with the diagram chase below
  (any use of the associator $\alpha$ of the symmetric monoidal 
  structure is kept implicit). 
  The only crucial argument involved is the fact that 
  \(\Dbimonm\in\Emoc\cL(\Dbimon\Times\Dbimon,\Dbimon)\). 
\[ 
\begin{tikzcd}
	{\oc X \tensor \Dbimon \tensor \Dbimon} & {!X \tensor ! \Dbimon \tensor \Dbimon} & {!(X \tensor \Dbimon) \tensor \Dbimon} & {!(X \tensor \Dbimon) \tensor !\Dbimon} & {!(X \tensor \Dbimon \tensor \Dbimon)} \\
	&& {!X \tensor ! \Dbimon \tensor !\Dbimon} & {!X \tensor ! (\Dbimon \tensor \Dbimon)} \\
	{!X \tensor \Dbimon} &&& {!X \tensor !\Dbimon} & {!(X \tensor \Dbimon)}
	\arrow["{!X \tensor \Sdlmate \tensor \Dbimon}", from=1-1, to=1-2]
	\arrow["{\Ocmont \tensor \Dbimon}", from=1-2, to=1-3]
	\arrow["{!(X \tensor \Dbimon) \tensor \Sdlmate}", from=1-3, to=1-4]
	\arrow["\Ocmont", from=1-4, to=1-5]
	\arrow["{!X \tensor \Dbimon \tensor \Sdlmate}"{description}, from=1-2, to=2-3]
	\arrow["{\Ocmont \tensor !\Dbimon}"{description}, from=2-3, to=1-4]
	\arrow["{!X \tensor \Ocmont}"', from=2-3, to=2-4]
	\arrow["\Ocmont"{description}, from=2-4, to=1-5]
	\arrow["{!X \tensor \Sdlmate \tensor \Sdlmate}"', from=1-1, to=2-3]
	\arrow["{!X \tensor \Dbimonm}"', from=1-1, to=3-1]
	\arrow["{!X \tensor \Sdlmate}"', from=3-1, to=3-4]
	\arrow["\Ocmont"', from=3-4, to=3-5]
	\arrow["{!(X \tensor \Dbimonm)}", from=1-5, to=3-5]
	\arrow["{!X \tensor !\Dbimonm}", from=2-4, to=3-4]
	\arrow[draw=none, from=2-3, to=3-1]
\end{tikzcd}  \]

The other computations are
  similar and can be found in Theorem 19 of~\cite{Ehrhard23-cohdiff} (they do not involve any 
  argument on the summability structure, so they directly carry to our setting). 
  Note that the proofs of \ref{ax:Sdlmate-Schwarz} and \ref{ax:Sdlmate-with} 
  do not rely on any
  of the assumptions on $\Dbimonca$.
\end{proof}

\begin{theorem} \label{thm:Sdlmate-to-coalgebra}
  Any left-sided analytic Taylor expansion
  induces an analytic coalgebra given by 
  \[ \Dbimonca = 
\begin{tikzcd}
	\Dbimon & {1 \tensor \Dbimon} & {!1 \tensor \Dbimon} & {!(1 \tensor \Dbimon)} & {!\Dbimon .}
	\arrow["{\tensorUnitL^{-1}}", from=1-1, to=1-2]
	\arrow["{\Ocmonz \tensor \Dbimon}", from=1-2, to=1-3]
	\arrow["{\Sdlmate_1}", from=1-3, to=1-4]
	\arrow["{! \tensorUnitL}", from=1-4, to=1-5]
\end{tikzcd} \]
Furthermore, $\Sdlmate =  \begin{tikzcd}
  {!X \tensor \Dbimon} & {!X \tensor !\Dbimon} & {!(X \tensor \Dbimon)}
  \arrow["{!X \tensor \Dbimonca}", from=1-1, to=1-2]
  \arrow["\Ocmont", from=1-2, to=1-3]
\end{tikzcd}$.
\end{theorem} 

\begin{proof} The fact that 
  $\Sdlmate = \Ocmont \compl (\oc X \tensor \Sdlmate)$ can be found in 
  Theorem 17 of~\cite{Ehrhard23-cohdiff}. Again, no argument on the summability
  structure are used so the proof carry directly to our setting.
  
  We prove that \(\Dbimonu\in\Emoc\cL(\Sone,\Dbimon)\)
  with the following diagram chase. The only crucial argument involved is 
  \ref{ax:Sdlmate-Taylor}.
\[\begin{tikzcd}
	1 &&&& {!1} \\
	& 1 && {!1} \\
	& {1 \tensor 1} & {!1 \tensor 1} & {!(1 \tensor 1)} \\
	\Dbimon & {1 \tensor \Dbimon} & {!1 \tensor \Dbimon} & {!(1 \tensor \Dbimon)} & {!\Dbimon}
	\arrow["\Dbimonu"', from=1-1, to=4-1]
	\arrow["{\tensorUnitL^{-1}}"', from=4-1, to=4-2]
	\arrow["{\Ocmonz \tensor \Dbimon}"', from=4-2, to=4-3]
	\arrow["{\Sdlmate_1}"', from=4-3, to=4-4]
	\arrow["{!\tensorUnitL^{-1}}"', from=4-4, to=4-5]
	\arrow["{! \Dbimonu}", from=1-5, to=4-5]
	\arrow["\Ocmonz", from=1-1, to=1-5]
	\arrow["{\tensorUnitL^{-1}}"', from=1-1, to=3-2]
	\arrow["{1 \tensor \Dbimonu}"{description}, from=3-2, to=4-2]
	\arrow["{\Ocmonz \tensor 1}", from=3-2, to=3-3]
	\arrow["{!1 \tensor \Dbimonu}"{description}, from=3-3, to=4-3]
	\arrow["\tensorUnitR", from=3-3, to=2-4]
	\arrow["{!\tensorUnitR^{-1}}"', from=2-4, to=3-4]
	\arrow["{!(1 \tensor \Dbimonu)}", from=3-4, to=4-4]
	\arrow["\tensorUnitR"', from=3-2, to=2-2]
	\arrow["\Ocmonz", from=2-2, to=2-4]
	\arrow[Rightarrow, no head, from=1-1, to=2-2]
	\arrow[Rightarrow, no head, from=1-5, to=2-4]
	\arrow["{!\tensorUnitL}"', from=3-4, to=1-5]
	\arrow["{\text{\ref{ax:Sdlmate-Taylor}}}"{description}, draw=none, from=4-3, to=3-4]
\end{tikzcd}\] 

We prove that \(\Dbimonm\in\Emoc\cL(\Dbimon\Times\Dbimon,\Dbimon)\) with the 
following diagram chase. The 
proof involves \ref{ax:Sdlmate-lin} and the fact that 
$\Sdlmate = \Ocmont \compl (\oc X \tensor \Sdlmate)$ (commutation $(a)$).
\[
\begin{tikzcd}
	{\Dbimon \tensor \Dbimon} & {!\Dbimon \tensor !\Dbimon} &&& {!\Dbimon \tensor !\Dbimon} & {!(\Dbimon \tensor \Dbimon)} \\
	& {1 \tensor !\Dbimon \tensor !\Dbimon} & {!1 \tensor !\Dbimon \tensor !\Dbimon} && {!(1 \tensor \Dbimon) \tensor !\Dbimon} \\
	&& {!1 \tensor !\Dbimon \tensor \Dbimon} & {!(1 \tensor \Dbimon) \tensor \Dbimon} & {!(1 \tensor \Dbimon) \tensor !\Dbimon} \\
	& {1 \tensor \Dbimon \tensor \Dbimon} & {!1 \tensor \Dbimon \tensor \Dbimon} & {!(1 \tensor \Dbimon) \tensor \Dbimon} & {!(1 \tensor \Dbimon \tensor \Dbimon)} \\
	\Dbimon & {1 \tensor \Dbimon} & {!1 \tensor \Dbimon} && {!(1 \tensor \Dbimon)} & {!\Dbimon}
	\arrow["\Ocmont", from=1-5, to=1-6]
	\arrow["\Dbimonm"', from=1-1, to=5-1]
	\arrow["{\tensorUnitL^{-1}}"', from=5-1, to=5-2]
	\arrow["{\Ocmonz \tensor \Dbimon}"', from=5-2, to=5-3]
	\arrow["{\Sdlmate_1}"', from=5-3, to=5-5]
	\arrow["{!\tensorUnitL}"', from=5-5, to=5-6]
	\arrow["{! \Dbimonm}", from=1-6, to=5-6]
	\arrow["{!1 \tensor \Dbimonm}", from=4-3, to=5-3]
	\arrow["{\Sdlmate_1 \tensor \Dbimon}"', from=4-3, to=4-4]
	\arrow["{\Sdlmate_{1 \tensor \Dbimon}}"', from=4-4, to=4-5]
	\arrow["{!(1 \tensor \Dbimonm)}"', from=4-5, to=5-5]
	\arrow["{\text{\ref{ax:Sdlmate-lin}}}"{description}, draw=none, from=5-3, to=4-5]
	\arrow["{!1 \tensor \Dbimonca \tensor \Dbimon}", from=4-3, to=3-3]
	\arrow["\Ocmont", from=3-5, to=4-5]
	\arrow["{!1 \tensor !\Dbimon \tensor \Dbimonca}", from=3-3, to=2-3]
	\arrow["{\Ocmonz \tensor \Dbimon \tensor \Dbimon}", from=4-2, to=4-3]
	\arrow["{1 \tensor \Dbimonm}", from=4-2, to=5-2]
	\arrow["{\tensorUnitL^{-1}}"', from=1-1, to=4-2]
	\arrow["{1 \tensor \Dbimonca \tensor \Dbimonca}"', from=4-2, to=2-2]
	\arrow["{\Ocmonz \tensor !\Dbimon \tensor !\Dbimon}", from=2-2, to=2-3]
	\arrow["{\Dbimonca \tensor \Dbimonca}", from=1-1, to=1-2]
	\arrow["{\tensorUnitL^{-1}}"', from=1-2, to=2-2]
	\arrow[Rightarrow, no head, from=1-2, to=1-5]
	\arrow["{\Ocmont \tensor \Dbimon}", from=3-3, to=3-4]
	\arrow["{!(1 \tensor \Dbimon) \tensor \Dbimonca}", from=3-4, to=3-5]
	\arrow[Rightarrow, no head, from=3-4, to=4-4]
	\arrow["{(a)}"{description}, draw=none, from=4-3, to=3-4]
	\arrow["{(a)}"{description}, draw=none, from=4-4, to=3-5]
	\arrow["{\Ocmont \tensor ! \Dbimon}", from=2-3, to=2-5]
	\arrow["{!\tensorUnitL^{-1} \tensor !\Dbimon}"', from=1-5, to=2-5]
	\arrow[Rightarrow, no head, from=2-5, to=3-5]
	\arrow["{! \tensorUnitL^{-1}}", from=1-6, to=4-5]
\end{tikzcd}\]
The other computations rely on similar arguments and can be found
in \cite{Ehrhard23-cohdiff}.
\end{proof}

\begin{corollary} \label{cor:Taylor-structure-induces-differential}
  The constructions of \cref{thm:coalgebra-to-Sdlmate} and \cref{thm:Sdlmate-to-coalgebra}
  are inverse of each other. Thus, for any representably $\Sigma$-additive category, there
  is a bijective correspondence between analytic Taylor expansions 
  and analytic coalgebras. 
\end{corollary}

\begin{proof} We already know from \cref{thm:Sdlmate-to-coalgebra}
  that if $\Dbimonca \defEq !\tensorUnitL \compl \Sdlmate_1 \compl 
  (\Ocmonz \tensor \Dbimon) \compl \tensorUnitL^{-1}$ then 
  $\Sdlmate = \Ocmont \compl (!X \tensor \Dbimonca)$. 
  Conversely, assume that 
  $\Sdlmate \defEq \Ocmont \compl (!X \tensor \Dbimonca)$. Then we can check by 
  a straightforward computation that 
  $\Dbimonca = !\tensorUnitL \compl \Sdlmate_1 \compl 
  (\Ocmonz \tensor \Dbimon) \compl \tensorUnitL^{-1}$. So 
  the constructions of \cref{thm:coalgebra-to-Sdlmate} and \cref{thm:Sdlmate-to-coalgebra}
  are inverse of each other. Then, we know from \cref{thm:Sdl-mate} that 
  the existence of a left-sided analytic Taylor expansion 
  $\Sdlmate$ is equivalent to the existence of an analytic Taylor 
  expansion $\Sdl$.
\end{proof}

\begin{definition} A representably analytic category is a representably
  $\Sigma$-additive category equipped with an analytic coalgebra (and equivalently
  an analytic Taylor expansion).
\end{definition}

\subsection{A remarkable isomorphism}
In this section we assume that \(\cL\) is a representably
$\Sigma$-additive category equipped with an analytic coalgebra.
We prove that under the assumption that the resource comonad is finitary
(see \cref{def:resource-finitary}), the coalgebra $(\Dbimon, \Dbimonca)$ 
is isomorphic to the free coalgebra $(\Oc 1, \dig)$.

Given a $\Oc$-coalgebra \(P\), let %
\(\Coalgcgen n\in\Emoc\cL(P,\Tensexp Pn)\) be the \(n\)-ary version of
the comultiplication of the comonoid \(P\), so that %
\(\Coalgcgen0=\Coalgw\), \(\Coalgcgen1=\Id\) and \(\Coalgcgen2=\Coalgc\).
We keep any use of the monoidal associativity $\tensorAssoc$ implicit.

In the case where \(P\) is the coalgebra \((\Dbimon,\Dbimonca)\), we
know by \Cref{th:dbimon-comon-contr} that
\(\Coalgc_\Dbimon=\Dbimoncm\) and hence we have
\begin{align*}
  \Coalgcgen n_\Dbimon\Compl\Win_k\Compl\Soneiso n
  =\sum_{\Biind{\Vect i\in\Nat^k}{i_1+\cdots+i_k=n}}
  \Win_{i_1}\Times\cdots\Times\Win_{i_k}
\end{align*}
where \(\Soneiso n\in\cL(\Tensexp\Sone n,\Sone)\) is the unique
canonical isomorphism induced by the symmetric monoidal structure of \(\cL\).

\begin{lemma}\label{lemma:dbimon-comon-projone}
  For all \(n\in\Nat\) we have
  \[
    \begin{tikzcd}
      \Dbimon
      \ar[r,"\Coalgcgen n_\Dbimon"]
      \ar[drr,swap,"\Wproj_n"]
      &
      \Tensexp\Dbimon n
      \ar[r,"\Tensexp{\Wproj_1}n"]
      &
      \Tensexp\Sone n
      \ar[d,"\Soneiso n"]
      \\
      &&
      \Sone
    \end{tikzcd}
  \]
\end{lemma}
\begin{proof}
  Given \(k\in\Nat\), we have
  \begin{align*}
    \Soneiso n
    \Compl\Tensexp{\Wproj_1}n
    \Compl\Coalgcgen n_\Dbimon
    \Compl\Win_k
    &=\Soneiso n
      \Compl\Tensexp{\Wproj_1}n
      (\sum_{\Biind{\Vect i\in\Nat^n}{i_1+\cdots+i_n=k}}
      \Win_{i_1}\Times\cdots\Times\Win_{i_n})
    \Compl\Invp{\Soneiso n}\\
    &=\Kronecker{k}{n}\Id_{\Sone}\\
    &=\Wproj_n\Compl\Win_k
  \end{align*}
  so that the diagram commutes by joint epicity of the \(\Win_k\)'s.
\end{proof}

\begin{definition}
  For all object \(X\), we define a natural morphism
  \(\Dergen n_X\in\cL(\Oc X,\Tensexp Xn)\) as
  \[
    \begin{tikzcd}
      \Oc X
      \ar[r,"\Coalgcgen n_{\Coalgfree X}"]
      &
      \Tensexp{(\Oc X)}n
      \ar[r,"\Tensexp{\Deru_X}n"]
      &
      \Tensexp Xn
    \end{tikzcd}
  \]
  and call it \emph{\(n\)-ary generalized dereliction}. This terminology
  is folklore in LL.
\end{definition}

For each \(n\in\Nat\) we define
\(\Dbimondegi n=\Soneiso n\Compl\Dergen
n_{\Sone}\in\cL(\Oc\Sone,\Sone)\), and then we define
\(\Dbimondeg=\Tuple{\Dbimondegi
  n}_{n\in\Nat}\in\cL(\Oc\Sone,\Dbimon)\).
In other words, thanks to the compatibility 
between the sum and the cartesian product, we
have %
\begin{align}\label{eq:dbimondeg-as-sum}
  \Dbimondeg=
  \sum_{k=0}^\infty
  \Win_k\Compl\Soneiso k
  \Compl\Tensexp{\Deru_{\Sone}}n
  \Compl\Coalgcgen n_{\Coalgfree\Sone}\,.
\end{align}

Conversely, we define
\(\Dbimondeginv\in\cL(\Dbimon,\Oc\Sone)\) as
\[
  \begin{tikzcd}
    \Dbimon
    \ar[r,"\Dbimonca"]
    &
    \Oc\Dbimon
    \ar[r,"\Oc{\Wproj_1}"]
    &
    \Oc\Sone .
  \end{tikzcd}
\]
\begin{lemma}\label{lemma:dbimondeg-iso-1}
  \(\Dbimondeg\Compl\Dbimondeginv=\Id_\Dbimon\).
\end{lemma}
\begin{proof}
  Given \(n\in\Nat\), we have
  \begin{align*}
    \Dbimondegi n\Compl\Dbimondeginv
    &=\Soneiso n
      \Compl\Tensexp{\Deru_{\Sone}}n
      \Compl\Coalgcgen n_{\Coalgfree{\Sone}}
      \Compl\Oc{\Wproj_{\Sone}}
      \Compl\Dbimonca\\
    &=\Soneiso n
      \Compl\Tensexp{\Deru_{\Sone}}n
      \Compl\Tensexp{\Oc{\Wproj_{\Sone}}}n
      \Compl\Coalgcgen n_{\Coalgfree{\Dbimon}}
      \Compl\Dbimonca\\
    &=\Soneiso n
      \Compl\Tensexp{\Wproj_{\Sone}}n
      \Compl\Tensexp{\Deru_{\Dbimon}}n
      \Compl\Coalgcgen n_{\Coalgfree{\Dbimon}}
      \Compl\Dbimonca\\
    &=\Soneiso n
      \Compl\Tensexp{\Wproj_{\Sone}}n
      \Compl\Coalgcgen n_\Dbimon\\
    &=\Wproj_n
  \end{align*}
  by \Cref{lemma:dbimon-comon-projone} and hence %
  \(\Dbimondeg\Compl\Dbimondeginv
  =\Tuple{\Dbimondegi n}_{n\in\Nat}\Compl\Dbimondeginv
  =\Tuple{\Dbimondegi n\Compl\Dbimondeginv}_{n\in\Nat}
  =\Tuple{\Wproj_n}_{n\in\Nat}
  =\Id_\Dbimon\).
\end{proof}

\begin{definition} %
  \label{def:resource-finitary}
  The resource category \(\cL\) is \emph{finitary} if the generalized
  derelictions \((\Dergen n_{\Sone})_{n\in\Nat}\) are jointly monic,
  that is \(\Dbimondeg\) is monic.
\end{definition}
The intuition is that the ``points'' of \(\Oc\Sone\) contain only a
finite amount of information.
This is typically the case when the definition of the exponential is
based on multisets.

\begin{theorem} %
  \label{th:dbimondeg-iso}
  If the representably analytic 
  category \(\cL\) is finitary, then
  \(\Dbimondeginv\Compl\Dbimondeg=\Id_{\Sone}\).
  Therefore, the $\Oc$-coalgebras \(\Dbimon\) and \(\Oc\Sone\) are isomorphic
  in \(\Emoc\cL\).
\end{theorem}
\begin{proof}
  We have \(\Dbimondeg\Compl\Dbimondeginv\Compl\Dbimondeg=\Dbimondeg\)
  by \Cref{lemma:dbimondeg-iso-1} and hence
  \(\Dbimondeginv\Compl\Dbimondeg=\Id\) by monicity of \(\Dbimondeg\).
  The second statement results from the fact that clearly
  \(\Dbimondeginv\in\Emoc\cL(\Dbimon,\Oc\Sone)\).
\end{proof}

\begin{remark}
  \label{rk:dbimon-ocone-iso-interp}
  So when \(\cL\) is a finitary representably analytic category
  which is closed, the objects \(\Oc\Sone\Limpl X\) and \(\Sfun X\)
  are isomorphic. This means that a morphism \(f\) from \(\Sone\) to
  \(X\) in \(\Kloc\cL\) is the same thing as a summable family
  \(\Vect x\) of elements of \(X\): \(f\) can be considered as a power
  series on the object \(\Sone\) of scalars, whose coefficients are
  the \(x_i\)'s, that is \(f(t)=\sum_{n=0}^\infty t^nx_n\).
  This power series is well known in traditional analysis and is 
  called the generating function of the sequence $\Vect x$.
  It validates the intuition developed throughout \cref{sec:summability-structure}
  that the elements of $\S X$ are power series.
  This observation is particularly remarkable in $\PCOH$ (see \cref{sec:pcoh}), in which 
  $\S 1$ is the set of (sub) probability distributions on $\N$. 
  The isomorphism 
  between $\S 1$ and $!1 \linarrow 1$ means that such (sub) probability distributions
  on $\N$
  can be seen as analytic maps from $[0,1]$ to $[0,1]$ with positive coefficients.
  This analytic map is exactly the \emph{probability generating function} found in 
  probability theory.

  More generally, if \(f\in\Kloc\cL(X,Y)=\cL(\Oc X,Y)\), we can define
  \(h\in\cL(\Tens{\Oc X}{\Oc\Sone},Y)\) as the following composition
  of morphism
  \[
    \begin{tikzcd}
      \Tens{\Oc X}{\Oc\Sone}
      \ar[r,"\Ocmont"]
      &[-1em]
      \Oc{\Tensp X\Sone}
      \ar[r,"\Oc\Rightu"]
      &[-1em]
      \Oc X
      \ar[r,"f"]
      &[-1em]
      Y
    \end{tikzcd}
  \]
  which can be seen as a two parameter analytic function which, by the
  isomorphism between \(\Sone\) and \(\Dbimon\), can be considered as a
  summable family \((h_n\in\Kloc\cL(X,Y))_{n\in\Nat}\).
  This means intuitively that we can write %
  \(f(tx)=\sum_{n=0}^\infty t^nh_n(x)\), that is, \(h_n\) is the
  \(n\)-homogeneous component of \(f\) which can be considered as a
  ``degree $n$ polynomial'' morphism \(X\to Y\).
  This morphism can also be obtained as %
  \(\Sproj_n\Compl\Tayfun(f)\Compl\Oc\Sinj_1\), using the Taylor
  functor.
\end{remark}

\subsection{The case of Lafont resource categories}
\label{sec:Lafont-categories}

If \(\cL\) is a symmetric monoidal category, one defines the category \(\Cm\cL\) of
commutative \(\Times\)-comonoids, whose objects are the commutative comonoid 
and whose morphisms are the comonoid morphisms (see~\cref{sec:bimonoid}). Identity 
and composition are defined as in $\cL$.
There is an obvious forgetful functor %
\(\Cmforg:\Cm\cL\to\cL\) which maps \(C\) to \(\Comonca C\) and acts
as the identity on morphisms.

\begin{definition} (\cite{Lafont88})
  The symmetric monoidal category \(\cL\) is a \emph{Lafont category} if the functor
  \(\Cmforg\) has a right adjoint.
\end{definition}

\begin{theorem} 
  Any cartesian Lafont symmetric monoidal category \(\cL\) has a canonical structure of
  resource category.
\end{theorem}
This is a standard result, see for instance~\cite{Mellies09}.
It means that \(\cL\) is endowed with a particular
resource structure %
\((\oc,\Deru,\Digg,\Seelyz,\Seelyt)\) that we describe now.
The resource modality which arises in that way is often called the
\emph{free exponential} of \(\cL\).

Let \(\Comonfree:\cL\to\Cm\cL\) be the right adjoint of
\(\Cmforg\).
By this adjunction, the functor \(\Cmforg\Comp\Comonfree:\cL\to\cL\)
has a structure of comonad: this is our resource comonad
\((\Oc,\Deru,\Digg)\).
For the Seely isomorphisms, we must first notice the following
property, see for example Corollary 18 of \cite{Mellies09}.
\begin{theorem}
  The category \(\Cm\cL\) is cartesian, with terminal object %
  \((\Sone,\Id_{\Sone},\Leftu_{\Sone}=\Rightu_{\Sone})\) and cartesian
  product of \(C_1\) and \(C_2\) the triple %
  \((\Comonca{C_1}\Times\Comonca{C_2},\Comonu,\Comonm)\) where the unit
  and the multiplication are given by
  \[
    \begin{tikzcd}
      \Tens{\Coalgca{C_1}}{\Coalgca{C_2}}
      \ar[r,"\Tens{\Comonu_{C_1}}{\Comonu_{C_2}}"]
      &[1.4em]
      \Tens{\Sone}{\Sone}
      \ar[r,"\Tens{\Leftu}{\Leftu}"]
      &
      \Sone
    \end{tikzcd}
  \]
  \[
    \begin{tikzcd}
      \Tens{\Tensp{\Coalgca{C_1}}{\Coalgca{C_2}}}
      {\Tensp{\Coalgca{C_1}}{\Coalgca{C_2}}}
      \ar[r,"\Sym_{2,3}"]
      &
      \Tens{\Tensp{\Coalgca{C_1}}{\Coalgca{C_1}}}
      {\Tensp{\Coalgca{C_2}}{\Coalgca{C_2}}}
      \ar[r,"\Tens{\Comonm_{C_1}}{\Comonm_{C_2}}"]
      &[1em]
      \Tens{\Coalgca{C_1}}{\Coalgca{C_2}}
    \end{tikzcd}
  \]
  where \(\Sym_{2,3}\) is defined using the coherence isos of the SM
  structure of \(\cL\).
  The first projection is
  \[
    \begin{tikzcd}
      \Tens{\Coalgca{C_1}}{\Coalgca{C_2}}
      \ar[r,"\Tens{\Comonu_{C_1}}{\Coalgca{C_2}}"]
      &[1em]
      \Tens{\Sone}{\Coalgca{C_2}}
      \ar[r,"\Leftu"]
      &[-1em]
      \Coalgca{C_2}
    \end{tikzcd}
  \]
  and similarly for the second one.
\end{theorem}

As a right adjoint, the functor \(\Comonfree\) preserves cartesian
products, and the Seely isomorphisms embody this preservation.

\begin{theorem}\label{th:Lafont-elem-summable-Taylor}
  Any Lafont representably \(\Sigma\)-additive category has an analytic
  coalgebra, so it is a representably analytic category.
\end{theorem}
\begin{proof}
  This is an immediate consequence of the fact that %
  \((\Dbimon,\Dbimoncu,\Dbimoncm)\) is a commutative comonoid and of
  the fact that %
  \(\Dbimoncu\), \(\Dbimoncm\), \(\Dbimonu\), \(\Dbimonm\) and
  \(\Win_0\) are comonoid morphisms.
\end{proof}

\begin{remark}
  There are Lafont categories which are not 
  $\Sigma$-additive, such as the category of Köthe spaces~(\cite{Ehrhard02}) or 
  the category of finiteness spaces~(\cite{Ehrhard05}) 
  (taking a semiring that is not positive, such as $\R$).
  As mentioned in \cref{rem:summability-general}, it might be possible
  to use the more general notion of partial commutative monoid in
  order to capture those models.
\end{remark}

\section{Models of Taylor expansion as representably analytic categories}
\label{sec:examples-taylor}

We present first fundamental models of LL which are already well known
models of Taylor expansion: the relational model and the weighted
relational model. We explain how the known Taylor expansion in those
models actually arises from a representably analytic structure.
This motivates that our theory of Taylor expansion is a generalization of the 
existing one.

\subsection{Some notation for multisets and semi-rings}
\label{sec:notation-multiset}
If \(A\) is a set, a (finite) multiset of elements of \(A\) is a
function \(m:A\to\Nat\) whose support %
\(\Msetsupp m=\{a\in A\St m(a)\not=0\}\) is a finite set.
Intuitively, \(m(a)\) is the number of occurrences of \(a\) in \(m\),
and we write \(a\in m\) if \(a\in\Msetsupp m\).
We use \(\Msetempty\) for the empty multiset such that
\(\Msetsupp\Msetempty=\emptyset\).
We use \(\Mfin A\) for the set of all finite multisets of all elements
of \(A\), that we consider as a commutative monoid (actually it is the
free commutative monoid generated by \(A\)), whose operation
is denoted additively: if $m_1, \ldots, m_n\in\Mfin A$, then we write
$m_1 + \cdots + m_n$ for their pointwise sum.

If \(\Vect a=(\List a1k)\in A^k\), we use %
\(\Mset{\Vect a}=\Mset{\List a1k}\) for the element \(m\) of
\(\Mset A\) which contains the elements of \(\Vect a\), taking
multiplicities into account, that is \(m(a)\) is the number of
\(i\in\{1,\dots,k\}\) such that \(a_i=a\).
For any multiset $m = \Mset{\List a1k} \in \Mfin A$ and $i \in I$,
we write $i \cdot m \in \Mfin {I \times A}$ the multiset 
$\Mset{(i, a_1), \ldots, (i, a_n)}$.

The size, or cardinality, of a multiset \(m\) is
\(\Mscard m=\sum_{a\in A}m(a)\in\Nat\).
We use \(\Mfinnz A\) for the set of all \(m\in\Mfin A\) such that
\(\Mscard m>0\).
We set \(\Factor m=\prod_{a\in A}\Factor{m(a)}\in\Nat\) and call this
number the factorial of \(m\).
For any \(\Vect m\in\Mfin A^n\), the quotient
\[ \Multinom{\Vect m}
=\frac{\Factor{(m_1+\cdots+m_n)}}{\Factor{m_1}\cdots\Factor{m_n}}\] is
an integer and is called the \emph{multinomial coefficient} of
\(\Vect m\). 

For any vector \(u\in \Rcal^A\) indexed by $A$, where \(\Rcal\) is a semiring 
(with standard algebraic notations) and \(m\in\Mfin A\) we set
\begin{equation} \label{eq:multiset-power}
  u^m=\prod_{a\in A}u_a^{m(a)}
\end{equation}
Any integer $k \in \N$ can be seen as the element 
$\sum_{i=1}^k 1 \in \Rcal$.
Then, for any \(u(1),\dots,u(k)\in \Rcal^A\), the usual multinomial formula generalizes to
\begin{equation} \label{eq:multinom-formula-multiset}
  \Big(\sum_{i=1}^k u(i)\Big)^m
  =\sum_{\Biind{\Vect m\in\Mfin A^k}{m_1+\cdots+m_k=m}}
  \Multinom{\Vect m}\prod_{i=1}^ku(i)^{m_i}\,.
\end{equation}

\begin{remark}
  Given a family of vectors \((u_i\in \Rcal^A)_{i \in I}\)
  we use the notations
  \(u_i\) or \(u(i)\) for the corresponding element of \(\Rcal^A\),
  depending on the context, so that for any $a \in A$ we can write 
  $u(i)_a$ instead of the ugly \({u_i}_a\).
\end{remark}

\subsection{Arbitrary sums are representable}

\label{sec:arbitrary-sum}

The usual axiomatization of categorical models of 
the syntactical Taylor expansion (\cite{Manzonetto12}) require 
the category to feature arbitrary countable sums.  
We show in this section that such infinite sum is always representable, 
assuming that the category satisfies \ref{ax:D-defined}.

It is well-known that any category enriched over commutative 
monoids that feature finite products (or finite coproducts) has
finite biproducts, see chapter 8 of \cite{Maclane71} for 
example\footnote{This reference assumes that the homset is 
an abelian group, but the proof carries directly to 
commutative monoids}.
Such categories are usually called semi-additive categories, but 
the naming conventions of differential categories call them
\emph{additive}.
We show a similar result on countable biproducts and countable sums.

\begin{definition}
  A (symmetric monoidal) $\Sigma$-additive category $\categoryLL$ has 
  \emph{total sums} if for all pairs of objects $X, Y$, 
  the sum of the $\Sigma$-monoid structure of $\categoryLL(X, Y)$ is a total operation. 
\end{definition}

Observe that if $\categoryLL$ is a $\Sigma$-additive category 
with total sum, then 
the sum is compatible with all cartesian products that exist in $\categoryLL$, 
by \cref{prop:sum-cartesian}.

\begin{proposition} \label{prop:countable-biproducts}
  Any countable cartesian 
product in a $\Sigma$-additive category with total sums is a countable biproduct.
\end{proposition}

\begin{proof} The proof mirrors the one on additive 
  categories, see chapter 8 of \cite{Maclane71}.
  Let $\Sinj_i \in \categoryLL(X_i, \withFam X_i)$ be the morphism 
  defined by 
  \[ \prodProj_j \compl \Sinj_i = \begin{cases}
    \id \text{ if } i = j \\
    0 \text{ otherwise.}
  \end{cases} \] 
  Then using the compatibility between the sums and the cartesian product, 
  we can easily show that 
  \[ \sum_{i \in I} \Sinj_i \compl \prodProj_i = \id \, .\] 
  Using this equation, it is easy to check that $\withFam X_i$ is also 
  a coproduct, with the injections maps given by the $\Sinj_i$ and 
  copairing of $\family<I>{f_i \in \category(X_i, Y)}$ defined as 
  $\coprodPairing{f_i} = \sum f_i \comp \prodProj_i$.
\end{proof}

\begin{theorem} \label{thm:omega-additive-representable}
  Any symmetric monoidal $\Sigma$-additive category 
  with total sums which satisfies \ref{ax:D-defined}
  is a representably $\Sigma$-additive category. 
  The $\Sigma$-summability structure $\S X = \Dbimon \linarrow X$ is a cartesian 
  product:
  $\S X = \withFam<\N> X$ and $\Sproj_i = \prodProj_i$.
\end{theorem}

\begin{proof} The category $\categoryLL$ satisfies \ref{ax:RS-mon} and 
  \ref{ax:RS-witness} by assumption and
  \ref{ax:RS-epi} by \cref{prop:countable-biproducts} above. 
  Furthermore, for any family $\vect f = \sequence{f_i \in \categoryLL(X, Y)}$ let 
  $h = \sum_{i \in \N}{\tensorUnitR \compl (f_i \tensor \prodProj_i)} \in 
  \categoryLL(X \tensor \Dbimon, Y)$.
  Then $h \comp (X \tensor \Sprojl_i) = f_i$ and 
  $h \compl (X \tensor \Wdiag) = \sum_{i \in \N} f_i$ so 
  $\categoryLL$ satisfies \ref{ax:RS-sum} and as such is a representably
  $\Sigma$-additive category.
  It follows from the fact that the sum is total that the witness 
  of $\sequence{f_i}$ always exists, so $\Dbimon \linarrow X$ is a cartesian 
  product with projection $\Sproj_i$ and tupling 
  $\Spairing{f_i}$. 
\end{proof}

\begin{remark}
Alternatively, \cref{thm:omega-additive-representable} is a direct consequence 
of \cref{thm:representably-additive-characterization}. In fact, the proof above 
is a direct simplification of the proof of the reverse direction of 
\cref{thm:representably-additive-characterization} using the fact that 
every indexed family is summable.
\end{remark}

\subsection{The relational model}

The most fundamental model of LL is doubtlessly the relational model $\REL$. 
In this category, objects are sets and morphisms are relations, that is,
$\REL(X, Y) = \sets{X \times Y}$. Composition and identity are defined 
as in \cref{sec:wcs}. If $s \in \REL(E, F)$ and $t \in \REL(F, G)$
\[ \id_E = \{(a, a) | a \in E\} \quad 
t \compl s = \{(a,c)\in E\times G
\St\exists b\in E \text{ s.t. } (a,b)\in s\text{ and }(b,c)\in t\} . \]
The category $\REL$ is symmetric monoidal. We set $\Sone = \{\Sonelem\}$ and
$ E_1\Times E_2= E_1 \times E_2$.
Then for any $s_1 \in \REL(E_1, F_1)$ and $s_2 \in \REL(E_2, F_2)$, we set
\[ s_1 \tensor s_2 = 
\{((a_1,a_2),(b_1,b_2))\St (a_i,b_i)\in s_i\text{ for }i=1,2\} . \]
This is a bifunctor, and the symmetric monoidal structure is given 
by the following natural isomorphisms.
\begin{align*}
  \Rightu_E &=\{((a,\Sonelem),a)\St a\in
  E\}\in\REL(E\Times\Sone,E) \\ 
  \Leftu_E &=\{((\Sonelem,a),a)\St a\in E\}\in\REL(\Sone\Times
  E,E) \\
    \Assoc_{E_1,E_2,E_3} &=\{(((a_1,a_2),a_3),(a_1,(a_2,a_3))) \St a_i
  \in E_i\text{ for }i=1,2,3\} \numberthis \label{eq:smc-rel} \\ 
  & \in\REL((E_1\Times E_2)\Times E_3,E_1\Times(E_2\Times E_3)) \\
  \Sym_{E_1,E_2} &=\{((a_1,a_2),(a_2,a_1)) \St a_i\in E_i\text{ for
  }i=1,2\} \in \REL(E_1 \tensor E_2, E_2 \tensor E_1) .
  \end{align*} 
This symetric monoidal category is closed, with internal hom of \(E\) and \(F\) the pair %
\((E\Limpl F,\Evlin)\) where $E \linarrow F = E \times F$ and
\[ \ev =\{(((a,b),a),b)\St a\in E\text{ and }b\in F\} \, . \]
The transpose of \(s \in \REL(G\Times E,F)\) is %
\(\Curlin(s)=\{(c,(a,b))\St ((c,a),b)\in s\}\in\REL(G,E\Limpl F)\).
The category is also \Staraut{} with \(\Sbot = \Sone 
= \{ \Sonelem \}\) as dualizing object.
This means that the morphism
\(\Curlin(\Evlin\Compl\Sym_{E\Limpl\Sbot,E})
\in\REL(E,(E\Limpl\Sbot)\Limpl\Sbot)\) %
is an iso.
In fact, if we define $\_^{\Sbot}$ as the functor $\_ \linarrow
\Sbot$, we have $E^{\Sbot} = E$ and 
$s^{\Sbot} = \{(b, a) \St (a, b) \in s\}$
up to a trivial isomorphism.
The morphism $s^{\bot}$ is called the transpose of $s$.

The category \(\REL\) is cartesian.
If \((E_i)_{i\in I}\) is a family of objects, then we define %
\[\withFam E_i = \Union_{i\in I}\{i\} \times E_i \]
as the disjoint union of the $E_i$.
The projections are given by the relations
\[ \Wproj_i=\{((i,a),a)\St i\in I\text{ and }a\in E_i\} 
\in \REL(\withFam E_i) .\] 
For any family
\((s_i\in\REL(F,E_i))_{i\in I}\) there is exactly one morphism %
\(\Tuple{s_i}_{i\in I}\in\REL(F,\withFam E_i)\) such that %
\(\forall j\in I\ \Wproj_j\Compl\Tuple{s_i}_{i\in I}=s_j\), namely
\begin{align*}
  \Tuple{s_i}_{i\in I}=\{(b,(i,a))\St i\in I\text{ and }(b,a)\in s_i\} .
\end{align*}
By \Starauty,
the category \(\REL\) is also cocartesian with
coproduct %
\[ \Bplus_{i\in I}E_i=\Orth{(\withFam {\Orth{E_i}})} 
= \withFam E_i . \] 
The corresponding injections are the dual of the projections, %
\((\Pin_i=\{(a,(i,a))\St i\in I\text{ and }a\in E_i\}
\in\REL(E_i,\Bplus_{j\in J}E_j))_{i\in I}\).
Observe that the coproduct and the product in $\REL$ are the same, even when the 
index set is infinite.
This means that $\REL$ has arbitrary biproduct (we will see that $\REL$ is a 
$\Sigma$-additive category with total sums).
Both the terminal object \(\Stop = \emptyset \) and the
initial object \(\Szero=\emptyset\) are given by the 
biproduct over the empty family $I = \emptyset$.

\begin{theorem}\label{th:rel-lafont-cat} 
  The symmetric monoidal category \(\REL\) is a Lafont category.
\end{theorem}
This result is well known folklore in LL. 
Any Lafont category is endowed with a canonical free exponential modality.
The free exponential modality 
\((\Oc{},\Deru,\Digg,\Seelyz,\Seelyt)\) of $\REL$ can be described as 
follows.
First \(\Oc E=\Mfin{E}\) is the set of finite multisets of
elements of \(E\).
If \(s\in\REL(E,F)\) then
\begin{equation} \label{eq:exponential-rel}
  \Oc s=\{(\Mset{\List a1n},\Mset{\List b1n})
  \St n\in\Nat\text{ and }((a_i,b_i)\in s)_{i=1}^n\} .
\end{equation}
The counit and comultiplication of the comonad are
\begin{align*}
  \Deru_E&=\{(\Mset a,a)\St a\in E\} \in \REL(\Oc E, E)\\
  \Digg_E&=\{(m_1+\cdots+m_n,\Mset{\List m1n})
           \St n\in\Nat\text{ and }(m_i\in\Mfin{E})_{i=1}^n\}
           \in \REL(\Oc E, \Oc \Oc E)\,.
\end{align*}
The Seely isomorphisms are %
\(\Seelyz=\{(\Sonelem,\Msetempty)\} \in \REL(1, \Oc \bot)\) and %
\[ \Seelyt_{E_1,E_2} =\{((m_1,m_2),\Msetact 1{m_1}+\Msetact 2{m_2})\St
(m_1,m_2) \in \Oc E_1\times \Oc E_1\} 
\in \REL(\Oc E_1 \Times \Oc E_2, \Oc(E_1 \with E_2))\] %
where \(\Msetact i{\Mset{\List a1n}}=\Mset{(i,a_1),\dots,(i,a_n)}\), 
recall \cref{sec:notation-multiset}.
The induced symmetric lax monoidal structure \((\Ocmonz,\Ocmont)\) is
easily proven to be
\begin{align*}
  \Ocmonz
  &=\{(\Sonelem,k\Mset{\Sonelem})\St k\in\Nat\}\in\REL(\Sone,\Oc\Sone)\\
  \Ocmont_{E,F}
  &=\{((\Mset{\List a1k},\Mset{\List b1k}),\Mset{(a_1,b_1),\dots,(a_k,b_k)})\\
  &\hspace{6em}
    \St k\in\Nat,\ \List a1k\in E\text{ and }\List b1k\in F\}
    \in\REL(\Tens{\Oc E}{\Oc F},\Oc(\Tens EF))\,.
\end{align*}

\subsubsection{Differential and Taylor expansion in $\REL$}
The category \(\REL\) has zero morphisms: take
\(0_{E,F}=\emptyset\in\REL(E,F)\).
The category $\REL$ is a differential category, see~\cite{Ehrhard18}. 
The addition of two relations is given by the union,
and the derivative is given by 
a natural transformation called the \emph{deriving transformation}.
\[ \coderbis_E = \{((m, a), m + [a]) \St a \in E, m \in \Oc E \} \in \REL(\Oc E \tensor E, 
\Oc E) . \]
The coKleisli category $\Kloc \REL$ is a cartesian differential category. 
The derivative of a morphism $s \in \Kloc \REL(E, F)$ is given by 
\[
\d s = \begin{tikzcd}
	{\Oc(E \with E)} & {\Oc E \tensor \Oc E} & {\Oc E \tensor E} & {\Oc E} & F
	\arrow["{(\seelyTwo)^{-1}}", from=1-1, to=1-2]
	\arrow["{\Oc E \tensor \der}", from=1-2, to=1-3]
	\arrow["{\coderbis_E}", from=1-3, to=1-4]
	\arrow["s", from=1-4, to=1-5]
\end{tikzcd} \] 
which is the standard definition of the derivative in the co-Kleisli category 
of a differential category, see \cite{Blute09}. 
We can check that 
\[ \d s = \{(0 \cdot m + 1 \cdot [a_1], b) \St (m+[a_1], b) \in s \} . \] 
Then, iterating the deriving transformation provides 
a natural transformation 
\[ \coderbisn_E : \{((m, a_1, \ldots, a_n), m + [a_1 + \cdots + a_n]) \St 
(a_i \in E)_{i=1}^n, 
m \in \Oc E\} \in \REL(\Oc E \tensor E^{\tensor n}, \Oc E) \] 
and a construction similar to the one above provides an $n$-th derivative for $s
\in \Kloc\REL (E, F)$.
\[
\begin{tikzcd}[column sep = small]
	\dn s = {\Oc(E \with E^{\with n})} & & & {\Oc E \tensor (\Oc E)^{\tensor n}} & & & {\Oc E \tensor E^{\tensor n}} & &  {\Oc E} & F
	\arrow["{(\seely{n+1})^{-1}}", from=1-1, to=1-4]
	\arrow["{\Oc E \tensor \der^{\tensor n}}", from=1-4, to=1-7]
	\arrow["{\coderbisn_E}", from=1-7, to=1-9]
	\arrow["s", from=1-9, to=1-10]
      \end{tikzcd} \]
We can check that $\dn s = \{(0 \cdot m + 1 \cdot [a_1]  \cdots + n \cdot [a_n], b) \St 
(m + [a_1, \ldots, a_n], b) \in s \}$.
In particular, the formal counterpart of $n$-homogeneous function 
$u \mapsto \deriven[0]{s}{n} \cdot (u, \ldots, u)$
is given by the morphism
\[
\begin{tikzcd}
	\dnz s = {\Oc E} & & {!(E \with E^{\with n})} & F .
	\arrow["\Oc \prodtuple{0, \id, \ldots, \id}", from=1-1, to=1-3]
	\arrow["{\dn s}", from=1-3, to=1-4]
      \end{tikzcd} \]
Observe that $\dnz s = \{ (m, b) \in s \St \Mscard m = n \}$.
Then the fact that the category $\REL$ is Taylor in the sense of \cite{Manzonetto12}
boils down to the equation 
\begin{equation} \label{eq:taylor-rel}
   s = \bigcup_{n \in \N} \dnz s
\end{equation}
which is a clear counterpart to the Taylor expansion at $0$
\[ f(u) = \sum_{n=0}^{\infty} \frac{1}{n!} \deriven[0]{f}{n} 
  (u, \ldots, u) \]
except that the coefficient $\frac{1}{n!}$ does not appear, the reason being that the 
sum in $\REL$ is idempotent. Note that for the sake of conciseness, we performed 
a Taylor expansion at $0$, but we could perform it at any point.

\subsubsection{Differential and Taylor expansion in $\REL$ as an analytic structure}

Let us explain how this Taylor expansion comes from an analytic structure.
First, observe that \cref{eq:taylor-rel} requires an infinite union.
This infinite union provides to $\REL$ the structure of 
a symmetric monoidal $\Sigma$-additive category.
Thus, by \cref{thm:omega-additive-representable}, $\REL$ is a 
representably $\Sigma$-additives category.
As such, the following result directly follows from \cref{th:coh-lafont-elem-sumable}.

\begin{theorem} $\REL$ is a Lafont representably $\Sigma$-additive category, 
  thus it is a representably analytic category.
\end{theorem}

The object \(\Dbimon=\withFam<\N> \Sone\) is the set of natural numbers.
The bimonoid structure of \(\Dbimon\) is given by
\begin{align*}
  \text{comonoid} \quad \Dbimoncu&=\{(0,\Sonelem)\}\in\REL(\Dbimon,\Sone)\\
  \Dbimoncm&=\{(n,(i,j))\St n,i,j\in\Nat\text{ and }n=i+j\}
             \in\REL(\Dbimon,\Dbimon\Times\Dbimon) \\
  \text{monoid} \quad 
  \Dbimonu&=\{(\Sonelem,n)\St n\in\Nat\}\in\REL(\Sone,\Dbimon)\\
  \Dbimonm&=\{((i,i),i) \St i\in\Nat\}
             \in\REL(\Dbimon\Times\Dbimon,\Dbimon)\,.
\end{align*}
We have seen that when the category is $\Sigma$-additive,  
\(\Sfun = \Dbimon \linarrow \_\) is given by 
the countable product of $X$ with itself. This functor 
has a bimonad structure induced by the
bimonoid structure of \(\Dbimon\), see \Cref{fig:bimonoid-to-bimonad}.
This structure is quite easy to describe.
\begin{align*} 
  \text{Monad unit: }\Sinj_0 &=\{(a,(0,a))\St a\in E\}\in\REL(E,\Sfun E) \\ 
  \text{Monad sum: } \SmonadSum &=\{((i,(j,a)),(i+j,a))\St i,j\in\Nat\text{ and }a\in E\}
  \in\REL(\Sfun^2E,\Sfun E) \\
  \text{Comonad unit: } \Ssum &=\{((i,a),a)\St i\in\Nat\text{ and }a\in E\}\in\REL(\Sfun
  E,S) \\
  \text{Comonad sum: } \Slift &=\{((i,a),(i,(i,a)))\St i\in\Nat\text{ and }a\in E\}
  \in\REL(\Sfun E,\Sfun^2E) \\ 
  \text{Distributive law: }\Sswap &=\{((i,(j,a)),(j,(i,a)))\St i,j\in\Nat\text{ and }a\in E\}
  \in\REL(\Sfun^2E,\Sfun^2E) .
  \end{align*}

By \Cref{th:Lafont-elem-summable-Taylor}, \(\Dbimon\) has an analytic coalgebra, 
that is, a \(\Oc\)-coalgebra structure
\(\Dbimonca\in\REL(\Dbimon,\Oc\Dbimon)\) which is given by
\begin{align*}
  \Dbimonca=\{(n,\Mset{\List i1k})
  \St k\in\Nat\text{ and }\List i1k\in\Nat\text{ with }i_1+\cdots+i_k=n\}\,.
\end{align*}
This is simply due to the fact that the \(k\)-ary version %
\(\Dbimon\to\Tensexp\Dbimon k\) is the relation %
\(\{(n,(\List i1k))\St i_1+\cdots+i_k=n\}\).
As seen in \cref{cor:Taylor-structure-induces-differential},
this coalgebra structure induces the distributive law %
\(\Sdl_E=\Curlin(\Sdl'_E)\in\REL(\Oc\Sfun E,\Sfun\Oc E)\) where %
\(\Sdl'_E\in\REL(\Oc(\Dbimon\Limpl E)\Times\Dbimon,\Oc E)\) is
\[
  \begin{tikzcd}
    \Oc(\Dbimon\Limpl E)\Times\Dbimon
    \ar[r,"\Tens{\Oc(\Dbimon\Limpl E)}{\Dbimonca}"]
    &[2em]
    \Oc(\Dbimon\Limpl E)\Times\Oc\Dbimon
    \ar[r,"\Ocmont"]
    &[-0.8em]
    \Oc((\Dbimon\Limpl E)\Times\Dbimon)
    \ar[r,"\Oc\Evlin"]
    &[-0.6em]
    \Oc E
  \end{tikzcd}
\]
so that
\begin{multline*}
  \Sdl_E
  =\{(\Mset{(i_1,a_1),\dots,(i_k,a_k)},(i_1+\dots+i_k,\Mset{\List a1k}))\\
  \St k\in\Nat,\ \List i1k\in\Nat\text{ and }\List a1k\in E\}\,.
\end{multline*}
Remember that the extension of \(\Sfun\) to the coKleisli category
\(\Kloc\REL\) is the functor \(\Tayfun:\Kloc\REL\to\Kloc\REL\) which
maps an object \(E\) to \(\Sfun E\)
and a morphism \(s\in\Kloc\REL(E,F)\) to
$\Tayfun s \defEq (\Sfun s)\Compl\Sdl_E$
\begin{align*}
  \Tayfun s
  &=\{(\Mset{(i_1,a_1),\dots,(i_k,a_k)},(i_1+\cdots+i_k,b))\\
  &\hspace{8em}\St k\in\Nat,\ \List i1k\in\Nat\text{ and }
    (\Mset{\List a1k},b)\in s\}\,.
\end{align*}
The morphism \(\Sinj_1=(\Wproj_1\Limpl E)\in\REL(E,\Sfun E)\) satisfies %
$\Sinj_1=\{(a,(1,a))\St a\in\Web E\}$
so that 
\[\Oc{\Sinj_1}=\{(\Mset{\List a1k},1 \cdot \Mset{a_1,\dots,a_k})
\St k\in\Nat\text{ and }\List a1k\in E\}\in\REL(\Oc E,\Oc{\Sfun E}) .\]
It follows that %
\[ \Tayfun(s)\Compl\Oc{\Sinj_1} =\{(\Mset{\List a1k},(k,b))\St
(\Mset{\List a1k},b)\in s\} \in\Kloc\REL(E,\Sfun F) . \] 
Then \(\Sproj_k\Compl\Tayfun(s)\Compl\Oc{\Sinj_1} = \dnz[k] s \). 
This means that the usual Taylor expansion at $0$ can be recovered from 
this functor, as seen in \cref{eq:recover-taylor} (remember that in this setting,
$\frac{1}{k!} = 1$ because the sum in $\REL$ is idempotent).
Then the axiom \ref{ax:Sdl-Taylor} implies that for any $s$, 
\[ s = \Ssum \compl \D s \compl \Oc \Sinj_1 
= \bigcup_{k \in \N} \Sproj_k\Compl\Tayfun(s)\Compl\Oc{\Sinj_1} 
= \bigcup_{k \in \N} \dnz[k] s. \] 
This is precisely \cref{eq:taylor-rel}, so \ref{ax:Sdl-Taylor} captures the fact 
that morphisms in $\REL$ are equal to their Taylor expansion.

\subsection{The weighted relation model}
\label{sec:wrel}

The weighted relational model introduced by \cite{Lamarche92}
and \cite{Laird13}\ is a 
quantitative generalization of the relational model.
Let $\Rcal$ be a complete semi-ring\footnote{In \cite{Laird13}, $\Rcal$ is 
assumed to be continuous, which is a stronger assumption than complete. See~\cite{Ong17} for an introduction 
of $\WREL$ over complete semirings.}. Then $\WREL$ is the category whose objects 
are sets whose morphisms are matrices, $\WREL(E, F) = \Rcal^{E \times F}$. 
Identity is given by the matrix $\id_{a, b} = \Kronecker a b$, and the composition 
of $s \in \WREL(E, F)$ with $t \in \WREL(F, G)$ is given by matrix 
composition\footnote{We write this product in the reverse order wrt.~the usual
algebraic conventions on matrices, because it is the notion of
composition in our category, and we respect the standard order of
factors when writing a composition in a category.
This is a well known and unfortunate mismatch of conventions.}.
\[(t \compl s)_{a, c} = \sum_{j \in F} s_{a, b} t_{b, c}.\]
This sum is well-defined, precisely because the semiring is complete
and allows arbitrary countable sums.
For any $s \in \WREL(E, F)$ and $x \in \Rcal^E$, we define $\Matappa sx\in\Rcal^F$ by 
\[ (\Matappa s x)_b = \sum_{a \in E} s_{a, b} x_a . \]
This operation is functorial: for any $E, F, G$, any $x \in \Rcal^E$, any $s \in \WREL(E, F)$ 
and $t \in \WREL(F, G)$,
$\Matappa t {(\Matappa s x)} = \Matappa {(t \compl s)} x$
and $\Matappa \id x = x$.

Observe that $\REL$ coincides with the category $\WREL[\Bring]$, where $\Bring = 
({0, 1}, \vee, \wedge)$ is the boolean semiring ($\vee$ is the disjunction, and 
$\wedge$ the conjunction). Indeed, 
any relation $s \in \REL(E, F)$ can be seen as a matrix 
$s' \in \WREL[\Bring](E, F)$ defined as 
\begin{equation} \label{eq:relation-as-matrix}
  s'_{a, b} = \begin{cases}
  1 \text{ if } (a, b) \in s \\
  0 \text{ otherwise.}
\end{cases} \end{equation}

\begin{remark} \label{rem:from-ll-to-wrel}
  In fact, any relation $r \subseteq E \times F$ can be seen as a morphism 
  $r \in \WREL(E, F)$ for any $\Rcal$, using \cref{eq:relation-as-matrix}
  above.
  This translation is functorial when restricted to 
  relations $r \subseteq E \times F$ 
  that are the graphs of bijective functions (that is, for 
  all $x \in E$, there is a unique $y \in F$ such that 
  $(x, y) \in r$). Most of the structural morphisms of the LL structure of
  $\REL$ consist in such relations, and thus this construction
  provides for free structural morphisms in $\WREL$ that satisfies 
  the commutations required by LL.  
  Note however that this construction does not preserve 
  the structure of $\REL$ in general. It does not preserve the sum, and it is not 
  even functorial when considered on all relations\footnote{Unless $1+1 = 1$, which is 
  not at all a common assumption}. 
\end{remark}

There is a bifunctor on $\WREL$ defined by $E \tensor F = E \times F$
(as in $\REL$) and for any $(s^i \in \WREL(E_i, F_i))_{i = 1,2}$,
\[ (s^1 \tensor s^2)_{(a_1, a_2), (b_1, b_2)} = s^1_{a_1, b_1} \cdot s^2_{a_2, b_2} 
\in \WREL(E_1 \tensor E_2, F_1 \tensor F_2) . \]
By \cref{rem:from-ll-to-wrel}, the symmetric monoidal structure of $\REL$ 
can be translated to $\WREL$, providing $\WREL$ 
with a symmetric monoidal structure. 
For example, $\tensorSym_{(a_1, a_2), (a_1', a_2')} = \Kronecker{a_1}{a_1'} \Kronecker{a_2}{a_2'}$.
Furthermore, $\WREL$ is closed with $E \linarrow F = E \times F$ and 
for any $s \in \WREL(G \tensor E, F)$, 
$\cur(s)_{a, (b, c)} = s_{(a, b), c}$.  The evaluation map 
$\ev \in \WREL(E \linarrow F \tensor E, F)$ 
is given by the translation of the evaluation
$\ev \in \REL(E \linarrow F \tensor E, F)$ into $\WREL$,
that is, $\ev_{((a, b), a'), b'} = \Kronecker {a} {a'} \Kronecker b {b'}$.

The category is $*$-autonomous, with $\bot = 1 = {*}$ as dualizing object. 
If $\_^{\bot}$ is the functor $\_ \linarrow \bot$, then as in $\REL$, 
$E^{\bot} = E$ and $s^{\bot}_{a, b} = s_{b, a}$ up to a trivial isomorphism
(which is given by the isomorphism in $\REL$ and by \cref{rem:from-ll-to-wrel}).
So the dual in $\WREL$ corresponds to 
the notion of matrix transposition. 
The category $\WREL$ is cartesian, with 
\[ \withFam E_i = \bigcup_{i \in I} \{i\} \times E_i \, . \]
The tupling of $(f(i) \in \WREL(E, F_i))$ is defined by 
\[(\prodPairing{f(i)})_{a, (i, b)} = f(i)_{a, b} \]
and the projections are the translations given by \cref{rem:from-ll-to-wrel}
of the projections in $\REL$, 
that is, $(\prodProj_i)_{(j, a), b} = \Kronecker i j \Kronecker a b$.
This cartesian product is dual to itself, and $\WREL$ has countable biproducts.

\begin{theorem} \label{th:wrel-lafont-cat}
  (Corollary III.6 of \cite{Laird13})
  The symmetric monoidal category \(\WREL\) is a Lafont category.
\end{theorem}

The canonical exponential modality associated to this free construction is 
given by the functor $\Oc\_$ that maps an object $E$ to $\Oc E = \Mfin E$
and a morphism $s \in \WREL(E, F)$ to
\begin{equation} \label{eq:exponential-v1} (\Oc s)_{m, [b_1, \ldots, b_n]} = 
\sum_{\substack{(a_1, \ldots, a_n) \text{ s.t. } \\ [a_1, \ldots, a_n] = m}} 
\prod_{i = 1}^n s_{a_i, b_i} 
\in \WREL(\Oc E, \Oc F) . \end{equation}
This formula is a straightforward generalization of 
the formula of the exponential in $\REL$ given in \cref{eq:exponential-rel}.
The comonad structure and the Seely isomorphisms correspond to the translation 
into $\WREL$
of the exponential modality structure of $\REL$.

The exponential given in \cref{eq:exponential-v1}
has an alternative formulation. For any $(a_1, \dots, a_n)$ 
such that $m = [a_1, \ldots, a_n]$, define 
$r_{a_1, \ldots, a_n} \in \Mfin{E \times F}$ by 
$r_{a_1, \ldots, a_n}(a, b) = \Scard{\{k \St a_k = a, b_k = b\}}$
so that by definition of $s^{r_{a_1, \ldots, a_n}}$ (recall \cref{eq:multiset-power})
\[ (\Oc s)_{m, [b_1, \ldots, b_n]} = 
\sum_{\substack{(a_1, \ldots, a_n) \text{ s.t. } \\ [a_1, \ldots, a_n] = m}} 
s^{r_{a_1, \ldots, a_n}} \, .\]
Observe that $r_{a_1, \ldots, a_n}(a, b) \in \Mstrans mp$ 
where \(\Mstrans mp\) for the set of all
\(r\in\Mfin{E\times F}\) such that
\begin{align*}
  \forall a\in E \ m(a)=\sum_{b\in F}r(a,b)
  \text{\quad and\quad}
  \forall b\in F \ p(b)=\sum_{a\in E}r(a,b)\,.
\end{align*}
Notice that if \(r\in \Mstrans mp\) then %
\(\Mscard r=\Mscard m=\Mscard p\) so that \(\Mstrans mp\) is non-empty
iff \(\Mscard m=\Mscard p\). Now, given any $r \in \Mstrans mp$, the number of tuples
$(a_1, \ldots, a_n)$ such that $m = [a_1, \ldots, a_n]$ and 
$r = r_{a_1, \ldots, a_n}$ is equal to 
\begin{align*}
  \Multinomb pr
  =\prod_{b\in F}\frac{\Factor{p(b)}}{\prod_{a\in E}\Factor{r(a,b)}}
\end{align*}
which belongs to \(\Rcal\); this is a generalized
multinomial coefficient. This enumeration can be obtained using the \emph{class equation}
of group theory.
Thus, we have
\begin{equation}  \label{eq:exponential-v2} 
  (\Oc s)_{m,p}=\sum_{r\in\Mstrans mp}\Multinomb pr s^r \, .
\end{equation}

Any morphism $s\in\WREL(\Oc E,F)$ induces a function 
$\Rcal^E \arrow \Rcal^F$ defined by a power series.
Let $\Prom x \in \Rcal^{\Oc E}$ by $(\Prom x)_m = x^m$, and
\begin{align*}
  \Fun s:\Rcal^E &\to\Rcal^F\\
  x&\mapsto\Matappa s{\Prom x} \text{ that is, } (\Fun s(x))_b
  =\sum_{m\in \Oc E,b\in F}s _{m,b}x^m .
\end{align*}

\begin{definition} \label{def:entire}
We call the functions defined in that way called \emph{entire functions}, following the terminology 
of \cite{DanosEhrhard11}.
\end{definition}

\begin{remark} \label{rem:analytic-not-bijective}
This operation is not a bijection between morphisms and entire functions.
For example, let 
$\Rcal = (\Rbar, +, \times)$. Then the function $\Rbar^{\{*\}} \arrow \Rbar^{\{*\}}$
that maps $0$ to $0$ and any $x > 0$ to $\infty$ is entire, and can be written in multiple 
ways, such as $f(x) = \infty \times x$ of $f(x) = \infty \times x^2$.
So if we define $s(1)_{(m, *) = \Kronecker{m}{[*]}}$ and $s(2)_{(m, *)} 
= \Kronecker{m}{[*, *]}$ we have $f = \Fun {s(1)} = \Fun{s(2)}$ even though
$s(1) \neq s(2)$.
We will see in~\cref{sec:pcoh} that the category of 
probabilistic coherent spaces is much better behaved from this point of view.
\end{remark}

\subsubsection{Differential and Taylor expansion in $\WREL$}

The category $\WREL$ is 
a differential category whose derivative has the same definition as the derivative 
of the historical model of Köthe spaces, see \cite{Ehrhard02}.
This derivative corresponds to the formal derivative 
of power series that generalizes the deriving transformation
of $\REL$ to arbitrary complete semirings.
\begin{align*}
  \coderbis_{(m, a), p} &= p(a) \Kronecker{m + [a]}p \in \WREL(\Oc E \tensor E, \Oc E) \\
  \coderbisn_{(m, a_1, \ldots, a_n), p} &= \frac{p!}{m!} \Kronecker{m + [a_1, \ldots, a_n]}p \in 
\WREL(\Oc E \tensor E^{\tensor n}, \Oc E) .
\end{align*}
The coefficient $\frac{p!}{m!}$ in this derivative enumerates the number of ways the multiset 
$[a_1, \ldots, a_n]$ can be inserted into $m$ in order to produce $p$.

We can check that the $n$-th derivative of $s \in \WREL(E, F)$ is equal to
\[ (\dn s)_{0 \cdot m + 1 \cdot m_1 + \cdots + n \cdot m_n, b} = \begin{cases}
  \frac{(m + [a_1, \ldots, a_n])!}{m!} s_{m + [a_1, \ldots, a_n], b} 
  \text{ if } \forall i,  m_i = [a_i] \\
  0 \text{ otherwise.}
\end{cases} \] It follows that 
\begin{align*} (\dnz s)_{m, b} 
  &= \sum_{\Biind{m_1, \ldots, m_n \text{ s.t.}}{m = m_1 + \ldots + m_n}} 
  (\dn s)_{0\cdot[] + 1 \cdot m_1 \cdots + n \cdot m_n} \\
  &= \sum_{\Biind{(a_1, \ldots, a_n) \text{ s.t.}}{m = [a_1, \ldots, a_n]}} m! \cdot s_{m, b} \\ 
  &= \begin{cases}
    n! \cdot s_{m, b} \text{ if } \Mscard m = n \\ 
    0 \text { otherwise} 
  \end{cases}
\end{align*}
upon observing that if $\Mscard{m} = n$ then the number of tuple $(a_1, \ldots, a_n)$ such that 
$m = [a_1, \ldots, a_n]$ is equal to $\frac{n!}{m!}$ (using the class formula of group theory).

So, \emph{assuming that $n!$ has a multiplicative inverse in $\Rcal$},  
$\frac{1}{n!} \dnz s$ is the $n$-th homogeneous component of $s$,
and we have the Taylor formula
\begin{equation} \label{eq:taylor-wrel}
   s = \sum_{n \in \N} \frac{1}{n!} \dnz s \, .
\end{equation}
This sum is well defined because $\Rcal$ is complete.
Taking $\Rcal = \Bring$, we recover the Taylor expansion in $\REL$ (in $\Bring$, 
the multiplicative inverse of $n! = 1$ is $1$).

As already mentionned, the notion of derivative in $\WREL$ corresponds to the notion of derivative of 
formal power series.
More precisely, we can check that
\begin{equation} \label{eq:explicit-derivative}
  \Fun{\dn s}(x, u^1, \ldots, u^n)_b
   =\sum_{m\in\Mfin{\Web E}}
   \sum_{\Vect a\in\Web E^d}
   \frac{\Factor{(m+\Mset{\Vect a})}}{\Factor m}
   s_{m+\Mset{\Vect a},b}x^m u^1_{a_1}\cdots u^d_{a_d} \, .
 \end{equation}
 We will write this quantity $\Deriv{s}{n}(x)(u^1, \ldots, u^n)$, 
 because it corresponds to the entire function given by 
 the formal derivative of the formal power series given by $s$.
 
 \begin{remark} Taking $s(1)_{(m, *) = \Kronecker{m}{[*]}}$ and $s(2)_{(m, *)} 
= \Kronecker{m}{[*, *]}$ 
in $\WREL[\Rbar](\oc \Sone, \Sone)$ 
as in \cref{rem:analytic-not-bijective}, we have 
$\Fun s(1) = \Fun s(2)$, yet 
\[ \id = \Fun {\dnz[1] s(1)} \neq \Fun{\dnz[1] s(2)} = 0 . \]
So the derivative $\d s$ is very dependent on the coefficients of the formal power series 
given by $s$, and is not determined at all by the induced function $\Fun s$ 
(observe that $\Fun {s(1)} = \Fun {s(2)}$ is not even derivable at $0$ in the sense 
of traditional real analysis).
Again, the situation is much more satisfying  in probabilistic coherence spaces.
 \end{remark}

 \subsubsection{Differential and Taylor expansion in $\WREL$ as an analytic structure}
 We now describe how this Taylor structure arises from an analytic coalgebra.
The category $\WREL$ is a symmetric monoidal $\Sigma$-additive category, where the sum 
of $\family<I>[i]{s^i \in \WREL(E, F)}$
is defined as the pointwise sum
\[ \left(\sum_{i \in I} s^i \right)_{a, b} \defEq \sum_{i \in I} s^i_{a, b}  \, . \]
Observe that when $\Rcal = \Bring$, this sum coincides with the sum in $\REL$.
By \cref{thm:omega-additive-representable}, $\WREL$ is a 
representably $\Sigma$-summable category, and thus the following result 
holds.

\begin{theorem} $\WREL$ is a Lafont representably $\Sigma$-additive category, 
  thus it is a representably analytic category.
\end{theorem}
The object $\Dbimon$ is again the set of natural numbers. Its bimonoid structure 
and its analytic coalgebra
is given by the translation of the bimonoid structure and the analytic 
coalgebra on $\Dbimon$ in $\REL$. 
The quantitative aspect of $\WREL$ only arises
when computing the induced distributive law
\(\Sdl_E=\Curlin(\Sdl'_E)\in\WREL(\Oc\Sfun E,\Sfun\Oc E)\) where %
\(\Sdl'_E\in\WREL(\Oc(\Dbimon\Limpl E)\Times\Dbimon,\Oc E)\) is
\[
  \begin{tikzcd}
    \Oc(\Dbimon\Limpl E)\Times\Dbimon
    \ar[r,"\Tens{\Oc(\Dbimon\Limpl E)}{\Dbimonca}"]
    &[2em]
    \Oc(\Dbimon\Limpl E)\Times\Oc\Dbimon
    \ar[r,"\Ocmont"]
    &[-0.8em]
    \Oc((\Dbimon\Limpl E)\Times\Dbimon)
    \ar[r,"\Oc\Evlin"]
    &[-0.6em]
    \Oc E \, .
  \end{tikzcd}
\]
One checks very easily that
$\Ocmont\Compl\Tensp{\Oc(\Dbimon\Limpl E)}{\Dbimonca}$ has only
\(0\) and \(1\) coefficients and is the translation
of the morphism \(\Ocmont\Compl\Tensp{\Oc(\Dbimon\Limpl E)}{\Dbimonca}\) 
of $\REL$:
\begin{align*}
  &(\Ocmont\Compl\Tensp{\Oc(\Dbimon\Limpl E)}{\Dbimonca}
  )_{(\Mset{(i_1,a_1),\dots,(i_k,a_k)},n),q}\\
  &\hspace{4em}=
  \begin{cases}
    1 & \text{if }q=\Mset{((i_1,a_1),j_1),\dots,((i_k,a_k),j_k)}\\
    &\text{with }\List j1k\in\Nat\text{ such that }j_1+\cdots+j_k=n\\
    0 & \text{otherwise,}
  \end{cases}
\end{align*}

Now, \(\Evlin\in\WREL((\Dbimon\Limpl E)\Times\Dbimon,E)\) is
characterized by \(\Evlin_{((i,a),j),b}=\Kronecker ij\Kronecker ab\)
and hence if \(\Oc\Evlin_{q,m}\not=0\) with %
\(q\in\Mfin{(\Nat\times E)\times\Nat}\) and \(m\in\Mfin{E}\),
we must have %
\(m=\Mset{\List a1k}\) and %
\(q=\Mset{((i_1,a_1),i_1),\dots,((i_k,a_k),i_k)}\) %
for some \(k\in\Nat\), \(\List a1k\in\Web E\) and \(\List i1k\in\Nat\).
For such multisets \(m\) and \(q\), the set \(\Mstrans qm\) has exactly
one element %
\(r\in\Mfin{((\Nat\times\Web E)\times\Nat)\times\Web E}\) such that %
\(\Evlin^r\not=0\), namely %
\(r=\Mset{(((i_1,a_1),i_1),a_1),\dots,(((i_k,a_k),i_k),a_k)}\), and we
have therefore
\begin{align*}
  \Oc\Evlin_{q,m}
  =\Multinomb mr
  =\prod_{a\in\Web E}\frac{\Factor{m(a)}}
  {\prod_{i\in\Nat}\Factor{p(i,a)}}=\frac{\Factor m}{\Factor p}\in\Nat
\end{align*}
where \(p\in\Mfin{\Nat\times\Web E}\) is defined by %
\(p(i,a)=q((i,a),i)=r(((i,a),i),a)\).
It follows that, for \(p\in\Mfin{\Nat\times\Web E}\), \(n\in\Nat\) and
\(m=\Mset{\List a1k}\in\Mfin{\Web E}\), we have
\begin{align*}
  \Sdl_{p,(n,m)}=
  \begin{cases}
    \frac{\Factor m}{\Factor p}
    &\text{if }p=\Mset{(i_1,a_1),\dots,(i_k,a_k)}\text{ with }
    i_1+\cdots+i_k=n\\
    0 & \text{otherwise.}
  \end{cases}
\end{align*}
This is a quantitative counterpart of the distributive law in 
$\REL$. 

\begin{remark} \label{rem:derivative-no-coef}
  It is interesting, and slightly puzzling, to observe that the numerical
  coefficients associated with the computation of derivatives (the
  \(42\) in the derivative \(42x^{41}\) of \(x^{42}\)) are generated,
  in the definition of \(\Sdl\), by the exponential combined with the
  evaluation map when computing \(\Oc\Evlin\) ---~which by the way is a
  purely LL morphism, not using any differential structure~---, and
  not by the \(\Dbimonca\) morphism itself, which seems to be ``the
  truly differential part'' of this definition.
\end{remark}

It follows that the induced Taylor functor \(\Tayfun:\Kloc\WREL\to\Kloc\WREL\)
maps a set $E$ to the set \(\Sfun E = \N \times E \) and 
a morphism \(s\in\Kloc\WREL(E,Y)\) to %
\(t=\Tayfun(s)\in\Kloc\WREL(\Sfun E,\Sfun Y)\) given by
\begin{align*}
  \Tayfun(s)_{\Mset{(i_1,a_1),\dots,(i_k,a_k)},(n,b)}=
  \Kronecker{n}{i_1+\cdots+i_k}
  \frac{\Factor{\Mset{\List a1k}}}
  {\Factor{\Mset{(i_1,a_1),\dots,(i_k,a_k)}}}
    s_{\Mset{\List a1k},b}\,.
\end{align*}
In particular, $(\Tayfun(s) \compl \Oc \Sinj_1)_{[a_1, \ldots, a_k], n} 
= \Kronecker{n}{k}  s_{[a_1, \ldots, a_k], b}$ 
so $\Sproj_n \compl \Tayfun(s) \compl \Oc \Sinj_1$
is the $n$-th homogeneous component of $s$. 
It follows from \ref{ax:Sdl-Taylor} that whenever the numbers $n!$ 
have multiplicative inverses,
\begin{equation} \label{eq:analytic-wrel}
  s = \Ssum \compl \D s \compl \Oc \Sinj_1 
= \sum_{n \in \N} \Sproj_n\Compl\Tayfun(s)\Compl\Oc{\Sinj_1} 
= \sum_{n \in \N} \frac{1}{n!} \dnz s
\end{equation}
so \cref{eq:taylor-wrel} boils down to \ref{ax:Sdl-Taylor}
(we did Taylor expansion at 
$0$ for the sake of conciseness, but we could do a Taylor expansion at any point
in a similar way).

\begin{remark} \label{rem:taylor-quotient}
  A key property of our Taylor expansion is that it structuraly 
  quotient the usual $n$-th derivative by the right coefficient, as 
  crucially used in \cref{eq:analytic-wrel}.
  In particular, it is quite remarkable that 
  our theory of analytic categories and analytic structure does not require $n!$ to 
  have a multiplicative inverse, in contrast to the Taylor formula given by 
  \cref{eq:taylor-wrel}. For example, $\WREL[\overline{\N}]$ is a representably
  analytic category, although it is not usually considered to be a model of 
  syntactical Taylor expansion.
\end{remark}

We give now an explicit formula for $\Fun{\D(s)} : \Rcal^{\Sfun E}\to\Rcal^{\Sfun Y}$ 
and we show that this function coincides with 
the intuitive formula given in \cref{sec:Taylor-operator}.
Notice that a multiset %
\(p=\Mset{(i_1,a_1),\dots,(i_k,a_k)}\) can be written \emph{in a
  unique way} \(p=\Msetact{j_1}{m_1}+\cdots+\Msetact{j_l}{m_l}\) 
  (recall the notation of \cref{sec:notation-multiset}),
  where
\(j_1<\cdots<j_l\in\Nat\) and %
\((m_i\in\Mfinnz{\Web E})_{i=1}^l\) are such that
\(m_1+\cdots+m_l=m=\Mset{\List a1k}\).
The condition that the \(m_i\)'s are non-empty is crucial for this
enumeration to be bijective: the only way to get \(p=\Msetempty\) is by
choosing \(l=0\).
It follows that the associated function %
\(\Fun{\Tayfun(s)}: \Rcal^{\Sfun E}\to\Rcal^{\Sfun Y}\) satisfies, for
each \(\Vect x\in\Rcal^{\Sfun E}\) and \((n,b)\in \Sfun Y\)
\begin{equation} \label{eq:T-functional}
  \Fun{\Tayfun(s)}(\Vect x)_{n,b}
  =
  \sum_{l=0}^\infty
  \sum_{\Vect m\in\Mfinnz{\Web E}^l}
  \sum_{\Biind{j_1<\cdots<j_l\in\Nat}
  {j_1\Mscard{m_1}+\cdots+j_l\Mscard{m_l}=n}}
  \Multinom{\Vect m}
  s_{m_1+\cdots+m_l,b}
  x(j_1)^{m_1}\cdots x(j_l)^{m_l}
\end{equation}
where we recall that
\(\Multinom{\Vect m}
=\frac{\Factor{(m_1+\cdots+m_l)}}{\Factor{m_1}\cdots\Factor{m_l}}\in\N\)
is the multinomial coefficient of the sequence \(\Vect m\) of multisets.

\begin{example} 
  For example, for \(n=0\), there are two ways to fulfill the condition %
\(j_1\Mscard{m_1}+\cdots+j_l\Mscard{m_l}=n\): either with \(l=0\), or
with \(l=1\) and \(j_1=0\), and then \(m_1\) can be any element of
\(\Mfinnz{\Web E}\).
Therefore, we have
\begin{align*}
 \Fun{\Tayfun(s)}(\Vect x)_{0,b}
 =s_{\Msetempty,b}+\sum_{m\in\Mfinnz{\Web E}}s_{m,b}x(0)^m
 =\Fun s(x(0))_b
\end{align*}
so that \(\Fun{\Tayfun(s)}(\Vect x)(0)=\Fun s(x(0))\). 
This equation corresponds to the axiom 
\ref{ax:Sdl-local}.
\end{example}

\begin{theorem} \label{thm:Taylor-expansion-wrel}
  For all \(n\in\Natnz\)
  \begin{align*}
    \Fun{\Tayfun(s)}(\Vect x)_{n,b}
    =\sum_{k=0}^\infty
    \sum_{\Vect n\in(\Natnz)^k}
    \sum_{\Biind{0<i_1<\cdots<i_k\in\Nat}{n_1i_1+\cdots+n_ki_k=n}}
    \frac 1{\prod_{i=1}^k\Factor{n_i}}
    \Deriv{s}{n_1+\cdots+n_k}(x(0))
    (\Rep{x(i_1)}{n_1},\dots,\Rep{x(i_k)}{n_k})_b
  \end{align*}
  where \(\Rep xl\) is the list of repeated arguments %
  \(\overbrace{x,\dots,x}^{n\times}\).
\end{theorem}
It follows from the right-hand expression above that
\[ \Fun{\Sproj_n \compl \Tayfun(s)}(\Vect x) =  \sum_{m \in \mpart{n}}
\frac{1}{m!} \Deriv{s}{\Mscard m}(x(0))
(\Rep{x(i_1)}{m(1)},\dots,\Rep{x(i_k)}{m(k)})  \]
where $\mpart{n} \defEq \{ m \in \mfin(\N^*) | \sum_{i \in \N^*} i\, m(i) =
n\}$ upon
identifying a pair \((\Vect i,\Vect n)\in\Nat^k\times(\Natnz)^k\) such
that \(0<i_1<\cdots<i_k\) with the element \(m\) of \(\Mfin\Natnz\)
such that \(\Msetsupp m=\{\List i1k\}\), which maps \(i_j\) to \(n_j\)
for \(j=1,\dots,k\). This corresponds to
\Cref{eq:Taylor-functor-formula} given in \cref{sec:Taylor-operator} 
and corresponds to a compositional formulation of Taylor expansion, 
thanks to the Faà di Bruno formula.
In particular, we recover all the terms of
the usual Taylor expansion by taking
\begin{align*}
  \Fun{\Tayfun(s)}(x,u,0,0,\dots)
  =\big(\frac 1{\Factor n}\Deriv{s}n(x)(\Rep un)\big)_{n\in\Nat} \, .
\end{align*}

\begin{proof}
Let \(v\) be the right-hand expression in the statement of the theorem, we have by \cref{eq:explicit-derivative}
\begin{align*}
  v&=\sum_{k=0}^\infty
  \sum_{\Vect n\in(\Natnz)^k}
  \sum_{\Biind{0<i_1<\cdots<i_k\in\Nat}{n_1i_1+\cdots+n_ki_k=n}}
  \frac 1{\prod_{i=1}^k\Factor{n_i}}
     \sum_{m\in\Mfin{\Web E}}\\
  &\hspace{2em}
  \sum_{\Vect a\in\Web E^{n_1+\dots+n_k}}
  \frac{\Factor{(m+\Mset{\Vect a})}}{\Factor m}
    s_{m+\Mset{\Vect a},b}{x(0)}^m 
    \prod_{j=1}^{n_1}x(i_1)_{a_j}
    \prod_{j=n_1+1}^{n_2}x(i_2)_{a_j}
    \cdots
    \prod_{j=n_1+\cdots+n_{k-1}+1}^{n_1+\cdots+n_k}x(i_k)_{a_j}\\
  &=\sum_{k=0}^\infty
  \sum_{\Vect n\in(\Natnz)^k}
  \sum_{\Biind{0<i_1<\cdots<i_k\in\Nat}{n_1i_1+\cdots+n_ki_k=n}}
  \frac 1{\prod_{i=1}^k\Factor{n_i}}
     \sum_{m\in\Mfin{\Web E}}\\
  &\hspace{2em}
  \sum_{\Biind{\Vect m\in\Mfin{\Web E}^k}{(\Mscard{m_i}=n_i)_{i=1}^k}}
  \frac{\Factor{(m+m_1+\cdots+m_k)}}{\Factor m}
    s_{m+m_1+\cdots+m_k,b}{x(0)}^m 
    x(i_1)^{m_1}\cdots x(i_k)^{m_k}
    \prod_{i=1}^k\frac{\Factor{n_i}}{\Factor{m_i}}
\end{align*}
because \(\frac{\Factor{n_i}}{\Factor{m_i}}\) is the number of
enumerations of the elements of \(m_i\) (taking multiplicities into
account).
So by simplification of the $\prod_{i=1}^k n_i!$ and definition of 
$\Multinom{m,\Vect m}$ we have
\begin{align*}
  v&=\sum_{k=0}^\infty
     \sum_{\Vect n\in(\Natnz)^k}
     \sum_{\Biind{0<i_1<\cdots<i_k\in\Nat}{n_1i_1+\cdots+n_ki_k=n}}
     \sum_{m\in\Mfin{\Web E}}\\
   &\hspace{4em}
     \sum_{\Biind{\Vect m\in\Mfin{\Web E}^k}{(\Mscard{m_i}=n_i)_{i=1}^k}}
     \Multinom{m,\Vect m}
     s_{m+m_1+\cdots+m_k,b}{x(0)}^m 
     x(i_1)^{m_1}\cdots x(i_k)^{m_k}
\end{align*}
This sum is the same as $\Fun{\Tayfun(s)}(\Vect x)_{n,b}$ given in 
\cref{eq:T-functional} up to a straightforward reindexing as we show now:
\begin{align*}
   v &=\sum_{k=0}^\infty
     \sum_{m\in\Mfinnz{\Web E}}
     \sum_{\Vect m\in\Mfinnz{\Vect E}^k}
     \sum_{\Biind{0<i_1<\cdots<i_k\in\Nat}
     {\Mscard{m_1}i_1+\cdots+\Mscard{m_k}i_k=n}}
     \Multinom{m,\Vect m}
     s_{m+m_1+\cdots+m_k,b}{x(0)}^m 
     x(i_1)^{m_1}\cdots x(i_k)^{m_k}\\
   &\hspace{3em}+\sum_{k=0}^\infty
     \sum_{\Vect m\in\Mfinnz{\Vect E}^k}
     \sum_{\Biind{0<i_1<\cdots<i_k\in\Nat}
     {\Mscard{m_1}i_1+\cdots+\Mscard{m_k}i_k=n}}
     \Multinom{\Msetempty,\Vect m}
     s_{m_1+\cdots+m_k,b} 
     x(i_1)^{m_1}\cdots x(i_k)^{m_k}\\
   &=\sum_{l=1}^\infty
     \sum_{\Vect m\in\Mfinnz{\Vect E}^l}
     \sum_{\Biind{i_1=0<i_2<\cdots<i_l\in\Nat}
     {\Mscard{m_1}i_1+\cdots+\Mscard{m_l}i_l=n}}
     \Multinom{\Vect m}
     s_{m_1+\cdots+m_l,b}
     x(i_1)^{m_1}\cdots x(i_l)^{m_l}\\
   &\hspace{3em}+\sum_{k=0}^\infty
     \sum_{\Vect m\in\Mfinnz{\Vect E}^l}
     \sum_{\Biind{0<i_1<\cdots<i_k\in\Nat}
     {\Mscard{m_1}i_1+\cdots+\Mscard{m_k}i_k=n}}
     \Multinom{\Vect m}
     s_{m_1+\cdots+m_k,b} 
     x(i_1)^{m_1}\cdots x(i_k)^{m_k}\\
   &=\Fun{\Tayfun(s)}(\Vect x)_{n,b} \quad
     \qedhere
\end{align*}
\end{proof}

\section{Examples of representably analytic categories with partial sums}
\label{sec:elementary-Taylor-examples}

We have shown that Taylor expansion in $\REL$ and $\WREL$ 
boils down to a representably analytic structure, suggesting that the theory of Taylor 
expansion developed in this article is a generalization of the existing one 
and provides structural insights on Taylor expansion.
We now exhibit a 
similar structure in multiple models of LL which are not even 
differential categories, suggestion that this new theory of Taylor expansion is much 
more general than the existing one.

\begin{remark}
  All of the examples considered are \emph{web based} models of LL, and 
  can be considered seen as refinements of $\REL$ and $\WREL$ by adding coherence (or quantitative)
  constraints.
  As such, their analytic structure is the same as the one in $\REL$ and $\WREL$.
  For now, we do not know examples of analytic category which are not of this kind.
  This requires further investigations.
  A good candidate would be the category of inegrable cones, recently introduced
  in \cite{Ehrhard22-cones}.
\end{remark}

\subsection{Weak coherence spaces}
\label{sec:wcs}

Our first example is based on a notion of coherence space introduced
in~\cite{Lamarche95} and which also arise naturally in the setting of
Indexed Linear Logic in~\cite{BucciarelliEhrhard01}.
The nice feature of this model, from the point of view of coherent
differentiation, is that it has a non-trivial and very simple notion
of summability.

\begin{definition}
  A \emph{weak coherence space} (WCS) is a structure %
  \(E=(\Web E,\Scoh E)\) where \(\Web E\) is a set and \(\Scoh E\) is
  a binary symmetric relation on \(\Web E\).
\end{definition}

The main difference with respect to the ordinary coherence spaces of 
\cref{sec:coh} is that in a WCS two elements of $\Web E$ 
can only be strictly coherent or strictly incoherent, 
whereas in ordinary coherence spaces there is a third 
possibility which occurs when the two points are equal. 
This ``ternary'' coherence setting is generalized in the non-uniform 
coherence spaces of \cref{sec:nucs}.

\begin{definition}
  A \emph{clique} of a weak coherence space \(E\) is a subset \(x\)
  of \(\Web E\) such that \(\forall a,a'\in X\ a\Scoh Ea'\).
  We use \(\Cl E\) for the set of all cliques of \(E\).
\end{definition}

\begin{definition}
  If \(E\) and \(F\) are WCS, we define a WCS \(E\Limpl F\) by %
  \(\Web{E\Limpl F}=\Web E\times\Web F\) and %
  \((a,b)\Scoh{E\Limpl F}(a',b')\) if \(a\Scoh Ea'\Implies b\Scoh Fb'\).
\end{definition}

Obviously \(\Id_E=\{(a,a)\St a\in\Web E\}\in\Cl{E\Limpl E}\) and if %
\(s\in\Cl{E\Limpl F}\) and \(t\in\Cl{F\Limpl G}\) then the relational
composition $t\Compl s$ 
belongs to \(\Cl{E\Limpl G}\).

\begin{definition}
  The category \(\WCS\) has the WCS as objects, and %
  \(\WCS(E,F)=\Cl{E\Limpl F}\), with identities defined as diagonal
  relations and composition as relational composition.
\end{definition}

\begin{definition}
  The \emph{dual} of \(E\) is \(\Orth E\), the WCS whose web is
  \(\Web E\) and \(a\Scoh{\Orth E}a'\) if \(\neg(a\Scoh Ea')\)
  (and then we also write \(a\Sincoh Ea'\)).
\end{definition}

The WCS \(\Sone\) is defined by \(\Web\Sone=\{\Sonelem\}\) and
\(\Sonelem\Scoh\Sone\Sonelem\).
So that \(\Cl\Sone=\{\emptyset,\{\Sonelem\}\}\).
Then one define \(\Sbot=\Orth\Sone\) so that
\(\Web\Sbot=\{\Sonelem\}\), with \(\Sonelem\Sincoh\Sbot\Sonelem\) so
that \(\Cl\Sbot=\{\emptyset\}\).
Observe that weak coherence spaces are not \emph{isomix}, the objects 
$\Sbot$ and $\Sone$ are not isomorphic. This contrasts with 
the coherence spaces of \cref{sec:coh} and the 
non uniform coherence spaces of \cref{sec:nucs}.

If \(E_1\) and \(E_1\) are WCS, we set
\(E_1\Times E_2=\Orth{(E_1\Limpl{\Orth{E_2}})}\), that is %
\((a_1,a_2)\Scoh{E_1\Times E_2}(a'_1,a'_2)\) if %
\(a_i\Scoh{E_i}a'_i\) for \(i=1,2\).

\begin{lemma}
  If \((s_i\in\WCS(E_i,F_i))_{i=1,2}\) then %
  \(s_1\Times s_2\subseteq\Web{E_1\Times E_2}\times\Web{F_1\Times
    F_2}\) defined by %
  \(s_1\Times s_2
  =\{((a_1,a_2),(b_1,b_2))\St (a_i,b_i)\in s_i\text{ for }i=1,2\}\)
  belongs to \(\WCS(E_1\Times E_2,F_1\Times F_2)\).
  The operation \(\Times\) defined in that way is a functor %
  \(\WCS^2\to\WCS\).
\end{lemma}
The proof is straightforward.
This bifunctor, together with its neutral element \(\Sone\), turns
\(\WCS\) into an symmetric monoidal category, by taking the same coherence isomorphisms
as the ones in $\REL$.
This symmetric monoidal category is closed, with internal hom of \(E\) and \(F\) the pair %
\((E\Limpl F,\Evlin)\) taken as in $\REL$, that is, %
\[ \Evlin=\{(((a,b),a),b)\St a\in\Web E\text{ and }b\in\Web F\} \, . \]
Let us check that \(\Evlin\in\WCS((E\Limpl F)\Times E,F)\), so let
\(a,a'\in\Web E\) and \(b,b'\in\Web F\) with %
\(((a,b),a)\Scoh{(E\Limpl F)\Times E}((a',b'),a')\).
This implies \(a\Scoh Ea'\) and also \((a,b)\Scoh{E\Limpl F}(a',b')\).
Therefore, we have \(b\Scoh Fb'\) as required.
The transpose of \(s\in\WCS(G\Times E,F)\) is the same as the one in $\REL$,%
\[ \Curlin(s)=\{(c,(a,b))\St ((c,a),b)\in s\}\in\WCS(G,E\Limpl F) .\]

It follows that the symmetric monoidal closed category \(\WCS\) is \Staraut{},
with \(\Sbot\) as dualizing object, because $\REL$ is.
The isomorphism \(\Curlin(\Evlin\Compl\Sym_{E\Limpl\Sbot,E})
\in\cL(E,(E\Limpl\Sbot)\Limpl\Sbot)\) 
can also be expressed more simply by observing that there is
a simple iso in %
\(\WCS(E\Limpl\Sbot,\Orth E)\), namely the relation %
\(\{((a,\Sonelem),a)\St a\in\Web E\}\) which is actually a bijection from %
\(\Web{E\Limpl\Sbot}\) to \(\Web{\Orth E}\).

The category \(\WCS\) is cartesian.
If \((E_i)_{i\in I}\) is a family of such objects, then we define %
\(\withFam E_i\)  as the WCS whose web is %
\(\Web{\withFam E_i}=\Union_{i\in I}\{i\}\times\Web{E_i}\) %
and whose coherence relation is defined by saying that %
\((i,a)\Scoh{\withFam E_i}(i',a')\) if %
\(i=i'\Implies a\Scoh{E_i}a'\).
Indeed, the projections in $\REL$ %
\[ (\Wproj_i=\{((i,a),a)\St i\in I\text{ and }a\in\Web{E_i}\})_{i\in I}\] %
satisfy \(\Wproj_i\in\WCS(\withFam[j] E_j,E_i)\).
The tupling is the same as the tupling in $\REL$: for any family %
\((s_i\in\WCS(F,E_i))_{i\in I}\) there is exactly one morphism %
\(\Tuple{s_i}_{i\in I}\in\WCS(F,\withFam E_i)\) such that %
\(\forall j\in I\ \Wproj_j\Compl\Tuple{s_i}_{i\in I}=s_j\), namely
\begin{align*}
  \Tuple{s_i}_{i\in I}=\{(b,(i,a))\St i\in I\text{ and }(b,a)\in s_i\}
\end{align*}
so the \(\Wproj_i\)'s are the projections of this cartesian product.

By \Starauty,
the category \(\WCS\) is also cocartesian with
coproduct %
\(\Bplus_{i\in I}E_i=\Orth{(\withFam \Orth{E_i})}\) whose web
is %
\(\Web{\Bplus_{i\in I}E_i}=\Union_{i\in I}\{i\}\times\Web{E_i}\) and
coherence is given by %
\((i,a)\Scoh{\Bplus_{j\in I}E_j}(i',a')\) if \(i=i'\) and
\(a\Scoh{E_i}a'\).
The corresponding injections are %
\((\Pin_i=\{(a,(i,a))\St i\in I\text{ and }a\in\Web{E_i}\}
\in\WCS(E_i,\Bplus_{j\in J}E_j))_{i\in I}\).

In the special case where \(I=\emptyset\), the product of the empty
family is the terminal object \(\Stop=(\emptyset,\emptyset)\) and the
coproduct of the empty family is the initial object
\(\Szero=(\emptyset,\emptyset)=\Stop\).

\begin{theorem}\label{th:wcs-lafont-cat}
  The symmetric monoidal category \(\WCS\) is a Lafont category.
\end{theorem}

Let us describe the free ressource modality. This modality is 
similar to the one of $\REL$.
The object $\oc E$
is defined as \(\Web{\Oc E}=\Mfin{\Web E}\) 
and \(m\Scoh{\Oc E}m'\) if \(a\Scoh Ea'\) for
all \(a\in m\) and \(a'\in m'\).
If \(s\in\WCS(E,F)\) then
\begin{align*}
  \Oc s=\{(\Mset{\List a1n},\Mset{\List b1n})
  \St n\in\Nat\text{ and }((a_i,b_i)\in s)_{i=1}^n\}
\end{align*}
and one checks easily that \(\Oc s\in\WCS(\Oc E,\Oc F)\):
let \((m,p),(m',p')\in\Oc s\), we must prove that %
\((m,p)\Scoh{\Oc E\Limpl\Oc F}(m',p')\), so assume that %
\(m\Scoh{\Oc E}m'\) and let us prove that \(p\Scoh{\Oc F}p'\).
Let \(b\in p\) and \(b'\in p'\).
There are \(a\in m\) and \(a'\in m'\) such that \((a,b),(a',b')\in s\).
We have \(a\Scoh Ea'\) and \((a,b)\Scoh{E\Limpl F}(a',b')\) %
and hence \(b\Scoh Fb'\).

The counit and comultiplication of the comonad are the same as in $\REL$,
\begin{align*}
  \Deru_E&=\{(\Mset a,a)\St a\in\Web E\}\\
  \Digg_E&=\{(m_1+\cdots+m_n,\Mset{\List m1n})
           \St n\in\Nat\text{ and }(m_i\in\Mfin{\Web E})_{i=1}^n\}\,.
\end{align*}
It is obvious that \(\Deru_E\in\cL(\Oc E,E)\), let us check that %
\(\Digg_E\in\cL(\Oc E,\Occ E)\) so let %
\((m,M),(m',M')\in\Digg_E\) and assume that \(m\Scoh{\Oc E}m'\), we
must prove that \(M\Scoh{\Oc E}M'\).
Let \(p\in M\) and \(p'\in M'\) so that \(m=p+m_1\) and \(m'=p'+m'_1\)
for some multisets \(m_1\) and \(m'_1\).
Since \(m\Scoh{\Oc E}m'\) we have \(p\Scoh{\Oc E}p'\).

The Seely isomorphisms are also the same as the ones of $\REL$, %
\(\Seelyz=\{(\Sonelem,\Msetempty)\}\) and %
\[ \Seelyt_{E_1,E_2} =\{((m_1,m_2),\Msetact 1{m_1}+\Msetact 2{m_2})\St
(m_1,m_2) \in\Web{\Oc E_1}\times\Web{\Oc E_1}\}\] %
where \(\Msetact i{\Mset{\List a1n}}=\Mset{(i,a_1),\dots,(i,a_n)}\).
We have \(\Seelyz\in\cL(\Sone,\Oc\Stop)\) because
\(\Msetempty\Scoh{\Oc\Stop}\Msetempty\), and \(\Seelyz\) is an iso
because \(\Sonelem\Scoh{\Sone}\Sonelem\).
The fact that the relation \(\Seelyt_{E_1,E_2}\) is a bijection is
obvious, and it is easy to check that it is an isomorphism in %
\(\WCS(\Oc{E_1}\Times\Oc{E_2},\Oc{(E_1\With E_2)})\).
Assume for instance that %
\((m_1,m_2)\Scoh{\Oc{E_1}\Times\Oc{E_2}}(m'_1,m'_2)\) and let us prove
that
\[ p=\Msetact 1{m_1}+\Msetact 2{m_2} \Scoh{\Oc{(E_1\With
    E_2)}}p'=\Msetact 1{m'_1}+\Msetact 2{m'_2}\] so let %
\((i,c)\in p\) and \((i',c')\in p'\), one must check that %
\((i,c)\Scoh{E_1\With E_2}(i',c')\), so assume that \(i=i'\).
Then, by definition of \(p\) and \(p'\) we have \(c\in m_i\) and
\(c'\in m'_i\) and hence \(c\Scoh{E_i}c'\).

The induced symmetric lax monoidal structure \((\Ocmonz,\Ocmont)\) coincides 
with the one of $\REL$.

\begin{remark}\label{rk:por-in-wcs}
  Although formally similar to Girard's coherence
  spaces~\cite{Girard87}, WCS have quite different properties and are
  closer to Scott semantics
  based on cpo's and continuous functions, see \cite{Scott76,Plotkin77},
  than to Berry stable semantics, see \cite{Berry78}.
  As an example, representing the type of booleans by %
  \(\Plus{\Sone}{\Sone}\), whose web is $\{\True, \False\}$ with 
  \(\True=(1,\Sonelem)\) and \(\False=(2,\Sonelem)\) and whose coherence 
  relation is given by $\True \Scoh{\Plus{\Sone}{\Sone}} \True$,
  $\False \Scoh{\Plus{\Sone}{\Sone}} \False$, and 
  $\True \Sincoh{\Plus{\Sone}{\Sone}} \False$,
  we can define a ``parallel or'' morphism
  \begin{align*}
    \mathsf{por}=\{
    ((\Mset\True,\Msetempty),\True),
    ((\Msetempty,\Mset\True),\True),
    ((\Mset\False,\Mset\False),\False)\}
    \in\WCS(
    \Oc{(\Plus{\Sone}{\Sone})}\Times\Oc{(\Plus{\Sone}{\Sone})},
    \Plus{\Sone}{\Sone}))
  \end{align*}
  which is a clique because \(\True\Scoh{\Plus{\Sone}{\Sone}}\True\).
  So WCS are compatible with this form of non-determinism (there is no
  deterministic implementation of this \(\mathsf{por}\) morphism), and
  nevertheless implement a non-trivial form of coherence since for
  instance
  \(\{\True,\False\}\notin\Cl{\Plus{\Sone}{\Sone}}
  =\{\emptyset,\{\True\},\{\False\}\}\).
\end{remark}

The category \(\WCS\) has zero morphisms: take
\(0_{E,F}=\emptyset\in\WCS(E,F)\). This morphism 
is also absorbing for the monoidal product.
The object \(\Dbimon=\withFam<\N> \Sone\) can be described as
follows (up to trivial iso): \(\Web\Dbimon=\Nat\) and
\(\forall i,j\in\Nat\ i\Scoh\Dbimon j\).
The injections \(\Win_i\in\WCS(\Sone,\Dbimon)\) are easy to describe:
\(\Win_i=\{(\Sonelem,i)\}\).

Notice that the resource modality is easily checked to be
finitary in the sense of \Cref{def:resource-finitary} so that we know
that \(\Dbimondeg\in\WCS(\Oc{\Sone},\Dbimon)\) is an isomorphism by
\Cref{th:dbimondeg-iso}.
This can also be checked directly:
\(\Web{\Oc\Sone}=\{k\Mset{\Sonelem}\St k\in\Nat\}\simeq\Nat\) and we
have \(k\Mset\Sonelem\Scoh{\Oc\Sone}k'\Mset\Sonelem\) for all
\(k,k'\in\Nat\) since \(\Sonelem\Scoh{\Sone}\Sonelem\).

\begin{theorem}
  \(\WCS\) is a representably $\Sigma$-additive category.
  A family $\family{s_a \in \WCS(E, F)}$ is summable 
  if $\bigcup_{a \in A} s_a \in \Cl{E \linarrow F}$, and then
  $\sum_{a \in A} s_a \defEq \bigcup_{a \in A} s_a$.
\end{theorem}

\begin{proof} It is very direct to check that 
  the sum described above is a $\Sigma$-monoid, using the fact that the 
  subset of a clique is also a clique. 
  Then $\WCS$ satisfies \ref{ax:RS-mon} because the $0$ of the $\Sigma$-additive 
  structure described above is exactly the zero-morphism $0 = \emptyset$. 
  Since \(\WCS\) is a symmetric monoidal closed category, 
  \ref{ax:RS-epi} amounts by
  \Cref{rk:elementary-closed-win-epicity} to saying that the
  \(\Win_i\)'s are jointly epic which is obvious since, given
  \(s\in\WCS(\Dbimon,E)\) we clearly have
  \(s=\{(i,a)\St s\Compl\Win_i\not=\emptyset\}\).

  We show that $\WCS$ satisfies \ref{ax:RS-sum}.
  Assume that \(\Vect s = \sequence{s_i \in \WCS(E, F)}\) is such that there exists %
  \(t\in\WCS(E\Times\Dbimon,F)\) such that %
  \[ s_i= t \compl (E \tensor \Sprojl_i) \compl \tensorUnitR^{-1}
  = \{(a,b)\in\Web{E\Limpl F}\St ((a,i),b)\in t\} \, . \]
  If this is the case then, given \((a,b)\in s_i\) and
  \((a',b')\in s_{i'}\), if \(a\Scoh Ea'\) then
  \((a,i)\Scoh{E\Times\Dbimon}(a',i')\) and hence \(b\Scoh Fb'\).
  It follows that \(\Union_{i\in\Nat}s_i\in\WCS(E,F)\) and we have 
  \[ ( t\compl (E \tensor \Wdiag) \compl \tensorUnitR^{-1})
  =\{(a,b)\St\exists i\ ((a,i),b)\in
  t\}=\Union_{i\in\Nat}s_i . \]
  Conversely, if \(\Union_{i\in\Nat}s_i\in\WCS(E,F)\) then %
  \[ t=\{((a,i),b)\St i\in\Nat\text{ and }(a,b)\in s_i\}
  \in\WCS(E\Times\Dbimon,F) . \] 
  Indeed, if $((a, i), b), ((a', i'), b') \in t$ and 
  $(a, i) \Scoh{E \tensor \Dbimon} (a', i')$ then $a \Scoh E a'$ and 
  $(a, b), (a', b') \in \Union_{i \in \Nat} s_i$ so $b \Scoh F b'$.
  And then, we have 
  \(s_i=t \compl (E \tensor \Sprojl_i) \compl \tensorUnitR^{-1} \) for each
  \(i\in\Nat\).
  Thus, $\WCS$ satisfies \ref{ax:RS-sum}. 
  
  Finally, we show that $\WCS$ satisfies \ref{ax:RS-witness}.
  By \cref{prop:reindexing}, it suffices to show that 
  \ref{ax:RS-witness} holds 
  on $\N$-indexed families.
  Let $\vect t = \sequence{t_i \in \WCS(E \tensor \Dbimon, F)}$ 
  be a $\N$-indexed family such that the family %
  \(\family<\N^2>[(i, j)]{t_i \Compl(X \tensor \Sprojl_j) \compl \tensorUnitR^{-1}}\) 
  is summable, that is
  \begin{align*}
    t=\Union_{i, j \in\Nat}t_i \Compl(E \tensor \Sprojl_j) \compl \tensorUnitR^{-1}
     =\{(a,b)\St\exists
    i, j\in\Nat\ ((a, j),b) \in t_i\}\in\WCS(E, F)\,,
  \end{align*}
  we contend that \(\Vect t\) is summable.
  This amounts to proving that \(u\in\WCS(E \Times \Dbimon\Times\Dbimon, F)\)
  where %
  \(u=\{((a, j,i),b)\St ((a, j),b)\in t_i\}\) since then we will have %
  \[ t_i=u\Compl(E \Times \Dbimon\Times\Win_i)\Compl\Inv\Rightu \]
  for all $i \in \N$.
  If \(((a, i,j), b),((a', i',j'),b')\in u\) 
  and $(a, i, j) \Scoh{E \tensor \Dbimon \tensor \Dbimon} 
  (a', i', j')$ then $a \Scoh E a'$, and we have \(((a, j),b)\in t_i\) and
  \(((a', j'),b')\in t_{i'}\) and hence \((a, b),(a', b')\in t\)
  so that \(b\Scoh Fb'\).
  It follows that \(u\in\WCS(E \tensor \Dbimon\Times\Dbimon, F)\).
\end{proof}

We describe the associated functor \(\Sfun=(\Dbimon\Limpl\_):\WCS\to\WCS\).
First, \(\Sfun E\) has \(\Nat\times\Web E\) as web, and
\((i,a)\Scoh{\Sfun E}(i',a')\) iff \(a\Scoh Ea'\).
It follows that we have the following order isomorphism
\begin{align*}
  \Cl{\Sfun E}\simeq
  \{\Vect x\in\prod_{i\in\Nat}\Cl{E_i}\St\Union_{i\in\Nat}x_i\in\Cl E\} 
\end{align*}
and we have \(\Sproj_i=\{((i,a),a)\St a\in\Web E\}\) and
\(\Ssum=\{((i,a),a)\St i\in\Nat\text{ and }a\in\Web E\}\).
Given \(s\in\WCS(E,F)\), we have %
\(\Sfun(s)=\{((i,a),(i,b))\St i\in\Nat\text{ and }(a,b)\in s\}\).

\begin{remark}
  The action of $\S$ on morphisms is the same as the functor 
  $\S$ on $\REL$. The key difference lies in the action on 
  object, which is neither the product, $\S E \neq  \withFam<\N> E$, nor 
  the coproduct, $\S E \neq \Bplus_{i \in \N} E$. 
  Recall that the product $\withFam<\N> E$ and the coproduct $\Bplus_{i \in \N} E$
  have the same web, but not the same coherence relation, we have 
  $(i, a) \Scoh{\withFam<\N> E} (j, b)$ if $i = j \imply a \Scoh E b$, and 
  $(i, a) \Scoh{\Bplus_{i \in \N} E} (j, b)$ if $i = j$ and $a \Scoh E b$, 
  Observe that $\Sfun E$ sits between the coproduct and the product, in the sense that 
  \[ (i, a) \Scoh{\Bplus_{i \in \N} E} (j, b) \imply (i,a)\Scoh{\Sfun E}(i',a') 
  \imply (i, a) \Scoh{\withFam<\N> E} (j, b) \, .\] 
  so the graphs of the identity maps provide a morphism in
  $\WCS(\Bplus_{i \in \N} E, \S E)$ and a morphism in $\WCS(\S E, \withFam<\N> E)$
\end{remark}

The functor \(\Sfun\) has a bimonad structure induced by the
bimonoid structure of \(\Dbimon\) as shown in \Cref{fig:bimonoid-to-bimonad}.
Both the bimonoid $\Dbimon$ and the bimonad $\S$ are defined as in $\REL$, 
because the cartesian symmetric monoidal structure 
of $\WCS$ and $\REL$ are the same, and the sums in $\WCS$ are equal (when defined)
to the sum in $\REL$.

By \Cref{th:Lafont-elem-summable-Taylor}, \(\Dbimon\) has an analytic coalgebra, 
that is, a \(\Oc\)-coalgebra structure
\(\Dbimonca\in\WCS(\Dbimon,\Oc\Dbimon)\).
This analytic coalgebra coincides with
the analytic coalgebra of $\REL$.
As seen in \cref{cor:Taylor-structure-induces-differential},
this coalgebra structure induces the distributive law %
\(\Sdl_E=\Curlin(\Sdl'_E)\in\WCS(\Oc\Sfun E,\Sfun\Oc E)\), which in 
turn induces the functor \(\Tayfun:\Kloc\WCS\to\Kloc\WCS\) 
that extends \(\Sfun\) to the coKleisli category
\(\Kloc\WCS\). Again, they are defined as in $\REL$.
So we just proved the following result.

\begin{theorem} $\WCS$ is a Lafont representably analytic category. The
  induced analytic structure is the same as in $\REL$.
\end{theorem}

\subsection{Girard's coherence spaces}
\label{sec:coh}
Just as in~\cite{Ehrhard23-cohdiff}, one can show that the usual
Girard's coherence spaces (CS) have an analytic coalgebra.
Remember that such a coherence space is a pair
\(E=(\Web E,\mathord{\Coh E})\) where \(\Web E\) is a set (the web)
and \(\Coh E\) is a binary, reflexive and symmetric relation on
\(\Web E\).
A clique of \(E\) is a subset \(x\) of \(\Web E\) such that
\(\forall a,a'\in x\ a\Coh Ea'\).
Given CS \(E\) and \(F\), one defines a CS \(E\Limpl F\) by
\(\Web{E\Limpl F}\) and \((a,b)\Coh{E\Limpl F}(a',b')\) if
\(a\Coh Ea'\Implies(b\Coh Fb'\text{ and }b=b'\Implies a=a')\), and the
category \(\COH\) has CS as objects, and \(\COH(E,F)=\Cl{E\Limpl F}\),
identity morphisms and composition being defined as in \(\REL\).

The category \(\COH\) is doubtlessly 
one of the most popular model of LL, and is a Lafont category with \(\Oc E\) defined as follows:
\(\Web{\Oc E}=\{\Mset{\List a1n}
\St n\in\Nat\text{ and }\{\List a1n\}\in\Cl E\}\)\footnote{There 
is also an exponential whose webs are sets instead of multisets, but 
this one is not free.}, 
see~\cite{MelliesTabreauTasson09}.
This is a major difference between all the other models presented in
this section, where \(\Web{\Oc E}=\Mfin{\Web E}\): one often says that
the CS exponential is \emph{uniform}, whereas the exponentials of the
other models are \emph{non-uniform}.
As far as we know, it is not possible to equip the category \(\COH\)
with a non-uniform exponential.
One has also to be careful with the definition of the action of this
functor on morphisms: given \(s\in\COH(E,F)\), one takes
\(\Oc s=\{(\Mset{\List a1k},\Mset{\List b1k})\St k\in\Nat\text{,
}\{\List a1k\}\in\Cl E\text{ and }(a_i,b_i)\in s\text{ for all }i\}\).

The object of degrees \(\Dbimon\) satisfies \(\Web\Dbimon=\Nat\) and
\(\forall i,j\in\Nat\ i\Coh\Dbimon j\).
Therefore, \(\Sfun E\) satisfies \(\Web{\Sfun E}=\Nat\times\Web E\) %
with \((i,a)\Coh{\Sfun E}(i',a')\) if \(a\Coh Ea'\) and
\(i\not=i'\Implies a\not=a'\).
It follows that %
\[ \Cl{\Sfun E}\Isom\{\Vect x\in\Cl E^\Nat\St
\Union_{i\in\Nat}x_i\in\Cl E\text{ and }i\not=j\Implies x_i\cap
x_j=\emptyset\} . \]

\begin{theorem}\label{th:coh-lafont-elem-sumable}
  \(\COH\) is a Lafont representably \(\Sigma\)-summable category, 
  thus it is a representably analytic category.
  A family $\family{s_a \in \COH(E, F)}$ is summable if the $s_a$ are pairwise disjoint, 
  and if their union is a clique. Their sum is given by the union.
\end{theorem}
\begin{proof} The proof is similar to the proof in \cref{sec:wcs} for $\WCS$.
\end{proof}

The canonical analytic coalgebra is
\[ \Dbimonca=\{(n,\Mset{\List i1k})\in\Nat\times\Mfin\Nat
\St k\in\Nat\text{ and }i_1+\cdots+i_k=n\} _, .\]
This is the same as the one in $\REL$ and $\WCS$, because 
$\Web{\Oc \Dbimon} = \Mfin{\Dbimon}$ (this is a consequence of the fact that 
$i \Coh \Dbimon j$ for all $i, j \in \Web{\Dbimon}$).
This is not the case of the induced Taylor functor 
\(\Tayfun:\Kloc\COH\to\Kloc\COH\) that maps a CS
\(E\) to \(\Sfun E\), and if \(s\in\Kloc\COH(E,F)\) then
\(\Tayfun(s)\in\Kloc\COH(\Sfun E,\Sfun F)\) is given by
\begin{multline*}
  \Tayfun(s)
  =\{(\Mset{(i_1,a_1),\dots,(i_k,a_k)},(n,b))
  \St k\in\Nat
  \text{, }\\
  \{(i_1,a_1),\dots,(i_k,a_k)\}\in\Cl{\Sfun E}
  \text{, }i_1+\cdots+i_k=n
  \text{ and }(\Mset{\List a1k},b)\in s\}\,.
\end{multline*}

So if \(s=\{(\Mset{a,a},b)\}\) is a simple ``quadratic'' morphism, for
having \(\{(i_1,a),(i_2,a)\}\in\Sfun E\), we need \(i_1=i_2\).
It follows that
\begin{align*}
  \Tayfun(s)=
  \{(\Mset{(i,a),(i,a)},(2i,b))\St i\in\Nat\}
\end{align*}
wherease in $\WCS$ and $\REL$ we had
\(\Tayfun(s)
=\{(\Mset{(i_1,a),(i_2,a)},(i_1+i_2,b))\St i_1,i_2\in\Nat\}\).
The precise meaning of this difference between the actions of the
\(\Tayfun\) functor in the uniform setting of coherence spaces and in
the non-uniform one of WCS has still to be fully understood.

\subsection{Nonuniform coherence spaces}

\label{sec:nucs}

Formally, nonuniform coherence spaces (NUCS) can be considered as a
refinement of WCS, but they have quite different properties, being
much closer to Girard's coherence spaces and to the stable semantics.
In particular the relation \(\mathsf{por}\) of \Cref{rk:por-in-wcs} is
rejected by NUCS.
We refer to~\cite{Ehrhard23-cohdiff}, Section~6.1 for a detailed presentation
of NUCS, we just recall the basic definitions.
\begin{definition} (\cite{BucciarelliEhrhard01})
  A non-uniform coherence space (NUCS) is a tuple \(E=(\Web E,\Scoh E,\Sincoh E)\) where
  \begin{enumerate}
  \item \(\Web E\) is a set called the web of \(E\);
  \item \(\Scoh E\) (strict coherence) and
  \(\Sincoh E\) (strict incoherence) are disjoint binary symmetric
  relations on \(\Web E\).
  \end{enumerate}
  The relation
  \(\mathord{\Neu E} =\Web E^2\setminus(\mathord{\Scoh
    E}\cup\mathord{\Sincoh E})\) (which is also symmetric) is called
  neutrality and the large coherence and incoherence relations are
  defined as %
  \(\mathord{\Coh E}=\mathord{\Scoh E}\cup\mathord{\Neu E}\) and
  \(\mathord{\Incoh E}=\mathord{\Sincoh E}\cup\mathord{\Neu E}\).
  A clique of a NUCS \(E\) is a subset \(x\) of \(\Web E\) such that %
  \(\forall a,a'\in x\ a\Coh Ea'\), and we use \(\Cl E\) for the set of
  all cliques of \(E\).

  The dual of a NUCS \(E\) is %
  \(\Orth E=(\Web E,\mathord{\Sincoh E},\mathord{\Scoh E})\) and one
  defines \(E\Limpl F\) by stipulating that
  \(\Web{E\Limpl F}=\Web E\times\Web F\) and by providing the large
  coherence relation and the neutrality: %
  \((a,b)\Neu{E\Limpl F}(a',b')\) if \(a\Neu Ea'\) and \(b\Neu Fb'\), %
  and \((a,b)\Coh{E\Limpl F}(a',b')\) if
  \begin{align*}
    a\Coh Ea'\Implies (b\Coh Fb' \text{ and }b\Neu Eb'\Implies a\Neu Ea')\,.
  \end{align*}
\end{definition}

Then the category \(\NUCS\) has the NUCS as objects and %
\(\NUCS(E,F)=\Cl{E\Limpl F}\), identity morphisms and composition
being defined as in the category \(\REL\).
The definition of the symmetric monoidal closed structure of \(\NUCS\) is completely similar
to that of \(\REL\) as well as the proof that the category \(\NUCS\)
(with dualizing object
\(\Sbot= \Sone= (\Sonelem,\emptyset,\emptyset)\)) is \Staraut{} and
cartesian.
Notice here that there is an important difference between \(\NUCS\) and
\(\WCS\): $\NUCS$ is \emph{isomix} (see \cite{Cockett97}), meaning that $\Sone \simeq \Sbot$,
whereas $\WCS$ is not.
Explicitly, the cartesian product in $\NUCS$ is given by the web
$\Web{\withFam E_i} = \Union_{i \in I} \{i\} \times \Web{E_i}$,
the strict coherence 
$(i, a) \Scoh{\withFam X_i} (i', a')$ if $i = i' \Rightarrow a \Scoh{X_i} a'$ and 
the neutrality
$(i, a) \Neu{\withFam X_i} (i', a')$ if $i = i'$ and $a \Neu{X_i} a'$.

As shown in~\cite{Boudes11}, the symetric monoidal category \(\NUCS\) is a Lafont category, the
induced resource modality is %
\((\Ocb,\Deru,\Digg,\Seelyz,\Seelyt)\)
where %
\(\Web{\Ocb E}=\Mfin{\Web E}\) and %
\(m\Coh{\Ocb E}m'\) if %
\(\forall a\in m\,\forall a'\in m'\ a\Coh Ea'\) and %
\(m\Neu{\Ocb E}m'\) if \(m\Coh Em'\) and %
\(m=\Mset{\List a1n}\) and \(m'=\Mset{\List{a'}1n}\) with %
\(a_i\Neu Ea'_i\) for \(i=1,\dots,n\). The morphisms
\(\Deru\), \(\Digg\), \(\Seelyz\) and \(\Seelyt\) are defined as in
$\REL$ and \(\WCS\).

\begin{theorem}\label{th:nucs-lafont-elem-sumable}
  \(\NUCS\) is a Lafont representably \(\Sigma\)-additive category, 
  thus it is a representably analytic category.
  A family $\family<I>[i]{s_i \in \NUCS(E, F)}$ is summable if 
  their union is a clique, and if for all $i \neq j$, if
   $(a, b) \in s_i$ and $(a', b') \in s_j$, then 
   $(a, b) \Scoh{E \linarrow F} (a', b')$.
   The sum is defined as the union.
\end{theorem}

\begin{proof} The proof is similar to the proof in \cref{sec:wcs} for $\WCS$.
\end{proof}
The object $\Dbimon$ has web $\Web{\Dbimon} = \N$, coherence 
$i \Coh{\Dbimon} j$ for all $i, j \in \N$, and neutrality $i \Neu{\Dbimon} j$ 
if $i = j$.
The induced functor \(\Sfun:\NUCS\to\NUCS\) is such that %
\(\Web{\Sfun E}=\Nat\times\Web E\) and %
\((i,a)\Scoh{\Sfun E}(i',a')\) if \(a\Scoh Ea'\) and %
\((i,a)\Neu{\Sfun E}(i',a')\) if \(i=i'\) and \(a\Neu Ea'\). 
This functor acts on morphisms exactly as in the setting of $\REL$ and $\WCS$, 
but observe that again $\S E \neq \withFam<\N> E$ and
$\S E \neq \Bplus_{i \in \N}$.
The analytic coalgebra and the analytic structure that follows are
the same as the ones in \(\WCS\) and $\REL$.

\begin{remark}
  For instance, for a family \((x_i\in\Cl{\Sone})_{i\in\Nat}\) to be
  summable, we need all the \(x_i\) to be empty but possibly for one (which
  is then equal to \(\{\Sonelem\}\)).
  This is very similar to the uniformity properties observed in the syntactical
  Taylor expansion.
\end{remark}

One interesting feature of \(\NUCS\) is that it admits another
resource modality \(\Ocbe\) whose structure morphisms \(\Deru\),
\(\Digg\), \(\Seelyz\) and \(\Seelyt\) are, again, defined as in
\(\REL\).
This exponential was actually the first one discovered for NUCS
because it arises naturally in the setting of Indexed Linear Logic,
see~\cite{BucciarelliEhrhard01} where NUCS were introduced as a
particular example of denotational models based on phase semantics.
One has \(\Web{\Ocbe E}=\Mfin{\Web E}\) and, given %
\(m=\Mset{\List a1n},m'=\Mset{\List a{n+1}k}\), one has %
\(m\Coh{\Ocbe E}m'\) if %
\(\forall i,j\in\{1,\dots,k\}\ i\not=j\Implies a_i\Coh Ea_j\) and %
\(m\Scoh{\Ocbe E}m'\) if \(m\Coh{\Ocbe E}m'\) and %
\(\exists i\in\{1,\dots,k\}\,\forall j\in\{1,\dots,k\}
\ i\not=j\Implies a_i\Scoh Ea_j\).

Let us describe \(\Ocbe\Sone\): %
we have \(\Web{\Ocbe{\Sone}}=\{i\Mset\Sonelem\St i\in\Nat\}\simeq\Nat\).
Next \(0\Scoh{\Ocbe\Sone}1\) and \(i\Neu{\Ocbe\Sone}j\) as
soon as \(i+j\neq 1\).
But \(\Dbimon\) is characterized by \(\Web{\Dbimon}=\Nat\) and %
\(i\Neu\Dbimon i\) and \(i\Scoh\Dbimon j\) when \(i\not=j\), and
therefore \(\Dbimon\) and \(\Ocbe\Sone\) are not isomorphic, and since
\(\Ocbe\) is easily seen to be finitary, the only possibility is that
this resource modality \(\Ocbe\) has no analytic coalgebra
\(\Dbimonca\) and thus no analytic structure $\Sdl$.

\begin{remark}
  On the other hand, setting \(\Dbimon_2=\Sone\With\Sone\), it is not
  hard to check that we have a coalgebra structure
  \(\delta\in\NUCS(\Dbimon_2,\Ocbe\Dbimon_2)\) given by
  \begin{align*}
    \delta=\{(i,\Mset{\List i1k})\St i,\List i1k\in\{0,1\}
    \text{ and }i=i_1+\cdots+i_k\}
  \end{align*}
  so that \(\NUCS\), equipped with the \(\Ocbe\) resource modality, is
  a representable model of coherent differentiation in the sense
  of~\cite{Ehrhard23-cohdiff}. 
\end{remark}

We just proved the following results.
\begin{theorem} \label{thm:differential-not-analytic}
  There are representably coherent differential categories\footnote{Those are called 
  \emph{differential elementary summable categories} in \cite{Ehrhard23-cohdiff}.} which 
  are not representably analytic.
\end{theorem}

We conjecture on the other hand that 
  under mild assumption (the existence of a binary 
  summabibility structure), any analytic structure induces
  a coherent differential structure.
  This coherent differential structure should be associated in 
  the representable case with the following coalgebra 
  \[ \begin{tikzcd}
	{\Dbimon_2} & & \Dbimon & {\oc \Dbimon} & & {\oc \Dbimon_2}
	\arrow["{\prodtuple{\prodProj_0, \prodProj_1, 0, \ldots}}", from=1-1, to=1-3]
	\arrow["\Dbimonca", from=1-3, to=1-4]
	\arrow["{\oc \prodtuple{\prodProj_0, \prodProj_1}}", from=1-4, to=1-6]
\end{tikzcd} \]

\subsection{Probabilistic coherence spaces}
\label{sec:pcoh}

The last example is the category of probabilistic coherence 
spaces of~\cite{Danos11}. In some sense, $\PCOH$ is to $\WREL[\Rbar]$ what 
\(\WCS\), $\COH$ and \(\NUCS\) are for $\REL$.
It is a fully abstract model of a probabilistic PCF, see~\cite{Ehrhard14},
and the model in which coherent differentiation was first 
discovered, see~\cite{Ehrhard22a}.

Given an at most countable set $A$ and $u,u'\in\Realpcto A$, we set
$\Eval u{u'}=\sum_{a\in A}u_au'_a\in\Realpc$ where \(\Realpc\) is the
completed half real line. Given $P\subseteq\Realpcto A$, we define
$\Orth P\subseteq\Realpcto A$ as
\begin{align*}
  \Orth P=\{u'\in\Realpcto A\St\forall u\in P\ \Eval u{u'}\leq 1\}\,.
\end{align*}
Observe that if $P$ satisfies
\( \forall a\in A\,\exists x\in P\ x_a>0 \) and
\( \forall a\in A\,\exists m\in\Realp \forall x\in P\ x_a\leq m \)
then $\Orth P\in\Realpto I$ and $\Orth P$ satisfies the same two
properties that we call \emph{local boundedness} which can also be
rephrased as
\begin{align*}
  \forall a\in A\quad 0<\sup_{x\in P}x_a<\infty\,.
\end{align*}

\begin{definition}
A probabilistic pre-coherence space (pre-PCS) is a pair
$X=(\Web X,\Pcoh X)$ where $\Web X$ is an at most countable set and
$\Pcoh X\subseteq\Realpcto{\Web X}$ satisfies
$\Biorth{\Pcoh X}=\Pcoh X$.
A probabilistic coherence space (PCS) is a
pre-PCS $X$ such that \(\Pcoh X\) is locally bounded.
\end{definition}

\begin{example} \label{ex:subprobability-distributions}
  The pairs $\probapcs{A} = 
  (A, \Proba(A))$ where $\Proba(A) = \{x \in \Rbar^A \St \sum_{a \in A} 
  x_a \leq 1\}$ is the set of sub probability distributions on $A$ is a PCS, hence the 
  name \emph{probabilistic} coherence space. However, a PCS is not always a 
  set of subprobability 
  distributions, especially when going at higher order.
\end{example}

We can define the pointwise order on $\Rpos^A$ for any set $A$: 
for any $u, v \in \Rpos^{A}$, $u \leq v$ if for all $a \in A$, 
$u_a \leq v_a$. 
Observe that if $u \leq v$ then for all $u' \in \Rpos^A$, 
$\Eval u {u'} \leq \Eval v {u'}$, whence the following result.

\begin{lemma} \label{lemma:pcs-downward-closed}
For any $v \in \Pcoh X$ and $u \in \Rpos^{\Web X}$, if $u \leq v$ then 
$u \in  \Pcoh X$.
\end{lemma}

Given a PCS $X$ and \(x\in\Pcoh X\) we set
$\Norm x_X=\sup_{x'\in\Pcoh{\Orth X}}\Eval x{x'}\in\Intcc01$. This
operation obeys the usual properties of a norm: %
\(\Norm x=0\Implies x=0\), %
\(\Norm{x_0+x_1}\leq\Norm{x_0}+\Norm{x_1}\) and %
\(\Norm{\lambda x}=\lambda\Norm x\) for all \(\lambda\in\Intcc 01\).

We recall some notations from $\WREL[\Rbar]$.
Given $t\in\Realpcto{A\times B}$ considered as a matrix
and $u\in\Realpcto A$, we define
$\Matappa tu\in\Realpcto B$ by
$(\Matappa tu)_b=\sum_{a\in A}t_{a,b}u_a$ (usual formula for applying
a matrix to a vector), and if $s\in\Realpcto{B\times C}$ we define the
product $\Matapp st\in\Realpcto{A\times C}$ of the matrix $s$ and $t$
by $(\Matapp st)_{a,c}=\sum_{b\in B}t_{a,b}s_{b,c}$. This is an
associative operation.

Let $X$ and $Y$ be PCSs, a morphism from $X$ to $Y$ is a matrix
$t\in\Realpto{\Web X\times\Web Y}$ such that
$\forall x\in\Pcoh X\ \Matappa tx\in\Pcoh Y$. It is clear that the
identity (diagonal) matrix is a morphism from $X$ to $X$ and that the matrix
product of two morphisms is a morphism and therefore, PCSs equipped
with this notion of morphism form a category $\PCOH$.

The condition $t\in\PCOH(X,Y)$ is equivalent to
\[
\forall x\in\Pcoh X,\forall y'\in\Pcoh{\Orth Y}, \Eval{\Matappa tx}{y'}\leq 1
\]
and observe that $\Eval{\Matappa tx}{y'}=\Eval t{\Tens x{y'}}$ where
$(\Tens x{y'})_{(a,b)}=x_ay'_b$. We define
\[\Limplf XY=(\Web X\times\Web Y,\{t\in\Realpto{\Web{\Limplf
    XY}}\St\forall x\in\Pcoh X\ \Matappa tx\in\Pcoh Y\}) \] 
this is a pre-PCS by this observation, and checking that it is indeed a PCS is
easy.

\begin{example}
  For example, if $X = (E, \Proba(E))$ and $Y = (F, \Proba(F))$
  as given in \cref{ex:subprobability-distributions}, then $X \linarrow Y$
  is the set of sub-stochastic matrices. In particular, $X \linarrow Y$ is not 
  a set of sub-probability distributions.
\end{example}

We define then $\Tens XY=\Orth{\Limplfp X{\Orth Y}}$; %
this is a PCS which satisfies
\[ 
\Pcohp{\Tens XZ}=\Biorth{\{\Tens xz\St x\in\Pcoh X\text{ and }z\in\Pcoh Z\}}
\]
where $(\Tens xz)_{(a,c)}=x_az_c$.

\begin{lemma}[\cite{DanosEhrhard11}] \label{lemma:pcs-tensor-morphism}
For any PCS $X_1, X_2, Y$ and $s \in \Rpos^{\Web{X_1 \tensor X_2 \linarrow Y}}$,
if for all $u_1 \in \Pcoh X_1$, $u_2 \in \Pcoh X_2$,
$s \cdot (u_1 \tensor u_2) \in \Pcoh Y$, then 
$s \in \PCOH(X_1 \tensor X_2, Y)$. 
\end{lemma}
Then it is easy to see that we have equipped in that way the category
$\PCOH$ with a symmetric monoidal structure for which it is
\Staraut{} with the dualizing object
$\Sbot=\Sone=(\{*\},[0,1])$, which coincides with the unit of $\Times$.
This structure is exactly the same as the one in $\WREL[\Rbar]$.
The \Starauty{} follows easily from the observation that
$(\Limplf X\Sbot)\Isom\Orth X$.

The category $\PCOH$ is cartesian: if $(X_j)_{j\in J}$ is an at most
countable family of PCSs, then the cartesian product of
the $X_j$'s is given by $(\withFam X_i,(\Wproj_i)_{i\in I})$ with
$\Web{\withFam X_i}=\Union_{i\in I}\{i\}\times\Web{X_i}$,
\[ (\Wproj i)_{(k,a),a'}= \begin{cases} 
  1 \text{ if $i=k$ and $a=a$} \\
  0 \text{ otherwise}
\end{cases} \] 
and for any $x\in\Realpto{\Web{\withFam X_i}}$ 
we have $x\in\Pcohp{\withFam {i\in I}X_i}$ if $\Matappa{\Wproj_i}x\in\Pcoh{X_i}$
for each $i \in I$. 
Given $(t_i\in\PCOH(Y,X_i))_{i\in I}$, the unique
morphism $t=\Tuple{t_i}_{i\in I}\in\PCOH(Y,\withFam X_i)$ such
that $\Wproj_j \Compl t=t_i$ is simply defined by
\[ t_{b,(i,a)}=(t_i)_{a,b} \, . \] 
The dual operation $\Bplus_{i\in I}X_i$,
which is a coproduct, is characterized by
$\Web{\Bplus_{i\in I}X_i}=\Union_{i\in I}\{i\}\times\Web{X_i}$ and
$x\in\Pcohp{\Bplus_{i\in I}X_i}$ if $x\in\Pcohp{\withFam{X_i}}$
and 
\[ \sum_{i\in I}\Norm{\Matappa{\Wproj_i}x}_{X_i}\leq 1 \, . \]
For example, the coproduct $\Bplus_{a \in A} 1$ corresponds to $\probapcs{A}$, the 
PCS of subprobability distributions on $A$ given in 
\cref{ex:subprobability-distributions}.

As to the exponentials, one sets $\Web{\Oc X}=\Mfin{\Web X}$ and
$\Pcohp{\Oc X}=\Biorth{\{\Prom x\St x\in\Pcoh X\}}$ where, given
$m\in\Mfin{\Web X}$, $\Prom x_m=x^m=\prod_{a\in\Web X}x_a^{m(a)}$. 
Using the same construction as in $\WREL$, a morphism
$t\in\PCOH(\Oc X,Y)=\Pcohp{\Limplf{\Oc X}{Y}}$ induces a function 
\begin{align*}
  \Fun t:\Pcoh X&\to\Pcoh Y\\
  x&\mapsto\Matappa t{\Prom x}
  \text{ that is, } (\Fun t(x))_b = \sum_{m\in\Web{\Oc
      X}}t_{m,b}x^m\,.
\end{align*}
A function $f : \Pcoh X \arrow \Pcoh Y$ such that $f = \Fun t$ 
for some $t \in \PCOH(\Oc X, Y)$ is called \emph{entire}, see 
\cite{DanosEhrhard11}.
The crucial difference with $\WREL$ is that this construction is now a bijection 
between morphisms and entire functions.

\begin{theorem}[\cite{DanosEhrhard11}] \label{thm:pcoh-characterization}
  The following assertions hold. \begin{enumerate}
  \item Let %
  \(t\in\Realpto{\Web{\Limplf{\Oc{X}}{Y}}}\).
  One has \(t\in\PCOH(\Oc{X},Y)\) %
  iff for all \(x\in\Pcoh{X}\) one has %
  \(\Matappa t{\Prom{x}}\in\Pcoh Y\).
  \item If \(s,t\in\PCOH(\Oc{X},Y)\)
  satisfy %
  \(
  \Matappa s{\Prom{x}}
  =
  \Matappa t{\Prom{x}}
  \) for all %
  \(x\in\Pcoh{X}\) then \(s=t\).
\end{enumerate}
Consequently, the construction $\Fun t$ given above induces a bijection between 
$\PCOH(\Oc X, Y)$ and the entire functions from $\Pcoh X$ to $\Pcoh Y$.
\end{theorem}

\begin{proof}
  The proof of $(2)$ uses crucially the local boundedness
  property of PCSs, which is why the proof does not carry to $\WREL$.
\end{proof}

The next two lemmas follow from \cref{thm:pcoh-characterization} by induction and 
monoidal closedness.
\begin{lemma}%
  \label{lemma:kl-tensor-maps-charact}
  Let %
  \(t\in\Realpto{\Web{\Limplf{\Oc{X_1}\Times\cdots\Times\Oc{X_k}}{Y}}}\).
  One has \(t\in\PCOH(\Oc{X_1}\Times\cdots\Times\Oc{X_k},Y)\) %
  iff for all \((x_i\in\Pcoh{X_i})_{i=1}^k\) one has %
  \(\Matappa t{(\Prom{x_1}\Times\cdots\Times\Prom{x_k})}\in\Pcoh Y\).
\end{lemma}

\begin{lemma}%
  \label{lemma:pcoh-kl-morph-charact}
  If \(s,t\in\PCOH(\Oc{X_1}\Times\cdots\Times\Oc{X_k},Y)\)
  satisfy %
  \(
  \Matappa s{(\Prom{x_1}\Times\cdots\Times\Prom{x_k})}
  =
  \Matappa t{(\Prom{x_1}\Times\cdots\Times\Prom{x_k})}
  \) for all %
  \((x_i\in\Pcoh{X_i})_{i=1}^k\) then \(s=t\).
\end{lemma}

Now given $t\in\PCOH(X,Y)$, the morphism 
$\Oc t\in\PCOH(\Oc X,\Oc Y)$ is defined using the same equation as 
$\WREL$, recall \cref{eq:exponential-v2}
\begin{align*}
  (\Oc t)_{m,p}=\sum_{r\in\Mstrans mp}\Multinomb pr t^r
\end{align*}
where we recall that %
\(t^r=\prod_{(a,b)\in\Web X\times\Web Y}t_{a,b}^{r(a,b)}\).
The main feature of this definition is that for all \(x\in\Pcoh X\) one
has %
\(
\Fun{\Oc t}(x)
=\Matappa{\Oc t}{\Prom x}=\Prom{(\Matappa tx)}
\). %
This property fully characterizes \(\Oc t\), thanks to 
\cref{thm:pcoh-characterization}.
     
The comonad structure is defined as in $\WREL$. That is, %
\(\Deru_X\in\Realpto{\Web{\Limplf{\Oc X}{X}}}\) is given by %
\((\Deru_X)_{m,a}=\Kronecker{m}{\Mset a}\) so that %
\(\forall x\in\Pcoh X\ \Matappa{\Deru_X}{\Prom x}=x\in\Pcoh X\) and
therefore \(\Deru_X\in\Pcoh(\Oc X,X)\).
Similarly, one defines %
\(\Digg_X\in\Realpto{\Web{\Limplf{\Oc X}{\Occ X}}}\) by %
\[ (\Digg_X)_{(m,\Mset{\List m1n})}=\Kronecker m{m_1+\cdots+m_n}\] so
that %
\(\forall x\in\Pcoh X\ \Matappa{\Digg_X}{\Prom x}=\Promm x\) and
hence, again, \(\Digg_X\in\Pcoh(\Oc X,\Occ X)\).
We prove that \((\Oc,\Deru{},\Digg{})\) is indeed a comonad,
using \cref{thm:pcoh-characterization}.
For instance, let \(t\in\Pcoh(X,Y)\), we have %
\[ \Matappa{(\Digg_Y\Compl\Oc t)}{\Prom x} =\Matappa{\Digg
  Y}{(\Prom{\Matappa tx})} =\Matappa{\Digg_X}{\Prom{(\Matappa tx)}}
=\Promm{(\Matappa tx)} \]
\[ \Matappa{(\Occ t\Compl\Digg_X)}{\Prom x} =\Matappa{\Occ
  t}{(\Matappa{\Digg_X}{\Prom x})} =\Matappa{\Occ t}{\Promm x}
=\Prom{(\Matappa{\Oc t}{\Prom x})} =\Promm{(\Matappa tx)} \] which
shows that \(\Digg{}\) is a natural transformation.
As another example, we have %
\[ \Matappa{(\Digg{\Oc X}\Compl\Digg_X)}{\Prom x}
=\Matappa{\Digg{\Oc X}}{\Promm x} =\Prommm x \] %
\[ \Matappa{(\Oc{\Digg_X}\Compl\Digg_X)}{\Prom x}
=\Matappa{\Oc{\Digg_X}}{\Promm x} =\Prom{(\Matappa{\Digg_X}{\Prom
    x})} =\Prom{(\Promm x)} =\Prommm x \] %
and hence %
\(\Digg{\Oc X}\Compl\Digg_X=\Oc{\Digg_X}\Compl\Digg_X\) %
which is one of the required comonad commutations.
The others are proven similarly.

\begin{remark}
It follows from those equations that the construction $\Fun t$ is functorial: 
for all $t \in \Kloc \PCOH(E, F)$ and $s \in \Kloc \PCOH(F, G)$,
$\Fun {t \comp s} = \Fun t \comp \Fun s$,
where $\comp$ is the composition in the coKleisli category $\Kloc \PCOH$.
So $\Kloc \PCOH$ is isomorphic to a category of entire functions.
\end{remark}

The Seely isomorphisms %
\(\Seelyz\in\PCOH(\Sone,\Oc\Top)\) and %
\( \Seelyt_{X_1,X_2}
\in\PCOH(\Tens{\Oc{X_1}}{\Oc{X_2}},\Ocp{{X_1}\With{X_2}}) \) %
are the same as in $\WREL[\Rbar]$ and
are given by %
\(\Seelyz_{\Sonelem,\Msetempty}=1\) and %
\[ \Seelyt_{((m_1,m_2),m)}=\Kronecker{\Msetact 1{m_1}+\Msetact
  2{m_2}}{m}\] %
where, for a multiset \(m=\Mset{\List a1k}\) we set %
\(\Msetact im=\Mset{(i,a_1),\dots,(i,a_k)}\) as in
\cref{sec:notation-multiset}.
It is obvious that \(\Seelyz\) is an iso. To check that
\(\Seelyt_{X_1,x_2}\) is a morphism we use
Lemma~\ref{lemma:kl-tensor-maps-charact}: let \(x_i\in\Pcoh{X_i}\) for
\(i=1,2\), one has %
\[ \Matappa{\Seelyt_{X_1,X_2}}{(\Tens{\Prom{x_1}}{\Prom{x_2}})}
=\Prom{\Tuple{x_1,x_2}} \in\Pcoh{\Ocp{{X_1}\With{X_2}}} \, . \] %
Conversely, defining
\(
s\in\Realpto{\Limplf{\Ocp{{X_1}\With{X_2}}}{\Tensp{\Oc{X_1}}{\Oc{X_2}}}}
\) by %
\(s_{m,(m_1,m_2)}=\Kronecker{\Msetact1{m_1}+\Msetact2{m_2}}{m}\) we
have %
\( \Matappa s{\Prom{\Tuple{x_1,x_2}}} =\Tens{\Prom{x_1}}{\Prom{x_2}}
\in\Pcohp{\Tens{\Oc{X_1}}{\Oc{X_2}}} \) for all
\(x_i\in\Pcoh{X_i}\) (\(i=1,2\)), and hence %
\(s\in\PCOH(\Ocp{{X_1}\With{X_2}},\Tensp{\Oc{X_1}}{\Oc{X_2}})\). %
It is obvious that \(s\) is the inverse of \(\Seelyt_{X_1,X_2}\) which
is therefore an isomorphism in \(\PCOH\).
Proving that it is natural and that it satisfies all the required
commutations for turning \(\PCOH\) into a model of \LL{} is routine
(using crucially Lemma~\ref{lemma:pcoh-kl-morph-charact}).

The induced lax monoidality %
\(
\Ocmong k\in\PCOH(
\Oc{X_1}\Times\cdots\Times\Oc{X_k},\Oc{(X_1\Times\cdots\Times X_k)})
\) is the same as in $\WREL$ and is such that %
\((\Ocmong k)_{(\List m1k),m}=1\)
if %
\(m=\Mset{(a_1^1,\dots,a_k^1),\dots,(a_1^n,\dots,a_k^n)}\) and %
\((m_i=\Mset{a_i^1,\dots,a_i^n})_{i=1}^k\), and %
\((\Ocmong k)_{(\List m1k),m}=0\) otherwise.

\begin{theorem} (\cite{CrubilleEhrhardPaganiTasson17})
  \label{th:pcoh-Lafont}
  The symetric monoidal category \(\PCOH\) is a Lafont category.
\end{theorem}

\subsubsection{Representable analytic structure of \(\PCOH\)}%
\label{sec:pcoh-can-diff-struct}
The category \(\PCOH\) has zero-morphisms (we have the \(0\) matrix
in \(\PCOH(X,Y)\) for any two objects \(X\) and \(Y\)).
The object \(\Dbimon=\withFam<\N> \Sone\) can be described as %
\(\Web\Dbimon=\Nat\) and %
\[ \Pcoh\Dbimon=\{x\in\Realpto\Nat\St \forall i\in\Nat\ x_i\in\Intcc01\} 
= \intcc{0}{1}^{\N} \, .\]
The morphisms \((\Win_i\in\PCOH(\Sone,\Dbimon))_{i\in\Nat}\) are
characterized by %
\(\Matappa{\Win_i}u=u\Base i\) for \(u\in\Pcoh\Sone=\Intcc01\),
where $\Base i \in \Pcoh \Dbimon$ is the element 
such that $\Base {i,i}= 1$ and $\Base {i,j} = 0$ for $j \neq i$.
These  morphisms are jointly epic because, for any
\(t\in\PCOH(\Dbimon,X)\) and \(x\in\Pcoh\Dbimon\) one has %
\(\Matappa tx=\sum_{i\in\Nat}x_i(\Matappa t{\Base i})\).
The morphism $\Wdiag$ is characterized by 
$\Wdiag \cdot u = u(1, 1, \ldots)$.

\begin{theorem}\label{th:pcoh-lafont-elem-sumable}
  \(\PCOH\) is a Lafont representably \(\Sigma\)-additive category, 
  thus it is a representably analytic category. 
  The sum is defined as in $\WREL[\Rbar]$,
  \[ \left(\sum_{i \in I} f(i)\right)_{a, b} =  \sum_{i \in I} f(i)_{a, b}\, .\]
  and $\family<I>[i]{f_i}$ is summable if this sum is in $\Pcoh Y$.
\end{theorem}

\begin{proof}
The sum given above is a $\Sigma$-monoid, this follows from the fact
that the sum in $\WREL[\Rbar]$ is a $\Sigma$-monoid and that 
a PCS is downward closed (if $x \in \Pcoh X$ and $y \leq x$ then $y \in \Pcoh X$).
The zero of this sum is the zero matrix and coincides with the zero 
morphism of $\PCOH$, and is absorbing for the monoidal product, so 
$\PCOH$ satisfies \ref{ax:RS-mon}.
Furthermore, \ref{ax:RS-epi} holds by 
\cref{rk:elementary-closed-win-epicity} and joint epicity of the 
$\Sprojl_i$.

We show that $\PCOH$ satisfies \ref{ax:RS-sum}.
Let $\sequence{f(i)} \in \PCOH(X, Y)^{\N}$ be a $\N$-indexed family. 
Define $h \in \Web{(X \tensor \Dbimon) \linarrow Y} = (\Web{X} \times \N) \times Y$
as $h_{(a, i), b} = f(i)_{a, b}$.
Then \[ h \cdot (u, \sequence{\lambda_i}) = \sum_{i \in \N} \lambda_i (f(i) \cdot u)
\leq \sum_{i \in \N} f(i) \cdot u \]
So if $\sequence{f_i}$ is summable then for all 
$u \in \Pcohp X$ we have $\sum_{i \in \N} f(i) \cdot u \in \Pcohp Y$ and hence 
$h \cdot (u, \sequence{\lambda_i}) \in \Pcohp Y$ by the equation above. Thus, 
$h \in \PCOH(X \tensor \Dbimon, Y)$. 
Then, by definition of $h$, $f(i) = h \compl (X \tensor \Sprojl_i) \tensor \tensorUnitR^{-1}$.
Conversely, if $h \in \PCOH(X \tensor \Dbimon, Y)$ then for all 
$u \in \Pcohp X$, 
\[ \sum_{i \in \N} f(i) \cdot u = h \cdot (u, (1, 1, \ldots)) \in \Pcohp Y \] so 
$\sequence{f(i)}$ is summable.
Thus, $\sequence{f(i)}$ is summable if and only if it has a witness 
(given by $h$ defined above), and 
$h \cdot (u, (1, 1, \ldots)) = (h \compl (X \tensor \Wdiag) \compl \tensorUnitR^{-1}) 
\cdot u$ so $\sum_{i \in \N} f(i) = h \compl (X \tensor \Wdiag) \compl \tensorUnitR^{-1}$
and $\PCOH$ satisfies \ref{ax:RS-sum}.

Finally, we show that $\PCOH$ satisfies \ref{ax:RS-witness}.
Let \(\family<I>[i]{s(i)\in\PCOH(X \tensor \Dbimon,Y)}\) be an $I$-indexed 
family.
Let $f(i,j) = s(i) \compl (X \tensor \Sprojl_j) \compl \tensorUnitR^{-1}_X$.
Observe that $f(i,j)_{a, b} = s(i)_{(a, j), b}$.
Assume that $\family<I \times \N>[(i,j)]{f(i,j)}$ is summable.
We want to prove that 
$\family<I>[i]{s(i)}$ is summable. By
\cref{lemma:pcs-tensor-morphism}, it suffices to prove that
that for all $u \in \Pcoh X$ and 
$\sequence[j]{\lambda_j} \in \Pcoh \Dbimon = {\intcc 0 1}^{\N}$,
\[ \sum_{i \in \N} s(i) \cdot (u \tensor \sequence[j]{\lambda_j}) \in \Pcoh Y \, . \]
We have
\[ \sum_{i \in I} s(i) \cdot (u \tensor \sequence[j]{\lambda_j}) 
= \sum_{\Biind{i \in I, a \in \Web X}{j \in \N, b \in \Web Y}} \lambda_j u_a s(i)_{(a,j), b} e_b 
= \sum_{i \in I, j \in \N} \lambda_j f(i, j) \cdot u
\leq \sum_{i \in I, j \in \N} f(i, j) \cdot u \in \Pcoh Y \]
so it follows by \cref{lemma:pcs-downward-closed} that 
$\sum_{i \in I} s(i) \cdot (u \tensor \sequence{\lambda_i}) \in \Pcoh Y$, 
hence $\family<I>[i]{s(i)}$ is summable.
\end{proof}

Observe that \(\Web{\Sfun X}=\Nat\times\Web X\).
The action of $\S$ on morphism is the same as in $\WREL[\Rbar]$,
$\S f$ is characterized by the equation $(\S f) \cdot (i, u) = (i, f \cdot u)$.
That is, $(\S f)_{(i, a), (j, b)} = \Kronecker i j f_{a, b}$.
The difference with respect to $\WREL[\Rbar]$ is that the object $\S X$ is not the cartesian product 
$\withFam<\N> X$. Indeed,
\[ \Pcohp{\Sfun X}\Isom\{\Vect x\in\Pcoh X^\Nat\St \sum_{i=0}^\infty
x(i)\in\Pcoh X\} \neq \Pcohp{X}^{\N} \simeq \Pcohp{\withFam<\N> X}\]

The bimonoid structure of \(\Dbimon\) is identical to that of
\(\Dbimon\) in \(\WREL[\Rbar]\), as well as the analytic coalgebra
\(\Dbimonca\in\PCOH(\Dbimon,\Oc\Dbimon)\), which is given by %
\(\Dbimonca_{n,\Mset{\List i1k}}=\Kronecker n{i_1+\cdots+i_k}\).
This coalgebra structure induces the same distributive law 
$\Sdl \in \PCOH(\Oc \S X, \S \Oc X)$
as in $\WREL[\Rbar]$, which is given by %
\begin{align*}
  \Sdl_{p,(n,m)}=
  \begin{cases}
    \frac{\Factor m}{\Factor p}
    &\text{if }p=\Mset{(i_1,a_1),\dots,(i_k,a_k)}\text{ with }
    i_1+\cdots+i_k=n\\
    0 & \text{otherwise.}
  \end{cases}
\end{align*}

The induced functor \(\Tayfun:\Kloc\PCOH\to\Kloc\PCOH\), which
maps a PCS \(X\) to \(\Sfun X\) has the same action on morphism as the one in 
$\Kloc {\WREL[\Rbar]}$: for any \(s\in\Kloc\PCOH(X,Y)\), 
\(\Tayfun(s)\in\Kloc\PCOH(\Sfun X,\Sfun Y)\) is given by
\begin{align*}
  \Tayfun(s)_{\Mset{(i_1,a_1),\dots,(i_k,a_k)},(n,b)}=
  \Kronecker{n}{i_1+\cdots+i_k}
  \frac{\Factor{\Mset{\List a1k}}}
  {\Factor{\Mset{(i_1,a_1),\dots,(i_k,a_k)}}}
    s_{\Mset{\List a1k},b}\,.
\end{align*}

For any $s \in \PCOH(\Oc X, Y)$, we introduce the
$n$-th derivative of the formal power series defined by $s$ as in $\WREL[\Rbar]$,
\[ \Deriv{s}{n}(x)(u^1, \ldots, u^n)
=\sum_{m\in\Mfin{\Web E}}
   \sum_{\Vect a\in\Web E^d}
   \frac{\Factor{(m+\Mset{\Vect a})}}{\Factor m}
   s_{m+\Mset{\Vect a},b}x^m u^1_{a_1}\cdots u^d_{a_d} \, .\]
Observe that $\Deriv{s}{n}$ is an entire function 
from $(\Pcoh X)^{n+1}$ to $\Rbar^{\Web Y}$, but is \emph{not necessarily} a
function from $(\Pcoh X)^{n+1}$ to $\Pcoh Y$.
Still, \cref{thm:Taylor-expansion-wrel} ensures that 
for any $n \in \N$,
\[ \Fun{\Sproj_n \compl \Tayfun(s)}(\Vect x) =  \sum_{m \in \mpart{n}}
\frac{1}{m!} \Deriv{s}{\Mscard m}(x(0))
(\Rep{x(i_1)}{m(1)},\dots,\Rep{x(i_k)}{m(k)}) \]
where \(\Rep xl\) is the list of repeated arguments %
\(\overbrace{x,\dots,x}^{n\times}\) and 
$\mpart{n} \defEq \{ m \in \mfin(\N^*) | \sum_{i \in \N^*} m(i) \ i = n\}$.
By construction, $\Fun{\D s}$ is an entire function from $\Pcoh \S X$ to 
$\Pcoh \S Y$.
This implies the following remarkable observation: the sum of all of 
the terms writen above is an element 
of $\Pcoh Y$.

As we saw, we have in particular that
\begin{align*}
  \Fun{\Tayfun(s)}(x,u,0,0,\dots)
  =\big(\frac 1{\Factor n}\Deriv{s}n(x)(\Rep un)\big)_{n\in\Nat}
\end{align*}
so although $\Deriv{s}n (x)(\Rep u n)$ is not necessarily in 
$\Pcoh Y$, we have $\frac{1}{n!} \Deriv{s}n (x)(\Rep u n) \in \Pcoh Y$, and 
we even have that the sum of these terms is in $\Pcoh Y$.
This mean that the regular Taylor expansion is perfectly compatible with the boundedness constraints 
of probabilistic coherence spaces, even when the $n$-th derivatives themselves are not.

\section{Conclusion}
We have developed a theory of Taylor expansion in categories which are
not necessarily additive.
The main motivations for this work are first that Taylor expansion has
been shown to be a useful tool in the analysis of functional programs,
see for instance~\cite{BarbarossaManzonetto20}, and second that most
concrete denotational models of such languages (such as coherence
spaces, probabilistic coherence spaces etc.)
feature only a \emph{partial} addition of morphisms for the very good
reason that full additivity is incompatible with 
the determinism of computations.
For instance, the values \(\True\) and \(\False\) of the object of
booleans should not be summable in a deterministic model.
In the very same line of idea, the uniformity of the Taylor expansion
observed in~\cite{Ehrhard08} seems to be closely related to the
summability constraints observed in coherence spaces and non-uniform
coherence spaces, and accounts syntactically for the fundamental
determinism of the \(\lambda\)-calculus (Church-Rosser and
Standardization theorems) and of the execution of terms in the Krivine
machine.

It turns out that all the categorical axiomatizations of denotational
models which account for the Taylor expansion of morphisms, and are
most often based on differential LL, make the assumption that homsets
are monoids where infinite summations are possible for the obvious
reason that infinite summations are an essential ingredient of Taylor
expansion.
We have shown in this paper that this strong form of additivity is not
a fatality: Taylor expansion can also exist in settings where only a
partial version of (finite and infinite) addition is available.

Differentiation was already accommodated in such partially additive
categories in~\cite{Ehrhard23-cohdiff} and the approach developed
here follows a similar pattern.
One main difference is that we had to develop a more subtle notion of
infinitary (countable) summability.
Beyond this main difference, the resulting theory of coherent Taylor
expansion is strikingly similar to that of coherent differentiation
---~and not essentially more complicated~---, with one additional nice
feature: the resulting Taylor functor is not only a monad (just as the
tangent functor in the tangent categories of~\cite{Rosicky84}) but
also a comonad.
This comonad structure, and more precisely the naturality of its
counit, reflects the fact that nonlinear morphisms coincide with their
Taylor expansion, expressing abstractly that they are analytic.

We have developed this theory in a LL setting of resource categories,
where the analytic structure arises as a distributive law wrt.~the
resource comonad, and also in general cartesian categories, following
the main ideas of~\cite{Walch23}.

This first denotational investigation of coherent Taylor expansion is a
strong incentive for developing now a syntactic analysis of this
operation, which might be similar to the coherent differential PCF
of~\cite{Ehrhard22-pcf}, this will be the object of further work.
Another natural question is whether this coherent Taylor expansion has
an associated resource calculus, just like Taylor expansion in the
setting of differential LL, see~\cite{Ehrhard08}.

The connection between coherent differentiation and coherent Taylor
expansion also deserves further study: such a study might be based on
the observation in \Cref{sec:elementary-Taylor-examples} that there
are simple models of LL which accommodate coherent differentiation but
not coherent Taylor expansion.

\section*{Acknowledgment}

This work was partly funded by the project %
\texttt{ANR-19-CE48-0014} %
\emph{Probabilistic Programming Semantics (PPS)} %
\url{https://www.irif.fr/anrpps}.

\bibliographystyle{alpha}
\bibliography{biblio.bib}

\begin{thebibliography}{LMMP13}

\bibitem[AHF18]{Aguiar18}
Marcelo Aguiar, Mariana Haim, and Ignacio Franco.
\newblock Monads on higher monoidal categories.
\newblock {\em Applied Categorical Structures}, 26, 06 2018.

\bibitem[AM80]{Arbib80}
Michael~A Arbib and Ernest~G Manes.
\newblock Partially additive categories and flow-diagram semantics.
\newblock {\em Journal of Algebra}, 62(1):203--227, 1980.

\bibitem[BCS06]{Blute06}
R.~Blute, Robin Cockett, and R.~Seely.
\newblock Differential categories.
\newblock {\em Mathematical Structures in Computer Science}, 16:1049 -- 1083,
  12 2006.

\bibitem[BCS09]{Blute09}
R.~Blute, Robin Cockett, and R.~Seely.
\newblock Cartesian differential categories.
\newblock {\em Theory and Applications of Categories}, 22:622--672, 01 2009.

\bibitem[BCS14]{Blute14}
Richard Blute, Robin Cockett, and Robert Seely.
\newblock Cartesian differential storage categories.
\newblock {\em arXiv: Category Theory}, 30, 05 2014.

\bibitem[BE01]{BucciarelliEhrhard01}
Antonio Bucciarelli and Thomas Ehrhard.
\newblock On phase semantics and denotational semantics: the exponentials.
\newblock {\em Ann. Pure Appl. Log.}, 109(3):205--241, 2001.

\bibitem[Bec69]{Beck69}
Jon Beck.
\newblock Distributive laws.
\newblock In B.~Eckmann, editor, {\em Seminar on Triples and Categorical
  Homology Theory}, pages 119--140, Berlin, Heidelberg, 1969. Springer Berlin
  Heidelberg.

\bibitem[Ber78]{Berry78}
G{\'e}rard Berry.
\newblock Stable models of typed $\lambda$-calculi.
\newblock In Giorgio Ausiello and Corrado B{\"o}hm, editors, {\em Automata,
  Languages and Programming}, pages 72--89, Berlin, Heidelberg, 1978. Springer
  Berlin Heidelberg.

\bibitem[Bie95]{Bierman95}
G.~M. Bierman.
\newblock What is a categorical model of intuitionistic linear logic?
\newblock In Mariangiola Dezani-Ciancaglini and Gordon Plotkin, editors, {\em
  Typed Lambda Calculi and Applications}, pages 78--93, Berlin, Heidelberg,
  1995. Springer Berlin Heidelberg.

\bibitem[BLV11]{Bruguieres11}
Alain Bruguières, Steve Lack, and Alexis Virelizier.
\newblock Hopf monads on monoidal categories.
\newblock {\em Advances in Mathematics}, 227(2):745--800, 2011.

\bibitem[BM20]{BarbarossaManzonetto20}
Davide Barbarossa and Giulio Manzonetto.
\newblock {Taylor subsumes Scott, Berry, Kahn and Plotkin}.
\newblock {\em Proc. {ACM} Program. Lang.}, 4({POPL}):1:1--1:23, 2020.

\bibitem[Bou11]{Boudes11}
Pierre Boudes.
\newblock Non-uniform (hyper/multi)coherence spaces.
\newblock {\em Math. Struct. Comput. Sci.}, 21(1):1--40, 2011.

\bibitem[CC14]{Cockett14}
Robin Cockett and G.~Cruttwell.
\newblock {Differential Structure, Tangent Structure, and SDG}.
\newblock {\em Applied Categorical Structures}, 22, 04 2014.

\bibitem[CEPT17]{CrubilleEhrhardPaganiTasson17}
Rapha{\"{e}}lle Crubill{\'{e}}, Thomas Ehrhard, Michele Pagani, and Christine
  Tasson.
\newblock {The Free Exponential Modality of Probabilistic Coherence Spaces}.
\newblock In Javier Esparza and Andrzej~S. Murawski, editors, {\em Foundations
  of Software Science and Computation Structures - 20th International
  Conference, {FOSSACS} 2017, Held as Part of the European Joint Conferences on
  Theory and Practice of Software, {ETAPS} 2017, Uppsala, Sweden, April 22-29,
  2017, Proceedings}, volume 10203 of {\em Lecture Notes in Computer Science},
  pages 20--35, 2017.

\bibitem[CLLW20]{Cockett17}
Robin Cockett, Jean-Simon~Pacaud Lemay, and Rory B.~B. Lucyshyn-Wright.
\newblock {Tangent Categories from the Coalgebras of Differential Categories}.
\newblock In Maribel Fern\'{a}ndez and Anca Muscholl, editors, {\em 28th EACSL
  Annual Conference on Computer Science Logic (CSL 2020)}, volume 152 of {\em
  Leibniz International Proceedings in Informatics (LIPIcs)}, pages
  17:1--17:17, Dagstuhl, Germany, 2020. Schloss Dagstuhl -- Leibniz-Zentrum
  f{\"u}r Informatik.

\bibitem[CS97]{Cockett97}
Robin Cockett and R.~Seely.
\newblock Proof theory for full intuitionistic linear logic, bilinear logic,
  and mix categories.
\newblock {\em Theory and Applications of Categories}, 3, 11 1997.

\bibitem[Day70]{Day70}
Brian Day.
\newblock On closed categories of functors.
\newblock In S.~MacLane, H.~Applegate, M.~Barr, B.~Day, E.~Dubuc, Phreilambud,
  A.~Pultr, R.~Street, M.~Tierney, and S.~Swierczkowski, editors, {\em Reports
  of the Midwest Category Seminar IV}, pages 1--38, Berlin, Heidelberg, 1970.
  Springer Berlin Heidelberg.

\bibitem[DE11a]{DanosEhrhard11}
Vincent Danos and Thomas Ehrhard.
\newblock Probabilistic coherence spaces as a model of higher-order
  probabilistic computation.
\newblock {\em Inf. Comput.}, 209(6):966--991, 2011.

\bibitem[DE11b]{Danos11}
Vincent Danos and Thomas Ehrhard.
\newblock Probabilistic coherence spaces as a model of higher-order
  probabilistic computation.
\newblock {\em Information and Computation}, 209(6):966--991, 2011.

\bibitem[Die69]{Dieudonne69}
J.~Dieudonné.
\newblock {\em {Foundations of Modern Analysis}}.
\newblock Boston, MA: Academic Press, 1969.

\bibitem[EG22]{Ehrhard22-cones}
Thomas Ehrhard and Guillaume Geoffroy.
\newblock {Integration in Cones}.
\newblock Technical report, {IRIF (UMR\_8243) - Institut de Recherche en
  Informatique Fondamentale}, December 2022.

\bibitem[Ehr02]{Ehrhard02}
Thomas Ehrhard.
\newblock {On Köthe Sequence Spaces and Linear Logic}.
\newblock {\em Mathematical Structures in Computer Science}, 12, 01 2002.

\bibitem[Ehr05]{Ehrhard05}
Thomas Ehrhard.
\newblock {Finiteness spaces}.
\newblock {\em {Mathematical Structures in Computer Science}}, 15(4):615--646,
  July 2005.
\newblock 32 pages.

\bibitem[Ehr18]{Ehrhard18}
Thomas Ehrhard.
\newblock An introduction to differential linear logic: proof-nets, models and
  antiderivatives.
\newblock {\em Mathematical Structures in Computer Science}, 28(7):995–1060,
  2018.

\bibitem[Ehr22]{Ehrhard22a}
Thomas Ehrhard.
\newblock Differentials and distances in probabilistic coherence spaces.
\newblock {\em Logical Methods in Computer Science}, 18(3), 2022.

\bibitem[Ehr23a]{Ehrhard22-pcf}
Thomas Ehrhard.
\newblock A coherent differential {PCF}.
\newblock {\em Logical Methods in Computer Science}, Volume 19, Issue 4, 10
  2023.

\bibitem[Ehr23b]{Ehrhard23-cohdiff}
Thomas Ehrhard.
\newblock Coherent differentiation.
\newblock {\em Mathematical Structures in Computer Science}, page 1–52, 2023.

\bibitem[ER03]{EhrhardRegnier02}
Thomas Ehrhard and Laurent Regnier.
\newblock The differential lambda-calculus.
\newblock {\em {Theoretical Computer Science}}, 309(1-3):1--41, 2003.

\bibitem[ER08]{Ehrhard08}
Thomas Ehrhard and Laurent Regnier.
\newblock {Uniformity and the Taylor expansion of ordinary lambda-terms}.
\newblock {\em Theoretical Computer Science}, 403(2):347--372, 2008.

\bibitem[ETP14]{Ehrhard14}
Thomas Ehrhard, Christine Tasson, and Michele Pagani.
\newblock Probabilistic coherence spaces are fully abstract for probabilistic
  pcf.
\newblock {\em SIGPLAN Not.}, 49(1):309–320, jan 2014.

\bibitem[EW23]{Walch23}
Thomas Ehrhard and Aymeric Walch.
\newblock {Cartesian Coherent Differential Categories}.
\newblock In {\em 2023 38th Annual ACM/IEEE Symposium on Logic in Computer
  Science (LICS)}, pages 1--13, Los Alamitos, CA, USA, jun 2023. IEEE Computer
  Society.

\bibitem[Fra78]{Fraenkel78}
L.~E. Fraenkel.
\newblock Formulae for high derivatives of composite functions.
\newblock {\em Mathematical Proceedings of the Cambridge Philosophical
  Society}, 83(2):159–165, 1978.

\bibitem[Gir87]{Girard87}
Jean{-}Yves Girard.
\newblock {Linear Logic}.
\newblock {\em {Theoretical Computer Science}}, 50:1--102, 1987.

\bibitem[GL21]{Garner21}
Richard Garner and {Jean-Simon Pacaud} Lemay.
\newblock Cartesian differential categories as skew enriched categories.
\newblock {\em Applied Categorical Structures}, 29(6):1099--1150, December
  2021.
\newblock Copyright the Author(s) 2021. Version archived for private and
  non-commercial use with the permission of the author/s and according to
  publisher conditions. For further rights please contact the publisher.

\bibitem[Gui80]{Guitart80}
René Guitart.
\newblock Tenseurs et machines.
\newblock {\em Cahiers de Topologie et Géométrie Différentielle
  Catégoriques}, 21(1):5--62, 1980.

\bibitem[Hag00]{Haghverdi00}
Esfandiar Haghverdi.
\newblock {Unique decomposition categories, Geometry of Interaction and
  combinatory logic}.
\newblock {\em Mathematical Structures in Computer Science}, 10(2):205–230,
  2000.

\bibitem[Hin13]{Hines13}
Peter Hines.
\newblock A categorical analogue of the monoid semiring construction.
\newblock {\em Mathematical Structures in Computer Science}, 23(1):55–94,
  2013.

\bibitem[Joh75]{Johnstone75}
P.~T. Johnstone.
\newblock {Adjoint Lifting Theorems for Categories of Algebras}.
\newblock {\em Bulletin of the London Mathematical Society}, 7(3):294--297, 11
  1975.

\bibitem[Kei75]{Keigher75}
William Keigher.
\newblock Adjunctions and comonads in differential algebra.
\newblock {\em Pacific Journal of Mathematics}, 59, 07 1975.

\bibitem[Kel74]{Kelly74-doctrinal}
G.~M. Kelly.
\newblock Doctrinal adjunction.
\newblock In Gregory~M. Kelly, editor, {\em Category Seminar}, pages 257--280,
  Berlin, Heidelberg, 1974. Springer Berlin Heidelberg.

\bibitem[KL23]{Kerjean23}
Marie Kerjean and Jean-Simon~Pacaud Lemay.
\newblock {Taylor Expansion as a Monad in Models of DiLL}.
\newblock In {\em 2023 38th Annual ACM/IEEE Symposium on Logic in Computer
  Science (LICS)}, pages 1--13, 2023.

\bibitem[Koc70]{Kock70}
Anders Kock.
\newblock Monads on symmetric monoidal closed categories.
\newblock {\em Archiv der Mathematik}, 21:1--10, 01 1970.

\bibitem[Koc71]{Kock71}
Anders Kock.
\newblock Closed categories generated by commutative monads.
\newblock {\em Journal of the Australian Mathematical Society},
  12(4):405–424, 1971.

\bibitem[Koc72]{Kock72}
Anders Kock.
\newblock Strong functors and monoidal monads.
\newblock {\em Archiv der Mathematik}, 23:113--120, 12 1972.

\bibitem[KS74]{Kelly74}
G.~M. Kelly and Ross Street.
\newblock Review of the elements of 2-categories.
\newblock In Gregory~M. Kelly, editor, {\em Category Seminar}, pages 75--103,
  Berlin, Heidelberg, 1974. Springer Berlin Heidelberg.

\bibitem[Laf88]{Lafont88}
Yves Lafont.
\newblock {\em {An Analysis of Example}}.
\newblock Phd thesis, Université Paris VII, January 1988.

\bibitem[Lam92]{Lamarche92}
François Lamarche.
\newblock Quantitative domains and infinitary algebras.
\newblock {\em Theoretical Computer Science}, 94(1):37--62, 1992.

\bibitem[Lam95]{Lamarche95}
Fran{\c{c}}ois Lamarche.
\newblock Generalizing coherent domains and hypercoherences.
\newblock In Stephen~D. Brookes, Michael~G. Main, Austin Melton, and Michael~W.
  Mislove, editors, {\em Eleventh Annual Conference on Mathematical Foundations
  of Programming Semantics, {MFPS} 1995, Tulane University, New Orleans, LA,
  USA, March 29 - April 1, 1995}, volume~1 of {\em Electronic Notes in
  Theoretical Computer Science}, pages 355--369. Elsevier, 1995.

\bibitem[Lem18]{Lemay18}
Jean-Simon~Pacaud Lemay.
\newblock {A tangent category alternative to the Faa di Bruno construction}.
\newblock {\em Theory and Applications of Categories}, 33(35):1072--1110, 2018.

\bibitem[LMMP13]{Laird13}
Jim Laird, Giulio Manzonetto, Guy McCusker, and Michele Pagani.
\newblock {Weighted Relational Models of Typed Lambda-Calculi}.
\newblock In {\em 2013 28th Annual ACM/IEEE Symposium on Logic in Computer
  Science}, pages 301--310, 2013.

\bibitem[MA86]{Manes86}
Ernest~G. Manes and Michael~A. Arbib, editors.
\newblock {\em {Algebraic Approaches to Program Semantics}}.
\newblock Springer-Verlag, Berlin, Heidelberg, 1986.

\bibitem[Mac71]{Maclane71}
Saunders MacLane.
\newblock {\em Categories for the Working Mathematician}.
\newblock Springer-Verlag, New York, 1971.
\newblock Graduate Texts in Mathematics, Vol. 5.

\bibitem[Man12]{Manzonetto12}
Giulio Manzonetto.
\newblock What is a categorical model of the differential and the resource
  $\lambda$-calculi?
\newblock {\em Mathematical Structures in Computer Science}, 22(3):451–520,
  2012.

\bibitem[Mar65]{Maranda65}
J.-M. Maranda.
\newblock Formal categories.
\newblock {\em Canadian Journal of Mathematics}, 17:758–801, 1965.

\bibitem[Mel09]{Mellies09}
Paul-Andr\'e Melli\`es.
\newblock {Categorical semantics of linear logic}.
\newblock {\em Panoramas et Synth\`eses}, 27:1--196, 2009.

\bibitem[Moe02]{Moerdijk02}
I.~Moerdijk.
\newblock Monads on tensor categories.
\newblock {\em Journal of Pure and Applied Algebra}, 168(2):189--208, 2002.
\newblock Category Theory 1999: selected papers, conference held in Coimbra in
  honour of the 90th birthday of Saunders Mac Lane.

\bibitem[MTT09]{MelliesTabreauTasson09}
Paul-Andr{\'e} Melli{\`e}s, Nicolas Tabareau, and Christine Tasson.
\newblock {An Explicit Formula for the Free Exponential Modality of Linear
  Logic}.
\newblock In Susanne Albers, Alberto Marchetti-Spaccamela, Yossi Matias,
  Sotiris Nikoletseas, and Wolfgang Thomas, editors, {\em Automata, Languages
  and Programming}, pages 247--260, Berlin, Heidelberg, 2009. Springer Berlin
  Heidelberg.

\bibitem[MW11]{Mesablishvili11}
Bachuki Mesablishvili and Robert Wisbauer.
\newblock {Bimonads and Hopf monads on categories}.
\newblock {\em Journal of K-theory K-theory and its Applications to Algebra
  Geometry and Topology}, 7:349--388, 04 2011.

\bibitem[Ong17]{Ong17}
C.-H.~Luke Ong.
\newblock {Quantitative semantics of the lambda calculus: Some generalisations
  of the relational model}.
\newblock In {\em 2017 32nd Annual ACM/IEEE Symposium on Logic in Computer
  Science (LICS)}, pages 1--12, 2017.

\bibitem[Plo77]{Plotkin77}
G.D. Plotkin.
\newblock {LCF} considered as a programming language.
\newblock {\em Theoretical Computer Science}, 5(3):223--255, 1977.

\bibitem[PW02]{Power02}
John Power and Hiroshi Watanabe.
\newblock Combining a monad and a comonad.
\newblock {\em Theoretical Computer Science}, 280(1):137--162, 2002.
\newblock Coalgebraic Methods in Computer Science.

\bibitem[Ros84]{Rosicky84}
J.~Rosicky.
\newblock Abstract tangent functors.
\newblock {\em Diagrammes}, 12:JR1--JR11, 1984.
\newblock talk:3.

\bibitem[Sco76]{Scott76}
Dana Scott.
\newblock {Data Types as Lattices}.
\newblock {\em SIAM Journal on Computing}, 5(3):522--587, 1976.

\bibitem[Str72]{Street72}
Ross Street.
\newblock The formal theory of monads.
\newblock {\em Journal of Pure and Applied Algebra}, 2(2):149 -- 168, 1972.

\bibitem[TA22]{Tsukada22}
Takeshi Tsukada and Kazuyuki Asada.
\newblock {Linear-Algebraic Models of Linear Logic as Categories of Modules
  over $\Sigma$-Semirings$*$}.
\newblock In {\em Proceedings of the 37th Annual ACM/IEEE Symposium on Logic in
  Computer Science}, LICS '22, New York, NY, USA, 2022. Association for
  Computing Machinery.

\bibitem[VO71]{VanOsdol71}
Donovan~H. Van~Osdol.
\newblock {\em Sheaves in regular categories}, pages 223--239.
\newblock Springer Berlin Heidelberg, Berlin, Heidelberg, 1971.

\bibitem[Wal22]{Walch22-internship}
Aymeric Walch.
\newblock {\em Coherent differentiation in models of Linear Logic}.
\newblock Internship report, ENS de Lyon, June 2022.

\bibitem[Whi65]{Whittlesey65}
E.~F. Whittlesey.
\newblock {Analytic Functions in Banach Spaces}.
\newblock {\em Proceedings of the American Mathematical Society},
  16(5):1077--1083, 1965.

\bibitem[Wis08]{Wisbauer08}
Robert Wisbauer.
\newblock {Algebras Versus Coalgebras}.
\newblock {\em Applied Categorical Structures}, 16:255--295, 04 2008.

\end{thebibliography}

\end{document}